\documentclass[11pt]{article}

\usepackage[utf8]{inputenc}
\usepackage{amsmath,amsthm,amsfonts,amssymb,mathtools,bbm,mathdots,tensor}
\usepackage[sans]{dsfont}
\usepackage[mathscr]{eucal}
\usepackage{authblk}
\usepackage{enumitem}
\usepackage{slashed,extarrows,bigints,cancel}
\usepackage[new]{old-arrows}
\usepackage{graphicx}
\usepackage{float,caption,multirow,makecell}
\usepackage{bm}
\usepackage[dvipsnames]{xcolor}
\usepackage[]{hyperref}
\hypersetup{colorlinks = true, linktoc=page, urlcolor=MidnightBlue, citecolor=red, linkcolor=ao!90!black}
\definecolor{ao}{rgb}{0.0, 0.5, 0.0}
\usepackage[hyperpageref]{backref}

\usepackage{tikz}
\usepackage{pgfplots}
\usetikzlibrary{decorations.markings, decorations.pathmorphing, arrows,calc,knots,hobby, decorations.pathreplacing, shapes.geometric, calc, arrows, arrows.meta,patterns,intersections, pgfplots.fillbetween,arrows.meta}

\tikzset{->-/.style={decoration={
			markings,
			mark=at position .58 with {\arrow{>[scale=2]}}},postaction={decorate}}}
\usepackage{tikz-cd}

\usepackage[top=20mm]{geometry}
\usepackage{ascmac}
\usepackage[font={footnotesize},labelfont={bf}]{caption}

\allowdisplaybreaks
\pgfplotsset{compat=1.18}
 
\newtheorem{thr}{Theorem}[section]

\newtheorem{lem}[]{Lemma}[section]
\newtheorem{cor}[]{Corollary}[section]
\newtheorem{prop}{Proposition}[section]
\newtheorem*{prop*}{Proposition}
\newtheorem{rmk}[]{Remark}[section]

\newtheorem*{theorem*}{Theorem}

\newenvironment{eqaligned}
{%
\begin{equation}
    \begin{aligned}
    } 
{%
\end{aligned}
\end{equation}
\ignorespacesafterend}

\newenvironment{eqaligned*}
{%
\begin{equation*}
    \begin{aligned}
    } 
{%
\end{aligned}
\end{equation*}
\ignorespacesafterend}

\newenvironment{eqgathered}
{%
\begin{equation}
    \begin{gathered}
    } 
{%
\end{gathered}
\end{equation}
\ignorespacesafterend}

\newenvironment{eqgathered*}
{%
\begin{equation*}
    \begin{gathered}
    } 
{%
\end{gathered}
\end{equation*}
\ignorespacesafterend}

\newcommand{\bmb}[1]{\bm\bar{#1}}
\newcommand{\mds}[1]{\mathds{#1}}
\newcommand{\mfk}[1]{\mathfrak{#1}}
\newcommand{\mbs}[1]{\boldsymbol{#1}}
\newcommand{\mbb}[1]{\mathbb{#1}}

\newcommand{\mbf}[1]{\mathbf{#1}}
\newcommand{\msf}[1]{\mathsf{#1}}
\newcommand{\mscr}[1]{\mathscr{#1}}
\newcommand{\mcal}[1]{\mathcal{#1}}

\newcommand{\rd}{\mathsf{d}}
\newcommand{\wt}[1]{\widetilde{#1}}
\newcommand{\ct}{\mbb{C}^\times}

\newcommand*\bbar[1]{%
  \hbox{%
    \vbox{%
      \hrule height 0.5pt %
      \kern0.2ex%
      \hbox{%
        \kern-0.1em%
        \ensuremath{#1}%
        \kern-0.1em%
      }%
    }%
  }%
}

\newcommand{\SU}[1]{\text{SU}(#1)}
\newcommand{\U}[1]{\text{U}(#1)}
\newcommand{\SO}{\text{SO}}

\newcommand{\ts}{\mcal{Z}_{\mbb{H}^3}}
\newcommand{\tsu}{\mcal{Z}_{\mcal{U}}}
\newcommand{\tsrf}{\mcal{Z}_{\R^4}}

\newcommand{\pc}{\Sigma_{N,n}}
\newcommand{\cpc}{\wt{\Sigma}_{N,n}}
\newcommand{\ipc}[1]{\Sigma_{N,n}^{(#1)}}
\newcommand{\icpc}[1]{\wt{\Sigma}_{N,n}^{(#1)}}
\newcommand{\iipc}[1]{{\Sigma'}_{N,n}^{(#1)}} %
\newcommand{\inpc}[2]{\Sigma_{N,n}^{(#1,#2)}}
\newcommand{\ismc}[1]{S_{#1}^{\text{sym}}}
\newcommand{\tcpn}{\wt{\Sigma}_N}
\newcommand{\cpn}{\Sigma_N}

\newcommand{\sth}{\textsuperscript{th}\,}
\newcommand{\R}{\mbb{R}}
\newcommand{\C}{\mbb{C}}
\newcommand{\Pbb}{\mbb{P}}
\newcommand{\Hbb}{\mbb{H}}
\newcommand{\ppm}{\Pbb^1_-\times\Pbb^1_+}
\newcommand{\wh}[1]{\widehat{#1}}
\newcommand{\CS}{\text{CS}}
\newcommand{\sbcm}{\Pbb S(\C^{1,3})}
\newcommand{\sbrm}{\Pbb S(\R^{1,3})}
\newcommand{\sbe}{\Pbb S(\R^4)}
\newcommand{\sbm}{\Pbb S(M)}

\newcommand{\rnk}[1]{\text{rnk}(\mfk{#1})}

\numberwithin{equation}{section}

\newcommand{\LatticeDualLatticegCPM}{

\begin{tikzpicture}[]
    \draw[pattern=crosshatch dots gray, opacity=.3] (-3.49,1.98) rectangle (3.99,-2.49);
    
    \draw[dashed,-Stealth, line width=1pt] (-5,2)node[left]{$q$} -- (+5.4,2);
    \draw[dashed,Stealth-, line width=1pt] (-5.3,.5) -- (+5,.5)node[right]{$q'$};
    \draw[dashed,-Stealth, line width=1pt] (-5,-1)node[left]{$q$} -- (+5.4,-1);
    \draw[dashed,Stealth-, line width=1pt] (-5.3,-2.5) -- (+5,-2.5)node[right]{$q'$};

    \draw[dashed,Stealth-, line width=1pt] (-3.5,3.2) -- (-3.5,-3)node[below]{$p$};
    \draw[dashed,-Stealth, line width=1pt] (-2,2.8) node[above]{$p'$} -- (-2,-3.4);
    \draw[dashed,Stealth-, line width=1pt] (-.5,3.2) -- (-.5,-3)node[below]{$p$};
    \draw[dashed,-Stealth, line width=1pt] (+1,2.8)node[above]{$p'$} -- (+1,-3.4);
    \draw[dashed,Stealth-, line width=1pt] (+2.5,3.2) -- (+2.5,-3)node[below]{$p$};
    \draw[dashed,-Stealth, line width=1pt] (+4,2.8)node[above]{$p'$} -- (+4,-3.4);

    \draw[fill=black] (-4.3,1.25) circle (3pt);
    \draw[fill=black] (-2.75,-.25) circle (3pt);
    \draw[line width=1.5pt,postaction={decoration={markings,mark=at position .58 with {\arrow{Stealth}}},decorate}] (-4.3,1.25) to (-2.75,-.25);
    
    \draw[fill=black] (-1.2,-1.75) circle (3pt);
    \draw[line width=1.5pt,postaction={decoration={markings,mark=at position .55 with {\arrow{Stealth}}},decorate}] (-2.75,-.25) to (-1.2,-1.75);

    \draw[fill=black] (+.28,-3.25) circle (3pt);
    \draw[line width=1.5pt,postaction={decoration={markings,mark=at position .55 with {\arrow{Stealth}}},decorate}] (-1.2,-1.75) to (+.28,-3.25);

    \begin{scope}[xshift=1.5cm,yshift=1.5cm]
    \draw[fill=black] (-4.3,1.25) circle (3pt);
    \draw[fill=black] (-2.75,-.25) circle (3pt);
    \draw[line width=1.5pt,postaction={decoration={markings,mark=at position .58 with {\arrow{Stealth}}},decorate}] (-4.3,1.25) to (-2.75,-.25);
    
    \draw[fill=black] (-1.2,-1.75) circle (3pt);
    \draw[line width=1.5pt,postaction={decoration={markings,mark=at position .55 with {\arrow{Stealth}}},decorate}] (-2.75,-.25) to (-1.2,-1.75);

    \draw[fill=black] (+.28,-3.25) circle (3pt);
    \draw[line width=1.5pt,postaction={decoration={markings,mark=at position .55 with {\arrow{Stealth}}},decorate}] (-1.2,-1.75) to (+.28,-3.25);

    \draw[fill=black] (+1.78,-4.75) circle (3pt);
    \draw[line width=1.5pt,postaction={decoration={markings,mark=at position .55 with {\arrow{Stealth}}},decorate}] (+.28,-3.25) to (+1.78,-4.75);
    \end{scope}

    \begin{scope}[xshift=4.5cm,yshift=1.5cm]
    \draw[fill=black] (-4.3,1.25) circle (3pt);
    \draw[fill=black] (-2.75,-.25) circle (3pt);
    \draw[line width=1.5pt,postaction={decoration={markings,mark=at position .58 with {\arrow{Stealth}}},decorate}] (-4.3,1.25) to (-2.75,-.25);
    
    \draw[fill=black] (-1.2,-1.75) circle (3pt);
    \draw[line width=1.5pt,postaction={decoration={markings,mark=at position .55 with {\arrow{Stealth}}},decorate}] (-2.75,-.25) to (-1.2,-1.75);

    \draw[fill=black] (+.28,-3.25) circle (3pt);
    \draw[line width=1.5pt,postaction={decoration={markings,mark=at position .55 with {\arrow{Stealth}}},decorate}] (-1.2,-1.75) to (+.28,-3.25);
    \end{scope}

    \draw[fill=black] (-4.3,1.25-3) circle (3pt);
    \draw[fill=black] (-4.3+1.5,1.25-4.5) circle (3pt);
    \draw[fill=black] (-4.3+7.5,1.25+1.5) circle (3pt);
    \draw[fill=black] (-4.3+9,1.25) circle (3pt);

    \draw[line width=1.5pt,postaction={decoration={markings,mark=at position .6 with {\arrow{Stealth}}},decorate}] (-4.3+3,1.25) -- +(-1.5,-1.5);
    \draw[line width=1.5pt,postaction={decoration={markings,mark=at position .6 with {\arrow{Stealth}}},decorate}] (-2.75+1.5,-.25+1.5) -- +(1.5,1.5);

    \draw[line width=1.5pt,postaction={decoration={markings,mark=at position .6 with {\arrow{Stealth}}},decorate}] (-4.3+6,1.25) -- +(-1.5,-1.55);
    
    \draw[line width=1.5pt,postaction={decoration={markings,mark=at position .6 with {\arrow{Stealth}}},decorate}] (-2.75+1.5,-.25-1.5) -- +(1.5,1.5);

    \draw[line width=1.5pt,postaction={decoration={markings,mark=at position .6 with {\arrow{Stealth}}},decorate}] (-4.3+6.05,1.25-3) -- +(-1.5,-1.5);
    
    \draw[line width=1.5pt,postaction={decoration={markings,mark=at position .6 with {\arrow{Stealth}}},decorate}] (-2.75+4.5,-.25-1.5) -- +(1.5,1.5);

    \draw[line width=1.5pt,postaction={decoration={markings,mark=at position .6 with {\arrow{Stealth}}},decorate}] (-4.3,1.25) -- +(1.5,1.5);

    \draw[line width=1.5pt,postaction={decoration={markings,mark=at position .6 with {\arrow{Stealth}}},decorate}] (-4.3,1.25-3) -- +(1.5,1.5);

    \draw[line width=1.5pt,postaction={decoration={markings,mark=at position .6 with {\arrow{Stealth}}},decorate}] (-4.3+3.05,1.25-3) -- +(-1.5,-1.45);

    \draw[line width=1.5pt,postaction={decoration={markings,mark=at position .6 with {\arrow{Stealth}}},decorate}] (-4.3,1.25-3) -- +(1.5,-1.5);

    \draw[line width=1.5pt,postaction={decoration={markings,mark=at position .6 with {\arrow{Stealth}}},decorate}] (-4.3+9,1.25-3.05) -- +(-1.5,-1.5);

    \draw[line width=1.5pt,postaction={decoration={markings,mark=at position .6 with {\arrow{Stealth}}},decorate}] (-4.3+6,1.25) -- +(1.5,1.5);

    \draw[line width=1.5pt,postaction={decoration={markings,mark=at position .6 with {\arrow{Stealth}}},decorate}] (-4.3+7.5,1.25+1.5) -- +(1.5,-1.5);

    \draw[line width=1.5pt,postaction={decoration={markings,mark=at position .55 with {\arrow{Stealth}}},decorate}] (-4.3+9,1.25) -- +(-1.5,-1.55);

\end{tikzpicture}
}

\newcommand{\BoltzmannWeightsgCPM}{

\begin{tikzpicture}
    \draw[line width=1.5pt,postaction={decoration={markings,mark=at position .35 with {\arrow{Stealth}}},decorate}] (0,-1.5)node[below,yshift=-.05cm]{$\mbs{v}$} -- (0,1.5) node[above,yshift=.05cm]{$\mbs{w}$};
    \draw[line width=1pt, dashed, -Stealth] (-1.5,-1.5)node[below]{$p$} -- ++(3,3);
    \draw[line width=1pt, dashed, Stealth-] (-1.5,+1.5) -- ++(3,-3)node[below]{$q$};

    \draw[fill=black] (0,-1.5) circle (3pt);
    \draw[fill=black] (0,1.5) circle (3pt);
    \node[] at (-3,0) {$\overline{\mscr{W}}_{pq}(\mbs{v},\mbs{w})=$};
    
    \begin{scope}[xshift=+5cm]
    \draw[line width=1.5pt,postaction={decoration={markings,mark=at position .8 with {\arrow{Stealth}}},decorate}] (0,-1.5)node[below,yshift=-.05cm]{$\mbs{v}$} -- (0,1.5) node[above,yshift=.05cm]{$\mbs{w}$};
    \draw[line width=1pt, dashed, Stealth-] (-1.5,-1.5)node[below]{$p$} -- ++(3,3);
    \draw[line width=1pt, dashed, -Stealth] (-1.5,+1.5) -- ++(3,-3)node[below]{$q$};
    \end{scope}

    \draw[fill=black] (5,-1.5) circle (3pt);
    \draw[fill=black] (5,1.5) circle (3pt);
    \node[] at (+8,0) {$=\overline{\mscr{W}}^{-1}_{qp}(\mbs{v},\mbs{w})$};

\end{tikzpicture}
}

\newcommand{\ElementaryBoxgCPM}{
    \begin{tikzpicture}
    \draw[line width=1pt, dashed, -Stealth] (-3,+.75)node[left]{$q$} -- +(4.9,0);
    \draw[line width=1pt, dashed, Stealth-] (-3.3,-.75) -- +(4.9,0)node[right]{$q'$};
    \draw[line width=1pt, dashed, Stealth-] (-1.5,1.9) -- +(0,-3.4)node[below]{$p$};
    \draw[line width=1pt, dashed, -Stealth] (-1.5+1.5,1.5)node[above]{$p'$} -- +(0,-3.4);

    \draw[pattern=crosshatch dots gray, opacity=.3] (-3+1.5,.74) rectangle ++(1.5,-1.48);

    \draw[fill=black] (-2.25,0) circle (3pt)node[left,xshift=-.05cm]{$\mbs{w}$};
    \draw[fill=black] (-2.25+1.5,1.5) circle (3pt)node[above,yshift=+.05cm]{$\mbs{v}'$};;
    \draw[fill=black] (-2.25+3,0) circle (3pt)node[right,xshift=+.05cm]{$\mbs{v}$};;
    \draw[fill=black] (-2.25+1.5,-1.5) circle (3pt)node[below,yshift=-.05cm]{$\mbs{w}'$};

    \draw[line width=1.5pt,postaction={decoration={markings,mark=at position .58 with {\arrow{Stealth}}},decorate}] (-2.25,0) -- ++(+1.5,+1.5);

    \draw[line width=1.5pt,postaction={decoration={markings,mark=at position .6 with {\arrow{Stealth}}},decorate}] (-2.25+1.5,1.5) -- ++(+1.5,-1.5);

    \draw[line width=1.5pt,postaction={decoration={markings,mark=at position .6 with {\arrow{Stealth}}},decorate}] (-2.25+3,0) -- ++(-1.5,-1.5);

    \draw[line width=1.5pt,postaction={decoration={markings,mark=at position .6 with {\arrow{Stealth}}},decorate}] (-2.25+1.5,-1.5) -- ++(-1.5,+1.5);

    \node at (3.6,0) {$=\,\mcal{R}^{\mbs{v}\mbs{v}'}_{\mbs{w}\mbs{w}'}(q,q';p,p')$};
    \end{tikzpicture}
}

\newcommand{\RMatrixgCPM}{
    \begin{tikzpicture}
        \draw[line width=1pt] (-1,-1.5)node[below]{$(\mbs{v},z)$} to[out=90, in=210] (0,0) to[out=30,in=270] (1,+1.5)node[above]{$(\mbs{w},z)$};
        \draw[line width=1pt] (+1,-1.5)node[below]{$(\mbs{w}',z')$} to[out=90, in=-30] (0,0) to[out=150, in=270] (-1,+1.5)node[above]{$(\mbs{v}',z')$};
        
        \draw[fill=black] (0,0) circle (2pt)node[xshift=1.3cm]{\small $\mcal{R}^{\mbs{vv}'}_{\mbs{ww}'}(z,z')$};
    \end{tikzpicture}
}

\newcommand{\YangBaxterEquation}{

    \begin{tikzpicture}
    \draw[line width=1pt] (-2,-2)node[below]{$(\mbs{v}_1,z)$} to[out=90, in=200] (0,0) to [out=20,in=270] (2,+2)node[above]{$(\mbs{v}_3,z)$};

    \draw[line width=1pt] (0,-2)node[below]{$(\mbs{v}'_1,z')$} to[out=90, in=270] (-2,0) to [out=90,in=270] (0,+2)node[above]{$(\mbs{v}'_3,z')$};

    \draw[line width=1pt] (2,-2)node[below]{$(\mbs{v}''_1,z'')$} to[out=90, in=-20] (0,0) to [out=160,in=270] (-2,+2)node[above]{$(\mbs{v}''_3,z'')$};

    \draw[fill=black] (-1.45,-.8) circle (3pt)node[left,xshift=-.1cm]{\tiny $\mcal{R}^{\mbs{v}_1\mbs{v}'_2}_{\mbs{v}_2\mbs{v}'_1}(z,z')$};
    
    \draw[fill=black] (0,0) circle (3pt)node[right,xshift=.3cm]{\tiny $\mcal{R}^{\mbs{v}_2\mbs{v}''_2}_{\mbs{v}_3\mbs{v}''_1}(z,z'')$};
    
    \draw[fill=black] (-1.45,.8) circle (3pt)node[left,xshift=-.1cm]{\tiny $\mcal{R}^{\mbs{v}'_2\mbs{v}''_3}_{\mbs{v}'_3\mbs{v}''_2}(z',z'')$};

    \node[] at (3.5,0){$=$};

    \begin{scope}[xshift=+7cm]
    \draw[line width=1pt] (-2,-2)node[below]{$(\mbs{v}_1,z)$} to[out=90, in=200] (0,0) to [out=20,in=270] (2,+2)node[above]{$(\mbs{v}_3,z)$};

     \draw[fill=black] (0,0) circle (3pt)node[left,xshift=-.3cm]{\tiny $\mcal{R}^{\mbs{v}_1\mbs{v}''_3}_{\mbs{v}_2\mbs{v}''_2}(z,z'')$};

    \draw[line width=1pt] (2,-2)node[below]{$(\mbs{v}''_1,z'')$} to[out=90, in=-20] (0,0) to [out=160,in=270] (-2,+2)node[above]{$(\mbs{v}''_3,z'')$};

    \begin{scope}[xscale=-1]
    \draw[line width=1pt] (0,-2)node[below]{$(\mbs{v}'_1,z')$} to[out=90, in=270] (-2,0) to [out=90,in=270] (0,+2)node[above]{$(\mbs{v}'_3,z')$};

    \draw[fill=black] (-1.45,-.8) circle (3pt)node[right,xshift=.1cm]{\tiny $\mcal{R}^{\mbs{v}'_1\mbs{v}''_2}_{\mbs{v}'_2\mbs{v}''_1}(z',z'')$};

    \draw[fill=black] (-1.45,.8) circle (3pt)node[right,xshift=.1cm]{\tiny $\mcal{R}^{\mbs{v}_2\mbs{v}'_3}_{\mbs{v}_3\mbs{v}'_2}(z,z')$};
    
    \end{scope}
    \end{scope}
    \end{tikzpicture}
}

\newcommand{\RMatrixCSTheory}{

\begin{tikzpicture}
    \draw[line width=1pt] (-1.5,0)node[left]{$(\mbs{v},z)$} -- (1.5,0)node[right]{$(\mbs{w},z)$};

    \draw[line width=1pt] (0,-1.5)node[below]{$(\mbs{w}',z')$} -- (0,+1.5)node[above]{$(\mbs{v}',z')$};

    \draw[decorate, decoration={coil, amplitude=.1cm, segment length=.125cm, aspect=.4},line width = .75pt] (.75,0) -- (1.2,1);

    \draw[decorate, decoration={coil, amplitude=.1cm, segment length=.125cm, aspect=.4},line width = .75pt] (0,.75) -- (1,1.2);

    \shade[ball color=gray!80] (1.5,1.5) circle (.6cm);
    
\end{tikzpicture}}

\newcommand{\ParametersOfCorrespondence}{
\begin{tabular}{c||c}
 Hyperbolic $\SU{n}$-Monopole &  $(\mbb{Z}_N)^{n-1}$ Generalized Chiral Potts Model \\
\hline
\hline
$\text{rnk}(\SU{n})=n-1$ & $n-1$ copies of $\mathbb{Z}_N$-spins \\
\hline
\makecell{monopole charges $m_i$ for \\ with $m_1=\ldots=m_{n-1}=N$} & \makecell{$n-1$  spins taking value \\ in $\mbb{Z}_N=\{0,\ldots,N-1\}$} \\ 
\hline
\makecell{monopole masses  $p_i,\,i=1,\ldots,n$ \\ with $p_1=\ldots=p_n=0$} & \makecell{no counterpart, \\ $p_1=\ldots=p_n=0$} \\
\hline
$n-1$ spectral curves & \makecell{$n-1$ curves constructed \\ out of $\cpc$}\\
\end{tabular}
}

\newcommand{\ElementaryBoxEqualRMatrix}{
    \begin{tikzpicture}
    \draw[line width=1pt, dashed, -Stealth] (-3,+.75)node[left]{$q$} -- +(4.9,0);
    \draw[line width=1pt, dashed, Stealth-] (-3.3,-.75) -- +(4.9,0)node[right]{$q'$};
    \draw[line width=1pt, dashed, Stealth-] (-1.5,1.9) -- +(0,-3.4)node[below]{$p$};
    \draw[line width=1pt, dashed, -Stealth] (-1.5+1.5,1.5)node[above]{$p'$} -- +(0,-3.4);

    \draw[pattern=crosshatch dots gray, opacity=.3] (-3+1.5,.74) rectangle ++(1.5,-1.48);

    \draw[fill=black] (-2.25,0) circle (3pt)node[left,xshift=-.05cm]{$\mbs{w}$};
    \draw[fill=black] (-2.25+1.5,1.5) circle (3pt)node[above,yshift=+.05cm]{$\mbs{v}'$};;
    \draw[fill=black] (-2.25+3,0) circle (3pt)node[right,xshift=+.05cm]{$\mbs{v}$};;
    \draw[fill=black] (-2.25+1.5,-1.5) circle (3pt)node[below,yshift=-.05cm]{$\mbs{w}'$};

    \draw[line width=1.5pt,postaction={decoration={markings,mark=at position .58 with {\arrow{Stealth}}},decorate}] (-2.25,0) -- ++(+1.5,+1.5);

    \draw[line width=1.5pt,postaction={decoration={markings,mark=at position .6 with {\arrow{Stealth}}},decorate}] (-2.25+1.5,1.5) -- ++(+1.5,-1.5);

    \draw[line width=1.5pt,postaction={decoration={markings,mark=at position .6 with {\arrow{Stealth}}},decorate}] (-2.25+3,0) -- ++(-1.5,-1.5);

    \draw[line width=1.5pt,postaction={decoration={markings,mark=at position .6 with {\arrow{Stealth}}},decorate}] (-2.25+1.5,-1.5) -- ++(-1.5,+1.5);

    \node at (2.5,0) {$=$};

    \begin{scope}[xshift=.5cm]
    \draw[line width = 1.5pt] (3.5,.05)node[left]{$\mbs{w}$} -- (6.5,.05)node[right]{$\mbs{v}$};
    \draw[line width = 1.5pt] (5,-1.5)node[below]{$\mbs{w}'$} -- (5,1.5)node[above]{$\mbs{v}'$};
    \end{scope}
    \end{tikzpicture}
}

\title{\bf\Large\centering Hyperbolic Monopoles, (Semi-)Holomorphic Chern--Simons Theories, and Generalized Chiral Potts Models}

\author[a]{Seyed Faroogh Moosavian\footnote{\href{mailto:sfmoosavian@gmail.com}{\texttt{sfmoosavian@gmail.com}}}}
\author[b,c,d]{Masahito Yamazaki\footnote{\href{mailto:masahito.yamazaki@ipmu.jp}{\texttt{masahito.yamazaki@ipmu.jp}}}}
\author[c,e]{Yehao Zhou\footnote{\href{mailto:yehaozhou1994@gmail.com}{\texttt{yehaozhou1994@gmail.com}}}}

\affil[a]{\small Mathematical Institute, University of Oxford, Woodstock Road, Oxford, OX2 6GG, UK}
\affil[b]{\small Department of Physics, University of Tokyo, Hongo, Tokyo 113-0033, Japan}
\affil[c]{\small Kavli Institute for the Physics and Mathematics of the Universe, UTIAS,  University of Tokyo, Kashiwa, Chiba 277-8583, Japan}
\affil[d]{\small Trans-Scale Quantum Science Institute, University of Tokyo, Tokyo 113-0033, Japan}
\affil[e]{\small Shanghai Institute for Mathematics and Interdisciplinary Sciences, Block A, International Innovation Plaza, No. 657 Songhu Road, Yangpu District, Shanghai, China}

\date{}

\begin{document}

\maketitle

\vspace*{-.65cm}
\begin{abstract}
    We study the relation between spectral data of magnetic monopoles in hyperbolic space and the curve of the spectral parameter of generalized chiral Potts models (gCPM) through the lens of (semi-)holomorphic field theories. We realize the identification of the data on the two sides, which we will call the hyperbolic monopole/gCPM correspondence. For the group $\SU{2}$, this correspondence had been observed by Atiyah and Murray in the 80s. Here, we revisit and generalize this correspondence, and establish its origin. By invoking the work of Murray and Singer on the construction of spectral data of hyperbolic monopoles with arbitrary boundary values of the Higgs field, we first generalize the observation of Atiyah and Murray to the group $\SU{n}$ and put the correspondence on solid ground. By embedding the classical and exceptional groups into $\SU{n}$, we explore the correspondence for these groups. Next, we propose a technology to engineer gCPM within the 4d Chern--Simons (CS) theory, which explains various features of the model including the lack of rapidity-difference property of its R-matrix, and its peculiarity of having a genus$\,\ge 2$ curve of the spectral parameter. Finally, we investigate the origin of the correspondence. We first clarify how the two sides of the correspondence can be realized from the 6d holomorphic CS theory formulated on $\sbm$, the projective spinor bundle of the Minkowski space $M=\R^{1,3}$, in the case of hyperbolic $\SU{n}$-monopoles, and the Euclidean space $M=\R^4$, in the case of the gCPM. We then establish that $\sbm$ for both $M=\R^{1,3}$ and $\R^4$ can be holomorphically embedded into $\sbcm$, the projective spinor bundle of the complexified Minkowski space $\C^{1,3}$, of complex dimension five equipped with a fixed complex structure. Finally, we explain how the 6d CS theory on $\sbm$ can be realized as the dimensional reduction of the 10d holomorphic CS theory formulated on $\sbcm$. As the latter theory is only sensitive to the complex structure of $\sbcm$, which has been fixed, we realize the correspondence as two incarnations of the same physics in ten dimensions. 
\end{abstract}

\pagenumbering{gobble}
\pagebreak 
\pagenumbering{roman}
\tableofcontents
\pagenumbering{arabic}

\section{Introduction}
\label{sec:introduction}

Due to the complexity of physical systems, it is rare to find physical situations that can be modeled and simultaneously solved exactly. As such, integrable models and their properties have fascinated physicists and mathematicians for centuries. In this work, we are concerned with two classes of integrable systems: (1) magnetic monopoles in hyperbolic space $\Hbb^3$, and (2) the generalized chiral Potts model (gCPM), a special two-dimensional integrable lattice model in which dynamical variables, called spins, live on lattice sites. Let us briefly explain each before elucidating the objectives of this work.

\paragraph{Monopoles in Hyperbolic Space.} In quantum field theory, the study of classical solutions such as instantons, monopoles, and skyrmions is indispensable for understanding nonperturbative physics. A fascinating feature of these configurations is that they are intimately connected with integrable systems. The most famous of all is an (anti-)instanton in four dimensions, which is a solution of the (anti-)self-duality equation
\begin{equation}\label{eq:self-duality equation,introduction}
    F=\pm\star F \;,
\end{equation}
where $F$ is the curvature of a connection $A$ on a principal $G$-bundle over the four-manifold and $\star$ is the Hodge star operation. Moduli space of solutions of \eqref{eq:self-duality equation,introduction} has been extensively studied, in particular by Donaldson \cite{Donaldson198302,Donaldson198501}. There is considerable evidence that many integrable systems can be derived from the (anti-)self-duality equation and as such can be regarded as the master integrable system \cite{Ward198508}. A useful tool in the study of integrability of (anti-)instantons is the twistor correspondence, according to which (anti-)instanton solutions give rise to holomorphic vector bundles over a certain complex manifold called the twistor space. For (anti-)instantons on $\R^4$, such a correspondence was first realized by Ward \cite{Ward1977,Ward197704} and from a mathematical perspective, by Atiyah, Hitchin, and Singer \cite{AtiyahHitchinSinger197707,AtiyahHitchinSinger197809}. 

\smallskip One of the most important integrable systems obtained from the (anti-)self-duality equation is the magnetic monopole. Traditionally, they have been realized as the reduction of (anti-)self-duality equation by assuming invariance under a one-dimensional group of symmetries. In the process, (anti-)self-duality equation reduces to the Bogomolny equation, the defining equation of magnetic monopoles \cite{Bogomolny197604}\footnote{The original paper of Bogomolny can be found in \cite[pg.\ 389]{RebbiSoliani198412}.}
\begin{equation}\label{eq:Bogomolny equations, introduction}
    F=\star D\phi\;,
\end{equation}
where the gauge covariant derivative $D$ is defined with the connection $A$,  $\phi$ is an adjoint-valued scalar (the Higgs field), and $\star$ is now the Hodge star operation on a three-manifold.\footnote{It would be more accurate to write the Bogomolny equations as $F=\pm\star D\phi$, where $\pm$ correspond to monopole/anti-monopole configurations, which come from the reduction of (anti-)instantons. To void cluttering, we just ignore this distinction in this work. It would not affect any of the results.} The Bogomolny equation \eqref{eq:Bogomolny equations, introduction} has been studied on general three-manifolds and in particular hyperbolic ones by Braam \cite{Braam1987,Braam198902}. Monopoles in the Euclidean space $\mbb{R}^3$, or Euclidean monopoles, have been the most well-studied case \cite{AtiyahHitchin198812}. They can be realized as translation-invariant (anti-)instantons. The study of monopoles in hyperbolic space $\mbb{H}^3$, or hyperbolic monopoles, was initiated by Atiyah \cite{Atiyah1984} who realized them as circle-invariant (anti-)instantons. The zero-curvature limit of hyperbolic monopoles gives Euclidean ones \cite{Atiyah1984,JarvisNorbury199711,Nash200810}. However, unlike Euclidean monopoles, hyperbolic monopoles are distinguished by boundary values of the Higgs field, which are also called magnetic monopole masses \cite{BraamAustin199008,MurraySinger199607,Norbury199911}. Therefore, it turns out to be convenient to distinguish two classes of hyperbolic monopoles based on the boundary values of the Higgs field: (1) Those with integer (or half-integer) boundary values of the Higgs field: these can be thought of as circle-invariant (anti-)instantons as conceived by Atiyah \cite{Atiyah1984,Chakrabarti198601,Nash198608,BraamAustin199008,Chan201506,Chan2017}, and (2) those with arbitrary real boundary values of the Higgs field: these can be thought of as dilatation-invariant (anti-)instantons, as have been realized by Murray and Singer \cite{MurraySinger199607}.\footnote{Hyperbolic monopoles with non-integral boundary values of Higgs field and appropriate boundary conditions could be interpreted as solutions to the (anti-)self-duality equation with certain singularities along a codimension-two subspace \cite{Atiyah1984,ForgacsHorvathPalla198102}. Furthermore, in general, the reduction of (anti-)self-duality equations on a generic curved spacetime does not need to lead to a Bogomolny-type equation once invariance under certain coordinate transformations are imposed \cite{ComtetForgacsHorvathy198407}.} The difference between the boundary values of the Higgs field comes exactly because of their particular realization. The crux of the construction in \cite{Atiyah1984,BraamAustin199008,Chan201506,Chan2017} is that hyperbolic monopoles are circle-invariant solutions of \eqref{eq:self-duality equation,introduction}. The quantization of the asymptotic values is then forced by the compactness of the circle. Therefore, it is conceivable that to construct hyperbolic monopoles with arbitrary boundary values of the Higgs field, one should consider solutions of the (anti-)self-duality equation invariant under a non-compact group. Indeed, the construction of Murray and Singer in \cite{MurraySinger199607} exploits this idea: they realize hyperbolic monopoles as dilatation (i.e.\ $\R_+$-) invariant solutions of \eqref{eq:self-duality equation,introduction} on an open subset of Minkowski space.

\smallskip Explicit examples of hyperbolic monopole configurations have been constructed in the literature \cite{Chakrabarti198601,MantonSutcliffe201207}. Interestingly, hyperbolic monopoles are also related to other integrable systems \cite{Ward199811}: the Braam--Austin or discrete Nahm equation, which is a one-dimensional lattice equation \cite{BraamAustin199008}, and a certain soliton system in three-dimensional AdS spacetime in the Lorentzian signature. Similar to (anti-)instantons, there is a twistor correspondence for magnetic monopoles. Through this correspondence, a monopole solution on a three-manifold determines a holomorphic vector bundle over the twistor space of that three-manifold \cite{Hitchin198212,Hitchin198302,Ward198109} (see also \cite{Prasad198003,PrasadRossi198103}). The monopole solution can eventually be reconstructed in terms of its so-called spectral data in a unique fashion. The twistor correspondence for monopoles on $\R^3$ was first pointed out by Ward \cite{Ward198109} and Hitchin \cite{Hitchin198212} for the gauge group $\SU{2}$, and then generalized to arbitrarily connected compact simple gauge groups in \cite{Murray198306,Murray198412,Murray1983,HurtubiseMurray1989,HurtubiseMurray1990}. Hence, this correspondence is sometimes called the Ward--Hitchin correspondence, and can be thought of as a non-linear version of one of the special instances of the Penrose transform \cite{Penrose196702,Penrose196805,Penrose196901,Eastwood198501}. The spectral data of hyperbolic $\SU{2}$-monopoles with integer (or half-integer) boundary values of the Higgs field is constructed in \cite{Atiyah1984,BraamAustin199008} (see also \cite{Sutcliffe202008} for the connection to ADHM construction), while for the group $\SU{n}$ with arbitrary boundary values of the Higgs field is worked out by Murray and Singer in \cite{MurraySinger199607}. There are partial results on the spectral data of hyperbolic $\SU{n}$-monopoles with integer boundary values of the Higgs field in \cite[\S 14]{Chan2017}. This work uses boundary conditions different than the ones used by Murray and Singer.

\paragraph{Integrability and the Yang--Baxter Equation.} $(1+1)$-dimensional integrable models are solvable models of statistical mechanics in which dynamical variables, or spins, live on sites of a quantum one-dimensional chain evolving in time. 
One can equivalently think of them as classical two-dimensional lattice models. Historically, the first systematic approach to solve the eigenvalue problem for such systems was the so-called coordinate Bethe ansatz \cite{Bethe193103}. Such models have been studied intensively since the complete analytic solution of the two-dimensional Ising model on a square lattice by Onsager \cite{Onsager194310}. Time and again it has been proven that the most successful and conceptual approach in uncovering the integrable structure of these models has been the so-called algebraic Bethe ansatz or quantum inverse scattering method \cite{FaddeevSklyaninTakhtajan1980,FaddeevTakhtajan1979,Sklyanin1982} based on concepts such as Lax pair, R-matrix, and the Yang--Baxter (YB) equation \cite{Lax196809,Onsager194310,McGuire196405,Yang196711,Yang196712,Baxter197203,Baxter197805}. 

\smallskip Spins of an integrable model take value in representations of a symmetry algebra $\mfk{g}$. This leads to a useful way of thinking about R-matrix, whose components are denoted as  $\mcal{R}^{\mbs{vv}'}_{\mbs{ww}'}(z,z')$: it is the intertwiner of the tensor product of two modules $\mcal{V}$ and $\mcal{V}'$ of $\mfk{g}$: $\mcal{R}_{\mcal{V}\mcal{V}'}(z,z'):\mcal{V}\otimes\mcal{V}'\to\mcal{V}'\otimes\mcal{V}$. It depends on two types of variables (1) two spectral/rapidity parameters $z\ne z'$, and (2) the spin states $\{\mbs{v},\mbs{v}',\mbs{w},\mbs{w}'\}$.
\begin{figure}[H]
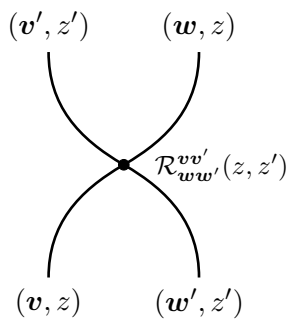

    \centering
    \RMatrixgCPM
        \caption{A graphical representation of the R-matrix of an integrable model.} %
    \label{fig:r-matrix of a general integrable model}
\end{figure}
The spectral parameters correspond to the inhomogeneity parameters of the integrable model and the states correspond to weights in representations of $\mfk{g}$. If one thinks of the R-matrix as $\mcal{R}_{\mcal{V}\mcal{V}'}(z,z'):\mcal{V}\otimes\mcal{V}'\to\mcal{V}'\otimes\mcal{V}$, then the YB equations can be obtained by acting on $\mcal{V}\otimes\mcal{V}'\otimes\mcal{V}''$. It is given by
\begin{equation}\label{eq:general YB equations}
    \mcal{R}_{\mcal{V}\mcal{V}'}(z,z')\mcal{R}_{\mcal{V}\mcal{V}''}(z,z'')\mcal{R}_{\mcal{V}'\mcal{V}''}(z',z'')=\mcal{R}_{\mcal{V}'\mcal{V}''}(z',z'')\mcal{R}_{\mcal{V}\mcal{V}''}(z,z'')\mcal{R}_{\mcal{V}\mcal{V}'}(z,z') \;.
\end{equation}
The YB equation is a sufficient condition for integrability whose origin goes back to the star-triangle relation mentioned in passing in the introduction of Onsager's paper \cite{Onsager194310}. It ensures the existence of enough number of conserved quantities in involution which in turns guarantees the complete quantum integrability of the model.\footnote{Strictly speaking, this is correct for the so-called ultra-local quantum integrable systems. The non-ultralocal quantum integrable systems are characterized by the braided quantum YB equation \cite[\S 8]{Kundu199612}.} By reading from bottom to top, a graphical representation of the YB equation is
\begin{figure}[H]
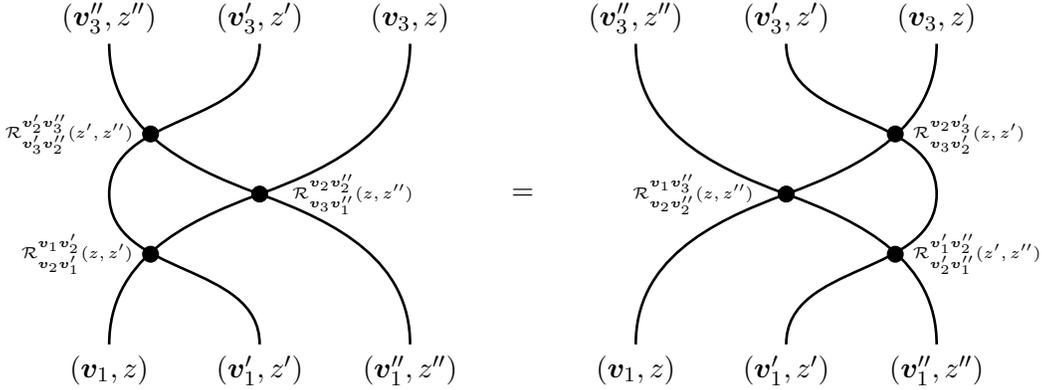

    \centering
    \YangBaxterEquation
    \caption{A graphical representation of the Yang--Baxter equation.}
    \label{fig:yang-baxter equation}
\end{figure}

\smallskip Due to the fundamental importance of the YB equation, finding its solutions and their classification have been actively pursued in the 70s and 80s \cite{KulishReshetikhinSklyanin198109,Sklyanin1982,BelavinDrinfeld198207,BelavinDrinfeld198307,BelavinDrinfeld1998,KulishSklyanin198207,BazhanovStroganov198207,Bazhanov198709,Bazhanov198710,Reshetikhin198710,Stolin199112,Kulish198611,LeitesSerganova198401,BazhanovShadrikov198712}.\footnote{An explicit elliptic solution for the Lie superalgebra $\mfk{sl}(1|1)$ was constructed in \cite{IshtiaqueMoosavianZhou202308} using techniques from the Bethe/Gauge correspondence.} 
For an integrable spin model associated with a Lie algebra $\mfk{g}$, its R-matrix can be expanded in terms of a formal parameter $\hbar$ as $\mcal{R}_{\mcal{V}\mcal{V}'}(z)=\mds{1}+\hbar\,r_{\mcal{V}\mcal{V}'}(z)+\mcal{O}(\hbar^2)$, where $r_{\mcal{V}\mcal{V}'}(z)\in\text{End}(\mfk{g}\otimes\mfk{g})$ is called the classical R-matrix. As a result of \eqref{eq:general YB equations}, it satisfies the classical YB equation
\begin{equation}\label{eq:classical YB equation}
    [r_{\mcal{V}\mcal{V}'}(z,z'),r_{\mcal{V}\mcal{V}''}(z,z'')]+[r_{\mcal{V}\mcal{V}'}(z,z'),r_{\mcal{V}'\mcal{V}''}(z',z'')]+[r_{\mcal{V}\mcal{V}''}(z,z''),r_{\mcal{V}'\mcal{V}''}(z',z'')]=0 \;.
\end{equation}
The classification of solutions of \eqref{eq:classical YB equation} is achieved in the celebrated work of Belavin and Drinfeld \cite{BelavinDrinfeld198207,BelavinDrinfeld198307,BelavinDrinfeld1998}. It was realized in loc.\ cit.\ that if one assumes the rapidity-difference property, i.e.\ $\mcal{R}(z,z')=\mcal{R}(z-z')$ or $\mcal{R}(z/z')$, then the rapidity parameters in \eqref{eq:classical YB equation} would take value in $\C,\C^\times,$ or $\mbb{E}$ (an elliptic curve), which are called the curve of spectral parameter of the model. It then turns out that the corresponding classical R-matrix is a rational, trigonometric, or elliptic function of the spectral parameter. Such solutions are called rational, trigonometric, and elliptic, respectively. These classical R-matrices can be quantized and the solutions of the YB equation \eqref{eq:general YB equations} can be constructed from the representation theory of Yangians, quantum affine, and quantum elliptic algebras, respectively \cite{EtingofKazhdan199506}. Over the years, it has been proven difficult to construct integrable spin models that evade this classification scheme and have rapidity parameters belonging to curves of genus$\,\ge 2$.\footnote{There are models whose rapidity parameter belong to higher-dimensional complex manifolds. An example of such models with three affine rapidity parameters lying on $\Pbb^3$ is constructed in \cite{Martins201410}.}

\paragraph{Chiral Potts Model and its Generalization.} The most important assumption in the Belavin--Drinfeld classification is the assumption of the rapidity-difference property of the R-matrix, which can be relaxed.\footnote{Once the assumption of rapidity-difference is removed, it is not difficult to imagine that the curve of spectral parameter could have genus$\,\ge 2$. The reason is that we cannot make sense of expressions like $z-z'$ or $z/z'$ on such curves without any additive structure. The closest analog objects on a higher-genus curve are the prime forms, which are sections of a holomorphic line bundle over the curve.} A model without this property was constructed by Shastry \cite{Shastry198801}. However, the most peculiar of such models is the chiral Potts (or chiral clock) model (CPM), first uncovered for $n=2$ in  \cite{Au-YangMcCoyBarryPerkTangYan1987,BaxterPerkAu-Yang1988}, and later generalized to $n\ge 2$ gCPM in \cite{BazhanovKashaevMangazeev1990,BazhanovKashaevMangazeevStroganov199105}. The interest in this model comes from its connection to commensurate-incommensurate phase transitions \cite{HowesKadanoffDenNijs198301,vonGehlenRittenberg1985} and also the so-called ripple phase in lipid-bilayer biological membranes \cite{ScottPearce198902}. We refer to \cite{Baxter198808,AlbertiniMcCoyPerkTang198903,BaxterBazhanovPerk199004,BazhanovStroganov199005,Baxter199005,Baxter199302,Cardy199210,Baxter200501,Baxter200504} for various features of the model and \cite{Au-YangPerk1989,AuYangPerk199609,Perk201511,Au-YangPerk201601} for extensive review of  its history. 

\smallskip The spins of the gCPM are $n-1$ copies of $\mbb{Z}_N$ variables on a square lattice. The curve of the spectral parameter of the model is an algebraic curve $\cpc$ in $\mbb{P}^{2n-1}$ characterized by the solution of the following equations
\begin{equation}\label{eq:curve of spectral parameter of gCPM, introduction}
	\begin{pmatrix}
		{z^+_{i}}^{N}
		\\
		{z^-_i}^{N}
	\end{pmatrix}
	= K_{ij}
	\begin{pmatrix}
		{z^+_j}^{N}
		\\
		{z^-_j}^{N}
	\end{pmatrix}\;,
	\qquad 
    i,j=1,\ldots,n \;,
\end{equation}
where $[z^+_1:z^-_1:\ldots:z^+_n:z^-_n]$ denotes the homogeneous coordinates on $\Pbb^{2n-1}$, and the matrices $K_{ij}\in\text{SL}(2,\C)$ satisfy certain relations (see \eqref{eq:relations satisfied by matrices Kij of gCPM}). The genus of the curve is given by
\begin{equation}
    g_{\cpc}=N^{2(n-1)}(N(n-1)-n)+1 \;.
\end{equation}
The original model in \cite{Au-YangMcCoyBarryPerkTangYan1987,BaxterPerkAu-Yang1988} corresponds to $n=2$. It is clear that $g$ is always greater or equal to one and the only situation with $g=1$ is $(n, N)=(2,2)$, which is the Ising model. The proof that the R-matrix of the gCPM satisfies the YB equation appeared in \cite{DateJimboKeiMiwa199008,DateJimboKeiMiwa199103,DateJimboMikiMiwa199104} for odd values of $N$, using the representation theory of quantum affine algebra $\text{U}_q(\wh{\mfk{sl}}(n,\C))$ with $q^N=1$ being a root of unity,\footnote{The fact that gCPMs are related to the representation theory of $\text{U}_q(\wh{\mfk{sl}}(n,\C))$ was conjectured in \cite{BazhanovKashaevMangazeevStroganov199105} and later proved in \cite{DateJimboMikiMiwa199104} for odd $N$.} and for general $N$ in \cite{KashaevMangazeevNakanishi199109}. More recently, it is also understood as a root-of-unity degeneration of more general integrable models with continuous spins \cite{BazhanovSergeev201006,BazhanovSergeev201106,KelsYamazaki201709},
which in turn can be reproduced from supersymmetric indices of supersymmetric quiver gauge theories 
\cite{Yamazaki201203,Terashima:2012cx,Yamazaki201307,Yamazaki201808}. Furthermore, the model has been studied in connection to knot theory \cite{DateJimboMikiMiwa199201}. However, 
 it is fair to say that we are still very far from a complete understanding of the gCPMs.

\paragraph{Hyperbolic Monopoles and CPM.} Shortly after the discovery of CPM, Atiyah and Murray \cite{Atiyah199106,AtiyahMurray1995} observed a remarkable connection between the curve of the spectral parameter of CPM and the spectral curve of hyperbolic monopoles associated with the gauge group $\SU{2}$: They noticed that by imposing a certain reality condition on the curve of CPM and taking the limit of vanishing boundary values of the Higgs field, the free $\mbb{Z}_N$-quotient of the resulting curve of the spectral parameter of CPM can be identified with the spectral curve of the hyperbolic monopole. Despite tremendous progress in the theory of integrable models, the origin of this mysterious connection remained unknown, and relatively very little attention has been paid to it over the past decades since the observation of Atiyah and Murray. The main objective of this work is to revisit this observation.  

\paragraph{4d Chern--Simons Theory.} One of the most remarkable progress in the theory of integrable models in the last two decades has been the discovery of the 4d version of Chern--Simons (CS) theory by Costello \cite{Costello201303}
\begin{equation}\label{eq:4d CS action,introduction}
    S=\frac{1}{2\pi}\bigintsss_{C\times\Sigma}\omega\wedge\CS(A) \;.
\end{equation}
The theory defined on a product manifold of the form $C\times\Sigma$: it is topological along $C$ and holomorphic along $\Sigma$.\footnote{It is expected that this theory makes sense on any four-manifold with a transverse holomorphic foliation.} The one-form $\omega$ is along $\Sigma$: $\omega=\omega(z)\rd z$ with $z$ being a local coordinate along $\Sigma$, and $\CS(A)$ is the CS three-form for a connection $A=A_C+A_{\bmb{z}}\rd\bmb{z}$, with $A_C$ being the connection along $C$. It has been realized that all integrable spin-chain models which fit into the Belavin--Drinfeld classification scheme, i.e.\ those whose curve of spectral parameter is either $\C,\C^\times,$ and $\mbb{E}$, can be realized within this theory \cite{Witten201611,CostelloWittenYamazaki201709,CostelloYamazakiWitten201802}. In particular, the R-matrix of these models can be constructed from a particular configuration of Wilson lines, equivalent to the configuration of Fig. \ref{fig:r-matrix of a general integrable model}.
\begin{figure}[H]
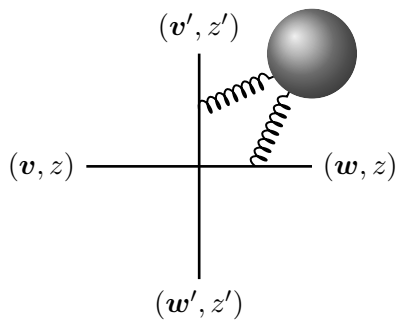

    \centering
    \RMatrixCSTheory
    \caption{The configuration of Wilson lines in the 4d CS theory for the computation of R-matrix. The wavy lines denote the gluon exchange and the gray circle shows the interaction in the bulk.}
    \label{fig:r-matrix from 4d CS theory}
\end{figure}
The two Wilson lines are stretched along the topological plane $C$ and are supported at different points $z,z'\in\Sigma$. The evaluation of Feynman diagrams of gluon exchanges between the lines computes the R-matrix where the leading-order nontrivial term can be identified with the classical R-matrix of the model \cite{CostelloWittenYamazaki201709}. From this perspective, the YB equation in Fig. \ref{fig:yang-baxter equation} is an immediate consequence of the existence of extra dimensions. Furthermore, many important classes of two-dimensional integrable field theories can be constructed, at least classically, by including certain surface defects in the theory \cite{CostelloYamazaki201908}. In this sense, it seems that the 4d CS theory plays the role of a unifying framework of a significant portion of integrable systems in two dimensions. For a selection of works related to the role of 4d CS theory in integrability see \cite{Vicedo201908,DelducLacroixMagroVicedo201909,BeniniSchenkelVicedo202008,LiniadoVicedo202301,FukushimaSakamotoYoshida202003,CostelloStefanski202005,Tian202005,TianHeChen202007,Stedman202009,FukushimaSakamotoYoshida202105,FukushimaSakamotoYoshida202112,HeTianChen202105,Levine202309,AshwinkumarSakamotoYamazaki202309,ColeWeck202407}.

\subsection{Statement of Problems}

Motivated by this circle of ideas, we explore three problems in this work.

\paragraph{Problem 1.} We ask whether the observation of Atiyah and Murray in \cite{Atiyah199106,AtiyahMurray1995} is restricted to the group $\SU{2}$ or there exists a general correspondence between the curve of the spectral parameter of gCPMs and hyperbolic $\SU{n}$-monopoles. 

\smallskip We realize that there is indeed a general correspondence for the group $\SU{n}$ with any $n\ge 2$. Hence, we will call it the hyperbolic monopole/gCPM correspondence. Furthermore, by embedding classical and exceptional groups into $\SU{n}$, the correspondence can be explored for these groups. As we will explain, this forces us to work in the limit of vanishing boundary values of the Higgs field. Such configurations are not trivial and are called flat monopoles, as elucidated in \S\ref{sec:general considerations}. The zero-mass limit of hyperbolic monopoles is also called nullarons, a terminology due to Murray \cite[Footnote 1]{NorburyRomao200512}. 

\paragraph{Problem 2.} Considering the unifying role of the 4d CS theory, we ask whether the gCPMs can be engineered within this theory. 

\smallskip It turns out that the answer is indeed yes. The main technique we use is to represent the curve of the spectral parameter of the gCPM as a branched cover of $\Pbb^1$, a construction that is always possible by virtue of Riemann's Existence Theorem.  We then realize that the requirement of topological invariance along $C$ together with  $\mbb{Z}_N^{n-1}$-invariance of the model fixes $\omega$ in \eqref{eq:4d CS action,introduction}.  

\paragraph{Problem 3.} As we elevate the connection between hyperbolic monopoles and gCPMs to the level of correspondence, we ask whether this is a coincidence or if there is a deeper explanation for it.

\smallskip We realize that there is indeed a deeper reason for this correspondence. The key role is played by (1) our engineering of the gCPM using the 
 4d CS theory; (2) the connection of hyperbolic monopoles and gCPMs with the 6d holomorphic CS (hCS) theories on certain complex three-dimensional manifolds; (3) the connection between the 6d and 10d hCS theories on complex manifolds. The fact that the 10d hCS theory depends only on the choice of complex structure unifies the two sides of the correspondence in ten-dimensional physics. 

\smallskip 

\subsection{Contributions of the Paper} We now summarize the main results of this work.

\subsubsection*{Problem 1. The Hyperbolic Monopoles/gCPM Correspondence (\S\ref{sec:generalized corresponence})} 
We first establish the correspondence for the gauge group $\SU{n}$. By starting from the curve of the spectral parameter of the gCPM, we reconstruct the spectral data of specific hyperbolic $\SU{n}$-monopoles. Conversely, we can start from the spectral data of hyperbolic $\SU{n}$-monopoles and construct the curve of the spectral parameter of the gCPM. In more details,

\paragraph{\small (1) Hyperbolic $\SU{n}$-Monopoles from gCPM (\S\ref{sec:hyperbolic monopoles from gCPMs}).} The procedure is as follows. From the curve of spectral parameter of the gCPM \eqref{eq:curve of spectral parameter of gCPM, introduction}, we construct $n-1$ curves $\ipc{1}, \dots, \ipc{n-1}$ of bidegree $(N,N)$  in $\mathbb{P}^1\times \mathbb{P}^1$, each of which can be written as 
    \begin{equation}\label{eq:spectral_data_gcpm} 
	\ipc{i}:\quad \begin{pmatrix}
		(z^+){}^{N}
		\\
		(z^-){}^{N}
	\end{pmatrix}
	= K_{i}
	\begin{pmatrix}
		(w^+){}^{N}
		\\
		(w^-){}^{N}
	\end{pmatrix}\;, \qquad i=1,\ldots,n-1 \;,
    \end{equation}
where $[z^+, z^-]$ and $[w^+, w^-]$ are the homogeneous coordinates on the two copies of $\Pbb^1$, and $K_i=K_{i,i+1}$ in \eqref{eq:curve of spectral parameter of gCPM, introduction} are restricted to be in $\text{PSL}(2,\R)$. We then show that these satisfy all the requirements of being the spectral curves of a hyperbolic $\SU{n}$-monopole. In particular, we show that 

\begin{enumerate}
    \item [(1)] By demanding $K_iK_{i+1}^{-1}$ to be an elliptic element of $\text{SL}(2,\R)$ for $i=1,\ldots,n-1$, we show that $\ipc{i}\cap\ipc{i+1}$ intersect at  $2N^2$ points. These points determine two divisors, $\ipc{i,i+1}$ and $\ipc{i+1,i}$, each associated with $N^2$ points. $\ipc{i,i+1}$ and $\ipc{i+1,i}$ are exchanged by a certain real structure of $\Pbb^1\times\Pbb^1$. 

    \item [(2)] Over $\ipc{i}$, certain line bundles determined by $(n,N)$ are isomorphic to line bundles determined by the divisors $\ipc{i,i+1}, \ipc{i+1,i}$, and $\ipc{i,i+1}+\ipc{i+1,i}$. 
\end{enumerate}

We then determine the hyperbolic $\SU{n}$-monopole configuration associated with this spectral data.

\paragraph{\small (2) GCPM from Hyperbolic $\SU{n}$-Monopoles (\S\ref{sec:gCPM from hyperbolic SU(n) monopoles}).} We then consider the reverse construction. We start from the spectral curves of hyperbolic $\SU{n}$-monopoles subject to maximal symmetry-breaking. We describe how from this collection of $n-1$ curves on $\Pbb^1\times\Pbb^1$, we can construct a curve on $\Pbb^{2n-1}$. After removing the reality condition, taking the limit of vanishing masses, and taking all the $n-1$ monopole charges to be equal to $N$, we recover the curve of the spectral parameter of the gCPM as given in \eqref{eq:curve of spectral parameter of gCPM, introduction}.

\paragraph{\small (3) The Correspondence.} We summarize our findings in the form of a proposition
\begin{prop*}\normalfont
    There is a one-to-one correspondence between 
    \begin{enumerate}
        \item [--] GCPMs whose curve of the spectral parameter $\Sigma\subset\Pbb^{2n-1}$ defined by $n-1$ $\text{SL}(2,\R)$-valued $K$-matrices $K_{i,i+1},\,i=1,\ldots,n-1$ in \eqref{eq:curve of spectral parameter of gCPM, introduction} such that $K_iK_{i+1}^{-1}$, with $K_i:=K_{i,i+1}$, is an elliptic element of $\text{SL}(2,\R)$; and

        \item [--] Hyperbolic $\SU{n}$-monopoles subject to maximal symmetry-breaking, magnetic charges $m_1=\ldots=m_{n-1}=N$, and vanishing asymptotic values of the Higgs field whose spectral data recovers the curve $\Sigma$.
    \end{enumerate}
\end{prop*}
For a more precise version of the proposition, see Proposition \ref{prop:hyperbolic monopole/gCPM correspondence}. The identification of parameters in the two sides is given in Table \ref{tab:identification of the parameters in the correspodence}.

\begin{table}[H]
    \centering
    \ParametersOfCorrespondence
    \caption{Parameters of the two sides of the hyperbolic monopole/gCPM correspondence.}
    \label{tab:identification of the parameters in the correspodence}
\end{table}

\paragraph{\small (4) The Correspondence for Classical and Exceptional Groups (\S\ref{sec:correspondence for classical and exceptional groups}).} We then explore the correspondence for other classical and exceptional Lie algebras by embedding them in $\SU{n}$. Denoting monopole charges for a Lie algebra $\mfk{g}$ by $m^{\mfk{g}}$, we find the values of charges for hyperbolic monopoles, associated with classical groups, which correspond to the gCPM
\begin{eqaligned}
    B_{n}&: &\qquad m_1^{B_n}&=\ldots=m^{B_n}_{n-1}=2N\;, &\qquad m^{B_n}_n&=N\;,
    \\
    C_{n}&: &\qquad m_1^{C_n}&=\ldots=m^{C_n}_{n}=N\;,
    \\
    D_{n}&: &\qquad m_1^{D_n}&=\ldots=m^{D_n}_{n-2}=2N\;, &\qquad m^{D_n}_{\pm}&=N\;,
\end{eqaligned}
while for exceptional groups $\mfk{g}_2$ and $\mfk{f}_4$, we find
\begin{eqaligned}
    \mfk{g}_2&: &\qquad  m_1^{\mfk{g}_2}&=-3N\;, &\qquad m^{\mfk{g}_2}_2&=-5N\;,
    \\
    \mfk{f}_4&: &\qquad m_1^{\mfk{f}_4}&=m_4^{\mfk{f}_4}=2N\;, &\qquad m_2^{\mfk{f}_4}&=m_3^{\mfk{f}_4}=3N\;.
\end{eqaligned}
We comment on the negative values of charges for $\mfk{g}_2$ in \S\ref{sec:correspondence for classical and exceptional groups}.

\subsubsection*{Problem 2. gCPM and 4d Chern--Simons Theory (\S\ref{sec:gCPM and 4d CS theory}).} We next explain how to realize the gCPM in the 4d CS theory. The punchline of our discussion is the following
\begin{enumerate}
    \item [(1)] We first realize the curve of the spectral parameter of gCPM as a branched cover of $\Pbb^1$ (\S\ref{sec:curve of spectral parameter as a branched cover of P1}). From this point onward, we will exclusively work with $C\times\Pbb^1$. From our considerations, it would turn out that the R-matrix of the gCPM can only be computed in a non-perturbative formulation of the 4d CS theory. 

    \item [(2)] We then construct the one-form $\omega_{\Pbb^1}$ of the 4d CS theory on $C\times\Pbb^1$ (\S\ref{sec:the one-form on the curve of spectral parameter}). We see that by demanding only (1) topological invariance and (2) $\mbb{Z}_N^{n-1}$-invariance, $\omega_{\Pbb^1}$ is strongly constrained and the minimal choice is given by
    \begin{equation}\label{eq:one-form relevant for engineering gCPM in 4d CS, introduction}
    \omega_{\Pbb^1}=\frac{1}{\sqrt[\leftroot{-2}\uproot{2}N^{n-1}]{\prod_{i=1}^{2N^{n-1}}(z-z_{Q_i})}}\rd z\;,
    \end{equation}
    where $\{Q_1,\ldots,Q_{2N^{n-1}}\}$ are the branched points of the branched covering map. From this and the branched covering map, we would be able to construct the one-form on $\cpc$ (see \S\ref{sec:the one-form on the curve of spectral parameter} and \S\ref{sec:generalization to gCPM}, in particular \eqref{eq:one-form relevant for engineering cpm on SigmaN depending on branched points of pi} and \eqref{eq:one-form on the curve of spectral parameter of gCPM}). 
    A remarkable feature of \eqref{eq:one-form relevant for engineering gCPM in 4d CS, introduction} which would be crucial for our developments is (see \S\ref{sec:the one-form on the curve of spectral parameter} for the derivation)
    \begin{equation}\label{eq:vanishing of zbar derivative of omegaP1 for gCPM, introduction}
    \partial_{\bmb{z}}\omega_{\Pbb^1}=0\;.
    \end{equation}

    This relation together with the form of $\omega_{\Pbb^1}$ in \eqref{eq:one-form relevant for engineering gCPM in 4d CS, introduction} would provide arguments about the peculiarity of the gCPM (see, in particular, Remarks \ref{rmk:why gCPM is peculiar?} and \ref{rmk:importance of vanishing of dbar derivative of omegaP1}). 
    
    \item[(3)] Although we are not looking to set up a perturbation theory, we see that as a consequence of \eqref{eq:vanishing of zbar derivative of omegaP1 for gCPM, introduction}, no requirement for imposing boundary conditions at $Q_i$ is implied (see \S\ref{sec:gauge Lie algebra, boundary conditions, R-matrix}). This is in sharp contrast to all the other integrable spin models where $\partial_{\bmb{z}}\omega_{\Pbb^1}$ would lead to (derivatives of) delta functions at poles of $\omega_{\Pbb^1}$. A proper definition of perturbation theory would then require the imposition of appropriate boundary conditions on the gauge fields at the poles of $\omega$ \cite[\S 9.1]{CostelloWittenYamazaki201709} (see \ref{sec:lack of rapidity-difference property}).

    \item[(4)] Finally, we explain how the absence of rapidity-difference property can be explained from the form on $\omega_{\Pbb^1}$ in \eqref{eq:one-form relevant for engineering gCPM in 4d CS, introduction}.  
 
\end{enumerate}

\subsubsection*{Problem 3. The Origin of the Correspondence (\S\ref{sec:on the origin of the correspondence})}

Finally, we will uncover the origin of the hyperbolic monopole/gCPM correspondence. In particular,

\paragraph{\small (1) Hyperbolic Monopole from 6d Holomorphic Chern--Simons Theory (\S\ref{sec:hyperbolic monopoles from six dimensions}).} We first explain how to construct hyperbolic monopoles from the 6d hCS theory on $\sbrm$, the projective spinor bundle of Minkowski space $\R^{1,3}$. This involves
\begin{enumerate}
    \item [--] Equipped with some results from Murray and Singer, we first establish in Proposition \ref{prop:one-to-one correspondence between CR bundles on ZU and holomorphic vector bundles on PS+U} and Corollary \ref{cor:one-to-one correspondence between self-dual connection on U and holomorphic vector bundle on PS+(U)} that solutions of the self-duality equation on Minkowski space $\R^{1,3}$ (or an open set thereof) can be constructed from solutions of equations of motion of the 6d hCS theory, which describe connections of type $(1,1)$;

    \item [--] We then see explicitly how hyperbolic monopoles (see \S\ref{sec:hyperbolic monopoles as dilatation-invariant instantons}), their charges, masses, and spectral data can be described in terms of six-dimensional gauge field (see \S\ref{sec:masses, charges, and spectral data from six dimensions}). In particular, this reduction shows why there is no parameter corresponding to hyperbolic monopole masses in the gCPM side of the correspondence (see Remark \ref{rmk:no parameters corresponding to masses in the gCPM side}), and also why there is no reality condition on the curve of the spectral parameter (see Remark \ref{rmk:no reality condition on the data}).
\end{enumerate}

\paragraph{\small (2) gCPM from 6d Holomorphic Chern--Simons Theory} We then show how to realize the gCPM from the 6d hCS theory on $\sbe$, the projective spinor bundle of Euclidean space $\R^4$. More concretely, 
    \begin{enumerate}
        \item [--] In \S\ref{sec:4d hBF theory from 6d hCS theory}, we show that the dimensional reduction of the 6d hCS on $\sbe$ to $\R^2\times\Pbb^1$ gives the 4d holomorphic BF (hBF) theory;

        \item [--] In \S\ref{sec:emergence of 4d CS theory}, we then show how the 4d CS theory emerges from the 4d hBF theory at the finite value of a cut-off dictated by the volume of extra dimensions of $\sbe$.
    \end{enumerate}

\paragraph{\small (3) 6d Holomorphic Chern--Simons Theory from Ten Dimensions.} Next, in \S\ref{sec:6d hCS on PS+(M) from 10d hCS on PS+C13}, we explain the followings

\begin{enumerate}
    \item [--] In Lemma \ref{lem:6d hCS from 10d hCS}, we show that the 6d hCS theory on a complex three-dimensional manifold can be realized from the dimensional reduction of the BV action of the 10d hCS theory on a complex five-dimensional manifold if the former manifold is embedded holomorphically in the latter;

    \item [--] To connect this result to the correspondence, we show, in Proposition \ref{prop:PS+(M) is holomorphically embedded in PS+C13}, that $\sbm$ for both $M=\R^{1,3}$ and $\R^4$ can be holomorphically-embedded into $\sbcm$, the projective spinor bundle of the complexified Minkowski space $\C^{1,3}$ equipped with a fixed complex structure $\mcal{J}_{\C^{1.3}}$ given in \eqref{eq:master complex structure on PS+C13}.
\end{enumerate} 

\paragraph{\small (4) The Correspondence from Ten Dimensions (\S\ref{sec:chasing the correspondence to ten dimensions}).} Putting all these results together, we finally explain how the two sides of the correspondence can be realized by starting from the 10d hCS theory on $\sbcm$, equipped with the complex structure $\mcal{J}_{\C^{1,3}}$, reducing it on $\sbrm$ and $\sbe$, which in turn reduce the two sides of the correspondence. In a naive sense, once everything is reduced to four dimensions, the correspondence is the result of a Wick or inverse Wick rotation followed by imposing or removing certain reality conditions, as we will explain in \S\ref{sec:chasing the correspondence to ten dimensions}.

\subsection{Contents of the Paper}

We start in \S\ref{sec:generalized CPM} by a brief review of the gCPM. The topic of hyperbolic monopole is not well-studied in the physics literature, and furthermore it forms an essential ingredient of our developments. Therefore, a rather extensive review of hyperbolic monopoles with arbitrary boundary values of the Higgs field and their spectral data is warranted. We fulfill this task in \S\ref{sec:hyperbolic monopoles and spectral data}. We then generalize the correspondence between hyperbolic monopoles and gCPM to the group $\SU{n}$ in \S\ref{sec:generalized corresponence}, with additional exploration of the correspondence for classical and exceptional Lie algebras in \S\ref{sec:correspondence for classical and exceptional groups}. The engineering of the gCPM within the 4d CS theory is the subject of \S\ref{sec:gCPM and 4d CS theory}. We establish the origin of the correspondence in \S\ref{sec:on the origin of the correspondence}. Several appendices complement the work. In Appendix \ref{sec:more details on twistor space of hyperbolic space}, we provide more details on the twistor space of hyperbolic space, especially from the perspective of the circle action on the twistor space of $S^4$. The curve of the spectral parameter of the gCPM is a special case of complete intersections and, as such, we include some basic details on complete intersections in Appendix \ref{sec:basic facts about complete intersections}. Finally, Appendix \ref{sec:some basic facts and Lie algebra manipulations} includes all the details of computations in \S\ref{sec:correspondence for classical and exceptional groups}.

\section{Generalized Chiral Potts Model}
\label{sec:generalized CPM}
We start with a brief review of the relevant facts about the gCPMs and their curve of the spectral parameter following the original construction by Bazhanov et al.\  \cite{BazhanovKashaevMangazeevStroganov199105}. With the aid of minimal cyclic representations of $\text{U}_q(\wh{\mfk{sl}}(n,\C))$, it was shown that the R-matrix satisfies the Yang-Baxter equation for odd $N$ in \cite{DateJimboKeiMiwa199008,DateJimboKeiMiwa199103}. The proof for the general case later appeared in \cite{BazhanovKashaevMangazeev1990,KashaevMangazeevNakanishi199109}. We refer the reader to these references for the details.

\subsection{The Model}
 
The gCMP is a model obtained by considering $(n-1)$ copies of spins taking values in $\mbb{Z}_{N}$ on a square lattice $\mcal{L}$, while the rapidity parameters live on its medial lattice $\mcal{L}'$:

\begin{figure}[H]
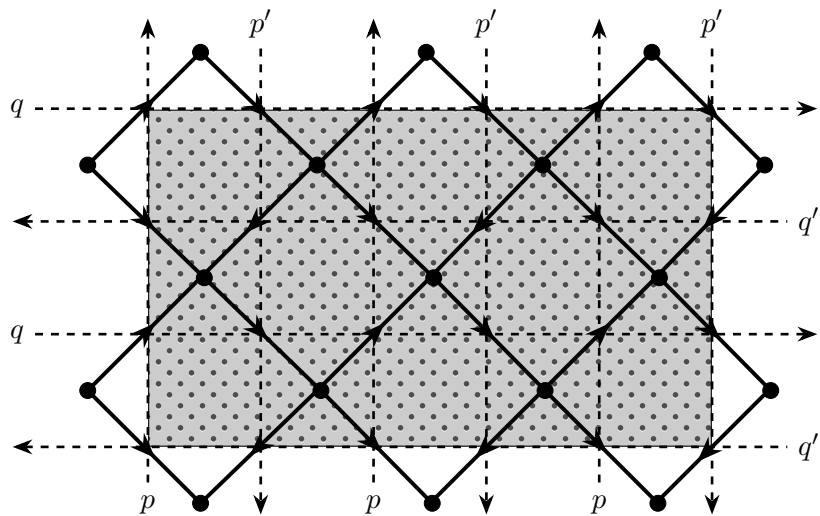

    \centering
    \LatticeDualLatticegCPM
    \caption{The solid and dashed lines represent the lattice $\mcal{L}$ and its median lattice $\mcal{L}'$, respectively, for gCPM. The spins are located on the square lattice while rapidities are located on its medial lattice.}
    \label{fig:lattice and dual lattice of gCPM}
\end{figure}

The spin variables are denoted as\footnote{What we will call $v_i$ corresponds to $v_i-v_{i+1}$ in the original article \cite{BazhanovKashaevMangazeevStroganov199105}.}
\begin{equation}
    \mbs{v}=(v_1,\ldots,v_n)\;, \qquad v_1,\ldots,v_n\in\mbb{Z}_N=\{0,\ldots,N-1\}\;,
\end{equation}
satisfying the relation
\begin{equation}
    \sum_{i=1}^n v_i=0\;, \quad (\text{mod}\,N)\;.
\end{equation}

The horizontal (vertical) lines of $\mcal{L}'$ carry rapidity (i.e.\ spectral) variables $p,p'$ ($q,q'$) in alternating order and also in different directions. The NW-SE edges of $\mcal{L}$ have the same NW-SE directions while NE-SW edges are in a checkerboard orientation. The rapidity variable lives on the algebraic curve of the spectral parameter, which we denote as $\cpc$. It is determined as the locus inside $\mbb{P}^{2n-1}$ determined by the following equations
\begin{equation}\label{eq:curve of spectral parameter of gCPM}
	\cpc:\qquad \begin{pmatrix}
		{z^+_{i}}^{N}
		\\
		{z^-_i}^{N}
	\end{pmatrix}
	= K_{ij}
	\begin{pmatrix}
		{z^+_j}^{N}
		\\
		{z^-_j}^{N}
	\end{pmatrix}\;,
	\qquad i,j=1,\cdots,n \;.
\end{equation}
where $(z^+_1:z^-_1:\ldots:z^+_n:z^-_n)$ are the set of $2n$ complex coordinates, and the $2\times 2$ matrices $K_{ij}$ satisfy
\begin{eqaligned}\label{eq:relations satisfied by matrices Kij of gCPM}
    \text{identity relation}&: \qquad \hphantom{\quad} K_{ii}=\mds{1}\;,
		\\
		\text{cocycle condition}&: \qquad K_{ij}K_{jk}K_{ki}=\mds{1}\;.
\end{eqaligned}
The cocycle condition implies that only $n-1$ of the matrices $K_{ij}$ with $i\ne j$ are independent. These $n-1$ matrices define $2n-2$ equations and hence \eqref{eq:curve of spectral parameter of gCPM} indeed defines a curve in $\Pbb^{2n-1}$. We can conveniently choose these to be the matrices $K_{i, i+1}$ for $i=1,\cdots,n-1$.  The relations \eqref{eq:relations satisfied by matrices Kij of gCPM} have obvious solutions 
\begin{equation}
    K_{ij}=M_iM_j^{-1}\;, \qquad i,j= 1,\ldots,n \;,
\end{equation}
with $M_i\in\text{SL}(2,\C)$. There are $n-1$ matrices $K_{i,i+1}$ each of which contributes three complex moduli, hence in total there are $3n-3$ complex parameters. However, the curve is invariant under the combined transformations
\begin{equation}\label{eq:redundancy in the curve of spectral parameter of gCPM}
    \begin{pmatrix}
        z_i^+
        \\
        z_i^-
    \end{pmatrix}\mapsto
    U_i 
    \begin{pmatrix}
        z_i^+
        \\
        z_i^-
    \end{pmatrix}\;, \qquad K_{ij}\mapsto U_i^NK_{ij}U_j^{-N}\;,
\end{equation}
with $U_i=\text{diag}(u_i,u_i^{-1})$. Removing these $n$ redundancies, we are left with $2n-3$ complex parameters that parameterize the moduli of the curve of spectral parameters. 

\smallskip For the gCPM, we only make the identification
\begin{equation}\label{eq:cyclic action of ZN(n-1) on coordinates of P2n-1}
	(z^+_{1},z^-_{1},\ldots,z^+_{n},z^-_{n})\sim \lambda(z^+_{1},z^-_{1},\ldots,z^+_{n},z^-_{n})\;,
    \qquad 
    \lambda\in\C^\times\;,
\end{equation}
i.e.\ these are homogeneous coordinates on $\Pbb^{2n-1}$. The equations \eqref{eq:curve of spectral parameter of gCPM} determines the algebraic curve $\cpc$ which is an $N^{n-1}$-fold covering of another curve, which is sometimes considered to be the curve of the spectral parameter of the gCMP. The latter is a further step quotient of $\cpc$ by a free action of $\mbb Z_{N}^{n-1}$:
\begin{eqaligned}\label{eq:ZNn-1 action on gCPM curve}
    (z^+_{1},z^-_{1}, \ldots,z^+_{n},z^-_{n})\mapsto (q^{k_1}z^+_{1},q^{k_1}z^-_{1},\ldots,q^{k_{n-1}}z^+_{n-1},q^{k_{n-1}}z^-_{n-1},z^+_{n},z^-_{n})\;,
\end{eqaligned}
where $q$ is a primitive $N$-th root of unity and $k_1,\cdots ,k_{n-1}\in \{0,1\cdots,N-1\}$.\footnote{The group of $N$\sth root of unity is isomorphic to $\mbb{Z}_N$.} We denote the actual curve of the spectral parameter of gCPM by $\cpc$ whose genus can be computed using the Riemann--Hurwitz formula as (see Appendix \ref{sec:basic facts about complete intersections} for some details)
\begin{equation}\label{eq:genus of the curve of spectral parameter of gCPM}
	g_{\cpc}=N^{2(n-1)}\left(N(n-1)-n\right)+1\;.
\end{equation}
The curve obtained by the $\mbb{Z}_N^{n-1}$-quotient of $\cpc$, where the action is given in \eqref{eq:ZNn-1 action on gCPM curve}, is denoted as $\pc$
\begin{equation}\label{eq:ZNxn-1 quotient of the curve of spectral parameter of gCPM}
    \pc:=\cpc/\mbb{Z}_N^{n-1}\;.
\end{equation}
 
\subsection{The Yang--Baxter Equation}

\smallskip Depending on the orientation of solid and dashed lines in Fig. \ref{fig:lattice and dual lattice of gCPM}, there are only two possible Boltzmann weights as follows

\begin{figure}[H]
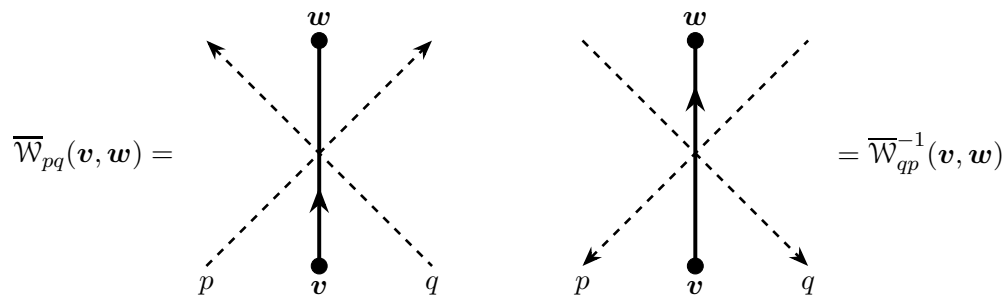

    \centering
    \BoltzmannWeightsgCPM
    \caption{The two possible Boltzmann weights of the gCPM.}
    \label{fig:possible Boltzmann weights of gCPMq}
\end{figure}

The Boltzmann weight $\overline{\mscr{W}}_{pq}(\mbs{v},\mbs{w})$ for rapidities $p$ and $q$ and the states $\mbs{v}$ and $\mbs{w}$ of the model can be written explicitly as follows
\begin{equation}
    \overline{\mscr{W}}_{pq}(\mbs{v},\mbs{w})=\omega^{Q(\mbs{v},\mbs{w})}\mscr{G}(\mbs{0},\mbs{v}-\mbs{w})\;,
\end{equation}
where $\omega:=\exp(2\pi\mfk{i}/N)$,
\begin{equation}
    Q(\mbs{v},\mbs{w}):=\sum_{i=1}^{n-1}v_i\sum_{j=1}^i(v_j-w_j) \;,
\end{equation}
and
\begin{equation}
    \mscr{G}(\mbs{0},\mbs{v}):=\frac{\prod_{I=1}^{|\mbs{v}|-1}\left(z_0^+(p)z_0^-(q)-z_0^+(q)z_0^-(p)\omega^I\right)}{\prod_{i=1}^{n-1}\prod_{I_i=0}^{v_i-1}\left(z_i^+(p)z_i^-(q)-z_i^+(q)z_i^-(p)\omega^{1+I_i}\right)} \;,
\end{equation}
with $|\mbs{v}|:=\sum_{i=1}^{n-1}v_i$ and $z^{\pm}_0(p)$ are local coordinates on the curve. $\overline{\mscr{W}}_{pq}(\mbs{v},\mbs{w})$ is defined up to an overall normalization. Finally, the R-matrix of the model is defined as the product of Boltzmann weights associated with the following elementary box of $\mcal{L}$:
\begin{figure}[H]
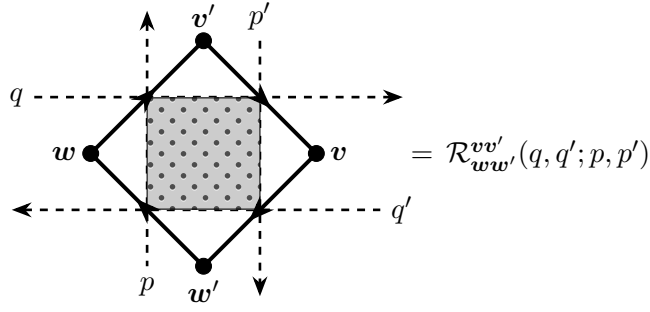

    \centering
    \ElementaryBoxgCPM
    \caption{The elementary box in the gCPM which gives the R-matrix of the model. The product of the Boltzmann weights as in \eqref{eq:R-matrix of gCPM} gives the R-matrix.}
    \label{fig:elementary box in gCPM which gives R-matrix}
\end{figure}
\noindent
and it is given by
\begin{equation}\label{eq:R-matrix of gCPM}
    \mcal{R}_{\mbs{ww'}}^{\mbs{vv'}}(q,q';p,p'):=\frac{\overline{\mscr{W}}_{qp}(\mbs{w},\mbs{v}')\overline{\mscr{W}}_{p'q}(\mbs{v}',\mbs{v})\overline{\mscr{W}}_{q'p'}(\mbs{v},\mbs{w}')}{\overline{\mscr{W}}_{q'p}(\mbs{w},\mbs{w}')} \;.
\end{equation}
Thinking of these R-matrices as intertwiners acting on a tensor product of some (representation) spaces $\mcal{V}_1\otimes \mcal{V}_2\otimes \mcal{V}_3$, one can show that they satisfy the usual form of the Yang--Baxter equation 
\begin{eqaligned}
    \mcal{R}_{12}(p, p'; q, q')\mcal{R}_{13}(p, p'; r, r')&\mcal{R}_{23}(q, q'; r, r')
    \\
    &=\mcal{R}_{23}(q,q';r,r')\mcal{R}_{13}(p,p';r,r')\mcal{R}_{12}(p,p';q,q') \;.
\end{eqaligned}

\section{Hyperbolic Monopoles and their Spectral Data}
\label{sec:hyperbolic monopoles and spectral data}
In this section, we review the construction of hyperbolic monopoles. We first explain some generalities on these monopoles. Next, we move to the elaboration of the twistor (or the Ward--Hitchin) correspondence for hyperbolic $\SU{n}$-monopoles with arbitrary boundary values of the Higgs field following \cite{MurraySinger199607,HurtubiseMurray1989} to which we refer for further details. The case with integer or half-integer values of the Higgs field, for the case of $\SU{2}$ has been treated in \cite{Atiyah1984,BraamAustin199008} while the case of $\SU{n}$ is considered in \cite{Chan201506} and in more details in \cite[\S 14]{Chan2017}.

\subsection{Monopoles on Hyperbolic Space}

We start with some generalities of hyperbolic monopoles. By a hyperbolic monopole, we mean a solution to the Bogomolny equation on the hyperbolic three space $\Hbb^3$. More concretely, consider a principal $G$-bundle over $\Hbb^3$ equipped with a connection $A$ and an adjoint-valued scalar $\phi$, the Higgs field. The Bogomolny equation reads
\begin{equation}\label{eq:Bogomolny equation}
    F=\star_{\Hbb^3}D\phi \;,
\end{equation}
where $F:=\rd A + A\wedge A$,
$D\phi:=\rd\phi+[A,\phi]_{\mfk{g}}$ and $\mfk{g}:=\text{Lie}(G)$.
We impose appropriate boundary conditions on the sphere at infinity  $S^2_\infty:=\partial\overline{\Hbb}^3$ with $\overline{\Hbb}^3$ denoting the closure of $\Hbb^3$.
The space of solutions of this equation, subject to appropriate boundary conditions and modulo gauge transformations, defines the moduli space of hyperbolic $G$-monopoles. 

\smallskip Choosing a local coordinate in the neighbourhood of $\partial\overline{\Hbb}^3$, we can expand the fields as
\begin{equation}\label{eq:expansion of higgs field and field strenght, near the boundary of H3}
    \phi=\phi_\infty+\mcal{O}(r^{n_\phi}) \;,
    \qquad F=F_\infty +\mcal{O}(r^{n_F}) \;,
\end{equation}
with $r$ being the radial coordinate away from $\partial\overline{\Hbb}^3$, $n_\phi$ and $n_F$ are positive numbers whose values depend on the boundary conditions, and we have used the notation
\begin{equation}
    \phi_\infty:=\phi\big|_{S^2_\infty} \;, \qquad 
    F_\infty:=F\big|_{S^2_\infty} \;. \qquad 
\end{equation}
A more precise version of this expansion will be given later (see \eqref{eq:expansion of connection near the boundary of H3}). 

\smallskip Hyperbolic monopoles have two characteristic parameters: (1) magnetic charges, (2) asymptotic values of the Higgs field $\phi$ or its masses \cite{Atiyah1984,BraamAustin199008,MurraySinger199607}. In the following, we explain how to compute these two parameters.

\hypertarget{magnetic charge of monopole}{\paragraph{Magnetic Charges of a Hyperbolic Monopole.}} The definition of the magnetic charges of a monopole depends on the breaking pattern of the gauge invariance at $S^2_\infty$ \cite{Weinberg198005,Weinberg198208,GoddardNuytsOlive1977,GoddardOlive197809,CorriganGoddard198108,WeinbergYi200609}.

\smallskip By an appropriate choice of basis, one can consider the Higgs field (and indeed any element of the Lie algebra $\mfk{g}$) in the Cartan subalgebra. We denote the generators of the
Cartan subalgebra of $\mfk{g}$ as $H_a,\,a=1,\ldots,\rnk{g}$, and write the asymptotic value of the Higgs field as
\begin{equation}\label{eq:expansion of higgs field at infinity}
    \phi_\infty=\sum_{i=1}^{\rnk{g}} p_iH_i \;,
    \qquad 
    p_1,\ldots,p_{\rnk{g}}\in\mbb{R} \;.
\end{equation}
Those generators of $\mfk{g}$ that do commute with $\phi_\infty$ form the remaining gauge invariance. As $\phi_\infty$ belongs to the chosen Cartan subalgebra, these unbroken generators could be of two types: (1) maximal symmetry-breaking: all the generators of the Cartan subalgebra and for all roots $\alpha$, $\alpha\cdot \mbs{p}=\sum_i \alpha_i p_i\ne 0$, with $\mbs{p}:=(p_1,\ldots,p_\rnk{g})$. In this case, the unbroken subgroup would be $\U{1}^{\rnk{g}}$; (2) Non-maximal symmetry-breaking: those generators $E_{\alpha}$ labeled by roots $\alpha$ such that $\alpha\cdot\mbs{p}=0$. Such generators would then generate a semi-simple Lie algebra $\mfk{k}$ whose corresponding group is $K$. The unbroken part of gauge invariance is then $K\times\U{1}^{\rnk{g}-\text{rnk}(\mfk{k})}$. In the following, we are exclusively considering the case of maximal symmetry-breaking. The reader interested in learning about the general case is referred to \cite[\S 5]{Murray198412}, \cite[\S 5]{GoddardOlive197809}, and \cite[\S 6.1]{WeinbergYi200609}.

\smallskip If the gauge group breaks from $G$ to a subgroup $H$, the vacuum manifold is $G/H$ \cite[\S 5.3-5.5]{GoddardOlive197809}. At infinity, the group of gauge transformations consists of maps $S^2_\infty\to G/H$\footnote{Note that this bundle, in this case, is trivial, as it is induced from a bundle on $\Hbb^3$, which is a contractible space and all bundles over which are trivial. If the three-space on which one studies the Bogomolny equation supports a non-trivial $G$-gauge bundle $P$, then the group of gauge transformations consists of global sections of $\text{Ad}P:=P\times_G G$.} and the non-trivial configurations are labeled by $\pi_2(G/H)$. Using the short exact sequence 
\begin{equation}
    0\to H\to G\to G/H\to 0 \;,
\end{equation}
and by assuming that $\pi_1(G)=0$ (by working with the universal cover of $G$ if necessary) and $\pi_2(G)=0$ (which holds for any connected Lie group), we have
\begin{equation}
    \pi_1(G/H)\simeq \pi_0(H) \;, 
    \qquad 
    \pi_2(G/H)\simeq \pi_1(H) \;.
\end{equation}
In the case of maximal symmetry-breaking, we have
\begin{equation}
    \pi_2(G/H)\simeq\pi_1(U(1)^{\rnk{g}})=\mbb{Z}^{\rnk{g}} \;,
\end{equation}
and the configuration is characterized by $\rnk{g}$ integer-valued topological quantum numbers, called magnetic charges. 

\smallskip For the maximal symmetry-breaking case, $H\simeq T$ with $T$ being the maximal torus of $G$, determining magnetic charges amounts to finding a basis for $\pi_2(G/H)$, which classifies the magnetic charges. In the case of maximal symmetry-breaking, there is a natural basis $\{b_1,\ldots,b_{\rnk{g}}\}$ of $\pi_2(G/H)$ labeled by simple roots $\{\alpha_1,\ldots,\alpha_{\rnk{g}}\}$ with respect to a choice of fundamental Weyl chamber. A convenient choice is obtained by fixing a point $p\in S^2_\infty$ which fixes a torus $T$, as the isotropy subgroup of $\phi_\infty(p)$, and we set the fundamental Weyl chamber to be the one that contains $\phi_\infty(p)$. A set of fundamental weights $\{\omega_1,\ldots,\omega_{\rnk{g}}\}$ would satisfy
\begin{equation}\label{eq:defining equation of fundamental weights}
    \frac{2\langle\omega_i,\alpha_j\rangle}{\langle\alpha_j,\alpha_j\rangle}=\delta_{ij} \;,
    \qquad 
    i,j=1,\ldots,\rnk{g}  \;. 
\end{equation}
Using the Killing form and identifying $\mfk{g}\simeq\mfk{g}^*$, we identify $H_i$ with the set of simple coroots $2\alpha_i/\langle\alpha_i,\alpha_i\rangle$, which are dual to the set of fundamental weights through \eqref{eq:defining equation of fundamental weights}. Hence,
\begin{equation}\label{eq:pairing of fundamental weights and simple coroots}
    \omega_i(H_j):=\langle\omega_i,H_j\rangle=\delta_{ij} \;, 
    \qquad 
    i,j=1,\ldots,\rnk{g} \;.
\end{equation}
Let $\mfk{t}:=\text{Lie}(T)$ be the Lie algebra of the maximal torus $T$ of $G$, and let $\Lambda\subset\mfk{t}^*$ be the weight lattice of $\mfk{g}$.\footnote{Strictly speaking, we should work with the complexification $G_\mbb{C}$ of $G$. But, when $H\simeq T$, then $G/H\simeq G_{\mbb{C}}/B$ for a Borel subgroup $B\subset G_\mbb{C}$ such that $G\cap B=H$.} Then, for any $\lambda\in\Lambda$, there is a one-dimensional representation $\lambda:T\to\U{1}$, and hence an associated line bundle $\mcal{L}_{-\lambda}:=G\times_{H}\U{1}$ to the principal $H$-bundle $G\to G/H$. It is defined as pairs $[g,x]$, for $g\in G$ and $x\in\U{1}$, satisfying the equivalence relation $(g,x)\sim(gh,\lambda(h)^{-1}x)$ over $G/H$ for any $h\in H$. One can define the isomorphism $\Lambda\otimes_{\mbb{Z}}\mbb{Q}\to H^2(G/H,\mbb{Q})$ which for a weight $\lambda$ takes the Chern class of $\mcal{L}_{-\lambda}$, i.e.\ $\lambda\to c(\mcal{L}_{-\lambda})$ \cite{BernsteinGelfandGelfand197306}. For a map $f:S^2_\infty\to G/H$ (or more precisely a homotopy class $[f]$ of such maps), there is a natural pairing $\langle\cdot,\cdot\rangle:\pi_2(G/H)\times H^2(G/H,\mbb{Q})\to H^2(S^2_\infty,\mbb{Q})\simeq\mbb{Q}$ given by
\begin{equation}
    \langle [f],c(\mcal{L}_{-\lambda})\rangle:=c(f^*\mcal{L}_{-\lambda}) \;,
\end{equation}
which in the basis $\{b_1,\ldots,b_{\rnk{g}}\}$ is given by \cite{BernsteinGelfandGelfand197306}
\begin{equation}\label{eq:the pairing of a basis element of pi2 with the chern-class of line bundle l}
    \langle b_i,c(\mcal{L}_{-\lambda})\rangle=\frac{2\langle\lambda,\alpha_i\rangle}{\langle\alpha_i,\alpha_i\rangle} \;,
    \qquad 
    i=1,\ldots,\rnk{g} \;.
\end{equation}
Note that $f^*\mcal{L}_{-\lambda}$ is a line bundle on $S^2_\infty$. The topological quantum numbers $\{m_1,\ldots,m_{\rnk{g}}\}$ can be defined by taking the class $[\phi_\infty]\in \pi_2(G/H)$ and set
\begin{equation}
    [\phi_\infty]=\sum_{i=1}^{\rnk{g}}m_i b_i \;.
\end{equation}
Using \eqref{eq:defining equation of fundamental weights}, \eqref{eq:the pairing of a basis element of pi2 with the chern-class of line bundle l}, and the fact that the pullback of $G\to G/H$ by $\phi_\infty$ is a principal $H$-bundle on $S^2_\infty$ with $H=T$, we have\footnote{Note that the Chern class of a line bundle $L$ is typically given by $c(L)=\frac{\mfk{i}}{2\pi}F$ with $F$ being the curvature of a connection on $L$. Here, we have chosen the normalization $c(L)=\frac{1}{2\pi}F$ instead.}
\begin{eqaligned}\label{eq:computation of magnetic charges}
        m_i&=\langle [\phi_\infty],c(\mcal{L}_{-\omega_i})\rangle \\
        &=c(\phi_\infty^*\mcal{L}_{-\omega_i}\rangle
        \\
        &=\frac{1}{2\pi}\bigintsss_{S^2_\infty} \omega_i(F_\infty)
        \\
        &=\frac{1}{4\pi}\bigintsss_{S^2_\infty}\rd\text{Vol}_{S^2_\infty} \omega_i(2\star_{S^2_\infty}F_\infty)
        \\
        &=\frac{1}{4\pi}\cdot 4\pi \cdot \omega_i(2\star_{S^2_\infty}F_\infty) \;,
\end{eqaligned}
which leads to 
\begin{equation}\label{eq:definition of magnetic charges of a monopole}
    m_i=\omega_i(2\star_{S^2_\infty}F_\infty) \;, 
    \qquad 
    i=1,\ldots,\rnk{g}  \;.
\end{equation}
Therefore, the $i$\textsuperscript{th} magnetic charge is the Chern number of the line bundle $\mcal{L}_{-\omega_i}$, associated with the $i$\textsuperscript{th} fundamental weight $\omega_i$, pulled back to $S^2_\infty$. 

\smallskip For the case of $G=\SU{n}$ with which we are mostly concerned, we take $\star_{S^2_\infty} F=\text{diag}(k_1,\ldots,k_n)/2$ for $k_1,\ldots,k_n\in\mbb{Z}$ and $\sum_i k_i=0$. 
Using the fact that for $\SU{n}$, $H_i=E_{i,i}-E_{i+1,i+1}$ with $E_{ij}$ being the elementary matrices, we have
\begin{equation}
    2\star_{S^2_\infty}F_\infty=\sum_{i=1}^n\left(\sum_{j=1}^ik_j\right)H_i \;,
\end{equation}
and hence using \eqref{eq:pairing of fundamental weights and simple coroots}, we arrive at
\begin{equation}\label{eq:magnetic charges for su(n) monopole}
    m_i=\omega_i(2\star_{S^2_\infty}F_\infty)=\sum_{j=1}^ik_j \;, 
    \qquad 
    i=1,\ldots,n \;.
\end{equation}
In particular, we have $m_n=\sum_{i=1}^n k_i=0$.

\paragraph{Masses of a Hyperbolic Monopole.} The second feature of a hyperbolic monopole is its masses. To define this notion, consider the expansion \eqref{eq:expansion of higgs field at infinity}. The $i$\textsuperscript{th} mass of the hyperbolic $G$-monopole is defined by the coefficient $p_i$ in \eqref{eq:expansion of higgs field at infinity}:
\begin{equation}\label{eq:definition of masses of hyperbolic monopoles}
    p_i=\omega_i(\phi_\infty) \;, 
    \qquad 
    i=1,\ldots,\rnk{g} \;.
\end{equation}

\subsection{Twistor Correspondence for Hyperbolic \texorpdfstring{$\SU{n}$}{SU(n)}-Monopoles}\label{sec:spectral data of hyperbolic su(n)-monopoles}

One of the most important aspects of the theory of hyperbolic monopole which best describes its integrability is the twistor correspondence. It states that there is a one-to-one correspondence between solutions of the Bogomolny equation on $\Hbb^3$ and certain holomorphic bundles over $\ts$, the twistor space of $\Hbb^3$. We begin in \S\ref{sec:twistor space of hyperbolic space} with the description of the twistor space of $\Hbb^3$, and then discuss the correspondence in \S\ref{sec:the twistor correspondence}. From now on, we will take $G=\SU{n}$. 

\subsubsection{Twistor Space of Hyperbolic Space}\label{sec:twistor space of hyperbolic space}

Let us describe the twistor space of hyperbolic three-space $\Hbb^3$. We provide two descriptions, first following \cite{Atiyah1984}, as a quotient space of some subspace of $\Pbb^3$ by $\ct$-action, and then by considering $\Hbb^3$ as the space of null rays in the upper cone of the Minkowski space, as described in \cite{MurraySinger199607}. We are mostly working with the latter description so we present it below, and relegate the first description to Appendix \ref{sec:more details on twistor space of hyperbolic space}.  More detailed studies of the space of geodesics on $\Hbb^3$ can be found in \cite{Georgiou2009,GeorgiouGuilfoyle200702}. 

\paragraph{Twistor Space from Null Rays in Minkowski Space.} A convenient description of the twistor space can be given in terms of thinking of $\Hbb^3$ as the space of rays of the future-pointing time-like vectors, denoted as $\mcal{U}$, in the Minkowski space $\R^{1,3}$. $\mcal{U}$, as an open cone in $\R^{1,3}$, is defined as
\begin{equation}\label{eq:definition of future-directed time-like vectors}
    \mcal{U}:=\left\{x:=(x_0,x_1,x_2,x_3)\in \R^{1,3}\,\big|\,-x_0^2+x_1^2+x_2^2+x_3^2<0,\,x_0>0\right\} \;.
\end{equation}
Furthermore, recall that the hyperboloid model of $\Hbb^3$ describes it as
\begin{equation}\label{eq:hyperboloid model of H3}
    \Hbb^3:=\left\{x\in \R^{1,3}\,\big|\,-x_0^2+x_1^2+x_2^2+x_3^2=-1\right\} \;.
\end{equation}
Therefore, if we consider the action of $\mbb{R}_+$, the set of positive real numbers, on $\mcal{U}$ given by rescaling
\begin{equation}\label{eq:r+ action on minkowski}
    x\to \lambda x \;,
    \qquad 
    \lambda\in\R_+ \;,
\end{equation}
the resulting quotient space can be identified with $\Hbb^3$. In the following, we use the notation
\begin{equation}\label{eq:length in minkowski space}
    |x|^2:=-x_0^2+x_1^2+x_2^2+x_3^2 \;, 
    \qquad 
    \forall x\in M  \;.
\end{equation}

\smallskip From the above description, $\ts$, the twistor space of $\Hbb^3$, can be described from $\tsu$, the twistor space of $\mcal{U}$, as the orbit space of the quotient of the latter by the $\R_+$-action. By definition, the twistor space of $\mcal{U}$ is the space of null geodesics (i.e.\ geodesics with null tangent vector) in $\mcal{U}$. A null geodesic in $\mcal{U}$ is the intersection of a straight line in $\R^{1,3}$, having a null tangent vector, with $\mcal{U}$. More precisely, $\tsu$ is defined as
\begin{equation}\label{eq:definition of twistor space of u}
    \tsu:=\left\{[z,w]\in\mbb{P}^3\big|\,z,w\in\mbb{C}^2,\,\langle z,w\rangle>0\right\} \;,
\end{equation}
where
\begin{equation}
    \langle z,w\rangle:=z^\dagger w=(\bmb{z}_0,\bmb{z}_1)\begin{pmatrix}
        w_0 \\ w_1 
    \end{pmatrix} \;.
\end{equation}

One considers a double fibration \cite{Wells197903,EastwoodPenroseWells198101}
\begin{equation}
\begin{tikzcd}\label{fig:correspondence space for u}
    & \arrow[ld,"\mu" swap] \mscr{C}_{\mcal{U}}\arrow[rd,"\nu"] &
    \\
    \tsu & & \mcal{U}
\end{tikzcd}
\end{equation}
where the correspondence space $\mscr{C}_{\mcal{U}}$ is defined as the space of all pairs $(\gamma,x)$, where $\gamma$ is a geodesic in $\mcal{U}$ together with a point  $x\in\gamma$
\begin{equation}\label{eq:correspondence space for U}
    \mscr{C}_{\mcal{U}}:=\left\{(x,\gamma)\in \mcal{U}\times\tsu\,\big|\,\gamma\in\tsu, \,x\in\gamma\right\} \;.
\end{equation}
If one identifies the Minkowski space with the space of Hermitian matrices through
\begin{equation}
    x\quad\longleftrightarrow \quad 
    \begin{pmatrix}
        x^0+x^3 & x^1-\mfk{i}x^2
        \\
        x^1+\mfk{i}x^2 & x^0-x^3
    \end{pmatrix} \;,
\end{equation}
then $\mscr{C}_{\mcal{U}}$ is defined through an incidence relation
\begin{equation}
    w=x\cdot z \;,
\end{equation}
where $\cdot$ is the matrix product, and the maps $\mu$ and $\nu$ in \eqref{fig:correspondence space for u} are restrictions of projections to the corresponding spaces. 

\smallskip Note that $\tsu$ is defined by one real condition \eqref{eq:definition of twistor space of u} in $\Pbb^3$, hence $\dim_{\R}\tsu=5$. While this space does not admit a complex structure, it instead admits a CR-structure, which describes a real hypersurface inside a complex manifold. One considers the CR-operator $\bar{\partial}_{\tsu}:\Omega^{0,\bullet}\to\Omega^{0,\bullet+1}$. A locally-defined function on $\tsu$ is called a CR function $f$ if $\bar{\partial}_{\tsu}f=0$. Given a CR-operator on $\tsu$, a CR-structure on $\tsu$ is induced by transition functions which are CR functions on $\tsu$. Next, consider a vector bundle $E\to\tsu$. A CR-operator $\bar{\partial}_E$ on $E$ is an operator which (1) satisfies the Leibniz rule $\bar{\partial}_E(f e)=\bar{\partial}_{\tsu}fe+f\bar{\partial}_Ee$ for any function $f$ on $\tsu$ and $e\in\Gamma(E,\tsu)$,\footnote{In the following, $\Gamma(E, M)$ would denote the space of sections of a vector bundle $E$ over $M$.} and (2) $\bar{\partial}_E^2=0$. One then says $\bar{\partial}_E$ is integrable and $(E,\bar{\partial}_E)$ is a CR vector bundle. 

\smallskip To define the twistor space of $\Hbb^3$, we need to consider the $\mbb{R}_+$-invariant part of $\tsu$. Hence, we first define the action of $\mbb{R}_+$ on $\tsu$ as (compare it with \eqref{eq:the C* action on CP3})
\begin{equation}\label{eq:R+ action on the twistor space of U}
    \lambda\cdot (z,w):=(\lambda^{-\frac{1}{2}}z,\lambda^{+\frac{1}{2}}w) \;, 
    \qquad 
    \forall\lambda\in\mbb{R}_+, \quad \forall (w,z)\in\tsu \;.
\end{equation}
The action has two obvious fixed lines $\Pbb^1_+$ and $\Pbb^1_-$ (see \eqref{eq:the fixed lines of the C* action on CP3}). Hence, the twistor space of $\Hbb^3$ can be then defined as 
\begin{equation}\label{eq:definition of twistor space of H3, second version}
    \ts:=\tsu/\mbb{R}_+=\left\{([z],[w])\in\mbb{P}^1_+\times\Pbb^1_-\,\big|\,\langle z,w\rangle\ne 0\right\},
\end{equation}
which leads to the following double fibration
\begin{equation}
\begin{tikzcd}\label{fig:correspondence space for twistor space}
    & \arrow[ld,"\mu" swap] \mscr{C}_{\Hbb^3}\arrow[rd,"\nu"] &
    \\
    \ts & & \Hbb^3
\end{tikzcd}
\end{equation}
where the correspondence space is defined analogous to $\mscr{C}_{\mcal{U}}$
\begin{equation}
    \mscr{C}_{\Hbb^3}:=\left\{(x,\gamma)\in\Hbb^3\times\ts\,\big|\,\gamma\in\ts,\,x\in\gamma\right\} \;.
\end{equation}

\smallskip Recall that geodesics on $\Hbb^3$ are parameterized with their end-points on $S^2_\infty$.
Let us denote the coordinates on $\Pbb^1_\pm$ as $z:=z_1/z_2$ and $w:=z_3/z_4$. From \eqref{eq:projection of fixed-lines of C*-action on CP3} and \eqref{eq:small and large limits of lambda}, a point $(z,w)$ of the minitwistor space \eqref{eq:minitwistor space of hyperbolic space} corresponds to a geodesics that starts at $z\in S^2_\infty$ and ends at $\bmb{w}\in S^2_\infty$. We would like to avoid geodesics that start and end at the same point. This motivates the introduction of a real structure $\sigma:\Pbb^1_-\times\Pbb^1_+\to\Pbb^1_-\times\Pbb^1_+$, induced from \eqref{eq:real structure on P3}
\begin{equation}\label{eq:real structure on twistor space}
    \sigma(z,w):=(\bmb{z},\bmb{w}) \;.
\end{equation}
This means that the transformed geodesic associated with the point $(\bmb{z},\bmb{w})$ starts at $\bmb{z}$ and ends at $w$. Therefore, the set of all geodesics that start and end at the same point are related by $\sigma$ and correspond to $(z,\bmb{z})\in\Pbb^1_-\times\Pbb^1_+$. We define the set of all such geodesics as
\begin{equation}\label{eq:set of geodesics that start and end at the same point}
    \Delta:=\left\{(z,w)\in\Pbb^1_+\times\Pbb^1_-\,\big|\,\bmb{w}=z\right\} \;.
\end{equation}
Excluding these points from the minitwistor space \eqref{eq:minitwistor space of hyperbolic space}, we end up with the twistor space of $\Hbb^3$
\begin{equation}
	\ts:= \Pbb_-^1\times\Pbb^1_+-\Delta \;,
\end{equation}
which is equipped with a real structure defined by 
$\sigma$.

\smallskip A line $\mbf{L}\subset\ts$ is called a real line if it is invariant under  $\sigma$
\begin{equation}\label{eq:definition of real lines}
    \sigma(\mbf{L})=\mbf{L} \;,
\end{equation}
or in more detail
\begin{equation}
        \sigma(\mbf{L}(z,w)):=\mbf{L}(\sigma(z,w))=\mbf{L}(\bar{z},\bar{w})=\mbf{L}(z,w) \;.
\end{equation}
In other words, a line is real if and only if $(z,w)\in\mbf{L}$ implies $\sigma(z,w)=(\bmb{z},\bmb{w})\in\mbf{L}$. 

\hypertarget{some relevant bundles}{\paragraph{Some Relevant Bundles.}} Bundles on spaces appearing in the double fibration \eqref{fig:correspondence space for u} and their $\mbb{R}_+$-invariant data play the main role in the twistor correspondence and also define the spectral data of hyperbolic monopoles. We thus spend some time on introducing some relevant bundles. 

\smallskip Holomorphic line bundles on $\ts$ are constructed by the pullback from the factors in $\ppm$. For $s_\mp\in\mbb{Z}$, we have\footnote{There are some differences between our conventions and the ones in \cite{MurraySinger199607}. For example, what we have denoted as $\mcal{O}_{\ts}(s_-,s_+)$ is denoted as $\mcal{O}(s_+,s_-)$ in loc.\ cit.}
\begin{equation}\label{eq:twisted structure sheaf of twistor space}
    \mcal{O}_{\ts}(s_-,s_+):=\left.\pi^*_{-}(\mcal{O}_{\mbb{P}^1_-}(s_-))\otimes\pi^*_+(\mcal{O}_{\mbb{P}^1_+}(s_+))\right|_{\ts}  \;,
\end{equation}
with the projections $\pi_\pm:\ppm\to\mbb{P}^1_\pm$, and then restriction to $\ts$. An important line bundle on $\ts$ is the (topologically) trivial one\footnote{In \cite[\S 3.1, Example 3]{MurraySinger199607}, this line bundle is denoted as $\wt{L}$.}
\begin{equation}\label{eq:line bundle L}
    L:=\mcal{O}_{\ts}(-1,1) \;.
\end{equation}
The space of sections of the corresponding bundle on $\tsu$ are CR-functions $f(z,w)$ satisfying
\begin{equation}
    f(\lambda^{-1/2} z,\lambda^{+1/2} w)=\lambda f(z,w) \;,
    \qquad 
    \forall \lambda\in\C^\times \;.
\end{equation}
More generally, for $p\in\mbb{C}$ and $q\in\mbb{Z}$, the most general holomorphic line bundles on $\ts$ have the generic form $L^p(q)$. A section of this line bundle corresponds to a CR-section $f$ of the corresponding bundle on $\tsu$ satisfying 
\begin{equation}\label{eq:sections of lp(q)}
    f(\zeta\lambda^{-1/2}z,\zeta\lambda^{+1/2}w)=\zeta^q\lambda^pf(z,w) \;, 
    \qquad 
    \mu,\lambda\in\C^\times \;. 
\end{equation}
These bundles on $\ts$ are isomorphic to $\mcal{O}_{\ts}(s_-,s_+)$ whose sections transform as
\begin{eqaligned}\label{eq:sections of o(q+,q-)}
    f(\zeta\lambda^{-1/2} z,\zeta\lambda^{+1/2} w)&=\left(\zeta\lambda^{-1/2}\right)^{s_-}\left(\zeta\lambda^{+1/2}\right)^{s_+}f(z,w)
    \\
    &=\zeta^{s_++s_-}\lambda^{\frac{s_+-s_-}{2}}f(z_-,z_+) \;.
\end{eqaligned}
Comparing with $\eqref{eq:sections of lp(q)}$, we conclude that
\begin{equation}\label{eq:relation between L and O bundles}
    p=\frac{s_+-s_-}{2} \;,
    \qquad  
    q=s_++s_- \;.
\end{equation}
We thus have
\begin{equation}\label{eq:Lp(q) in terms of O(q_+,q_-)}
    L^p(q)\simeq \mcal{O}_{\ts}((q-2p)/2,(q+2p)/2) \;.
\end{equation}
In particular, we have a topologically trivial bundle for $s_++s_-=0$, such as $L$ in \eqref{eq:line bundle L}. Furthermore, any tensor power of $L^{p}$ is trivial. From \eqref{eq:relation between L and O bundles}, we can identify
\begin{eqaligned}\label{eq:basis for holomorphic line bundles}
    L^{\frac{1}{2}}(1):=\mcal{O}_{\ts}(0,1) \;, 
    \qquad 
    L^{-\frac{1}{2}}(1):=\mcal{O}_{\ts}(1,0) \;.
\end{eqaligned}
These bundles are (topologically) isomorphic: $L^{-p}$ and $L^p$ are both topologically trivial, and hence $L^{p}(q)\simeq L^{-p}(q)$. In particular, 
\begin{equation}\label{eq:topological equivalence of O(n,0) and O(0,n)}
    L^{n/2}(n)\simeq\mcal{O}_{\ts}(0,n)\simeq\mcal{O}_{\ts}(n,0)\simeq L^{-n/2}(n) \;.
\end{equation}
By comparing \eqref{eq:sections of lp(q)} and \eqref{eq:sections of o(q+,q-)} shows that\footnote{While the two sides are bundles over two different bases, one should think of \eqref{eq:identification of OZH3(p,p) with OZU(2p)} as the correspondence between holomorphic vector bundles on $\ts$ and CR-vector bundles on $\tsu$.}
\begin{equation}\label{eq:identification of OZH3(p,p) with OZU(2p)}
    \mcal{O}_{\ts}(p,p)\longleftrightarrow \mcal{O}_{\tsu}(2p) \;.
\end{equation}

\smallskip We need some basic results about these bundles. Consider the following diagram associated with $\sigma:\ts\to\ts$
\begin{equation}\label{eq:pullback of real line with sigma}
    \begin{tikzcd}
        \sigma^*\mcal{O}_{\ts}(m,n) \arrow[r]\arrow[d,"\pi" swap] & \arrow[d,"\pi"]\mcal{O}_{\ts}(m,n)
        \\
        \ts \arrow[r,"\sigma" swap] & \ts 
    \end{tikzcd},
\end{equation}
where 
\begin{equation}
    \sigma^*\mcal{O}_{\ts}(m,n)={\mcal{O}}_{\ts}(m,n) \;.
\end{equation}
Note that the complex structure on $\sigma(\ts)$ is the complex conjugate of $\ts$, and a section of $\mcal{O}_{\ts}(m,n)$ on $\sigma(\ts)$ would transform as $\bmb{z}^m\bmb{w}^n$. Pull this back with $\sigma$ gives a section which transforms with $z^m w^n$, and hence a section of $\mcal{O}_{\ts}(m,n)$. 

\smallskip Consider a real line defined by \eqref{eq:definition of real lines}, for which we have
\begin{equation}
    \begin{tikzcd}
        \sigma^*\mcal{O}_{\mbf{L}}(m,n) \arrow[r]\arrow[d,"\pi" swap] & \arrow[d,"\pi"]\mcal{O}_{\mbf{L}}(m,n)
        \\
        \mbf{L} \arrow[r,"\sigma" swap] & \sigma(\mbf{L})=\mbf{L}
    \end{tikzcd}.
\end{equation}
The bundle $\mcal{O}_{\mbf{L}}(m,n)$ is defined by the embedding $\iota_{\mbf{L}}:\mbf{L}\hookrightarrow\ts$ as $\mcal{O}_{\mbf{L}}(m,n):=\iota^*_{\mbf{L}}\mcal{O}_{\ts}(m,n)$. In particular, on a real line $\mbf{L}$, the isomorphism \eqref{eq:pullback of real line with sigma} holds. It furthermore implies
\begin{equation}\label{eq:invariance of O(n,n) under sigma}
    \sigma^*\mcal{O}_{\mbf{L}}(n,n)\simeq\mcal{O}_{\mbf{L}}(n,n) \;.
\end{equation}

\begin{rmk}\normalfont
    As we will explain below, spectral curves $S$ of hyperbolic $\SU{n}$-monopoles are real lines which are sections of $\mcal{O}_{S}(m,m)$ for some $m$. Therefore, this bundle should be invariant under $\sigma$, a fact which is implied by \eqref{eq:invariance of O(n,n) under sigma}.
\end{rmk}

In the following, the bundle $\wt{L}\to\mcal{U}$ would denote the bundle of homogeneous functions on $\mcal{U}$ of degree $+1$, i.e.\ all functions satisfying $f(\lambda X)=\lambda f(X)$ for $\lambda\in\mbb{R}_+$. Note that $\lambda$ cannot be negative since it then maps $X_0>0$ to $\lambda X_0<0$, and hence does not preserve $\mcal{U}$.  

\smallskip A holomorphic vector bundle $E\to\ts$ is called non-degenerate if the restriction of $E$ to any real line in $\ts$ is a holomorphically-trivial bundle. Furthermore, $E$ is called a real bundle if there is an anti-holomorphic involution $\sigma_E:E\to E^*$, where $E^*$ denotes the conjugate vector bundle of $E$, that covers the map $\sigma:\ts\to\ts$. 

\smallskip Next consider a vector bundle $V\to\mscr{C}_{\mcal{U}}$. A $\bar{d}$-operator on $V$ is a linear differential operator $\bar{d}_V:\Omega^{p,q}\to\Omega^{p,q+1}$, which satisfies the Leibniz rule. $(V,\bar{d}_V)$ is called an integrable bundle if $\bar{d}^{\,^2}_V=0$, and is called non-degenerate if its restriction to a fiber of $\nu$ in \eqref{fig:correspondence space for u}, which is isomorphic to $\Pbb^1$, is a holomorphically-trivial bundle. 

\smallskip Finally, we point out that holomorphic bundles over the twistor space $\ts$ are key parts of the spectral data of hyperbolic monopoles. The holomorphic structure on vector bundles on $V\to\ts$ is induced by the CR-structure on the corresponding bundle on $\tsu$. The reason for this one-to-one correspondence is \cite[\S 3.1, Example 2]{MurraySinger199607}
\begin{equation}
    \Lambda^{0,1}_{\tsu}\simeq\pi^*\Lambda^{0,1}_{\ts} \;,
\end{equation}
where the isomorphism is given by the pullback. Hence, a section of $V$ is holomorphic if and only if the section of the corresponding bundle on $\tsu$ is CR. 

\subsubsection{The Correspondence} \label{sec:the twistor correspondence}
We now discuss the twistor correspondence, i.e.\ the one-to-one correspondence between solutions to the Bogomolny equation on $\Hbb^3$ for the gauge group $\SU{n}$ and holomorphic vector bundle on $\ts$.

\smallskip The basic idea of the twistor correspondence is encoded in the double fibration \eqref{fig:correspondence space for u}: mathematical structures on $\mcal{U}$ are pulled-back by $\nu$ to $\mscr{C}_{\mcal{U}}$ and then sent to the corresponding mathematical structures on $\tsu$ by $\mu$. Consider the self-duality equation on $\mcal{U}$, which is a connection on a principal $G$-bundle on $\mcal{U}$. It is more convenient to work with the corresponding associated vector bundle equipped with a self-dual connection. The main result is that there is a one-to-one correspondence between such vector bundles and integrable bundles on $\mscr{C}_{\mcal{U}}$, and in turn, the latter bundles are in one-to-one correspondence with CR-bundles on $\tsu$. Restricting to $\mbb{R}_+$-invariant data will provide the twistor correspondence between the solutions of the Bogomolny equation on $\Hbb^3$ and holomorphic vector bundles on $\ts$.

\hypertarget{correspondence between bundles over the double fibration}{\paragraph{Correspondence Between Bundles over the Double Fibration \eqref{fig:correspondence space for u}.}} We now briefly explain the correspondence between bundles with additional structures defined over spaces in the double fibration \eqref{fig:correspondence space for u}.

\smallskip Consider an integrable bundle $(\wt{E}_{\mscr{C}_{\mcal{U}}},\bar{d}_{\wt{E}_{\mscr{C}}})$ on $\mscr{C}_{\mcal{U}}$. Such a bundle naturally corresponds to a CR-bundle $E\to \tsu$ where the CR-structure is determined by the CR-operator $\bar{\partial}_E$ as follows: One first chooses a section $s:\tsu\to \mscr{C}_{\mcal{U}}$, and then set
\begin{equation}\label{eq:bundles on ZU from bundles on CU}
    E_s:=s^*\wt{E}_{\mscr{C}_{\mcal{U}}} \;, 
    \qquad 
    \bar{\partial}_{E_s}:=s^*\bar{d}_{\wt{E}_{\mscr{C}}} \;.
\end{equation}
For two different sections, say $s_1$ and $s_2$ of $\tsu\to\mscr{C}$, there is a bundle map between bundles constructed by $s_1$ and $s_2$ as \eqref{eq:bundles on ZU from bundles on CU}, which is given by the parallel transport of $\bar{d}_{\wt{E}}$ along the fibers of $\mu$ in \eqref{fig:correspondence space for u}. On the other hand, having a CR-bundle $E\to\tsu$ equipped with the CR-operator $\bar{\partial}_{E}$, a bundle $(\wt{E},\bar{d}_{\wt{E}})$ on $\mscr{C}$ is constructed by the pull-back operation
\begin{equation}
    \wt{E}:=\mu^* E \;,
    \qquad 
    \bar{d}_{\wt{E}}:=\mu^* \partial_{E} \;.
\end{equation}
There is a one-to-one correspondence between CR-bundles on $\tsu$ and integrable bundles on $\mscr{C}$, all modulo corresponding equivalences. For further details see \cite[Proposition 3.2]{MurraySinger199607}.

\smallskip Next, consider a vector bundle $\wt{E}\to\mcal{U}$ equipped with a self-dual connection $\wt{A}$. Then, one can define the bundle $\wt{E}_{\mscr{C}_{\mcal{U}}}\to\mscr{C}_{\mcal{U}}$, with $\wt{E}_{\mscr{C}_{\mcal{U}}}:=\nu^*\wt{E}$. The pulled-back connection induces a $\bar{d}$-operator on $\wt{E}_{\mscr{C}_{\mcal{U}}}$, which we denote as $\bar{d}_{\wt{E}_{\mscr{C}}}$. This bundle is integrable if and only if
\begin{equation}\label{eq:curvature of connection on C}
    \nu^*\wt{F}\in\Gamma(\mu^*\Lambda^{1,0}_{\tsu}\wedge\Lambda^1_{\mscr{C}_{\mcal{U}}}) \;,
\end{equation}
with $\wt{F}$ being the curvature of $\wt{A}$. Since $\wt{A}$ is self-dual, this is always the case. On the other hand, by construction, any vector bundle constructed from the pull-back of a vector bundle on $\mcal{U}$ would automatically be non-degenerate. This shows that there is a one-to-one correspondence between vector bundle equipped with a self-dual connection $(\wt{E},\wt{A})$ and non-degenerate integrable vector bundles $(\wt{E}_\mscr{C},\bar{d}_{\wt{E}_{\mscr{C}}})$ given explicitly by 
\begin{equation}\label{eq:one-to-one correspondence between self-dual connections on U and inegrable bundles on correspondence space}
    \wt{E}_{\mscr{C}}:=\nu^*\wt{E} \;, 
    \qquad 
    \bar{d}_{\wt{E}_{\mscr{C}}}:=\nu^*\mcal{D}_{\wt{A}} \;,
\end{equation}
with $\mcal{D}_{\wt{A}}$ is the covariant derivative defined by $\wt{A}$. This correspondence holds modulo corresponding equivalences.  For further details see \cite[Proposition 3.3]{MurraySinger199607}.

\smallskip Putting these two results together, one concludes that there is a one-to-one correspondence between rank-$n$ vector bundle on $\mcal{U}$ equipped with a self-dual connection and CR-bundles on $\tsu$. More concretely, one consider a section $f:\tsu\to\mcal{U}$ which assigns the point $f(\gamma)\in\gamma\subset\mcal{U}$ to $\gamma\in\tsu$. The choice of such section will provide a section $s$ of the projection $\mu$ in \eqref{fig:correspondence space for u} by setting 
\begin{equation}
    s(\gamma):=(\gamma,f(\gamma)) \;, 
    \qquad 
    \forall\gamma\in\tsu  \;.
\end{equation}
Conversely, having such a section $s$, one can construct a section $f:\tsu\to\mcal{U}$ by setting $f=\nu\circ s$.

\smallskip Now assume $\wt{E}\to\mcal{U}$ is a vector bundle equipped with a self-dual connection $\wt{A}$. The corresponding CR-bundle $E\to\tsu$ can be defined by
\begin{equation}\label{eq:inverse ward transform, the bundle}
    E:=f^*\wt{E} \;,
\end{equation}
where the CR-structure on $E$ is induced by the CR-operator $\bar{\partial}_E$
\begin{equation}\label{eq:inverse ward transform, the cr operator}
    \bar{\partial}_{E}:=(f^*\mcal{D}_{\wt{A}})^{0,1} \;.
\end{equation}
The CR-structure is defined as usual: a section $\psi$ of $E$ is CR if $\bar{\partial}_E\psi=0$. The two relations \eqref{eq:inverse ward transform, the bundle} and \eqref{eq:inverse ward transform, the cr operator} are called the Inverse Ward Transform (IWT) \cite{MurraySinger199607}.

\paragraph{$\mbb{R}_+$-Invariant Data.} The last step of establishing the twistor correspondence for hyperbolic $\SU{n}$-monopoles is to extract the $\R_+$-invariant part of the IWT (i.e.\ relations \eqref{eq:inverse ward transform, the bundle} and \eqref{eq:inverse ward transform, the cr operator}). Recall that the $\R_+$-action on $\mcal{U}$ is given by \eqref{eq:r+ action on minkowski}. The crucial point is that the pullback under $\pi:\mcal{U}\to\Hbb^3$ commutes with the $\R_+$-action. This will give rise to the following picture: Consider a vector bundle $\wt{E}_{\mcal{U}}\to\mcal{U}$ equipped with a self-dual connection $\wt{A}$. The $\mbb{R}_+$-invariant part of the data of this vector bundle gives rise to a rank-$n$ vector bundle $\wt{E}_{\Hbb^3}\to\Hbb^3$ together with a pair $(\wt{A}_+,\phi)$, where $\wt{A}_+$ is a connection on $\wt{E}_{\Hbb^3}$ and $\phi$\footnote{Murray and Singer use $\mfk{i}\phi$ for the self-adjoint Higgs field \cite{MurraySinger199607}.} is an adjoint-valued scalar, satisfying the Bogomolny equation
\begin{equation}
    \wt{F}_+=\star_{\Hbb^3}{\mcal{D}_{\wt{A}_+}}\phi \;,
\end{equation}
where $\wt{F}_+={\msf{d}}_{\Hbb^3}\wt{A}_++\frac{1}{2}[\wt{A}_+,\wt{A}_+]$. More explicitly, having the vector bundle $(\wt{E}_{\Hbb^3},\wt{A}_+,\phi)$, a vector bundle $(\wt{E}_\mcal{U},\wt{A})$, with a self-dual $\wt{A}$, can be constructed where
\begin{equation}
    \wt{E}_{\mcal{U}}=\pi^*\wt{E}_{\Hbb^3} \;,
\end{equation}
and
\begin{equation}
    \mcal{D}_{\wt{A}}=\pi^*\mcal{D}_{\wt{A}_+}+\pi^*\phi\,\rd_{\Hbb^3}\log|x| \;,
\end{equation}
with $|x|$ is defined in \eqref{eq:length in minkowski space}. The $\R_+$-invariance of the connection from this expression is easy to check. Conversely, having the vector bundle $(\wt{E}_{\mcal{U}},\wt{A})$, one can construct the corresponding vector bundle $(\wt{E}_{\Hbb^3},\wt{A}_+,\phi)$ by first choosing a section $s:\Hbb^3\to\mcal{U}$ of the projection $\pi$, and setting
\begin{equation}\label{eq:bundles on H3 from bundles on U}
    \wt{E}_{\Hbb^3,s}:=s^*\wt{E}_{\mcal{U}} \;,
\end{equation}
For two sections $s$ and $s'$, we get $\wt{E}_{\Hbb^3,s}$ and $\wt{E}_{\Hbb^3,s'}$, and we have 
\begin{equation}
    \wt{E}_{\Hbb^3,s}=s^*\wt{E}_{\mcal{U}}=s^*\circ\pi^*\wt{E}_{\Hbb^3,s'}=\wt{E}_{\Hbb^3,s'} \;,
\end{equation}
so the two vector bundles are isomorphic. 

\hypertarget{reality condition}{\paragraph{The Reality Condition.}} The bundle of (anti)self-dual two-forms on $\R^{1,3}$ is not real, hence to get real solutions of the self-duality equation, one needs to impose additional reality conditions on the gauge field $\wt{\mcal{A}}$ on $\mcal{U}$. This means that there would not be a real solution of the Bogomolny equations by looking into the $\R_+$-invariant data. To get a real solution, one needs to impose appropriate real conditions on the gauge field $\wt{A}_+$ and the Higgs field $\phi$. An appropriate real condition turns out to be
\begin{equation}\label{eq:reality condition on gauge and Higgs fields}
    \wt{A}^*_+=-\wt{A}_+ \;, \qquad \phi^*=\phi \;.
\end{equation}
This condition can be expressed in a more invariant manner, which the reader can find in \cite[\S 3.3]{MurraySinger199607}. Furthermore, it imposes a condition on the corresponding CR bundle $E$ on $\tsu$ defined in \eqref{eq:inverse ward transform, the bundle}, which can be found in loc.\ cit.  

\paragraph{Boundary Conditions.} Up to now, we have explained the twistor correspondence for hyperbolic $\SU{n}$-monopoles, through which rank-$n$ holomorphic vector bundles (with a real structure) on $\ts$ is assigned to a solution of the Bogomolny equation on $\Hbb^3$. However, to be able to solve the Bogomolny equation, one needs to impose boundary conditions. In turn, these boundary conditions make sure that the bundle $E$, defined through the inverse Ward transform \eqref{eq:inverse ward transform, the bundle}, exists. Furthermore, boundary conditions, through the twistor correspondence, would lead to two filtrations of $E$, which are the main ingredients in defining the spectral data of the monopole later on. Let us explain these ideas in some detail.

\smallskip Let $\wt{E}_{\Hbb^3}\to\Hbb^3$ be a rank-$n$ vector bundle equipped with a connection $\wt{A}_+$ and an $\text{ad}\, E$-valued Higgs field $\phi$. The corresponding vector bundle on $\mcal{U}$ is denoted as $\wt{E}_\mcal{U}$. On this data, we impose the following boundary conditions:

\hypertarget{bc1}{\paragraph{\small Boundary Condition 1.}} The first boundary condition is that the asymptotic value of the Higgs field exists on $\partial\overline{\Hbb}^3$, and it is given by \eqref{eq:definition of masses of hyperbolic monopoles}
\begin{equation}
    \phi_\infty=
    \begin{pmatrix}
        p_1
        \\
         & \ddots 
         \\
         & & p_n
    \end{pmatrix} \;, 
    \qquad 
    p_i\in\R  \;,
    \qquad 
    \sum_{i=1}^np_i=0  \;,
\end{equation}
where we hereafter assume without losing generality that
$p_1\ge p_2 \ge \dots \ge p_n$.

To specify other boundary conditions, we note that the gauge group $\SU{n}$ breaks to $\U{1}^{n-1}$ on $\partial\overline{\Hbb}^3$, hence in a neighborhood $\mscr{N}$ of $\partial\overline{\Hbb}^3$ homotopic to $S^2$, 
where the vector bundle $\wt{E}_{\mcal{U}}$ can be written as
\begin{equation}\label{eq:vector bundle EU as a direct sum of EUi}
    \wt{E}_\mcal{U}\Big|_{\mscr{N}}=\bigoplus_{i=1}^n \wt{E}_{\mcal{U},i} \;,
\end{equation}
with each $\wt{E}_{\mcal{U},i}$ is a one-dimensional complex line bundle defined by
\begin{equation}\label{eq:definition of 1d bundles tildeE on U}
    \wt{E}_{\mcal{U},i}:\quad \det(\phi-\lambda_i)=0 \;, \qquad i=1,\ldots,n \;,
\end{equation}
and we have
\begin{equation}
    \lambda_i\Big|_{\partial\overline{\Hbb}^3}=p_i \;, \qquad i=1\,\ldots,n \;.
\end{equation}
Furthermore, the connection $\wt{A}_+$ upon projection into the $i$\textsuperscript{th} factor of the summand \eqref{eq:vector bundle EU as a direct sum of EUi} gives a $\U{1}$-connection on $\wt{E}_{\mcal{U},i}|_{\partial\overline{\Hbb}^3}$. As $\overline{\mscr{N}}$, the closure of $\mscr{N}$, is homeomorphic to $S^2$, one can integrate the curvature of $\wt{A}_i$ to obtain the degree of $\wt{E}_{\mcal{U},i}$ 
\begin{equation}
    k_i:=\text{deg}\, \wt{E}_{\mcal{U},i}=\frac{\mfk{i}}{2\pi}\bigintsss_{\overline{\mscr{N}}}F_{\wt{A}_i}=\bigintsss_{\overline{\mscr{N}}}c_1({\wt{E}_{\mcal{U},i}}) \;, 
    \qquad 
    i=1,\ldots,n \;,
\end{equation}
where $F_{\wt{A}_i}$ denotes the curvature of $\wt{A}_i$, and $c_1({\wt{E}_{\mcal{U},i}})$ denotes the first Chern class of ${\wt{E}_{\mcal{U},i}}$. Noting that $c_1(\oplus_{j=1}^i\wt{E}_{\mcal{U},j})$ is $\sum_{j=1}^ik_j$, comparing with \eqref{eq:computation of magnetic charges} shows that $\sum_{j=1}^i k_i$, by construction, can be identified with \eqref{eq:magnetic charges for su(n) monopole}, the $i$\textsuperscript{th} magnetic charge of the monopole.

\smallskip The connection on $\wt{E}|_{\mscr{N}}$ would then have the following expansion
\begin{equation}\label{eq:expansion of connection near the boundary of H3}
    \nabla_{\wt{A}_+}=\bigoplus_{i=1}^n\nabla_{\wt{A}_i}+\text{diag}(\lambda_1,\ldots,\lambda_n)\rd_{\mcal{U}}\log|x|+ \mscr{E} \;,
\end{equation}
with  $\mscr{E}\in\Lambda^1(\overline{\mscr{N}},\text{Hom}(\wt{E}_{\mcal{U}}))$, and $x$ denotes a local coordinate on $\overline{\mscr{N}}$. The components of $\mscr{E}$ are denoted as $\mscr{E}_{ij}\in\Lambda^1(\overline{\mscr{N}},\text{Hom}(\wt{E}_{\mcal{U},i},\wt{E}_{\mcal{U},j}))$. We can then express the rest of the boundary conditions as

\hypertarget{bc2}{\paragraph{\small Boundary Condition 2.}} It consists of three parts

\begin{enumerate}
    \item [{\small\bf 2a.}] The connection $\bigoplus_{i=1}^n\nabla_{\wt{A}_i}$ on $S^2_\infty$ has a smooth extension to $\overline{\mscr{N}}$. 
    
    \item [{\small\bf 2b.}] One demands that $p_i-\lambda_i$ extend as sections $\phi_j\in \Gamma(\overline{\mscr{N}},\wt{L}^{-2})$, so that $\phi_i=|x|^{-2}(p_i-\lambda_i),\,i=1,\ldots,n$.

    \item [{\small  \bf 2c.}]  Finally, one requires that $\mscr{E}_{ij}$ extends as a section $\mscr{E}_{ij}^{\text{ext}}$ of $\wt{L}^{-|p_i-p_j|}$, so that $\mscr{E}_{ij}^{\text{ext}}=|x|^{-|p_i-p_j|}\mscr{E}_{ij}$.
\end{enumerate}
All integral hyperbolic monopoles satisfy these boundary conditions, and they guarantee the existence of (1) solutions of the Bogomolny equation subject to the reality condition \eqref{eq:reality condition on gauge and Higgs fields}, and (2) a well-defined notion of spectral data. Another justification for these boundary conditions is given in \cite[\S 4.2]{MurraySinger199607}. Using these boundary conditions, we write the connection \eqref{eq:expansion of connection near the boundary of H3} as
\begin{eqaligned}
    \nabla_{\wt{A}_+}&=\bigoplus_{i=1}^n\nabla_{\wt{A}_i}+\text{diag}(p_1-|x|^2\phi_1,\ldots,p_n-|x|^2\phi_n)\rd_{\mcal{U}}\log|x|,
    \\
    &+\mbs{x_p}^{-1}\cdot\mscr{E}^{\text{ext}}_++\mbs{x_p}\cdot\mscr{E}^{\text{ext}}_- \;,
\end{eqaligned}
where
\begin{equation}
    \mbs{x_p}:=
    \begin{pmatrix}
        |x|^{p_1} 
        \\
        & \ddots
        \\
        && |x|^{p_n}
    \end{pmatrix} \;.
\end{equation}

\smallskip To state the correspondence, we need to introduce a few more extra concepts. First consider the two maps $\pi_\pm:\tsu\to\partial\mcal{U}$ \footnote{By an abuse of notation, we have used the same notation as \eqref{eq:twisted structure sheaf of twistor space} for projection onto $\Pbb^1_\pm$. This hopefully does not cause any confusion for the cautious reader.} for $\forall z,w\in\C$ (recall the definition of $\tsu$ in \eqref{eq:definition of twistor space of u})
\begin{equation}\label{eq:sections pipm of the correspondence}
    \pi_+(z,w):=\frac{zz^\dagger}{\langle z,z\rangle} \;, 
    \qquad 
    \pi_-(z,w):=\frac{{\sigma}(w) \sigma(w)^\dagger}{\langle w,w\rangle} \;,
\end{equation}
where $\sigma z$ is defined through the real structure on $\mbb{P}^3$ \eqref{eq:real structure on P3}
\begin{equation}
    \sigma(w)=\begin{pmatrix}
         \bmb{w}_0 \\ \bmb{w}_1
    \end{pmatrix} \;, 
    \qquad 
    \sigma (w^\dagger)=(w_0,w_1) \;. 
\end{equation}
Recalling \eqref{eq:R+ action on the twistor space of U}, we see that $\pi_+$ and $\pi_-$ are both $\mbb{R}_+$-invariant (basically the numerator and denominator of the right-hand sides of \eqref{eq:sections pipm of the correspondence} are multiplied by $\lambda^{\pm 1}$ for $\pi_\pm$, and hence in total are $\R_+$-invariant). Therefore, they descend to maps $\pi_\pm:\ts\to S^2_\infty\subset\overline{\Hbb}^3$: For a geodesic $\gamma\in\ts$, $\pi_+(\gamma)$ and $\pi_-(\gamma)$ are its initial and final points, respectively. Therefore, they are sections of the correspondence space $\mscr{C}_{\mcal{U}}$ and can be used to perform the IWT \eqref{eq:inverse ward transform, the bundle} and \eqref{eq:inverse ward transform, the cr operator}. It is clear from \eqref{eq:sections pipm of the correspondence} that $\pi_+$ is just the projection to $S^2_\infty$, so it is orientation-preserving, while $\pi_-$ is the projection to $S^2_\infty$ combined with the antipodal map, hence it is orientation-reversing in total (compare it with \eqref{eq:projection of fixed-lines of C*-action on CP3}). 

\smallskip We can now state the twistor correspondence for the hyperbolic monopoles (see \cite[Theorem 4.1]{MurraySinger199607} for a proof)

\begin{thr}[The Twistor Correspondence for Non-Integral Hyperbolic $\SU{n}$-Monopoles]\normalfont \label{thr:twistor correspondence theorem}
    Let the pair $(\wt{A},\phi)$ be a solution of the Bogomolny equation \eqref{eq:Bogomolny equation} on $\Hbb^3$ with $\wt{A}$ being a connection on a vector bundle $\wt{E}$ and $\phi$ belonging to the adjoint bundle of $\wt{E}$, satisfying the boundary conditions \hyperlink{bc1}{1} and \hyperlink{bc2}{2}. For the sections \eqref{eq:sections pipm of the correspondence} of the correspondence restricted to $\ts$, there are two holomorphic vector bundles on $\ts$ defined by
    \begin{equation}
        E^{\pm}:=(\pi_{\pm}^*\wt{E},\bmb{\partial}_{E^\pm}) \;,
    \end{equation}
    with
    \begin{equation}
    \bmb{\partial}_{E^\pm}:=\pi_\pm^*\left(\bigoplus_{i=1}^n\nabla_{\wt{A}_i}+\mscr{E}^{\text{ext}}_\pm\right)\mp \text{diag}(p_1,\ldots,p_n)\bmb{\partial}\chi_\pm \;,
    \end{equation}
    and the real-valued functions $\chi_\pm:\tsu\to\mbb{R}$ are defined as
    \begin{equation}
        \chi_+(z,w):=\frac{\langle w,z\rangle}{\langle w,w\rangle} \;, 
        \qquad 
        \chi_-(z,w):=\frac{\langle z,w\rangle}{\langle z,z\rangle} \;.
    \end{equation}
    As a result, there are two holomorphic filtrations 
    \begin{equation}\label{eq:two filtrations of E}
        0\simeq E^\pm_0\subset E^\pm_1\subset \ldots\subset E^\pm_{n-1}\subset E^\pm_n\simeq E^\pm \;,
    \end{equation}
    where
    \begin{equation}\label{eq:definition of bundles Epmi over ZH3}
        E^+_i:=\pi_+^*\left(\bigoplus_{j=1}^{i}\wt{E}_{\mcal{U},n-i+j}\right) \;, 
        \qquad 
        E^-_i:=\pi_-^*\left(\bigoplus_{j=1}^{i}\wt{E}_{\mcal{U},j}\right) \;,
    \end{equation}
    with the bundle $\wt{E}_{\mcal{U},i}$ is defined in \eqref{eq:definition of 1d bundles tildeE on U}. Furthermore, the following isomorphisms hold
    \begin{equation}
        E^+_{n-i+1}/E^+_{n-i} \simeq L^{p_i+k_i/2}(k_i) \;, 
        \qquad 
        E^-_{i}/E^-_{i-1}\simeq L^{p_i+k_i/2}(-k_i) \;,
    \end{equation}
    where the bundle $L$ is defined in \eqref{eq:line bundle L}. \qed 
\end{thr}

\begin{rmk} \normalfont
   For the sake of comparison, recall that in the context of Euclidean $G$-monopoles, $E^\pm$ are the two bundles obtained by reductions of the structure group of the bundle $E$ from $G_\C$ (the complexification of $G$) to a Borel subgroup $B$ (for $E^+$) and its opposite subgroup $\bmb{B}$ (for $E^-$) \cite{Murray1983,Murray198306,Murray198412,HurtubiseMurray1989}. \qed 
\end{rmk}

Note that the degree of pullback of the bundle by projection $\pi_\pm$ has a bidegree $(\cdot,\cdot)$ referring to the degree over $(\Pbb^1_-,\Pbb^1_+)$. Therefore, the degree of $E^\pm_i$ over $\Pbb^1_{\mp}$ is zero. On the other hand, the degree of $E^+_i$ over $\Pbb^1_+$, which we denote by $\text{deg}_\pm$, is
\begin{equation}
  \text{deg}_+E^+_i=\sum_{j=1}^ik_{n-i+j}=-\sum_{j=1}^i k_j \;,
\end{equation}
while for the degree of $E^-_i$ over $\Pbb^1_-$, one needs to take into account that the map $\pi_-$ is orientation-reversing
\begin{equation}
   \text{deg}_-E^-_i=-\sum_{j=1}^ik_j \;,
\end{equation}
where the orientation-reversal is encoded in the additional $-1$ prefactor. Recalling \eqref{eq:magnetic charges for su(n) monopole}, we get 
\begin{equation}\label{eq:degree of Epm}
    \text{deg}\, E^+_i=(0,-m_i) \;, 
    \qquad 
    \text{deg}\, E^-_i=(-m_i,0) \;.
\end{equation}
Therefore, it follows that
\begin{equation}\label{eq:degree of quotient of Epm}
    \text{deg}(E^+/E^+_i)=(0, +m_i) \;, 
    \qquad  
    \text{deg}(E^-/E^-_i)=(+m_i, 0) \;.
\end{equation}
Also, recall that $E$ is a real bundle with the real structure $\sigma_E:E \to E^*$. In a complete analogy with the Euclidean case, we have \cite{Hitchin198212,Murray1983,Murray198306,Murray198412}
    \begin{equation}
        \sigma_E(E^+_i)\simeq (E^-_{n-i})^\perp \;, 
        \qquad 
        i=1,\ldots,n  \;,
    \end{equation}
where $(E^\pm_i)^\perp$ denotes the subspace of linear functionals on $E$ vanishing on $E^\pm_i$.

\subsubsection{Spectral Data and Recovering Monopole Solutions}\label{sec:spectral data and recovering monopole solutions}

The twistor correspondence gives a holomorphic bundle on $\ts$ from the data of a monopole configuration on $\Hbb^3$ and vice versa. However, one question remains to be answered, which is: {\it Given a solution of the Bogomolny equation (with maximal gauge-symmetry breaking), does it always come from such a construction?} To answer this question, one needs to introduce the notion of spectral data, a concept first introduced by Hitchin \cite{Hitchin198212}. For hyperbolic monopoles, this concept is defined canonically through the two filtrations $E^\pm_i$ constructed in Theorem \ref{thr:twistor correspondence theorem}. Then one argues that any solution of the Bogomolny equation on $\Hbb^3$ can be completely determined from its spectral data, which is the answer to the question we wished to respond to.

\paragraph{Spectral Curves.} 
Suppose that we have 
a hyperbolic $\SU{n}$-monopole associated with the holomorphic bundle $E$ on $\ts$ satisfying the boundary conditions \hyperlink{bc1}{1} and \hyperlink{bc2}{2}, giving rise to filtrations \eqref{eq:two filtrations of E} of $E$.
The spectral curves of the hyperbolic monopole are defined in a complete analogy with the Euclidean case \cite{Hitchin198212,Murray1983,Murray198306,Murray198412}. 
Consider the maps
\begin{equation}\label{eq:definition of ith spectral curve of monopole}
    \Gamma^-_i:\bigwedge\nolimits^{i} E^-_i\to \bigwedge\nolimits^{i}\left(E^+/E_{n-i}^+\right), \qquad i=1,\ldots,n-1,
\end{equation}
where $\bigwedge\nolimits^{i}$ denotes the $i$\textsuperscript{th} exterior power. The $i$\textsuperscript{th} spectral curve is defined as
\begin{equation}
    S_i:=\text{div}(\Gamma_i^-) \;, \qquad i=1,\ldots,n-1 \;,
\end{equation}
with $\text{div}(\Gamma^-_i)$ denotes the divisor of the map $\Gamma^-_i$, i.e.\ its zero locus. Recall that $\bigwedge\nolimits^{i}E^-_i=\det E^-_i$, the determinant line bundle of $E^-_i$. From \eqref{eq:degree of Epm}, we see that the degree of $\det E^-_i$ is just $(-m_i,0)$. From $\text{deg}(\mcal{O}_{\Pbb^1}(D))=\text{deg}(D)$ for any divisor $D$, we conclude that the associated line bundle is isomorphic to $\mcal{O}_{\ts}(-m_i,0)$ \cite{Grothendieck195701}. Similarly, $\bigwedge\nolimits^{i}(E^+/E^+_{n-i})=\det(E^+/E^+_{n-i})$ with degree $+m_i$, and the associated line bundle is isomorphic to $\mcal{O}_{\ts}(0,+m_i)$. Therefore, the spectral curves are determined by the divisors of the following maps
\begin{equation}\label{eq:the map Gamma-}
    \Gamma^-_i:\det E^-_i\to\det(E^+/E^+_{n-i}) \;, \qquad i=1,\ldots,n-1 \;,
\end{equation}
and thus $S_i$ is a section of $\mcal{O}_{\ts}(m_i,m_i)$. Taking into account \eqref{eq:identification of OZH3(p,p) with OZU(2p)}, we see that
\begin{equation}\label{eq:spectral curves of hm as divisors of gamma-}
    S_i\in\Gamma(\mcal{O}_{\ts}(m_i,m_i))\simeq\Gamma(\mcal{O}_{\tsu}(2m_i)) \;, 
    \qquad 
    i=1,\ldots,n-1 \;,
\end{equation}
where $\simeq$ should be understood as we explained around \eqref{eq:identification of OZH3(p,p) with OZU(2p)}. Alternatively, $S_i$ can be described as the divisor of
\begin{equation}
    \Gamma_i^+:\bigwedge\nolimits^{n-i}E^+_{n-i}\to\bigwedge\nolimits^{n-i}\left(E^-/E_{i}^-\right) \;,
    \qquad 
    i=1,\ldots,n-1 \;,
\end{equation}
where $\bigwedge\nolimits^{n-i}E^+_{n-i}\simeq\det E^+_{n-i}$ and $\bigwedge\nolimits^{n-i}\left(E^-/E_{i}^-\right)\simeq\det\left(E^-/E_{i}^-\right)$ are of degree $(0,-m_i)$ and $(m_i,0)$, respectively. Therefore, the spectral curve $S_i$ is again a section of $\mcal{O}_{\ts}(m_i,m_i)$, and the genus of each curve is given by
\begin{equation}\label{eq:genus of Si}
    g_{S_i}=(m_i-1)^2 \;,
    \qquad 
    i=1,\ldots,n-1 \;.
\end{equation}
Furthermore, $S_i$ is compact \cite[Lemma on pg.\ 990]{MurraySinger199607}, and real in the following sense
\begin{equation}\label{eq:reality of spectral curves of hyperbolic monopoles}
    \sigma(S_i)=S_i \;, \qquad i=1,\ldots,n-1 \;,
\end{equation}
where this equation should be understood as in \eqref{eq:definition of real lines}.

\hypertarget{spectral data}{\paragraph{Spectral Data.}} It turns out the spectral curves $S_i$ are not the only data needed for the recovery of a specific monopole solution, and the following extra details are necessary. Consider the sets $S_i\cap S_{i+1}$ of intersections of $S_i$ and $S_{i+1}$. It can be decomposed as $S_{i,i+1}\cup S_{i+1,i}$ where
\begin{eqaligned}\label{eq:definition of Sii+1 and Si+1i}
    S_{i,i+1}&:=\left\{s\in S_i\,\Big|\,\dim\left(E^+_{n-i-1}\cap E^-_i\right)\big|_s\ge 1\right\} \;,
    \\
    S_{i+1,i}&:=\left\{s\in S_i\,\Big|\,\dim\left(E^+_{n-i}\cap E^-_{i+1}\right)\big|_s\ge 2\right\} \;,
\end{eqaligned}
which are exchanged by the real structure $\sigma$. A monopole is called generic if $S_i$ and $S_{i+1}$ intersect transversally and $S_{i,i+1}$ and $S_{i+1,i}$ are distinct set of points. Each one of the sets $S_{i,i+1}$ and $S_{i+1,i}$ consists generically of $m_im_{i+1}$ points, which are exchanged by the real structure $\sigma$, hence in total the intersection locus
\begin{equation}\label{eq:intersection locus of adjacent spectral curves of monopoles}
    S_i\cap S_{i+1}=S_{i,i+1}\cup S_{i+1,i} \;, 
    \qquad 
    i=1,\ldots,n-1 \;,
\end{equation}
is a set of $2m_im_{i+1}$ points. As there are $n-1$ spectral curves $\{S_1,\ldots,S_{n-1}\}$, the condition of genericity means that there is a spectral curve associated with each node of the Dynkin diagram of $\mfk{su}(n)=\text{Lie}(\SU{n})$, and curves associated with adjacent nodes are intersecting transversally. This means that spectral curves are labeled by simple roots or the associated fundamental weights (see \eqref{eq:defining equation of fundamental weights}) of $\mfk{su}(n)$. The spectral curves $\{S_1,\ldots, S_{n-1}\}$ together with the information of the decomposition of their intersection is called the spectral data of hyperbolic $\SU{n}$-monopoles.  

\paragraph{Hyperbolic Monopoles from Spectral Data.} The next and final question to wrap up the twistor correspondence of hyperbolic $\SU{n}$-monopoles is how to recover a monopole from its spectral data. Here, we sketched the results following \cite{HurtubiseMurray1989,MurraySinger199607}.\footnote{While \cite{HurtubiseMurray1989} discussed monopoles in Euclidean space, their considerations can be adopted to 
hyperbolic space, as noted in \cite{MurraySinger199607}.} 

\smallskip One starts from the following sequence of sheaves 
\begin{equation}\label{eq:short exact sequence of sheaves to construct the bundle E}
    0\to E^+ \to \bigoplus_{i=1}^n  E^+/(E^+_{n-i}+E^-_{i-1})\to\bigoplus_{i=1}^{n-1} E^+/(E^+_{n-i}+E^-_i) \;.
\end{equation}
The first map is simply the projection to the corresponding equivalence class while the second map is given by \cite[Proposition 1.12]{MurraySinger199607}
\begin{equation}
    (e_1,\ldots,e_n)\mapsto (\delta_1(e_1)-\gamma_2(e_2),\ldots,\delta_{n-1}(e_{n-1})-\gamma_n(e_n)) \;,
\end{equation}
with the maps $\gamma_i$ and $\delta_i$ are the natural projections
\begin{eqaligned}\label{eq:maps gammai and deltai}
    \gamma_i&:\frac{E^+}{E^+_{n-i}+E^-_{i-1}}\mapsto\frac{E^+}{E^+_{n-i+1}+E^-_{i-1}} \;,
    \qquad &i&=1,\ldots,n-1 \;,
    \\
    \delta_i&:\frac{E^+}{E^+_{n-i}+E^-_{i-1}}\mapsto \frac{E^+}{E^+_{n-i+1}+E^-_{i}} \;,
    \qquad &i&=1,\ldots,n-1 \;.
\end{eqaligned}
One can show that this is a short exact sequence \cite[Proposition 1.12]{HurtubiseMurray1989}. According to loc.\ cit.\ the first and second sums in the sequence are given by \cite[eq. (1.11)]{HurtubiseMurray1989}
\begin{equation}
    \bigoplus_{i=1}^n E^+/(E^+_{n-i}+E^-_{i-1})=\bigoplus_{i=1}^n L^{p_i+k_i/2}(m_{i-1}+m_i)\otimes \mcal{I}(S_{i-1,i}) \;,
\end{equation}
and with $\mcal{I}(S_{i-1,i})$ denoting the sheaf of functions vanishing on $S_{i-1,i}$, and
\begin{equation}
    \bigoplus_{i=1}^n E^+/(E^+_{n-i}+E^-_{i})=\bigoplus_{i=1}^n L^{p_{i+1}+k_{i+1}/2}(m_{i}+m_{i+1})\otimes [-S_{i,i+1}]\Big|_{S_i} \;,
\end{equation}
where $L$ is the trivial line bundle \eqref{eq:line bundle L}, and $[D]$ denotes the line bundle associated with the divisor $D$. Therefore, the maps $\gamma_i$ and $\delta_i$ in  \eqref{eq:maps gammai and deltai} can be understood as
\begin{eqaligned}\label{eq:maps gammai and deltai, second form}
    \gamma_i&:L^{p_i+k_i/2}(m_{i-1}+m_i)\otimes \mcal{I}(S_{i-1,i})\mapsto L^{p_i+k_i/2}(m_{i-1}+m_i)\otimes [-S_{i-1,i}]\Big|_{S_{i-1}} \;,
    \\
    \delta_i&:L^{p_i+k_i/2}(m_{i}+m_{i+1})\otimes \mcal{I}(S_{i,i+1})\mapsto L^{p_{i+1}+k_{i+1}/2}(m_{i}+m_{i+1})\otimes [-S_{i,i+1}]\Big|_{S_{i}} \;.
\end{eqaligned}
This means that $\gamma_i$ is simply the restriction to $S_{i-1}$ while $\delta_i$ is the restriction to $S_i$ accompanied by multiplying with some meromorphic section. Its holomorphic part $\zeta_i$ can be read off from \eqref{eq:maps gammai and deltai, second form}. This is a section of 
\begin{equation}\label{eq:holomorphic part of the section that determines deltai}
    \zeta_i\in\Gamma(L^{p_{i+1}+k_{i+1}/2-p_i-k_i/2}(m_{i-1}+m_{i+1})\otimes[-S_{i-1,i}-S_{i,i+1}],S_i) \;.
\end{equation}
Therefore, over the $i$\textsuperscript{th} spectral curve $S_i$, there is a holomorphic section of the bundle $L^{p_{i+1}+k_{i+1}/2-p_i-k_i/2}(m_{i-1}+m_{i+1})$ with zeroes along $S_{i-1,i}\cup S_{i,i+1}$, i.e.\ \cite[Lemma 1.8]{HurtubiseMurray1989}
\begin{equation}\label{eq:trivial bundle over ith spectral curve}
    L^{p_{i+1}+k_{i+1}/2-p_i-k_i/2}(m_{i-1}+m_{i+1})\otimes[-S_{i-1,i}-S_{i,i+1}]\Big|_{S_i}\simeq\mcal{O}_{S_i} \;,
\end{equation}
for $i=1,\ldots,n-1$. This feature uniquely determines $\zeta_i$. This completes the construction of the bundle $E^+$, as the kernel of the second map of \eqref{eq:short exact sequence of sheaves to construct the bundle E}, from the spectral data.

\begin{rmk}\normalfont
    For the case of $\SU{2}$ hyperbolic monopole, there is only one spectral curve $S$ and hence no intersection sets $S_{i-1,i}$ and $S_{i,i+1}$. The condition \eqref{eq:trivial bundle over ith spectral curve} reduces to
    \begin{equation}
        \left.L^{p_2+k_2/2-p_1-k_1/2}(m_0+m_2)\right|_S\simeq \mcal{O}_S \;.
    \end{equation}
    We have $m_0=0=m_2=k_1+k_2$, which gives $k_2=-k_1=:k$ and also $p_2=-p_1=:p$. Hence, 
    \begin{equation}
        \left.L^{2p+k}\right|_S\simeq \mcal{O}_S \;,
    \end{equation}
    which is precisely the condition that defines the spectral curve $S$ \cite[eq. (3.7) and Proposition 3.5]{Atiyah1984}.\qed 
\end{rmk}

\smallskip Let us recapitulate the procedure of constructing the monopole solution 
\begin{itemize}
    \item[(1)] From the spectral data, i.e.\ curves together with the data of the splitting of their intersection locus, one constructs the sequence of sheaves \eqref{eq:short exact sequence of sheaves to construct the bundle E};

    \item[(2)] Then, the bundle $E^+$ can be recovered as the kernel of the second map in \eqref{eq:short exact sequence of sheaves to construct the bundle E};

    \item[(3)] The real structure on $E^+$ can be constructed using spectral data. We avoid explaining this and instead refer to \cite[\S 1.1e, pg. 53]{HurtubiseMurray1989};

    \item[(4)] One follows the procedure in \cite[Theorem 1.19]{HurtubiseMurray1989} to construct the monopole solution $(\wt{A},\phi)$ from the bundle $E^+$.
\end{itemize}

For our purpose of comparison of the data above to the data constructed from the curve of the spectral parameter of the gCPM, \eqref{eq:trivial bundle over ith spectral curve} is essential and provides one of the key relations. We are now equipped with enough background and are in a position to establish the correspondence between the gCPM and hyperbolic $\SU{n}$-monopoles to which we now turn.

\section{The Hyperbolic Monopole/gCPM Correspondence}
\label{sec:generalized corresponence}
We are now in a position to establish the correspondence between the spectral data of hyperbolic $\SU{n}$-monopoles and the curve of spectral parameter $\pc$, defined in \eqref{eq:ZNxn-1 quotient of the curve of spectral parameter of gCPM}, of the $\mbb{Z}_N^{n-1}$ generalized CPM. A point we would like to emphasize in the very beginning is the following: the correspondence works only when we consider 
special choices of the data, either on the gCPM side or the hyperbolic $\SU{n}$-monopole side, as we will explain below. 

\subsection{General Considerations}\label{sec:general considerations}

Before delving into the construction let us state an important assumption first. 

\paragraph{An Important Assumption.}  In the process of the construction, it becomes mandatory to take the boundary values of the Higgs field to vanish. To define the spectral data of such a configuration, we will work with {\it non-integral} hyperbolic monopoles and their spectral data, as we explained in \S\ref{sec:spectral data of hyperbolic su(n)-monopoles}. Such monopoles, as their name suggests, can have arbitrary masses, hence it is possible to take the zero-mass limit. The existence of hyperbolic $\SU{2}$-monopoles with an arbitrary mass has been proven in \cite{SibnerSibner199202,SibnerSibner1995,SibnerSibner201210}, and it has been generalized by Murray and Singer to hyperbolic $\SU{n}$-monopoles \cite{MurraySinger199607}. The construction in loc.\ cit.\ involves certain analytical subtleties. In particular, the boundary conditions used in \cite{MurraySinger199607} to define the spectral data is a different set of boundary conditions compared to their integral counterpart \cite{Chan201506,Chan2017}. However, as has been emphasized in \cite{Atiyah199106,AtiyahMurray1995,MurraySinger199607,JarvisNorbury199711}, taking the zero-mass limit is a subtle issue. An important point to emphasize is that a hyperbolic monopole with $p_i=0,\,i=1,\ldots$, is not a trivial configuration. For a gauge group $G$ with the Lie algebra $\mfk{g}=\text{Lie}(G)$, this can be seen by noting that 
\begin{eqgathered}\label{eq:laplacian of the square of higgs field}
    \Delta|\phi|^2=2\sum_{a=1}^{\dim\mfk{g}}\Delta\phi^a\phi_a+2|D\phi|^2=2|D\phi|^2\ge 0 \;,
\end{eqgathered}
where the second equality follows from the Bianchi identity $D F=0$ and the Bogomolny equation. Here, $\Delta:=D^2$ is the gauge covariant Laplacian, $D:=\rd\,+\,[A,\cdot]$, and the norm of the Higgs field is defined by\footnote{Here, we have $\text{tr}_{\mfk{g}}(\phi^2)=\text{tr}_{\mfk{g}}\left(\phi^a\phi^bJ_aJ_b\right)=\phi^a\phi^b\delta_{ab}=\phi^a\phi_a$, with $\{J^a,\,a=1,\ldots,\dim(\mfk{g})\}$ is a set of generators for $\mfk{g}$ with the chosen normalization.}
\begin{eqaligned}
    |\phi|^2:=\text{tr}_{\mfk{g}}(\phi^2)=\sum_{a=1}^{\dim\mfk{g}}\phi_a\phi^a \;.
\end{eqaligned}
From the Maximum Principle \cite[p.\ 248]{JaffeTaubes1980}, therefore, the value of the Higgs field is bounded by its value at the boundary, where the gauge group breaks to a subgroup. As we set $p_i=0$, it follows from \eqref{eq:laplacian of the square of higgs field} that $D\phi=0$. Then, the Bogomolny equation \eqref{eq:Bogomolny equation} leads to $F=0$. Such monopoles are called flat. 

\smallskip On the other hand, there are separate constructions of hyperbolic monopoles for integer (or half-integer) masses. The reason is that these monopoles can be thought of as circle-invariant instantons, as we briefly reviewed in Appendix \ref{sec:more details on twistor space of hyperbolic space}. The spectral data for integral hyperbolic $\SU{2}$-monopoles and the relation to the discrete Nahm equations is described in \cite{Atiyah1984,BraamAustin199008}. These results have been generalized to hyperbolic $\SU{n}$-monopoles in \cite{Chan201506,Chan2017}. A shortcoming of the latter references is that it has not been proven that the spectral data on the integral hyperbolic $\SU{n}$-monopole, with $n>2$, uniquely determines the monopole solution, a result that has been established for $n=2$ by Atiyah \cite{Atiyah1984}. As the construction of both integral and non-integral hyperbolic $\SU{n}$-monopoles is available, there are two ways to construct hyperbolic $\SU{n}$-monopole configurations with vanishing masses: (1) One can either consider the Murray--Singer construction in \cite{MurraySinger199607} and take its zero-mass limit or (2) one can work with integral hyperbolic $\SU{n}$-monopoles, as formulated in \cite{Chan201506,Chan2017} and then simply set all the masses to zero. Presumably, either way, the resulting spectral data is well-defined and uniquely determines a configuration of flat hyperbolic $\SU{n}$-monopole.\footnote{In the Euclidean setting, solutions with branch-type singularity along $S^2$ has been constructed in \cite{ForgacsHorvathPalla198102}. Such solutions have a nontrivial holonomy around $S^2$.} In this work, we exclusively consider the zero-mass limit of the construction of non-integral hyperbolic monopoles by Murray and Singer in \cite{MurraySinger199607} and assume that this limit exists and defines the spectral data of a flat hyperbolic $\SU{n}$-monopole. \qed

\smallskip With the understanding of this assumption, let us briefly summarize what we are aiming for in the following subsections
\begin{itemize}
    \item [(1)] In \S\ref{sec:hyperbolic monopoles from gCPMs}, we start from the curve of spectral parameter $\pc$ for the gCMP defined in $\Pbb^{2n-1}$ and is given in \eqref{eq:ZNxn-1 quotient of the curve of spectral parameter of gCPM}. As there are $n-1$ spectral curves for a hyperbolic $\SU{n}$-monopole on $\ts$, we construct $n-1$ curves out of $\pc$ from a natural $\C^\times$-action on $\Pbb^{2n-1}$, which we call $\ipc{1},\ldots,\ipc{n-1}$ on $\Pbb^1\times\Pbb^1$. We furthermore impose a reality condition on these curves which makes them candidates for the spectral curves of hyperbolic $\SU{n}$-monopoles. We show that these curves satisfy all the properties of the curves for a {\it particular} configuration of a hyperbolic $\SU{n}$-monopole with $n-1$ charges $(N,\ldots,N)$. Furthermore, we show that the bundle defined in the left-hand side of \eqref{eq:trivial bundle over ith spectral curve} is trivial upon restriction to $\ipc{i}$ and setting $p_i=0$. This identifies the particular configuration of hyperbolic $\SU{n}$-monopoles and spectral curves $S_1,\ldots,S_{n-1}$ corresponding to curves $\ipc{1},\ldots,\ipc{n-1}$. In particular, we realize
    \begin{equation}
        \ipc{i}\simeq S_i \;, \qquad i=1,\ldots,n-1 \;.
    \end{equation}

    \item [(2)] In \S\ref{sec:gCPM from hyperbolic SU(n) monopoles}, we reverse the construction, i.e.\ we start from $n-1$ curves $S_1,\ldots,S_{n-1}$ of hyperbolic $\SU{n}$-monopoles. By removing the reality condition and gluing these curves together, we reconstruct the curve of the spectral parameter of the gCPM.
\end{itemize}

This provides a one-to-one correspondence between the spectral data of a (special) hyperbolic $\SU{n}$-monopole and the data coming from the curve of the spectral parameter of the gCPM. 

\subsection{Establishing the Correspondence for \texorpdfstring{$\SU{n}$}{SU(n)}}

In this section, we provide the generalization of the observation in \cite{Atiyah199106,AtiyahMurray1995} and put it into a sound ground.

\subsubsection{Hyperbolic \texorpdfstring{$\SU{n}$}{SU(n)}-Monopoles from gCPMs}\label{sec:hyperbolic monopoles from gCPMs}

We present the construction of the spectral data of a hyperbolic $\SU{n}$-monopole from the curve of the spectral parameter of the gCPM in three steps. 

\hypertarget{step1}{\paragraph{Step 1: Constructing $n-1$ Curves in $\Pbb^1\times\Pbb^1$ out of $\pc$.}} As we explained in \S\ref{sec:generalized CPM}, the curve of the spectral parameter of the model is defined in $\Pbb^{2n-1}$. However, further scrutiny shows that \eqref{eq:curve of spectral parameter of gCPM} can indeed be regarded as a curve in the product of $n$ copies of $\Pbb^1$. This can be seen as follows. Recall that, as we have stated around \eqref{eq:cyclic action of ZN(n-1) on coordinates of P2n-1}, the curve of spectral parameter for the gCPM can be quotiented to $\pc=\cpc/\mbb{Z}_N^{n-1}$, where  $\cpc\subset\Pbb^{2n-1}$. As $\mbb{Z}_N^{n-1}\subset (\C^\times)^{n-1}$, we can generalize \eqref{eq:cyclic action of ZN(n-1) on coordinates of P2n-1} in a natural way and instead consider the following more general action of $(\C^{\times})^{n-1}$ on $\Pbb^{2n-1}$
\begin{equation}
    [z^+_1:z^-_1:\ldots:z^+_n:z^-_n]\mapsto [\lambda_1 z^+_1:\lambda_1 z^-_1:\ldots:\lambda_{n-1}z^+_{n-1}:\lambda_{n-1}z_{n-1}^-:z^+_n:z^-_n] \;,
\end{equation}
for any $(\lambda_1,\ldots,\lambda_{n-1})\in(\C^{\times})^{n-1}$. This action has $n$ fixed lines given by
\begin{equation}\label{eq:fixed lines of Cx-action on P2n-1}
    \Pbb_i := [z_i^+:z_i^-] \simeq [0:\ldots:0:z_i^+:z_i^-:0:\ldots:0] \;, 
    \qquad 
    i=1,\ldots,n  \;.
\end{equation}
Therefore, the fixed locus $\mcal{F}$ is the union of $n$ copies of $\Pbb^1$,  which for convenience we label as $\Pbb^1_i,\,i=1,\ldots,n$
\begin{equation}
    \mcal{F}:=\Pbb_1^1\cup\cdots\cup\Pbb^1_n \;.
\end{equation}
The action is free elsewhere. A $(\C^\times)^{n-1}$-orbit intersects $\Pbb_1^1,\ldots,\Pbb^1_n$ and the set of free  orbits can thus be identified with
\begin{equation}
    \{\Pbb^{2n-1}-\mcal{F}\}/{(\C^\times)^{n-1}}\simeq \Pbb_1^1\times\cdots\times\Pbb^1_n \;.
\end{equation}
For a point $([a_1:b_1],\ldots,[a_n:b_n])\in \Pbb_1^1\times\cdots\times\Pbb^1_n$, the generic fiber of the $(\C^\times)^{n-1}$-bundle $\{\Pbb^{2n-1}-\mcal{F}\}\to\Pbb_1^1\times\cdots\times\Pbb^1_n$ is given by 
\begin{equation}
    [a_1/\lambda_1:b_1/\lambda_1:\ldots:a_{n-1}/\lambda_{n-1}:b_{n-1}/\lambda_{n-1}:a_n:b_n], \quad (\lambda_1,\ldots,\lambda_{n-1})\in(\C^\times)^{n-1} \;.
\end{equation}
As we vary $(\lambda_1,\ldots,\lambda_{n-1})$, we generate the entire fiber, which is thus isomorphic to $(\C^\times)^{n-1}$, as it should be. Therefore, the projection $\{\Pbb^{2n-1}-\mcal{F}\}\to\Pbb_1^1\times\cdots\times\Pbb^1_n$ defines a rank-$(n-1)$ $(\C^\times)^{n-1}$-vector bundle over $\Pbb^1_i\times\cdots\times\Pbb^1_n$. For generic values of the matrices $K_{ij}$ in \eqref{eq:curve of spectral parameter of gCPM}, the curve $\cpc$ does not intersect $\mcal{F}$, which implies that the action of $\mbb{Z}_N^{n-1}$ on $\cpc$ is free. We thus conclude that $\pc$ indeed defines a curve in $\Pbb^1_1\times\cdots\times\Pbb^1_n$. 

\smallskip Next, recall that only $n-1$ equations in \eqref{eq:curve of spectral parameter of gCPM} are independent due to the cocycle conditions \eqref{eq:relations satisfied by matrices Kij of gCPM}. These equations can be conveniently chosen to be the ones defined by $K_{i,i+1}$ for $i=1,\ldots,n-1$ and the rest will be determined through \eqref{eq:relations satisfied by matrices Kij of gCPM}. We then define the following sequence of projections
\begin{equation}\label{eq:sequence of projection and forgetful maps}
    \Pbb_1^1\times\cdots\times\Pbb^1_n \xrightarrow{\text{pr}_i}\Pbb^1_i\times\Pbb^1_{i+1}\xrightarrow{\text{f}_i}\Pbb^1\times\Pbb^1 \;.
\end{equation}
Under ${\text{pr}}_i$, the equations for $\cpc$ are projected onto a curve in $\Pbb^1_i\times\Pbb^1_{i+1}$ and under the forgetful map $\text{f}_i$, the $i$-dependence of coordinates of the curve is forgotten, and we land on a curve in $\Pbb^1\times\Pbb^1$. Notice that there is still a hidden $i$-dependence in the matrices $K_{i,i+1}$, and hence we generate $n-1$ distinguished curves by this procedure. Denoting a point in $\Pbb^1\times\Pbb^1$ as $([z^+:z^-],[w^+:w^-])$, we can write these maps more explicitly. Under $\text{pr}_i$, we get the equations for the curve $\iipc{i}$ in $\Pbb^1_i\times\Pbb^1_{i+1}$
\begin{equation}\label{eq:equations for curves C'i in P1xP1}
	\iipc{i}:\quad \begin{pmatrix}
		(z^+_{i}){}^{N}
		\\
		(z^-_i){}^{N}
	\end{pmatrix}
	= K_{i}
	\begin{pmatrix}
		(z^+_{i+1}){}^{N}
		\\
		(z^-_{i+1}){}^{N}
	\end{pmatrix} \;,
    \qquad i=1,\ldots,n-1 \;,
\end{equation}
and under $\text{f}_i$, we get the equations for the curve $\icpc{i}$ in $\Pbb^1\times\Pbb^1$
\begin{equation}\label{eq:equations for curves tCi in P1xP1}
	\icpc{i}:\quad \begin{pmatrix}
		(z^+){}^{N}
		\\
		(z^-){}^{N}
	\end{pmatrix}
	= K_{i}
	\begin{pmatrix}
		(w^+){}^{N}
		\\
		(w^-){}^{N}
	\end{pmatrix} \;, \qquad i=1,\ldots,n-1 \;,
\end{equation}
where for the sake of brevity, we have denote $K_{i,i+1}$ with $K_i,\,i=1,\ldots,n-1$. Considering $f_i\circ\text{pr}_i$ for $i=1,\ldots,n-1$ produces $n-1$ curves. These curves have an obvious $\mbb{Z}_N\times\mbb{Z}_N$ symmetry, acting as
\begin{equation*}
    ([z^+,z^-],[w^+,w^-])\mapsto ([\omega z^+, z^-],[\omega' w^+, w^-]) \;,
\end{equation*}
with $\omega^N={\omega'}^N=1$ are two $N$\sth root of unity. In the following, we
consider the $n-1$ quotient curves
\begin{equation}\label{eq:Zn quotient of ith tilde curve}
    \ipc{i}:=\icpc{i}/(\mbb{Z}_N\times\mbb{Z}_N)\subset\Pbb^1\times\Pbb^1 \;, 
    \qquad 
    i=1,\ldots,n-1 \;.
\end{equation}
These are $n-1$ curves in $\Pbb^1\times\Pbb^1$ we wished to construct. 

\smallskip Alternatively, we can think of these curves as follows. We consider the following embeddings $\iota_i:\Pbb^3_i\hookrightarrow\Pbb^{2n-1}$ for $i=1,\ldots,n-1$
\begin{equation}\label{eq:embedding of p3 in p2n-1}
    \Pbb^3_i:\,\,[z_i^+:z^-_i:z_{i+1}^+:z_{i+1}^-]\simeq [0:\ldots:z_i^+:z^-_i:z_{i+1}^+:z_{i+1}^-:0:\ldots:0]\in\Pbb^{2n-1} \;.
\end{equation}
We next consider the $\C^\times$-action on $\Pbb^3_i$ as
\begin{equation}
    [z_i^+:z^-_i:z_{i+1}^+:z_{i+1}^-]\mapsto[\lambda z_i^+:\lambda z_{i}^-:z_{i+1}^+:z_{i+1}^-] \;, 
    \qquad 
    \lambda\in\C^\times  \;.
\end{equation}
Following the same argument as the case of $(\C^\times)^{n-1}$ action on $\Pbb^{2n-1}$, we end up with $n-1$ $\C^\times$-bundles
\begin{equation}
    \{\Pbb^3-(\Pbb^1_i\cup\Pbb^1_{i+1})\}/\C^\times\simeq \Pbb^1_i\times\Pbb^1_{i+1} \;, 
    \qquad 
    i=1,\ldots,n-1 \;,
\end{equation}
with generic fibers $\C^\times$. Here $\Pbb^1_i$ and $\Pbb^1_{i+1}$ denote the two fixed lines of $\C^\times$-action defined analogous to \eqref{eq:fixed lines of Cx-action on P2n-1}. Therefore, the projection $\{\Pbb^3-(\Pbb^1_i\cup\Pbb^1_{i+1})\}\to\Pbb^1_i\times\Pbb^1_{i+1}$ defines a $\C^\times$-line bundle over $\Pbb^1_i\times\Pbb^1_{i+1}$. Under the embedding $\iota_i$, the equations \eqref{eq:curve of spectral parameter of gCPM} of $\cpc$ become the ones for $\iipc{i}$ given in \eqref{eq:equations for curves C'i in P1xP1}.
If we now apply the forgetful map $\text{f}_i$, as defined in \eqref{eq:sequence of projection and forgetful maps}, we land on the equations for the curve $\icpc{i}$ given \eqref{eq:equations for curves tCi in P1xP1}, that by passing to the $\mbb{Z}_N\times\mbb{Z}_N$ quotient produces the curve $\ipc{i}$. Considering all the embeddings $\iota_i$ gives the set of curves $\{\ipc{1},\ldots,\ipc{n-1}\}$. From this point of view, one can see each of the curves in the set $\{\ipc{1},\ldots,\ipc{n-1}\}$ as the curve of spectral parameter for an ordinary chiral Potts model, which would be useful later for identifying the corresponding hyperbolic $\SU{n}$-monopole configuration.   

\smallskip Finally, we would like to point out that the curves $\ipc{1},\ldots,\ipc{n-1}$ enjoy a global symmetry under which 
\begin{equation}\label{eq:global symmetry of curves sigmai}
    \begin{pmatrix}
        {z^+}^N
        \\
        {z^-}^N
    \end{pmatrix}
    \mapsto 
    U 
    \begin{pmatrix}
        {z^+}^N
        \\
        {z^-}^N
    \end{pmatrix} \;,
    \qquad 
    \begin{pmatrix}
        {w^+}^N
        \\
        {w^-}^N
    \end{pmatrix}
    \mapsto 
    V
    \begin{pmatrix}
        {w^+}^N
        \\
        {w^-}^N
    \end{pmatrix} \;, 
    \qquad K_i\mapsto UK_iV^{-1} \;,
\end{equation}
for some diagonal matrices $U=\text{diag}(u,u^{-1})$ and $V=\text{diag}(v,v^{-1})$. It is a global symmetry in the sense that, unlike \eqref{eq:redundancy in the curve of spectral parameter of gCPM}, it acts the same way on all matrices $K_i,\,i=1,\ldots,n-1$. Another difference of this transformation with \eqref{eq:redundancy in the curve of spectral parameter of gCPM} is that it does not preserve the cocycle condition \eqref{eq:relations satisfied by matrices Kij of gCPM}. This is expected since the curves $\ipc{1},\ldots,\ipc{n-1}$ are considered independent. Defining $K_i$ as
\begin{equation}\label{eq:entries of matrix Ki}
    K_i=
    \begin{pmatrix}
        \alpha_i & \beta_i
        \\
        \gamma_i & \delta_i
    \end{pmatrix} \;, 
    \qquad i=1,\ldots,n-1 \;.
\end{equation}
Under \eqref{eq:global symmetry of curves sigmai}, we have
\begin{eqaligned}
    \alpha_i&\mapsto (uv)\alpha_i \;, 
    &\qquad \beta_i&\mapsto \left(\frac{u}{v}\right)\beta_i \;,
    \\
    \gamma_i&\mapsto \left(\frac{v}{u}\right)\gamma_i \;, 
    &\qquad \delta_i&\mapsto \left(\frac{1}{uv}\right)\delta_i \;.
\end{eqaligned}

To summarize, we started from the curve $\pc$, and constructed $n-1$ curves $\ipc{1},\ldots,\ipc{n-1}$ satisfying \eqref{eq:equations for curves tCi in P1xP1}. 

\paragraph{Step 2: Properties of the Curves $\{\ipc{1},\ldots,\ipc{n-1}\}$.} We now determine the properties of the  curves $\{\ipc{1},\ldots,\ipc{n-1}\}$. Once we described these properties it becomes evident that they are related to the spectral data of a {\it particular configuration} of hyperbolic $\SU{n}$-monopoles. In the next step, we then identify these monopole configurations. 

\smallskip The curves $\{\ipc{1},\ldots,\ipc{n-1}\}$ have the following properties.

\paragraph{\small $\ipc{i}$ is a section of $\mcal{O}_{\ts}(N,N)$.} From \eqref{eq:equations for curves tCi in P1xP1}, it is clear that $\ipc{i}$ has bi-degree $(N,N)$, and hence is a section of $\mcal{O}_{\ts}(N,N)\simeq\mcal{O}_{\tsu}(2N)$ (see \eqref{eq:identification of OZH3(p,p) with OZU(2p)}). Each curve has genus
\begin{equation}\label{eq:genus of sigmai}
    g_{\ipc{i}}=(N-1)^2 \;,
    \qquad
    i=1,\ldots,n-1 \;.
\end{equation}

\paragraph{\small $\ipc{i}$ defines a real curve.} The matrix $K_i$ in \eqref{eq:equations for curves tCi in P1xP1} which defines the curve $\ipc{i}$ is in general complex matrices belonging to $\text{SL}(2,\C)$. However, we would like to impose a constraint on them as follows. Recall that $\Pbb^1\times\Pbb^1$ had a real structure, given in \eqref{eq:real structure on twistor space}, descended from the real structure \eqref{eq:real structure on P3} on $\Pbb^3$. Under this real structure, the equations \eqref{eq:equations for curves tCi in P1xP1} of the curve $\ipc{i}$ transform to\footnote{Under $\sigma$, the coordinates on the two copies of $\Pbb^1$ is complex-conjugated; the coordinate $z$ in \eqref{eq:real structure on twistor space} is the inhomogeneous coordinate on $\Pbb^1$ and hence $\bar{z}$ corresponds to $[\bar{z}^+:\bar{z}^-]$. The same comment applies to $w, w^+,$ and $w^-$.}
\begin{equation}\label{eq:transformation of ith curve under real structure}
    \ipc{i}\quad\mapsto\quad\sigma(\ipc{i}):\qquad 
    \begin{pmatrix}
        \bar{z}^{+N}
        \\
        \bar{z}^{-N}
    \end{pmatrix}
    =K_i
    \begin{pmatrix}
        \bar{w}^{+N}
        \\
        \bar{w}^{-N}
    \end{pmatrix} \;.
\end{equation}
The extra constraint we put on our curve is the reality condition, i.e. we demand
\begin{equation}\label{eq:demanded reality condition on gCPM curves}
    \sigma(\ipc{i})=\Sigma^{(i)}_{N,n} \;, 
    \qquad i=1,\ldots,n-1 \;,
\end{equation}
where this equation should be understood as in the sense of \eqref{eq:definition of real lines}. Recall that this was one of the conditions on the spectral curves of hyperbolic $\SU{n}$-monopole (see \eqref{eq:reality of spectral curves of hyperbolic monopoles}). Taking the complex conjugate of both sides of \eqref{eq:equations for curves tCi in P1xP1} and comparing the result with \eqref{eq:transformation of ith curve under real structure} show that \eqref{eq:demanded reality condition on gCPM curves} implies the following condition on matrices $K_i$ 
\begin{equation}\label{eq:consequece of imposing reality condition for matrices Ki}
    \overline{K}_i=K_i \;,
    \qquad 
    i=1,\ldots,n-1 \;. 
\end{equation}
We thus have $n-1$ real curves defined by \eqref{eq:equations for curves tCi in P1xP1} in which all matrices $K_i$ satisfy \eqref{eq:consequece of imposing reality condition for matrices Ki} and belong to $\text{SL}(2,\R)$ (rather than $\text{SL}(2,\C)$). 

\paragraph{\small\bf $\ipc{i}\cap\ipc{i+1}$ consists of $N^2+N^2$ points exchanged by $\sigma$.} Recall that each of $\ipc{i}$ or $\ipc{i+1}$ has bi-degree $(N,N)$ in $\Pbb^1\times\Pbb^1$. Hence, by degree reasons (B\'ezout's Theorem), they intersect generically at $N^2+N^2=2N^2$ points. With the aim of comparing the result with the monopole side (see \hyperlink{spectral data}{Spectral Data} in \S\ref{sec:spectral data and recovering monopole solutions}), we would like to see what extra condition matrices $K_i$ need to satisfy such that the intersection points can be decomposed into two sets each of which contains $N^2$ points, and furthermore these sets are exchanged by the real structure $\sigma$ given in \eqref{eq:real structure on twistor space}. 

\smallskip Let us first determine the intersection loci of $\ipc{i}$ and $\ipc{i+1}$. At intersection points, the equations for $\ipc{i}$ and $\ipc{i+1}$ coincide. Considering \eqref{eq:equations for curves tCi in P1xP1}, we arrive at 
\begin{equation}\label{eq:equation for the intersection of curves}
    \ipc{i}\cap\ipc{i+1}:\qquad 
	(K_{i}-\lambda_i K_{i+1})
	\begin{pmatrix}
		(w^+){}^{N}
		\\
		(w^-){}^{N}
	\end{pmatrix}=0 \;,
    \qquad \lambda_i\in\mbb{C}^\times \;, 
\end{equation}
where we have subtracted the equations of $\ipc{i}$ and $\ipc{i+1}$, so
\begin{equation}\label{eq:matrix of Ki-lambdaiKi+1}
    K_i-\lambda_i K_{i+1}=
    \begin{pmatrix}
        \alpha_i-\lambda_i\alpha_{i+1} &  \beta_i-\lambda_i\beta_{i+1}
        \\
        \gamma_i-\lambda_i\gamma_{i+1} & \delta_i-\lambda_i\delta_{i+1}
    \end{pmatrix} \;.
\end{equation}
Furthermore, we considered that we are working in the projective line where we can rescale homogeneous coordinates by $\lambda_i\in\mbb{C}^\times$. \eqref{eq:equation for the intersection of curves} will have a nontrivial solution if and only if
\begin{equation}
    \det\left(K_{i}-\lambda_i K_{i+1}\right)=0 \;, 
    \qquad 
    \lambda\in\mbb{C}^\times \;,
\end{equation}
which can be written as
\begin{equation}\label{eq:eigenvalue equation for lambdai}
    \lambda_i^2-2A_i\lambda_i+B_i=0 \;,
\end{equation}
where %
\begin{equation}\label{eq:values of coefficients Ai and Bi}
    A_i:=\frac{\text{Tr}(K_iK_{i+1}^{-1})}{2\det K_{i+1}}=\frac{1}{2}\text{Tr}(K_iK_{i+1}^{-1}) \;, 
    \qquad 
    B_i:=\frac{\det K_i}{\det K_{i+1}}=1 \;,
\end{equation}
and we have used the fact that $\det K_i=1,\,i=1,\ldots,n-1$. \eqref{eq:eigenvalue equation for lambdai} has the following solutions
\begin{equation}\label{eq:solutions lambdapm}
    \lambda_i^{\pm}=A_i\pm\sqrt{A_i^2-1} \;.
\end{equation}
Each of these values leads to a solution $({w^+}^N,{w^-}^N)$. For $\lambda^\pm_i$, we can eliminate the second row of $K_i-\lambda_iK_{i+1}$ in \eqref{eq:matrix of Ki-lambdaiKi+1}, which leads to the equation
\begin{equation}\label{eq:equation of intersection in terms of w+ and w-}
    (\alpha_i-\lambda^\pm_i\alpha_{i+1}){w^+}^N+(\beta_i-\lambda^\pm_i\beta_{i+1}){w^-}^N=0 \;,
\end{equation}
and we then take ${w^+}^N$ as a parametric solution of ${w^-}^N$. The inhomogeneous coordinates $z$ and $w$ on $\Pbb^1_z\times\Pbb^1_w$ are given by 
\begin{equation}\label{eq:inhomogeneous coordinates on P1xP1}
    z:=\frac{z^+}{z^-} \;,
    \qquad w:=\frac{w^+}{w^-} \;,
\end{equation}
in terms of which, we can write \eqref{eq:equation of intersection in terms of w+ and w-} as\footnote{Note that the index $\lambda^\pm_i$ takes value from $i=1,\ldots,n-2$, and hence the range in \eqref{eq:equations determining the Nth roots of w}.}
\begin{equation}\label{eq:equations determining the Nth roots of w}
    w^N_{\pm,i}=-\frac{\beta_i-\lambda_i^\pm\beta_{i+1}}{\alpha_i-\lambda_i^\pm\alpha_{i+1}} \;,
    \qquad i=1,\ldots,n-2 \;.
\end{equation}
For each of $\lambda^\pm_i$, this equation has $N$ solutions, which we denote as $w_{\pm,i,I},\,I=1,\ldots,N$. Putting these solutions into 
equation \eqref{eq:equations for curves tCi in P1xP1}, and noting that $\det K_i\ne 0$, we arrive at the following equation for $z$
\begin{equation}\label{eq:z coordinates of intersection points in terms of w}
    z^N_{\pm,i,I}=\frac{\alpha_iw^N_{\pm,i,I}+\beta_i}{\gamma_iw_{\pm,i,I}^N+\delta_i} \;,
\end{equation}
which for each value of $I$ has $N$ solutions, and we denote them as $z_{\pm,i,IJ},\,I,J=1,\ldots,N$. Introducing the following notation
\begin{equation}\label{eq:Ppmi in terms of coordinates on P1xP1}
    P^\pm_{i,IJ}:=(z_{\pm,i,IJ},w_{\pm,i,I}) \;,
    \qquad 
    i=1,\ldots,n-2 \;,
    \quad
    I,J=1,\ldots,N \;,
\end{equation}
the intersection points associated with $\lambda^\pm_i$ define two divisors on $\ipc{i}$
\begin{eqaligned}\label{eq:intersection divisors for lambdaipm}
    \lambda^+_i&:&\quad \inpc{i}{i+1}&:=P^+_{i,11}+\ldots+P^+_{i,NN} \;,
    \\
    \lambda^-_i&:&\quad \inpc{i+1}{i}&:=P^-_{i,11}+\ldots+P^-_{i,NN} \;.
\end{eqaligned}
Each of these divisors contains $N^2$ distinct points and hence the total number of intersection points is $2N^2$.

\smallskip Next, we would like to impose the constraint that these two sets of $N^2$ points are exchanged by the real structure $\sigma$ in \eqref{eq:real structure on twistor space}. By a suitable ordering of the points in the two sets in \eqref{eq:intersection divisors for lambdaipm}, we get the following action of $\sigma$ on the solution set
\begin{equation}
    \sigma(P^+_{i,IJ})=P^-_{i,IJ} \;,
    \qquad 
    i=1,\ldots,n-2,\quad I,J=1,\ldots,N \;,
\end{equation}
which, by \eqref{eq:Ppmi in terms of coordinates on P1xP1} and \eqref{eq:real structure on twistor space}, gives
\begin{eqaligned}
    \sigma(z_{+,i,IJ})&=\bar{z}_{+,i,IJ}=z_{-,i,IJ} \;,
    \\
    \sigma(w_{+,i,I})&=\bar{w}_{+,i,I}=w_{-,i,I} \;.
\end{eqaligned}
As entries of $K_i$ are real, from \eqref{eq:equations determining the Nth roots of w} and $\bar{w}_{+,i,I}=w_{-,i,I}$, we arrive at
\begin{equation}\label{eq:complex conjugate of lambdapm}
    \bar{\lambda}^\pm_i=\lambda^{\mp}_i \;,
    \qquad 
    i=1,\ldots,n-2 \;.
\end{equation}
Once this condition is satisfied, \eqref{eq:z coordinates of intersection points in terms of w} implies that $\bar{z}_{+,i,IJ}=z_{-,i,IJ}$ is automatically satisfied. We thus need to make sure that \eqref{eq:complex conjugate of lambdapm} can always be satisfied. For this to happen, by \eqref{eq:values of coefficients Ai and Bi} and \eqref{eq:solutions lambdapm}, we should have 
\begin{equation}
    A_i^2-B_i=\frac{1}{4}\left(\text{Tr}(K_iK_{i+1}^{-1})\right)^2-1\le 0  \;.
\end{equation}
Since $K_i\in\text{SL}(2,\R)$, the classification of elements of $\text{SL}(2,\mbb{R})$\footnote{An element $M\in\text{SL}(2,\R)$ is called hyperbolic, parabolic, or elliptic depending on whether $|\text{Tr}M|-2$ is positive, zero, or negative. Hyperbolic elements are squeezing, parabolic elements are translation (shear), and elliptic elements are conjugate to a rotation.} implies that $K_iK_{i+1}^{-1}$ should be either an elliptic or parabolic element of $\text{SL}(2,\R)$, a condition that can always be satisfied by an appropriate choice of the matrices $K_1,\ldots,K_{n-1}$. In the following we will discuss the case where $K_iK_{i+1}^{-1}$ is elliptic, which is generically expected unless $\lambda_i^{+}=\lambda_i^{-}$.

\smallskip In summary, we demonstrated that $\ipc{i}\cap\ipc{i+1}$ consists of $2N^2$ points which can be assembled into the union of two sets $\inpc{i}{i+1}\cup \inpc{i+1}{i}$, defined in \eqref{eq:intersection divisors for lambdaipm}, with the property that 
\begin{equation}\label{eq:divisors associated with curves are exchanged by real structure}
    \sigma(\inpc{i}{i+1})=\inpc{i+1}{i} \;,
    \qquad 
    \sigma(\inpc{i+1}{i})=\inpc{i}{i+1} \;,
    \qquad 
    i=1,\ldots,n-1 \;. 
\end{equation}

\paragraph{\small\bf $\ipc{i}$ is associated with the $i$\sth simple root of $\SU{n}$.}

The curves $\{\ipc{1},\ldots,\ipc{n-1}\}$ are independent and distinct curves intersect at a finite number of points. We can label these curves as follows. Since $\mfk{su}(n)=\text{Lie}(\SU{n})$ has $n-1$ simple roots given by $\{e_1-e_2,\ldots,e_{n-1}-e_n\}$ with $e_i$ being the standard basis of $\R^n$, we can make the following assignment
\begin{equation}\label{eq:assignment of curves to simple roots of su(n)}
    \ipc{i}\longleftrightarrow e_i-e_{i+1} \;, 
    \qquad 
    i=1,\ldots,n \;. 
\end{equation}
Taking into account the properties of $\{\ipc{1},\ldots,\ipc{n-1}\}$, we see that the curves intersect at $2N^2$ points whenever they are associated with simple roots that label the adjacent nodes of $\mfk{su}(n)$ Dynkin diagram. Finally, notice that for $\ipc{i}$ and $\ipc{j}$ with $j\ne i\pm 1$, we do not impose any specific requirement of transversality and finiteness of their intersection. However, in our case, it turns out that those curves also intersect at $2N^2$ points for the same reason as for the case of $\ipc{i}$ and $\ipc{i\pm 1}$. According to \cite{MurraySinger199607}, spectral curves of a hyperbolic $\SU{n}$-monopole are labeled by simple roots of $\SU{n}$, and those associated with simple roots labeling the adjacent nodes of Dynkin diagram intersect at a finite number of points. As our main aim is to establish the connection to hyperbolic $\SU{n}$-monopoles, it motivates our assignment \eqref{eq:assignment of curves to simple roots of su(n)}.

\begin{rmk}\normalfont
    Regarding the correspondence established in \cite{Atiyah199106,AtiyahMurray1995}, the curve $\ipc{i}$ can be thought of as the curve of spectral parameter for the CPM with $\mbb{Z}_N$ spins, and as such, as the spectral curve of a hyperbolic $\SU{2}$-monopoles of magnetic charge $N$. We thus have a collection of $n-1$ curves associated with $n-1$ hyperbolic $\SU{2}$-monopoles, each of which has the charge $N$. These collection of $\SU{2}$ monopole configurations realizes a hyperbolic $\SU{n}$-monopole through the embedding of the $i$\sth hyperbolic $\SU{2}$-solution, called a fundamental monopole, along the $i$\sth simple root of $\SU{n}$ \cite{Weinberg198005,Weinberg198208}.  \qed 
\end{rmk}

\paragraph{$\ipc{i}$ and $\ipc{i\pm 1}$ determine special line bundles.} We would like to establish that the divisors associated with the intersection of curves $\ipc{i}$ define special line bundles. To see this, let us set $N=1$ in \eqref{eq:equations for curves tCi in P1xP1}. These are linear equations in $\Pbb^1\times\Pbb^1$, which in the homogeneous coordinates are given by
\begin{equation}
    [z_+:z_-]=[\alpha_iw_++\beta_i:\gamma_iw+\delta_i] \;, 
    \qquad 
    i=1,\ldots,n-1 \;,
\end{equation}
and imply $\lambda z=(\alpha_iw+\beta_i)/(\gamma_iw+\delta_i)$ for $\lambda\in\C^\times$ or the set of points $[z:\lambda z]\in\Pbb^1\times\Pbb^1$. Hence, these equations determine $n-1$ copies of $\Pbb^1$. The divisors $\inpc{i}{i+1}$ and $\inpc{i+1}{i}$ defined in \eqref{eq:intersection divisors for lambdaipm} for a fixed value of $i$ and $N=1$ contain a single point, and hence each has degree $+1$. Therefore, the corresponding line bundle on $\Pbb^1$ is $\mcal{O}_{\mbb{P}^1}(1)$. On the other hand, under the embedding $\iota:\Pbb^1\hookrightarrow\Pbb^1\times\Pbb^1$, we have $\iota^*\mcal{O}_{\ts}(m,n)\simeq\mcal{O}_{\Pbb^1}(m+n)$. Therefore, the corresponding line bundle on $\Pbb^1\times\Pbb^1$, restricted to $S_i$, is either $\mcal{O}_{\ts}(0,1)$ or $\mcal{O}_{\ts}(1,0)$. Recall that the curves $\ipc{i}$ are real in the sense of \eqref{eq:demanded reality condition on gCPM curves}, hence $\sigma^*(\mcal{O}_{\ts}(1,0))|_{\ipc{i}}=\mcal{O}_{\ts}(1,0)|_{\ipc{i}}\simeq\mcal{O}_{\ts}(0,1)|_{\ipc{i}}$ where we have used \eqref{eq:holomorphic isomorphism on a curve Sigmai}. Using \eqref{eq:divisors associated with curves are exchanged by real structure}, we therefore get the isomorphism $[\Sigma_{1,n}^{{(i,i+1)}}]\simeq\mcal{O}_{\ts}(0,1)|_{\ipc{i}}\simeq L^{-\frac{1}{2}}(1)|_{\ipc{i}}$ and $[\Sigma_{1,n}^{{(i+1,i)}}]\simeq\mcal{O}_{\ts}(1,0)|_{\ipc{i}}\simeq L^{+\frac{1}{2}}(1)|_{\ipc{i}}$ where the notation $[\cdots]$ denotes the line bundle associated with a divisor and we have used \eqref{eq:basis for holomorphic line bundles}.

\smallskip Switching to the general $N$, from \eqref{eq:equations for curves tCi in P1xP1}, \eqref{eq:entries of matrix Ki}, and \eqref{eq:inhomogeneous coordinates on P1xP1}, the curves are given by
\begin{equation}
    \gamma_iz^Nw^N+\delta_iz^N-\alpha_iw^N-\beta_i=0 \;, 
    \qquad i=1,\ldots,n-1 \;. 
\end{equation}
These are equations for $N$-fold branched covers of $\Pbb^1$ as the coordinate transformations $(z,w)\mapsto (z^N,w^N)$ shows. Therefore, from the discussion of the case of $N=1$, the degree of the pull-back of $\mcal{O}_{\ts}(0,1)$ to the cover is simply $N$, and hence we can identify it as
\begin{eqaligned}\label{eq:special line bundles defined by gCMP curves I}
    [\inpc{i}{i+1}]&\simeq\mcal{O}_{\ts}(N,0)\big|_{\ipc{i}}\simeq L^{-\frac{N}{2}}(N)\big|_{\ipc{i}} \;, 
    \\
    [\inpc{i+1}{i}]&\simeq\mcal{O}_{\ts}(0,N)\big|_{\ipc{i}}\simeq L^{+\frac{N}{2}}(N)\big|_{\ipc{i}}\simeq L^{-\frac{N}{2}}(N)\big|_{\ipc{i}} \;,
\end{eqaligned}
where we have used \eqref{eq:Lp(q) in terms of O(q_+,q_-)} and \eqref{eq:holomorphic isomorphism on a curve Sigmai}.

\begin{rmk}\normalfont
    Let us now explain the following isomorphism
    \begin{equation}\label{eq:holomorphic isomorphism on a curve Sigmai}
        \mcal{O}_{\ts}(kN,k'N)\big|_{\ipc{i}} \simeq \mcal{O}_{\ts}((k_--k)N,(k_++k)N)\big|_{\ipc{i}},
    \end{equation}
    for integers $k,k_\pm$ as holomorphic bundles. Consider the following sequence of maps
    \begin{equation}
        \ipc{i}\xrightarrow[]{c^{(i)}_{N,N}} \mbf{C}^{(i)}\xrightarrow[]{\iota^{(i)}}\ts,
    \end{equation}
    where $c_{N,N}^{(i)}$ is the branched covering map of degree $N^2$,\footnote{For more on this map, see \S\ref{sec:curve of spectral parameter as a branched cover of P1} and Appendix \ref{sec:basic facts about complete intersections}.} and $\iota^{(i)}$ is the embedding in $\ts$. $\mbf{C}^{(i)}\simeq\Pbb^1$. From 
    \begin{eqaligned}
        \#([\mbf{C}^{(i)}]\cap[\Pbb^1\times\{\text{pt}\}])&=1
        \\
        \#([\mbf{C}^{(i)}]\cap[\{\text{pt}\}\times\Pbb^1])&=1
    \end{eqaligned}
    where $\#$ denotes the intersection number, on $\ts$, $[\mbf{C}^{(i)}]\simeq\mcal{O}_{\ts}(1,1)$. $\mbf{C}^{(i)}\cap \mbf{C}^{(i+1)}$ is a divisor on $\mbb{C}^{(i)}$ or $\mbf{C}^{(i+1)}$, whose corresponding line bundle $[\mbf{C}^{(i)}\cap \mbf{C}^{(i+1)}]$ is $\mcal{O}_{\ts}(1,1)$, restricted to either of these curves. Then,
    \begin{eqaligned}
        [\ipc{i}\cap\ipc{i+1}]&=c^{(i)*}_{N,N}[\mbf{C}^{(i)}\cap \mbf{C}^{(i+1)}]
        \\
        &=c^{(i)*}_{N,N}\mcal{O}_{\ts}(1,1)\big|_{\mbf{C}^{(i)}}
        \\
        &=\mcal{O}_{\ts}(N,N)\big|_{\ipc{i}}.
    \end{eqaligned}
    From the isomorphism of $\mbb{C}^{(i)}$ and the fact that $\mbf{C}^{(i)}\cap \mbf{C}^{(i+1)}$, we see that one of these points correspond to $\mcal{O}_{\Pbb^1}(1)$, hence we have
    \begin{equation}
        \mcal{O}_{\ts}(1,0)\big|_{\mbf{C}^{(i)}}\simeq\mcal{O}_{\Pbb^1}(1)\simeq \mcal{O}_{\ts}(0,1)\big|_{\mbf{C}^{(i)}},
    \end{equation}
    a fact that on $\ipc{i}$ translates to 
    \begin{equation}
        \mcal{O}_{\ts}(N,0)\big|_{\ipc{i}}\simeq 
        \mcal{O}_{\ts}(0,N)\big|_{\ipc{i}}.
    \end{equation}
    More generally for any integer $r$
    \begin{equation}
        \mcal{O}_{\ts}(r,0)\big|_{\mbf{C}^{(i)}}\simeq \mcal{O}_{\ts}(0,r)\big|_{\mbf{C}^{(i)}},
    \end{equation}
    which leads to
    \begin{equation}
        \mcal{O}_{\ts}(Nr,0)\big|_{\ipc{i}}\simeq 
        \mcal{O}_{\ts}(0,Nr)\big|_{\ipc{i}}.
    \end{equation}
    Note that
    \begin{eqaligned}\label{eq:equivalence of bundles O(rN,0) and O((r-s)N,sN) on curves Sigmai}
        \mcal{O}_{\ts}(Nr,0)\big|_{\ipc{i}}&\simeq \mcal{O}_{\ts}((r-1)N,0)\big|_{\ipc{i}}\otimes \mcal{O}_{\ts}(N,0)\big|_{\ipc{i}}
        \\
        &=\mcal{O}_{\ts}((r-1)N,N)\big|_{\ipc{i}}
        \\
        &=\mcal{O}_{\ts}((r-s)N,sN)\big|_{\ipc{i}}.
    \end{eqaligned}
    Even more generally, we can start from the bundle $\mcal{O}_{\ts}(k_+,k_-)|_{\mbf{C}^{(i)}}$, which leads to the bundle $\mcal{O}_{\ts}(k_+N,k_-N)|_{\ipc{i}}$ upon pulling back by the branched covering map $c^{(i)}_{N,N}$. A computation similar to \eqref{eq:equivalence of bundles O(rN,0) and O((r-s)N,sN) on curves Sigmai} shows the holomorphic isomorphism \eqref{eq:holomorphic isomorphism on a curve Sigmai}. In particular, we have
    \begin{eqaligned}
        \mcal{O}_{\ts}(N,0)\big|_{\ipc{i}}&\simeq \mcal{O}_{\ts}(0,N)\big|_{\ipc{i}},
        \\
        \mcal{O}_{\ts}(2N,0)\big|_{\ipc{i}}&\simeq \mcal{O}_{\ts}(N,N)\big|_{\ipc{i}}\simeq \mcal{O}_{\ts}(0,2N)\big|_{\ipc{i}}.
    \end{eqaligned}
    We would like to emphasize that \eqref{eq:holomorphic isomorphism on a curve Sigmai} does not hold away from $\ipc{i}$, generically. Instead, as the line bundle $L$, defined in \eqref{eq:line bundle L} is topologically trivial, \eqref{eq:topological equivalence of O(n,0) and O(0,n)} holds.
    \qed
\end{rmk}

\smallskip Finally, due to the correspondence between line bundles and divisors on a curve $X$, we have the following
\begin{equation}
    D_1+D_2\longleftrightarrow [D_1]\otimes_{\mcal{O}_X}[D_2] \;,
\end{equation}
for any two divisors $D_1$ and $D_2$ on $X$. Applying this to our case, we get (where all bundles are restricted to $\ipc{i}$) 
\begin{eqaligned}\label{eq:special line bundles defined by gCMP curves II}
    [\inpc{i}{i+1}+\inpc{i+1}{i+2}]&=[\inpc{i}{i+1}]\otimes_{\mcal{O}_{\ts}}[\inpc{i+1}{i+2}]
    \\
    &=\mcal{O}_{\ts}(N,0)\otimes_{\mcal{O}_{\ts}}\mcal{O}_{\ts}(N,0)
    \\
    &=\mcal{O}_{\ts}(N,0)\otimes_{\mcal{O}_{\ts}}\mcal{O}_{\ts}(0,N)
    \\
    &\simeq\mcal{O}_{\ts}(N,N) \;,
\end{eqaligned}
where in the third line, we have used \eqref{eq:holomorphic isomorphism on a curve Sigmai}. The line bundles \eqref{eq:special line bundles defined by gCMP curves I} associated with the divisors \eqref{eq:intersection divisors for lambdaipm} are the special line bundles we were looking for.

\paragraph{\small $\ipc{i}$ is compact.} $\ipc{i}$ is the solution of an algebraic equation \eqref{eq:equations for curves tCi in P1xP1} in $\Pbb^1\times\Pbb^1$. Therefore, it is a closed set in the Zariski topology\footnote{Recall that the Zariski closed sets on $\Pbb^1\times\Pbb^1$ are given precisely by the zero locus of homogeneous polynomials of some degree in homogeneous coordinates on the two copies of $\Pbb^1$.} on $\Pbb^1\times\Pbb^1$, and as such compact. 
\qed

\begin{rmk}\normalfont
    For compactness, arguments similar to \cite[First Lemma of \S 5]{MurraySinger199607} can also be used. 
\end{rmk}

In summary, from the curve of the spectral parameter of the gCPM $\pc$ defined in \eqref{eq:curve of spectral parameter of gCPM}, we have constructed $n-1$ curves $\{\ipc{1},\ldots,\ipc{n-1}\}$ whose equations are given by \eqref{eq:equations for curves tCi in P1xP1} (where the $\mbb{Z}_N\times\mbb{Z}_N$ quotient in \eqref{eq:Zn quotient of ith tilde curve} is understood). These curves are of bi-degree $(N,N)$ in $\Pbb^1\times\Pbb^1$, real, in the sense of \eqref{eq:demanded reality condition on gCPM curves}, and compact. Their intersection defines divisors \eqref{eq:intersection divisors for lambdaipm} whose corresponding line bundles upon restriction to $\ipc{i}$ satisfy \eqref{eq:special line bundles defined by gCMP curves I}. 

\paragraph{Step 3: Identifying the Hyperbolic $\SU{n}$-Monopole Solution.} We now would like to show that the curves $\{\ipc{1},\ldots,\ipc{n-1}\}$ constructed above determine the spectral data of a particular class of hyperbolic $\SU{n}$-monopoles. By a reordering of the curves, if necessary, we claim
\begin{equation}\label{eq:claimed isomorphism between curves}
    \ipc{i}\simeq S_i \;, \qquad i=1,\ldots,n-1 \;,
\end{equation}
for some specific choice of magnetic charges, related to $N$, of the hyperbolic $\SU{n}$-monopole determined by curves $\{S_1,\ldots,S_{n-1}\}$.

\paragraph{\small Proof of \eqref{eq:claimed isomorphism between curves}.} We present the argument by comparing the properties of curves in the two sides of \eqref{eq:claimed isomorphism between curves} and then identifying these properties as follows. 

\begin{itemize}    
    
    \item[--] As we have explained in \S\ref{sec:spectral data and recovering monopole solutions}, for the case of maximal symmetry-breaking, there are $n-1$ spectral curves given by \eqref{eq:spectral curves of hm as divisors of gamma-}. Therefore, our curves correspond to such a monopole configuration.

    \item[--] The spectral curves $S_i$ are sections of $\mcal{O}_{\ts}(m_i,m_i)$ while $\ipc{i}$ are sections of $\mcal{O}_{\ts}(N,N)$, both on $\ts$. Therefore, the monopole configuration corresponding to spectral curves $\{\ipc{1},\ldots,\ipc{n-1}\}$ has the following magnetic charges
    \begin{equation}\label{eq:magnetic charges for curves coming from gCPM}
        m_i=N \;, \qquad i=1,\ldots,n-1 \;. 
    \end{equation}
    Then, \eqref{eq:genus of Si} and \eqref{eq:genus of sigmai} imply that $S_i$ and $\ipc{i}$ have the same genus.
    
    \item [--] Both $S_i$ (see \cite{MurraySinger199607}) and $\ipc{i}$ (see \eqref{eq:assignment of curves to simple roots of su(n)}) are associated with simple root of $\SU{n}$. Furthermore, those curves corresponding to adjacent nodes of the Dynkin diagram intersect at $2m_im_{i+1}=2N^2$ points, where we have used \eqref{eq:magnetic charges for curves coming from gCPM}. 
  
    \item[--] Both $S_i$ and $\ipc{i}$ are real curves in the sense of \eqref{eq:reality of spectral curves of hyperbolic monopoles} and \eqref{eq:demanded reality condition on gCPM curves}, respectively, and also compact. 

    \item[--] Consider the line bundle \eqref{eq:trivial bundle over ith spectral curve} with the values of $m_i$ given in \eqref{eq:magnetic charges for curves coming from gCPM}
    \begin{equation}\label{eq:line bundles on curves S_i}
        L^{p_{i+1}+k_{i+1}/2-p_i-k_i/2}(m_{i-1}+m_{i+1})\Big|_{S_i}\simeq [S_{i-1,i}+S_{i,i+1}] \;,
    \end{equation}
    On the other hand, \eqref{eq:magnetic charges for su(n) monopole} and \eqref{eq:magnetic charges for curves coming from gCPM} imply 
    \begin{equation}
        k_1=N \;, \qquad k_2=\ldots=k_{n-1}=0 \;, \qquad k_n=-N \;, 
    \end{equation}
    where the last relation follows from $m_n=0$. Then, \eqref{eq:line bundles on curves S_i} implies
    \begin{eqaligned}\label{eq:line bundles over spectral curves of a monopole}
        L^{p_2-p_1-N/2}(N)\Big|_{S_1}&\simeq [S_{1,2}] \;,
        \\
        L^{p_3-p_2}(2N)\Big|_{S_2}&\simeq [S_{1,2}+S_{2,3}] \;,
        \\
        \qquad &\vdots 
        \\
        L^{p_{n-1}-p_{n-2}}(2N)\Big|_{S_{n-2}}&\simeq [S_{n-3,n-2}+S_{n-2,n-1}] \;,
        \\
        L^{p_{n}-N/2-p_{n-1}}(N)\Big|_{S_{n-1}}&\simeq [S_{n-2,n-1}] \;,
    \end{eqaligned}
    while \eqref{eq:special line bundles defined by gCMP curves I} and \eqref{eq:special line bundles defined by gCMP curves II} give
    \begin{eqaligned}\label{eq:line bundles over curves coming from gCPM}
        L^{-N/2}(N)\Big|_{\ipc{1}}&\simeq [\inpc{1}{2}]\simeq\mcal{O}_{\ts}(N,0) \;,
        \\
        L^0(2N)\Big|_{\ipc{2}}&\simeq [\inpc{2}{1}+\inpc{2}{3}]\simeq\mcal{O}_{\ts}(N,N) \;,
        \\
        \qquad &\qquad \vdots 
        \\
        L^0(2N)\Big|_{\ipc{n-2}}&\simeq [\inpc{n-3}{n-2}+\inpc{n-2}{n-1}]\simeq\mcal{O}_{\ts}(N,N) \;,
        \\
        L^{-N/2}(N)\Big|_{\ipc{n-1}}&\simeq [\inpc{n-2}{n-1}]\simeq\mcal{O}_{\ts}(N,0)\Big|_{\ipc{n-1}}\simeq \mcal{O}_{\ts}(0,N)\Big|_{\ipc{n-1}} \;,
    \end{eqaligned}
    where we have used \eqref{eq:holomorphic isomorphism on a curve Sigmai} and \eqref{eq:Lp(q) in terms of O(q_+,q_-)}, from which it follows that $L^{-N/2}(N)\simeq\mcal{O}_{\ts}(0,N)$ and $L^0(2N)\simeq\mcal{O}_{\ts}(N,N)$. Therefore, if we take the limit
    \begin{equation}
        p_i\to 0 \;, \qquad i=1,\ldots,n-1 \;,
    \end{equation}
    we see that both sets of curves give the same line bundles upon restriction to the corresponding curves. Note that can identify the divisors $S_{i,i+1}$ and $S_{i+1,i}$ with those defined in \eqref{eq:intersection divisors for lambdaipm}
    \begin{equation}\label{eq:identification of divisors defined by curves Sigmai and Si}
        \inpc{i}{i+1}=S_{i,i+1} \;,
        \qquad \inpc{i+1}{i}=S_{i+1,i} \;,
        \qquad i=1,\ldots,n-1 \;,
    \end{equation}
    and hence the corresponding line bundles
    \begin{equation}\label{eq:identification of line bundles defined by curves Sigmai and Si}
        [\inpc{i}{i+1}]=[S_{i,i+1}] \;,
        \qquad [\inpc{i+1}{i}]=[S_{i+1,i}] \;,
        \qquad i=1,\ldots,n-1 \;.
    \end{equation}
    
\end{itemize}

Putting all these together, we see that the set of curves $\{\ipc{1},\ldots,\ipc{n-1}\}$ indeed define the spectral data of a configuration of flat hyperbolic $\SU{n}$-monopoles. We thus proved the claimed identification \eqref{eq:claimed isomorphism between curves}. \qed 

\subsubsection{gCPM from Hyperbolic \texorpdfstring{$\SU{n}$}{SU(n)}-Monopole}\label{sec:gCPM from hyperbolic SU(n) monopoles}

 So far, we have seen that starting from the curve of the spectral parameter of the gCPM, we can construct the spectral data of a hyperbolic $\SU{n}$-monopole. A natural question is to reverse the construction: we start from the spectral data of a hyperbolic $\SU{n}$-monopole and construct the curve of the spectral parameter of the gCPM. 

\smallskip Consider the spectral data of a hyperbolic $\SU{n}$-monopole with a maximal symmetry-breaking $\SU{n}\to\U{1}^{n-1}$ at $S^2_\infty$. This means (1) the curves $\{S_1,\ldots,S_{n-1}\}$ where $S_i$ is a section of $\mcal{O}_{\ts}(m_i,m_i)$; Therefore, they are described by equations of the form 
\begin{equation}\label{eq:equation for spectral curves of hyperbolic su(n)-monopoles}
    S_i:\quad F_i(z,w)=\sum_{\alpha,\beta=0}^{m_i}a_{\alpha,\beta}^{(i)} z^{\alpha}w^{\beta}=0 \;, 
    \qquad 
    i=1,\ldots,n-1 \;,
\end{equation}
where $z$ and $w$ are defined in \eqref{eq:inhomogeneous coordinates on P1xP1}, and with the assumption that
\begin{equation}
    a^{(i)}_{m_i,m_i}\ne 0 \;, 
    \qquad 
    i=1,\ldots,n-1 \;,
\end{equation}
and (2) $S_i\cap S_{i+1}$ consists of $2m_im_{i+1}$ points. The asymptotic values of Higgs fields are $\{p_1,\ldots,p_n\}$ subject to the constraint $\sum_{i=1}^np_i=0$. The construction of the corresponding curve of the spectral parameter of the gCPM goes as follows:

\paragraph{\small $S_i$ Determines a $\mbb{Z}_{m_i}\times\mbb{Z}_{m_i}$-Invariant Curve.} From $S_i$, we construct a $\mbb{Z}_{m_i}\times\mbb{Z}_{m_i}$-symmetric curve as follows. $\mbb{Z}_{m_i}\times\mbb{Z}_{m_i}$-invariance simply means
\begin{equation}
    F_i(\omega z,\omega' w)=F_i(z,w) \;, \qquad i=1,\ldots,n-1 \;,
\end{equation}
for two $m_i$\sth roots of unity $\omega$ and $\omega'$. From \eqref{eq:equation for spectral curves of hyperbolic su(n)-monopoles}, $F_i(z,w)$ transforms as
\begin{equation}
    F_i(\omega z,\omega' w)=\sum_{\alpha,\beta=0}^{m_i}a_{\alpha,\beta}^{(i)} \omega^\alpha\omega'^\beta z^{\alpha}w^{\beta} \;,
\end{equation}
and equating this with $F_i(z,w)$ gives the following constraints
\begin{equation}
    \omega^\alpha\omega'^\beta=1 \;,
\end{equation}
Setting $\omega=\exp(2\pi\mfk{i}s/m_i)$ and $\omega'=\exp(2\pi\mfk{i}s'/m_i)$ for $s,s'=0,\ldots,m_i-1$, we arrive at the following relation
\begin{equation}
    s\alpha+s'\beta=0 \;, \qquad \mod{m_i}\;.
\end{equation}
This relation must hold for all $s,s'=0,\ldots,m_i-1$. Hence, we conclude
\begin{equation}
    \alpha,\beta=0,m_i \;.
\end{equation}
Therefore, the $\mbb{Z}_{m_i}\times\mbb{Z}_{m_i}$-invariant curve $\ismc{i}$ constructed out of $S_i$ has the following equation
\begin{equation}\label{eq:equation for symmetric monopole curve}
    \ismc{i}:\quad F^{\text{sym}}_i(z,w)=\gamma_i z^{m_i}w^{m_i}+\delta_iz^{m_i}-\alpha_iw^{m_i}-\beta_i=0 \;,
    \qquad i=1,\ldots,n-1 \;,
\end{equation}
for some constants $\alpha_i,\beta_i,\gamma_i,\delta_i$. Recall that $S_i$ (and $\ismc{i}$) is real in the sense of \eqref{eq:reality of spectral curves of hyperbolic monopoles}, therefore $\alpha_i,\beta_i,\gamma_i,\delta_i\in\R$. Hence, we can write the equation in a more succinct form
\begin{equation}\label{eq:equation for symmetrized curves}
    S^\text{sym}_i:\quad 
    \begin{pmatrix}
        {z^+}^{m_i} 
        \\
        {z^-}^{m_i} 
    \end{pmatrix}
    =K_i
    \begin{pmatrix}
        {w^+}^{m_i}
        \\
        {w^-}^{m_i} 
    \end{pmatrix} \;,
    \qquad i=1,\ldots,n-1 \;,
\end{equation}
where the homogeneous coordinates $z^\pm$ and $w^\pm$ on $\Pbb^1\times\Pbb^1$ are defined in \eqref{eq:inhomogeneous coordinates on P1xP1}, and $K_i$ is of the form \eqref{eq:entries of matrix Ki}. 

\paragraph{\small $\{\ismc{1},\ldots,\ismc{n-1}\}$ Define a Curve in the Product of $n$ Copies of $\Pbb^1$.} Consider the product of $n$ copies of $\Pbb^1$. For the sake of distinguishing them, we denote them as $\Pbb^1_i$ for $i=1,\ldots,n$, their homogeneous coordinates as $z^{\pm}_i$, and their inhomogeneous coordinate as $z_i:=z^+_i/z^-_i$. Then, we have the following embeddings
\begin{equation}
    \iota_i:\ts\hookrightarrow \Pbb^1_1\times\ldots\times\Pbb^1_{i}\times\Pbb^1_{i+1}\times\ldots\times\Pbb^1_n \;,
    \qquad i=1,\ldots,n-1\;.
\end{equation}
$\iota_i$ embeds $\ts$ inside $\Pbb^1_i\times\Pbb^1_{i+1}$ and is given explicitly as
\begin{equation}
    \iota_i(z,w):=(0,\dots,0,z,w,0,\ldots,0) \;. 
\end{equation}
Under this embedding, the equation \eqref{eq:equation for symmetrized curves} of $\ismc{i}$ becomes
\begin{equation}\label{eq:equation for ith Ssym curve embedded in n copies of P1}
    S^\text{sym}_{n,\iota_i}:\quad 
    \begin{pmatrix}
        (z^+_i)^{m_i}
        \\
        (z^-_i)^{m_i} 
    \end{pmatrix}
    =K_{i,i+1}
    \begin{pmatrix}
        (z^+_{i+1})^{m_i}
        \\
        (z^-_{i+1})^{m_i}
    \end{pmatrix} \;,
    \qquad i=1,\ldots,n-1 \;,
\end{equation}
where we have redefined $K_i\to K_{i,i+1}$. Considering all these curves, we can consider their union, which is again a curve in $\Pbb^1\times\ldots\Pbb^1$
\begin{equation}\label{eq:curve in the product of n copies of P1}
    S^{\text{sym}}_{n;m_1,\ldots,m_{n-1}}:=\bigcup_{i=1}^{n-1}S_{n,\iota_i}^{\text{sym}} \;.
\end{equation}
Let us compute the genus of this curve. We first compute the degree of its canonical bundle. Let $P_i$ be the divisor in $\mbb{P}^1_1\times\ldots\times\mbb{P}^1_n$ pull-backed from a point in $\mbb{P}^1_i$, and consider the embedding $\iota:S^{\text{sym}}_{n;m_1,\ldots,m_{n-1}}\hookrightarrow \Pbb^1_1\times\ldots\times\Pbb^1_n$. By the adjunction formula, the canonical bundle of the curve is
\begin{equation}
    K_{S^{\text{sym}}_{n;m_1,\ldots,m_{n-1}}}=\iota^*\left(K_{\mbb{P}^1_1\times\ldots\times\mbb{P}^1_n}\otimes_{\mcal{O}_{\Pbb^1_1\times\ldots\times\Pbb^1_n}} \bigotimes_{i=1}^n\mcal{O}(m_i P_i+m_i P_{i+1})\right) \;,
\end{equation}
where 
\begin{equation}
    \mcal{O}(m_i P_i+m_i P_{i+1}):=\mcal{O}_{\mbb{P}^1_1\times\ldots\times\mbb{P}^1_n}(0,\ldots,0,m_i P_i,m_i P_{i+1},0,\ldots,0) \;,
\end{equation}
with no summation over $i$ is understood and
\begin{equation}
    \mcal{O}_{\Pbb^1_1\times\ldots\times\Pbb^1_n}(s_1,\ldots,s_n):=\bigotimes_{i=1}^n\pi^*_i\mcal{O}_{\Pbb^1_i}(s_i), \qquad s_1,\ldots,s_n\in\mbb{Z} \;,
\end{equation}
with $\pi_i:\Pbb^1_1\times\ldots\times\Pbb^1_n\to\Pbb^1_i$ is the natural projection. Using the fact\footnote{Recall that the canonical bundle of $\Pbb^1$ is $\mcal{O}_{\Pbb^1}(-2P)$ for any point $P\in\Pbb^1$.}
\begin{equation}
    K_{\Pbb^1_1\times\ldots\times\Pbb^1_n}=\mcal{O}_{\Pbb^1_1\times\ldots\times\Pbb^1_n}(-2P_1,\ldots,-2P_n) \;,
\end{equation}
we have
\begin{equation}\label{eq:canonical bundle of the curve}
    K_{S^{\text{sym}}_{n;m_1,\ldots,m_{n-1}}}=\iota^*\mcal{O}_{\Pbb^1_1\times\ldots\times\Pbb^1_n}(D) \;,
\end{equation}
where
\begin{equation}
    D:=((m_1-2)P_1,(m_2+m_1-2)P_2,\ldots,(m_{n-1}+m_{n-2}-2)P_{n-1},(m_{n-1}-2)P_n) \;.
\end{equation}
On the other hand, $\deg\iota=m_1\ldots m_{n-1}$. We thus can compute the degree of $K_{S^{\text{sym}}_{n;m_1,\ldots,m_{n-1}}}$ from \eqref{eq:canonical bundle of the curve} to be ($m_n=0$)
\begin{eqaligned}\label{eq:degree of the canonical bundle of the curve in the product of n copies of P1}
    \deg K_{S^{\text{sym}}_{n;m_1,\ldots,m_{n-1}}}&=\deg\iota\times \deg\mcal{O}_{\Pbb^1_1\times\ldots\times\Pbb^1_n}(D)
    \\
    &=\deg\iota\times\deg D
    \\
    &=2m_1\ldots m_{n-1}\left(\sum_{i=1}^{n-1}m_i-n\right) \;.
\end{eqaligned}
By equating this with $2g_{S^{\text{sym}}_{n;m_1,\ldots,m_{n-1}}}-2$, we get $g_{S^{\text{sym}}_{n;m_1,\ldots,m_{n-1}}}$, the genus of $S^{\text{sym}}_{n;m_1,\ldots,m_{n-1}}$, to be
\begin{equation}
    g_{S^{\text{sym}}_{n;m_1,\ldots,m_{n-1}}}=m_1\ldots m_{n-1}\left(\sum_{i=1}^{n-1}m_i-n\right)+1 \;.
\end{equation}

\begin{rmk}\normalfont 
    The curve we eventually compare with the curve of the spectral parameter of the gCPM is the $\mbb{Z}_{m_1}\times\ldots\times\mbb{Z}_{m_{n-1}}$ covering of $S^{\text{sym}}_{n;m_1,\ldots,m_{n-1}}$. Denoting that curve by $\wt{S}^{\text{sym}}_{n;m_1,\ldots,m_{n-1}}$, we have
    \begin{equation}\label{eq:Zm1x...xZmn-1 cover of Snm1...mn-1}
        S^{\text{sym}}_{n;m_1,\ldots,m_{n-1}}=\wt{S}^{\text{sym}}_{n;m_1,\ldots,m_{n-1}}/(\mbb{Z}_{m_1}\times\ldots\times\mbb{Z}_{m_{n-1}}) \;.
    \end{equation}
    From \eqref{eq:degree of the canonical bundle of the curve in the product of n copies of P1}, we have
    \begin{equation}
        \deg K_{\wt{S}^{\text{sym}}_{n;m_1,\ldots,m_{n-1}}}=2m^2_1\ldots m^2_{n-1}\left(\sum_{i=1}^{n-1}m_i-n\right) \;,
    \end{equation}
    hence its genus is
    \begin{equation}\label{eq:genus of Zm1x...Zmn-1 cover of the curve in product of n copies of P1 with of arbitrary degree}
        g_{\wt{S}^{\text{sym}}_{n;m_1,\ldots,m_{n-1}}}=m^2_1\ldots m^2_{n-1}\left(\sum_{i=1}^{n-1}m_i-n\right)+1 \;.
    \end{equation}
\end{rmk}

\paragraph{Recovering the Curve of the Spectral Parameter of the gCPM.} The comparison between \eqref{eq:equation for ith Ssym curve embedded in n copies of P1} and \eqref{eq:equations for curves tCi in P1xP1} shows that to connect the curve \eqref{eq:curve in the product of n copies of P1} to the curve of spectral parameter of the gCPM, we have to take all the magnetic charges to be the same
\begin{equation}
    m_1=\ldots=m_{n-1}=N. 
\end{equation}
This would then implement certain restrictions on the matrices $K_{i,i+1}$. We can summarize them as 
\begin{eqgathered}
    K_{i,i}=\mds{1} \;, 
    \qquad K_{i,j}K_{j,i}=\mds{1} \;,
    \qquad K_{i,j}K_{j,k}K_{k,i}=\mds{1} \;.
\end{eqgathered}
The second condition implies that the matrices $K_{i,j}$ have to be invertible. Taking the determinant of both sides of the cocycle condition tells us that all $K_i$s should have a unit determinant for the cocycle condition to be satisfied for all values of $i,j,$ and $k$. Therefore, $K_i$s belong to $\text{SL}(2,\mbb{R})$ subject to the redundancy \eqref{eq:redundancy in the curve of spectral parameter of gCPM}. Furthermore, the discussion in \hyperlink{step1}{Step 1: Constructing $n-1$ Curves in $\Pbb^1\times\Pbb^1$ out of $\pc$} in \S\ref{sec:hyperbolic monopoles from gCPMs} shows that the curve $\wt{S}^{\text{sym}}_{n;m_1,\ldots,m_{n-1}}$, defined in \eqref{eq:Zm1x...xZmn-1 cover of Snm1...mn-1}, can indeed be thought of as a curve in $\Pbb^{2n-1}$ whose genus can be computed from \eqref{eq:genus of Zm1x...Zmn-1 cover of the curve in product of n copies of P1 with of arbitrary degree} to be
\begin{equation}
    g_{\wt{S}^{\text{sym}}_{n;N,\ldots,N}}=N^{2(n-1)}(N(n-1)-n)+1 \;,
\end{equation}
which coincides with \eqref{eq:genus of the curve of spectral parameter of gCPM}. Finally, for \eqref{eq:line bundles over spectral curves of a monopole} and \eqref{eq:line bundles over curves coming from gCPM} to match, all $p_i$s have to be sent to zero. We thus successfully recovered the curve of the spectral parameter of a particular class of the gCPMs
\begin{equation}\label{eq:identification of curves constructed from spectral data of monopoles with gCPM curve}
    \wt{S}^{\text{sym}}_{n;N,\ldots,N}\simeq\cpc \;,
\end{equation}
where $\cpc$ is defined in \eqref{eq:curve of spectral parameter of gCPM}. Let us summarize our findings in the form of a proposition

\begin{prop}\normalfont\label{prop:hyperbolic monopole/gCPM correspondence}
    There is a one-to-one correspondence between 
    \begin{enumerate}
        \item [--] the curve $\cpc$ of the spectral parameter of the gCPM defined by $n-1$ $\text{SL}(2,\R)$-valued $K$-matrices $K_{i,i+1},\,i=1,\ldots,n-1$ in \eqref{eq:curve of spectral parameter of gCPM} such that $K_iK_{i+1}^{-1}$ with $K_i:=K_{i,i+1}$, is an elliptic element of $\text{SL}(2,\R)$; and

        \item [--] a generic hyperbolic $\SU{n}$-monopole  subject to maximal symmetry-breaking, $n-1$ magnetic charges $m_1=\ldots=m_{n-1}=N$, and vanishing boundary values of the Higgs field whose spectral data defines the curve \eqref{eq:Zm1x...xZmn-1 cover of Snm1...mn-1} which, upon setting the parameters, can be identified with $\cpc$, as in \eqref{eq:identification of curves constructed from spectral data of monopoles with gCPM curve}. 
    \end{enumerate}
\end{prop}

\subsection{Generalization to Other Lie Algebras}\label{sec:correspondence for classical and exceptional groups}

In this section, we explain how to extend the results on hyperbolic monopoles to other classical and exceptional Lie algebras. Spectral data of Euclidean monopoles associated with classical Lie groups is constructed in \cite{HurtubiseMurray1989}. The method of loc.\ cit.\ can be applied to hyperbolic monopoles associated with classical Lie groups, as has been used for the case of $\mfk{su}(n)$ by Murray and Singer in \cite{MurraySinger199607}. We use this construction for other classical Lie algebras, $\mfk{so}(n)$ and $\mfk{sp}(n)$ in this section. The basic idea is that the spectral curves $S^{\mfk{g}}_i$ of a $\mfk{g}$-monopole, where $\mfk{g}$ is a classical Lie algebra, are labeled by fundamental weights of $\mfk{g}$, which we denote as $\omega^{\mfk{g}}_i$ (see \hyperlink{spectral data}{Spectral Data} in \S\ref{sec:spectral data and recovering monopole solutions}). Once we embed $\mfk{g}$ inside $\mfk{su}(n)$, the main point is to find the relation between the fundamental weights of the algebra and $\mfk{su}(n)$ under this embedding. The number of spectral curves depends on the symmetry-breaking pattern at infinity. We are exclusively concerned with the maximal symmetry-breaking case in which there are $\rnk{\mfk{g}}$ spectral curves $\{S^{\mfk{g}}_1,\ldots,S^{\mfk{g}}_{\rnk{\mfk{g}}}\}$ and $\rnk{g}$ number of magnetic charges $(m_1^{\mfk{g}},\ldots,m_{\rnk{g}}^{\mfk{g}})$ associated with these fundamental weights, as can be seen from \eqref{eq:definition of magnetic charges of a monopole}. This will provide the relation between spectral curves and magnetic charges of $\mfk{g}$-monopoles and the corresponding monopoles associated with classical Lie algebras.\footnote{For the relation between hyperbolic monopole solutions for orthogonal and symplectic groups and those for the group $\SU{n}$ in terms of discrete Nahm data, see \cite[pg. 105, Theorem 56]{Chan2017}.} All the details, which are based on some basic facts about Lie algebras, are collected in Appendix \ref{sec:some basic facts and Lie algebra manipulations}. Note that for the case of orthogonal and symplectic groups, the holomorphic vector bundles carry additional structures. Here, we impose these additional requirements on the curve \cite{HurtubiseMurray1989}.  

\smallskip Furthermore, to the best of our knowledge, a detailed study of the construction of spectral data for monopoles associated with exceptional Lie groups has not been carried out in the literature.\footnote{In this regard see the related works \cite{Mendizabal202306} and \cite[\S 5.4]{Mendizabal2024}. The Euclidean $\mfk{g}_2$-monopoles have been studied in \cite{ShnirZhilin201508}.} In particular, we do not know (1) whether the spectral curves are labeled by fundamental weights; (2) whether the spectral data depends on additional conditions; and (3) whether the spectral data uniquely determines the monopole solution. Here, we assume that by embedding these Lie algebras in a classical Lie algebra the spectral data would characterize the monopole configuration. Therefore, whatever we present for the case of exceptional Lie groups should be considered as mere speculation, and a proper investigation is warranted. 

\smallskip Given the correspondence between the gCPM and hyperbolic monopoles for the $\mathfrak{su}{n}$ theories, it is natural to speculate on the existence of the generalizations of the gCPM associated with a general semisimple Lie algebra, whose spectral data matches that of the hyperbolic monopoles associated with the same Lie algebra. It seems very few are known in this direction and we hope that our paper stimulates future exploration in this direction, see \S\ref{sec:discussion and future direction} for further discussion.

\paragraph{Lie Algebra $B_n:=\mfk{so}(2n+1)$.} Let us first consider the Lie algebra $B_n=\mfk{so}(2n+1)$, which can be embedded inside $\mfk{su}(2n+1)$. For maximal symmetry-breaking, there are $n$ spectral curves $\{S_1^{B_n},\ldots,S_n^{B_n}\}$ for the case of $\mfk{so}(2n+1)$ and $2n$ spectral curves $\{S_1^{A_{2n}},\ldots,S_{2n}^{A_{2n}}\}$ for $\mfk{su}(2n+1)$. These curves are related as (see \eqref{eq:relation between fundamental weights of Bn and A2n})
\begin{eqaligned}\label{eq:relation of spectral curves: so(2n+1) and su(2n+1)}
    S_i^{A_{2n}}&=S^{A_{2n}}_{2n+1-i}=S^{B_n}_i \;, 
    &\qquad a&=1,\cdots,n-1 \;,
		\\
		S^{A_{2n}}_n&=S^{A_{2n}}_{n+1}=2S^{B_{n}}_n \;, &\qquad (&\text{with multiplicity $2$}) \;,
\end{eqaligned}
while the magnetic charges of the configuration related to the gCPMs are given by (see \eqref{eq:relation between Bn and A2n charges I} and \eqref{eq:relation between Bn and A2n charges II})
\begin{equation}\label{eq:charges of so(2n+1) monopoles associated with a gCPM}
	m^{B_n}_1=\cdots=m^{B_n}_{n-1}=2N \;, \qquad m^{B_n}_{n}=N \;,
\end{equation}
where $N$ is some positive integer. The spectral curves associated with $\SO(2n+1)$ should satisfy some extra conditions: (1) Over $S^{B_n}_n$, we have the isomorphism of line bundles $L^0(2N)\simeq\mcal{O}_{\ts}(N,N)\simeq[S^{A_{2n}}_{n-1,n}]$ \cite[pg.\ 41, (C-$1_k$)]{HurtubiseMurray1989}. Recall from \eqref{eq:definition of Sii+1 and Si+1i} and \eqref{eq:intersection locus of adjacent spectral curves of monopoles} that  $S^{A_{2n}}_{n-1}\cap S^{A_{2n}}_{n}=S^{A_{2n}}_{n-1,n}\cup S^{A_{2n}}_{n,n-1}$. Furthermore, from \eqref{eq:special line bundles defined by gCMP curves I} and $N\to 2N$,\footnote{Recall that the corresponding $\mfk{su}(2n+1)$ monopole has $2n$ magnetic charges of value $2N$, see \eqref{eq:relation between Bn and A2n charges I} and \eqref{eq:relation between Bn and A2n charges II}.} we have $[S_{n-1,n}^{A_{2n}}]|_{S_i}\simeq\mcal{O}_{\ts}(0,2N)|_{S^{A_{2n}}_n}=\mcal{O}_{\ts}(N,0)|_{S^{A_{2n}}_n}\otimes_{\mcal{O}_{\ts}}\mcal{O}_{\ts}(0,N)|_{S^{A_{2n}}_n}\simeq\mcal{O}_{\ts}(N,N)|_{S^{A_{2n}}_n}$, where we have the isomorphism of $\mcal{O}_{\ts}(N,0)$ and $\mcal{O}_{\ts}(0,N)$ on $S^{A_{2n}}_n$. This condition is thus satisfied on $S^{A_{2n}}_n$ and in particular over $S^{B_n}_n$; (2) Furthermore, there is a positivity condition on a real constant, stated in \cite[pg. 41, (C-$4_{k}$)]{HurtubiseMurray1989}, which we assume it holds. Due to the second relation in \eqref{eq:relation of spectral curves: so(2n+1) and su(2n+1)}, there are $n$ independent $K_{i,j}$, which are $K_{i,i+1},\,i=1,\ldots,n-1$ and the rest are determined from the cocycle condition. However, we need all matrices $K_{i,i+1}\,i=1,\ldots,2n-1$ to define the spectral data of the corresponding $A_{2n}$-monopole. The rest of the construction of the corresponding gCPM follows the case of $\mfk{su}(2n+1)$ and the discussion in \S\ref{sec:gCPM from hyperbolic SU(n) monopoles}. 

\paragraph{Lie Algebra $C_n:=\mfk{sp}(2n)$.} We next consider the algebra $C_n=\mfk{sp}(2n)$, which can be embedded in $\mfk{su}(2n)$. 
The relation between spectral curves is (see \eqref{eq:relation between fundamental weights of Cn and A2n-1})
\begin{equation}\label{eq:relation between spectral curves of hyperbolic sp(n) and su(2n) monopoles}
	S^{A_{2n-1}}_i=S^{A_{2n-1}}_{2n-i}=S^{C_n}_i \;,
    \qquad i=1,\cdots,n \;,
\end{equation}
and the charges are given by (see \eqref{eq:relation between Cn and A2n-1 charges})
\begin{equation}\label{eq:charges of Cn monopoles}
	m^{C_n}_1=\cdots=m^{C_n}_n=N \;,
\end{equation}
for some positive integer $N$. There are no extra requirements for $\mfk{sp}(2n)$ monopoles \cite{HurtubiseMurray1989}, and the construction of the curve of the spectral parameter of the corresponding gCPM follows the case of $\mfk{su}(2n)$ in \S\ref{sec:gCPM from hyperbolic SU(n) monopoles}.

\paragraph{The Lie Algebra $D_n:=\mfk{so}(2n)$.} 
Finally consider $D_n=\mfk{so}(2n)$ that can embedded in $\mfk{su}(2n)$.
Their spectral curves are related as (see \eqref{eq:relation between fundamental weights of Dn and A2n-1})
\begin{equation}\label{eq:relation of spectral curves of hyperbolic so(2n) and su(2n) monopoles}
	\begin{gathered}
		S^{A_{2n-1}}_i=S^{A_{2n-1}}_{2n-i}=S^{D_n}_i \;, \qquad i=1,\cdots,n-2 \;,
		\\
		S^{A_{2n-1}}_{n-1}=S^{A_{2n-1}}_{n+1}=S^{D_n}_+\cup S^{D_n}_- \;,
		\\
		S^{A_{2n-1}}_n=2S^{D_n}_+ \;,\quad (\text{with multiplicity $2$}) \;,
	\end{gathered}
\end{equation}
with charges of the $\mfk{so}(2n)$ monopole are (see \eqref{eq:relation between Dn and A2n-1 charges I}, \eqref{eq:relation between Dn and A2n-1 charges II}, and \eqref{eq:relation between Dn and A2n-1 charges III})
\begin{equation}\label{eq:charges of so(2n) monopoles}
	m^{D_n}_1=\cdots=m^{D_n}_{n-2}=2N \;, \qquad m^{D_n}_{\pm}=N \;,
\end{equation}
Similar to the case of $B_n$, there are additional conditions on spectral data: (1) Over $S^{D_n}_+$, $L^0(2N)\simeq[S^{A_{2n-1}}_{n-2,n-1}]$ and over $S^{D_n}_-$, $L^0(2N)\simeq[S^{A_{2n-1}}_{n-1,n-2}]$ \cite[pg.\ 40, (C-$1_{\pm}$)]{HurtubiseMurray1989}. Recall that $S^{A_{2n-1}}_{k-2}\cap S^{A_{2n-1}}_{k-1}$ consists of $8N^2$ points where $S^{A_{2n-1}}_{n-2,n-1}$ contains $4N^2$ of them. The divisor associated with these points defines a line bundle that, from \eqref{eq:special line bundles defined by gCMP curves I}, is $\mcal{O}_{\ts}(2N,0)$ on both curves. On the other hand, $L^0(2N)\simeq\mcal{O}_{\ts}(N,N)\simeq\mcal{O}_{\ts}(2N,0)$. Therefore, $L^0(2N)\simeq[S^{A_{2n-1}}_{n-2,n-1}]$ holds on the entire curve $S^{A_{2n-1}}_{k-1}$ including its $S^{D_n}_+$ segment. For the same reason $L^0(2N)\simeq[S^{A_{2n-1}}_{n-1,n-2}]$ also holds;  (2) We also assume the extra positivity constrains \cite[pg 40, (C-$1_{\pm}$)]{HurtubiseMurray1989}. The rest of the construction follows similar to the case of $B_n$ and $C_n$ groups.

\paragraph{Exceptional Lie Algebra $\mfk{g}_2$.}  The Lie algebra $\mfk{g}_2$ has rank 2 and dimension 14, with fundamental representations of dimensions $7$ and $17$. This Lie algebra can be embedded into $\mfk{so}(7)$, and in turn in $\mfk{su}(7)$. The relation between magnetic charges imply that if we take $m_i^{\mfk{su}(7)}=2N,\,i=1,\ldots,6$, then (see \eqref{eq:magnetic charges of g2 monopoles in terms of those of so(7) and su(7) monopoles})
\begin{equation}\label{eq:magnetic charges of the g2-monopole associated with a gCPM}
    m_1^{\mfk{g}_2}=-3N \;, \qquad m^{\mfk{g}_2}_2=-5N \;.
\end{equation}
In the mathematical literature, it is customary to take magnetic charges to have only positive values \cite[\S 3]{Murray198412} and \cite{AtiyahHitchin198812}. However, physically, nothing prevents us from taking negative magnetic charges and it is consistent with Dirac's quantization condition \cite[\S 2.1]{Shnir2005}. Therefore, \eqref{eq:magnetic charges of the g2-monopole associated with a gCPM} makes sense and would correspond to anti-monopoles with negative charges.

\paragraph{Exceptional Lie Algebra $\mfk{f}_4$.} The Lie algebra $\mfk{f}_4$ has rank $4$, dimension $52$, and has a fundamental representation of dimension $26$. As such it can be embedded in $\mfk{so}(26)$, and in turn $\mfk{su}(26)$. Considerations similar to what we have explained above show that the magnetic charges of an $\mfk{f}_4$-monopole associated with a gCPM are
\begin{equation}
    m_1^{\mfk{f}_4}=2N \;,
    \qquad m_2^{\mfk{f}_4}=3N \;,
    \qquad m_3^{\mfk{f}_4}=3N \;,
    \qquad  m_4^{\mfk{f}_4}=2N \;.
\end{equation}

\paragraph{Exceptional Lie Algebra $\mfk{e}_6, \mfk{e}_7,$ and $\mfk{e}_8$.} The details for $\mfk{e}_6, \mfk{e}_7,$ and $\mfk{e}_8$ follow the same arguments as above and we leave them to interested readers.

\section{Generalized CPM and 4d Chern--Simons Theory}
\label{sec:gCPM and 4d CS theory}

In recent years, it has been discovered that the 4d Chern--Simons (CS) theory provides a framework to unify and explain many aspects of discrete and continuous integrable models \cite{Costello201303,Costello201308,CostelloWittenYamazaki201709, CostelloYamazakiWitten201802,CostelloYamazaki201908}. This theory is defined on a product manifold of the form $\Sigma\times C$.\footnote{In fact, it should be possible to define the theory along any four-dimensional manifolds that admit a transverse holomorphic foliation.} The theory is topological along the topological plane $C$ and holomorphic along the holomorphic plane $\Sigma$. The action is
\begin{equation}\label{eq:4dCS action}
    S=\frac{1}{2\pi}\bigintsss_{\Sigma\times C}\omega\wedge \text{CS}(A) \;,
\end{equation}
where $\omega$ is a meromorphic one-form on $\Sigma$, and
\begin{equation}
    \text{CS}(A)\equiv \text{Tr}_{\mfk{g}}\left(A\wedge \rd A+\frac{2}{3}A\wedge A\wedge A\right)
\end{equation}
is the CS three-form and the trace is taken over the gauge Lie algebra $\mfk{g}$ of the theory. A $(1+1)$d quantum integrable spin-chain model (equivalently a two-dimensional lattice model) could be constructed by considering an arrangement of line defects lying along the topological plane $C$ and supported at different points of $\Sigma$. These defects carry representations of the gauge Lie algebra (or quantum-mechanical deformations extension), which in turn gives the representations at sites of the spin-chain model. The Yang--Baxter (YB) equation is an immediate consequence of two extra dimensions. Furthermore, the R-matrix can be computed by a perturbative analysis of the exchange between line defects. This simple yet powerful method could reproduce many of the properties of integrable lattice models and integrable field theories. Related to our previous discussion, there are two natural questions then: 
\begin{enumerate}
    \item [(1)] Can the gCPM be realized within 4d CS theory?

    \item [(2)] Can the 4d CS theory provide a conceptual explanation of the hyperbolic monopole/gCPM correspondence uncovered in \S\ref{sec:generalized corresponence}?
\end{enumerate}

In this section, we will investigate the first question and provide arguments that it is indeed possible to incorporate the gCPM into the 4d CS theory. We investigate the second question in \S\ref{sec:on the origin of the correspondence} and discover that the 4d CS theory is part of a bigger picture within which we can unify the two sides of the correspondence. 

\subsection{CPM from 4d Chern--Simons Theory}\label{sec:CPM from 4d CS theory}

As we just explained, the spectral parameter of an integrable model takes value in the holomorphic plane $\Sigma$. Therefore, applying the usual wisdom of the 4d CS theory, we would like to study the 4d CS theory on $\Sigma\times C$, where $\Sigma$ is the curve $\cpc$ of the spectral parameter of the gCPM, defined in \eqref{eq:curve of spectral parameter of gCPM}. In this section, we exclusively consider the case of $n=2$, corresponding to the ordinary CPM, and denote the corresponding curve as $\tcpn$, and its free $\mbb{Z}_N$ quotient by $\Sigma_N$. For the gCPMs, we only collect the results in \S\ref{sec:generalization to gCPM} and leave the obvious generalization of details to the interested reader.

\subsubsection{The Curve of Spectral Parameter as a Ramified Cover of \texorpdfstring{$\Pbb^1$}{P(1)}}\label{sec:curve of spectral parameter as a branched cover of P1}

As a consequence of the Riemann--Roch theorem, any meromorphic one-form $\omega_{\tcpn}$ on $\wt{\Sigma}_{N}$ satisfies\footnote{The left-hand side of \eqref{eq:Riemann-Roch theorem on tildeSigma} is the degree of any canonical divisor on $\tcpn$, which can be computed from the Riemann--Roch theorem, while the right-hand side is, by definition, the degree of the canonical divisor associated with $\omega_{\tcpn}$.}
\begin{equation}\label{eq:Riemann-Roch theorem on tildeSigma}
   \text{on $\wt{\Sigma}_{N}$}:\qquad  2N^3-4N^2=\text{$\#$ of zeroes of $\omega_{\tcpn}$}-\text{$\#$ of poles of $\omega_{\tcpn}$} \;,
\end{equation}
Since there is no fixed-point for the $\mbb{Z}_N$-action on $\tcpn$, we have an unramified covering map, and the one-form $\omega_{\cpn}$ on $\Sigma_N$, corresponding to $\omega_{\tcpn}$, satisfies
\begin{equation}\label{eq:Riemann-Roch theorem on Sigma}
    \text{on ${\Sigma}_{N}$}:\qquad 2N^2-4N=\text{$\#$ of zeroes of $\omega_{\cpn}$}-\text{$\#$ of poles of $\omega_{\cpn}$} \;.
\end{equation}
This implies that the number of zeroes and poles of the one-form $\omega$ on $\Sigma_{N}$ is not uniquely determined. Regarding the counting in \cite[eq.\ (3.14)]{CostelloWittenYamazaki201709} and $N\ge 2$, this implies that the engineering of the CPM within the 4d CS theory has a unique feature: the meromorphic one-form $\omega$ must have zeroes as well as poles, while for the rest of integrable lattice models studied in \cite{CostelloWittenYamazaki201709,CostelloYamazakiWitten201802} it was sufficient to have a one-form with only poles. Furthermore, only the ratio $\omega_{\cpn}/\hbar$ appears in the path integral, hence the existence of a zero for $\omega$ corresponds to the effective limit $\hbar\to\infty$ near the zero. As $\hbar$ is the parameter in the series expansion of R-matrix computations in 4d CS theory, the perturbative methods in $\hbar\to\infty$ cannot be trusted and one has to resort to a non-perturbative formulation of the 4d CS theory. It should be possible to achieve such a non-perturbative definition along the lines of \cite{Witten201101}. The construction of the 4d CS theory within string theory is known for the theory with a bosonic gauge Lie algebra \cite{CostelloYagi201810,AshwinkumarTanZhao201806} and with a gauge Lie superalgebra \cite{IshtiaqueMoosavianRaghavendranYagi202110}. However, none of these works provide a practical recipe for constructing the R-matrix of an integrable spin-chain model within a non-perturbative formulation of the 4d CS theory. This is an outstanding question, which we will not pursue it in this work.

\smallskip The strategy we use in this work is to realize $\cpn$ as a branched cover of $\Pbb^1$. This will allow us to work over $\Pbb^1$ instead of $\cpn$ and construct the one-form on $\cpn$ from the one on $\Pbb^1$. Recall that $\tcpn$ is obtained by the intersection of two hypersurfaces in $\Pbb^3$, and the genus of the curve is $N^3-2N^2+1$. By the  Riemann's Existence Theorem, every Riemann surface admits a nonconstant meromorphic function and as such can be realized as a branched cover of $\Pbb^1$, and $\tcpn$ is no exception. Let $\wt{\pi}_{N}:\tcpn\to\Pbb^1$ denote the covering map ramified\footnote{A map $\pi:X\to Y$ between two Riemann surfaces is ramified at a point $P\in X$ if in local coordinates around $P$ and $\pi(P)$, $\pi$ can be written as the map $z\mapsto z^r$ with $r\ge 2$. The point $\pi(P)$ is called the branch point of $\pi$. Such maps are called branched or ramified covering maps.} at a finite number of points $\{P_1, \ldots, P_s\}$ for some positive integer $s$, which we determine momentarily. It is given explicitly by $\wt{\pi}_{N}([x_0:x_1:x_2:x_3])=[x_0:x_3]$ (see Appendix \ref{sec:basic facts about complete intersections}). We have (see \eqref{eq:data of gCPM curve as a complete intersection})
\begin{equation}\label{eq:degree of ramification divisor on tildeSigmaN}
    \deg R_{\wt{\pi}_{N}}=\sum_{i=1}^se_{P_i}-1=2N^2(N-1) \;,
\end{equation}
where $R_{\wt{\pi}_{N}}$ is the ramification divisor of $\wt{\pi}_{N}$ and $e_{P}$ denotes the ramification index at $P\in \tcpn$. If we consider the $\mathbb{Z}_N$-quotient of this curve (which is an unramified degree-$N$ covering map $\text{pr}_N:\tcpn\to\cpn$), we can think of $\cpn$ as a ramified cover of $\Pbb^1$, given by $\pi_N:\cpn\to\Pbb^1$ with $\wt{\pi}_N=\pi_N\circ\text{pr}_N$, where\footnote{Note that the genera of $\tcpn$ and $\cpn$ are related as $g_{\tcpn}-1=N(g_{\cpn}-1)$.}
\begin{equation}\label{eq:degree of ramification divisor on SigmaN}
    \deg R_{\pi_N}=\sum_{i=1}^se_{P_i}-1=2N(N-1) \;.
\end{equation}
The ramification locus $\{P_1, \ldots, P_{2N}\}$ can be determined by evaluating the Hessian of the defining equation (see Appendix \ref{sec:basic facts about complete intersections}):
\begin{equation}
    \det\begin{pmatrix}
    0 & -N\beta x_2^{N-1}
        \\
    -Nx_1^{N-1} & -N\gamma x_2^{N-1}
    \end{pmatrix}=-N^2\beta x_1^{N-1}x_2^{N-1}=0 \;.
\end{equation}
For obvious reasons, both $x_1$ and $x_2$ cannot vanish. Considering two sets of solutions by setting $x_1=0,x_2=1$ and $x_1=1,x_2=0$ and focusing on $\mbb{Z}_N$-invariant solutions, we get the ramification loci to be the points $[\delta^{-1/N}:0:1:y]$ and $[(\beta\gamma^{-1})^{1/N}:1:0:y']$, where $y^N=-\gamma\delta^{-1}$ and $y'^N=\gamma^{-1}$, hence each of which contains $N$ points. As such, there are $2N$ points at which $\pi_N:\cpn\to\Pbb^1$ is ramified, hence $s=2N$, with ramification indices $e_{P_i}=N,\,i=1,\ldots,s=2N$. This means that the branched covering map has the following explicit form, in a suitable coordinate $z$
around each point:
\begin{equation}\label{eq:branched covering map of cpm over P1}
    \pi_N(z)\sim
    \left\{\begin{aligned}
        &z \;, &\qquad &\text{away from all $P_i$} \;,
        \\
        &z^N \;, &\qquad &\text{near one of $P_i$} \;.
    \end{aligned}\right.
\end{equation}
With this description in hand, the action of the 4d CS theory on $C\times \cpn$ can be written as an action on $C\times \Pbb^1$ as follows. From the identity,
\begin{equation}
    \bigintsss_{\phi(X)}\Omega=\bigintsss_X\phi^*\Omega \;,
\end{equation}
for a map $\phi:X\to\phi(X)$ and a differential form $\Omega$, and using the fact that the branched covering map $\pi$ is surjective, we arrive at
\begin{equation}\label{eq:4d cs on SigmaN in terms of the 4d CS on P1}
    \bigintsss_{C\times\cpn}\omega_{\cpn}\wedge\text{CS}(A_{\cpn})=\bigintsss_{C\times\Pbb^1}\omega_{\Pbb^1}\wedge\text{CS}(A_{\Pbb^1}) \;,
\end{equation}
where we have denoted the one-form on $\Pbb^1$ by $\omega_{\Pbb^1}$, and
\begin{equation}
    \omega_{\cpn}\wedge\text{CS}(A_{\cpn})=\pi^*_N\left(\omega_{\Pbb^1}\wedge\text{CS}(A_{\Pbb^1})\right) \;.
\end{equation}
We have distinguished the gauge field on $\cpn$ and $\Pbb^1$ by denoting them as $A_{\cpn}$ and $A_{\Pbb^1}$, respectively. We thus get,
\begin{equation}\label{eq:one-form on cpm curve as the pullback of one-form on P1}
    \omega_{\cpn}=\pi^*_N\omega_{\Pbb^1} \;,
    \qquad 
    \text{CS}(A_{\cpn})=\pi^*_N\text{CS}(A_{\Pbb^1}) \;.
\end{equation}

\begin{rmk}[The Necessity of Nonperturbative Formulation of the 4d CS Theory] \normalfont We have explained above that the existence of zero for $\omega_{\Sigma_N}$ implies the necessity for a non-perturbative treatment of the 4d CS theory for the purpose of computation of the R-matrix of CPM. The reason was that the existence of zeroes demands the divergence of the loop-counting parameter $\hbar$, which appears in the path integral as $\hbar^{-1}S$ with $S$ given in \eqref{eq:4dCS action}. However, the mapping \eqref{eq:4d cs on SigmaN in terms of the 4d CS on P1} and the formulation on $C\times\Pbb^1$ does not change the value of the loop-counting parameter and it still is divergent as $\hbar\to\infty$. Therefore, the formulation on $C\times\Pbb^1$ does not help with this major challenge. However, it is helpful in some respects as we will explain below (see in particular \S\ref{sec:the one-form on the curve of spectral parameter}).  
\end{rmk}

\subsubsection{The One-Form on the Curve of Spectral Parameter}\label{sec:the one-form on the curve of spectral parameter}

The next task is to obtain the one-form $\omega_{\Pbb^1}$ and $\omega_{\Sigma_N}$. What are the principles based on which we can proceed? However, we proceed by invoking the following two reasonable principles: 
\begin{enumerate}
    \item [(1)] {\bf Preserving Topological Invariance along $C$.} As the CPM is a discrete integrable system, the construction of its R-matrix using the 4d CS theory demands the preservation of topological invariance along $C$. On the other hand, the experience of engineering integrable fields theories using 4d CS theory has taught us that the existence of zeroes for the one-form $\omega$ in \eqref{eq:4dCS action} often requires breaking topological invariance \cite{CostelloYamazaki201908}. Therefore, to ensure the preservation of topological invariance along $C$, we require that the one-form $\omega$ in \eqref{eq:4dCS action} does not have a zero. However \eqref{eq:Riemann-Roch theorem on Sigma} implies that it is impossible to require this on $\cpn$. On the other hand, on $\Pbb^1$, as a consequence of the Riemann--Roch theorem, we have
    \begin{equation}
    -2=\text{$\#$ of zeroes of $\omega_{\Pbb^1}$}-\text{$\#$ of poles of $\omega_{\Pbb^1}$} \;.
    \end{equation}
    We can thus certainly require the one-form $\omega_{\Pbb^1}$ not to have a zero, which we do in our construction of $\omega_{\Pbb^1}$.

    \item[(2)] {\bf Imposing $\mbb{Z}_N$-Invariance.} Recall that $\cpn$ is obtained from $\tcpn$ by the $\mbb{Z}_N$-covering map $\text{pr}_N$. Therefore, any well-defined object in general, and $\omega_{\cpn}$ in particular, must be $\mbb{Z}_N$-invariant up to possibly an overall constant phase. 
\end{enumerate}
As we will explain below, these two requirements will almost uniquely determine $\omega_{\Pbb^1}$ (up to possibly an overall constant phase). 

\smallskip Let us start with assuming that $\omega_{\Pbb^1}$ does not have a zero. Then, the most general meromorphic one-form on $\Pbb^1$ is given by
\begin{equation}
    \omega_{\Pbb^1}=\frac{1}{\prod_{i=1}^{\#}(z^{r_i}-z_i)}\rd z \;,
\end{equation}
for some integers $\#,\{r_1,\ldots,r_{\#}\}$, and points $\{z_1,\ldots,z_{\#}\}$ of $\Pbb^1$. Performing $z\to 1/z$ transformation, we get
\begin{equation}
    \omega_{\Pbb^1}\mapsto \frac{z^{\sum_ir_i-2}}{\prod_{i=1}^{\#}(1-z_iz^{r_i})}\rd z \;.
\end{equation}
If we require that $\omega_{\Pbb^1}$ does not have a zero, $\sum_i r_i-2=0$ should be zero. This gives the two possibilities $\#=1$ and $r_1=2$ or $\#=2$ and $r_1=r_2=1$
\begin{equation}\label{eq:two possibilities of omegaP1 without branch cut}
    \omega^{(1)}_{\Pbb^1}=\frac{1}{z^2-z_1}\rd z \;, 
    \qquad 
    \omega^{(2)}_{\Pbb^1}=\frac{1}{(z-z_1)(z-z_2)}\rd z \;,
\end{equation}
Using \eqref{eq:branched covering map of cpm over P1} and \eqref{eq:one-form on cpm curve as the pullback of one-form on P1}, we see that
\begin{equation}
    \omega^{(1)}_{\cpn}\sim
    \left\{
    \begin{aligned}
        &\frac{1}{z^2-z_1}\rd z \;, &\qquad &\text{away from all $P_i$} \;,
        \\
        &\frac{Nz^{N-1}}{z^{2N}-z_1}\rd z \;, &\qquad &\text{near one of $P_i$} \;,.
    \end{aligned}
    \right.
\end{equation}
and 
\begin{equation}
    \omega^{(2)}_{\cpn}\sim
    \left\{
    \begin{aligned}
        &\frac{1}{(z-z_1)(z-z_2)}\rd z \;, &\qquad &\text{away from all $P_i$} \;,
        \\
        &\frac{Nz^{N-1}}{(z^N-z_1)(z^{N}-z_2)}\rd z \;, &\qquad &\text{near one of $P_i$} \;.
    \end{aligned}
    \right.
\end{equation}
If we perform a $\mbb{Z}_N$-transformation $z\to q z$, with $q$ being an $N$\textsuperscript{th} root of unity, none of these forms are invariant. Hence, neither is the correct form of the required one-form on $\cpn$. 

\smallskip As the second attempt, we consider the following one-form
\begin{equation}\label{eq:most general form of omegaP1 without zero}
    \omega_{\Pbb^1}=\frac{1}{\prod_{i=1}^{\#}(z^{r_i}-z_i)^{1/s_i}}\rd z \;,
\end{equation}
for two sets of positive integers $\{r_1,\ldots,r_{\#}\}$ and $\{s_1,\ldots,s_{\#}\}$. $s_i$s cannot be negative as that would lead to a zero of $\omega_{\Pbb^1}$. The above analysis tells us that for having $\mbb{Z}_N$-invariance, we should set\footnote{Instead of \eqref{eq:possible values of r_i for engineering cpm}, one can instead take the non-minimal choice $r_i=n_i N$ for positive integers $\{n_1,\ldots,n_{\#}\}$. However, this is not needed as we demand the minimal choice for $\mbb{Z}_N$-invariance.\label{ftn:choice of ris}}
\begin{equation}\label{eq:possible values of r_i for engineering cpm}
    r_1=\ldots=r_\#=N \;.
\end{equation}
Performing $z\mapsto 1/z$ transformation and demanding no zero for the resulting form, we get the constraint
\begin{equation}
    N\sum_{i=1}^{\#}\frac{1}{s_i}-2=0 \;.
\end{equation}
One of the simplest integer solutions of this constraint is\footnote{Similar to the choice of $r_i$s, explained in Footnote \ref{ftn:choice of ris}, this is the minimal solution. One can always choose non-minimal solutions such as $\#=4$ and $s_i=2N$.}
\begin{equation}
    \#=2 \;, \qquad s_1=s_2=N \;.
\end{equation}
We thus arrive at\footnote{Here, we have assumed that all the $N$\textsuperscript{th} root factors belong to the same branch, i.e. $\sqrt[\leftroot{-2}\uproot{2}N]{z^{t_1}}\times\sqrt[\leftroot{-2}\uproot{2}N]{z^{t_2}}=\sqrt[\leftroot{-2}\uproot{2}N]{z^{t_1+t_2}}$, for some $t_1$ and $t_2$, and similarly for other factors. If different factors belong to different branches of the $N$\textsuperscript{th} root function, we get an overall constant phase that does not affect manipulations.}
\begin{equation}\label{eq:one-form relevant for engineering cpm on P1}
    \omega_{\Pbb^1}=\frac{1}{\sqrt[\leftroot{-2}\uproot{2}N]{(z^N-z_1)(z^N-z_2)}}\rd z \;,
\end{equation}
from which, by using \eqref{eq:branched covering map of cpm over P1} and \eqref{eq:one-form on cpm curve as the pullback of one-form on P1}, we get
\begin{equation}\label{eq:one-form relevant for engineering cpm on SigmaN}
    \omega_{\cpn}\sim
    \left\{
    \begin{aligned}
        &\frac{1}{\sqrt[\leftroot{-2}\uproot{2}N]{(z^N-z_1)(z^N-z_2)}}\rd z \;, 
        &\qquad &\text{away from all $P_i$} \;,
        \\
        &\frac{Nz^{N-1}}{\sqrt[\leftroot{-2}\uproot{2}N]{(z^{N^2}-z_1)(z^{N^2}-z_2)}}\rd z \;, &\qquad &\text{near one of $P_i$} \;,
    \end{aligned}
    \right.
\end{equation}
which, by construction, is $\mbb{Z}_N$-invariant.\footnote{Note that away from a ramification point, $\omega_{\cpn}$ is invariant only up to an overall irrelevant constant phase that would not affect computations.} It is clear from this form that $\omega_{\cpn}$ has zeroes of order $N-1$ near a ramification point of $\pi_N$. The one-form $\omega_{\Pbb^1}$ \eqref{eq:one-form relevant for engineering cpm on P1} and the corresponding $\mbb{Z}_N$-invariant one-form $\omega_{\cpn}$ \eqref{eq:one-form relevant for engineering cpm on SigmaN} are the unique one-forms required for engineering the CPM within the 4d CS theory. However, these forms have the undesirable feature of having branch cuts, which introduce a sort of non-locality into the problem. This is expected since we are trying to describe a higher-genus curve in terms of data on $\Pbb^1$. On the other hand, as we will see in \S\ref{sec:lack of rapidity-difference property}, this structure explains the main property of the R-matrix of the model, namely the lack of rapidity-difference. 

\begin{rmk}[Why is the gCPM a Peculiar Integrable Spin System?]\label{rmk:why gCPM is peculiar?}\normalfont 
    Our construction of $\omega_{\Pbb^1}$ suggests why the CPM (also gCPM, as we will see below) is so special. We derived $\omega_{\Pbb^1}$ by demanding the topological invariance along $C$ and the $\mbb{Z}_N$ (or $\mbb{Z}_N^{n-1}$ in the case of the gCPM) invariance, which then fixes  $\omega_{\Pbb^1}$. To engineer any integrable spin model (not an integrable field theory), even those that are currently unknown, within 4d CS theory, it is not conceivable that the first requirement can be relaxed, as it is essential for the derivation of Yang--Baxter equation. However, this restricts the possibilities to either the two cases in \eqref{eq:two possibilities of omegaP1 without branch cut}, which has already been studied in \cite{CostelloWittenYamazaki201709}, or \eqref{eq:one-form relevant for engineering cpm on P1} (and \eqref{eq:one-form relevant for engineering gcpm on P1} below for the gCPM, also the more general situation discussed around \eqref{eq:requirement of not having a zero for omegaP1}). Any other choice of $\omega_{\Pbb^1}$ would necessarily have a zero, with the possibility of breaking the topological invariance along $C$. This seems to be a plausible explanation of why the gCPM is such a peculiar integrable spin model. Hence, it is not surprising that there are no other known integrable spin models whose curve of the spectral parameter has a genus greater than one. See also Remark \ref{rmk:importance of vanishing of dbar derivative of omegaP1} and also the discussion around \eqref{eq:requirement of not having a zero for omegaP1}.
\end{rmk}

\smallskip We can fix the branched points of \eqref{eq:one-form relevant for engineering cpm on P1} and \eqref{eq:one-form relevant for engineering cpm on SigmaN} as follows. Let $\{z_{1,1},\ldots,{z_{1,N}}\}$ and $\{z_{2,1},\ldots,{z_{2,N}}\}$ be the set of solutions of $z^N=z_1$ and $z^N=z_2$ and hence the $2N$ branched points of $\omega_{\Pbb^1}$. Let us first consider the case of $N=2$, where $\cpn$ is an elliptic curve, and see whether we can recover the standard one-form $\rd z$ on an elliptic curve. In this case, near a branch point, say $z_{1,1}$, $\omega_{\Pbb^1}$ takes the form $\omega_{\Pbb^1}\sim 1/\sqrt{z\times G(z)}\,\rd z$, for some holomorphic function $G(z)$ that is not zero at $z_{1,1}$. Upon pulling back to $\Sigma_2$, we have $\omega_{\Sigma_2}\sim 1/\sqrt{z\times G(z)}\,\rd z$ away from the ramification locus of $\pi$ and $\omega_{\Sigma_2}\sim (2z)/\sqrt{z^2\times G(z^2)}\,\rd z$ near a ramification point say $z_{1,1}$. Therefore, to recover the standard form $\rd z$, we need to take $z_{1,1}$ to be one of the branched points of $\pi$ (the factor $1/\sqrt{G(z^2)}$ is irrelevant). This observation suggests the fixing of the branched points $\{z_{1,1},\ldots,{z_{1,N}}\}$ and $\{z_{2,1},\ldots,{z_{2,N}}\}$ as follows: Let $\{Q_1,\ldots,Q_{2N}\}$ be the set of branched points of $\pi$ on $\Pbb^1$ defined by $Q_i:=\pi_N(P_i)$, we then set
\begin{equation}
    z_{1,i}=z_{Q_i} \;, \qquad z_{2,i}=z_{Q_{N+i}} \;, \qquad i=1,\ldots,N \;.
\end{equation}
We can thus write the one-forms \eqref{eq:one-form relevant for engineering cpm on P1} and \eqref{eq:one-form relevant for engineering cpm on SigmaN} in a more compact form as
\begin{equation}\label{eq:one-form relevant for engineering cpm on P1 depending on branched points of pi}
    \omega_{\Pbb^1}=\frac{1}{\sqrt[\leftroot{-2}\uproot{2}N]{\prod_{i=1}^{2N}(z-z_{Q_i})}}\rd z \;,
\end{equation}
and
\begin{equation}\label{eq:one-form relevant for engineering cpm on SigmaN depending on branched points of pi}
    \omega_{\cpn}\sim
    \left\{
    \begin{aligned}
        &\frac{1}{\sqrt[\leftroot{-2}\uproot{2}N]{\prod_{i=1}^{2N}(z-z_{Q_i})}}\rd z \;, &\qquad &\text{away from all $P_i$s} \;,
        \\
        &\frac{Nz^{N-1}}{\sqrt[\leftroot{-2}\uproot{2}N]{\prod_{i=1}^{2N}(z^N-z_{Q_i})}}\rd z \;, &\qquad &\text{near one of $P_i$s} \;.
    \end{aligned}
    \right.
\end{equation}
This concludes our construction of one-forms on $\Pbb^1$ and $\Sigma_N$.

\subsubsection{Gauge Lie Algebra, Boundary Conditions, and R-Matrix}\label{sec:gauge Lie algebra, boundary conditions, R-matrix}

In this section, we discuss the choice of the gauge Lie algebra for the theory, boundary conditions for the gauge fields, and the R-matrix.

\paragraph{The Gauge Lie Algebra.} We need to specify the gauge Lie algebra $\mfk{g}$ for the 4d CS theory, which appears explicitly in the definition of the theory in \eqref{eq:4dCS action}. In general, the spin models are related to the representation theory of quantum groups. In the case of the gCPM, it is known that the R-matrix of the model arises as the intertwiner of certain $N^{n-1}$-dimensional representations, the so-called minimal cyclic representations, of $\text{U}_q(\wh{\mfk{sl}}(n,\C))$ with $q$ being an $N$\textsuperscript{th} root of unity \cite{BazhanovKashaevMangazeevStroganov199105,DateJimboKeiMiwa199008,DateJimboKeiMiwa199103,DateJimboMikiMiwa199104}. Therefore, for $n=2$, one has to take the gauge Lie algebra of the 4d CS to be $\mfk{g}=\mfk{sl}(2,\C)$.

\paragraph{Boundary Conditions.} To formulate a perturbation theory of the 4d CS theory on $C\times\Pbb^1$, we need to impose boundary conditions on gauge fields. The location at which one needs to impose boundary conditions can be understood as follows. The only relevant term in the action is $\omega_{\Pbb^1}\wedge\text{Tr}_{\mfk{g}}(A\wedge\rd A)$.  Indeed only $\omega_{\Pbb^1}\wedge\text{Tr}_{\mfk{g}}(A\wedge\partial_{\bar{z}} A)$ is relevant since, upon integration by parts, it leads to a term $\partial_{\bar{z}}\omega_{\Pbb^1}\text{Tr}_{\mfk{g}}(A\wedge\delta A)$ in the variation of the action. In the case of other integrable spin-chain models, $\partial_{\bar{z}}\omega_{\Pbb^1}$ involves (a $\partial_{\bmb{z}}$-derivative of) delta functions at the location of the poles which, in turn, would lead to the term $\text{Tr}_{\mfk{g}}(A\wedge\delta A)|_{C\times\{\text{set of poles of $\omega_{\Pbb^1}$}\}}$ \cite[\S 9.1]{CostelloWittenYamazaki201709}. One thus needs to impose appropriate boundary conditions on the gauge fields at the poles of $\omega_{\Pbb^1}$ to formulate a well-defined perturbation theory around a classical solution to the equations of motion. In our case, however, regarding the form \eqref{eq:one-form relevant for engineering cpm on P1 depending on branched points of pi}, the situation is different due to the existence of branch points $\{z_{Q_1},\ldots,z_{Q_{2N}}\}$. Let us first compute $\partial_{\bar{z}}\omega_{\Pbb^1}$. From \eqref{eq:one-form relevant for engineering cpm on P1 depending on branched points of pi}, it is 
\begin{equation}\label{eq:zbar partial derivative of omegaP1}
    \partial_{\bar{z}}\omega_{\Pbb^1}=2\pi\mfk{i}\sum_{i=1}^{2N}\left\{\sqrt[\leftroot{-2}\uproot{2}N]{\frac{(z-z_{Q_i})^{N-1}}{\prod_{j=1,j\ne i}^{2N}(z-z_{Q_j})}}\delta^{(2)}(z-z_{Q_i},\bar{z}-\bar{z}_{Q_i})\right\}\rd z \;,
\end{equation}
with $\delta^{(2)}(z,\bar{z})$ denotes the delta-function distribution.
However, this expression simply vanishes: away from the branched points, it vanishes due to delta functions, and at a branched point, say $z_{Q_i}$, it is zero due to the factor $(z-z_{Q_i})^{(N-1)/N}$ in the numerator. Hence
\begin{equation}\label{eq:vanishing of zbar derivative of omega}
    \partial_{\bar{z}}\omega_{\Pbb^1}=0 \;.
\end{equation}
Therefore, the term $\partial_{\bar{z}}\omega_{\Pbb^1}\wedge\text{Tr}_{\mfk{g}}(A\wedge\delta A)$ in the variation of the action is automatically zero, and we do not need to impose boundary conditions on the gauge field to kill it. 

\smallskip For completeness, let us explain the derivation of \eqref{eq:zbar partial derivative of omegaP1}. Without loss of generality, we take a branch of the $N$\textsuperscript{th}-root function and take all factors involving the $N$\textsuperscript{th}-root function to belong to the same branch. We first multiply the numerator and denominator of $\omega_{\Pbb^1}$ in \eqref{eq:one-form relevant for engineering cpm on P1 depending on branched points of pi} by\footnote{We can choose different branches for different factors but the final result of vanishing of $\partial_{\bar{z}}\omega_{\Pbb^1}$ will not be affected as different branches are related by an overall phase. Here, we have chosen the same branch to be ably to multiply $\sqrt[\leftroot{-2}\uproot{2}N]{z}\times\sqrt[\leftroot{-2}\uproot{2}N]{z}=\sqrt[\leftroot{-2}\uproot{2}N]{z^2}$ and similar type of manipulations without dealing with extra phases which would not affect the final result.}
\begin{equation}
    \left(\prod_{j=1}^{2N}(z-z_{Q_j})\right)^{(N-1)/N} \;,
\end{equation}
which gives
\begin{equation}\label{eq:different expression for omegaN}
    \omega_{\Pbb^1}=\frac{\left(\prod_{j=1}^{2N}(z-z_{Q_j})\right)^{(N-1)/N}}{\left(\prod_{j=1}^{2N}(z-z_{Q_j})\right)}\rd z \;.
\end{equation}
As the numerator is holomorphic, its $\partial_{\bar{z}}$-derivative vanishes, and we only need to deal with the denominator. By using the identity
\begin{equation}
    \partial_{\bar{z}}\left(\frac{1}{z-z_0}\right)=2\pi\mfk{i}\,\delta^{(2)}(z-z_0,\bar{z}-\bar{z}_0) 
\end{equation}
in the $\partial_{\bar{z}}$-derivative of the denominator of \eqref{eq:different expression for omegaN} repeatedly, we arrive at the expression \eqref{eq:zbar partial derivative of omegaP1}.

\smallskip As $\cpn$ and $C$ do not have boundaries, we do not have any boundary condition. Since we are not setting up a perturbation theory, we would not pursue this further in this work. It is further evidence that the perturbation theory is not the right approach to the computation of R-matrix for the CPM and one indeed needs to resort to a non-perturbative definition. 

\begin{rmk}[Behavior of the Gauge Field at Ramification Points]\normalfont
    To make the action finite near a ramification locus, we demand the mildest vanishing behavior that cures the issue. From \eqref{eq:one-form relevant for engineering cpm on P1}, we see that near a branched point, $\omega_{\Pbb^1}$ behaves as $z^{-1/N}$. As the action is quadratic in the gauge field, even a linear vanishing behavior would do the job. Note that this is {\it not} a boundary condition and we demand it to make the integrand of the action \eqref{eq:4dCS action} of the 4d CS theory well-defined.
\end{rmk}

\paragraph{The R-Matrix and Bazhanov--Stroganov Procedure.}

The R-matrix of CPM can be presented in various ways: As the original weights $\mscr{W}$ and $\overline{\mscr{W}}$, as vertex weights, or as IRF weights. These are related to each other \cite{Wu197110,KadanoffWegner197112}. The computation of the R-matrix of the CPM from the 4d CS theory should follow the same procedure as the other integrable lattice models \cite{CostelloWittenYamazaki201709}. As the R-matrix is the intertwiner of tensor product of two $N$-dimensional minimal cyclic representations of $\text{U}_q(\wh{\mfk{sl}}(2,\C))$ at $q$ being a root of unity, one takes two lines in this representation, and then computes the R-matrix as in Fig. \ref{fig:reprentation of R-matrix of CPM as a vertex model}.

\begin{figure}[H]
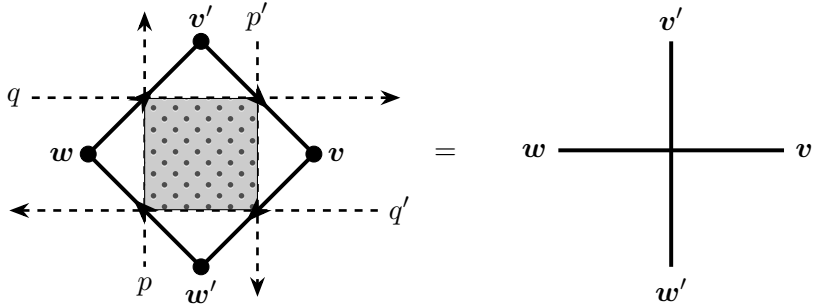

    \centering
    \ElementaryBoxEqualRMatrix
    \caption{The representation of R-matrix of CPM as a vertex model. The two lines in the right carry minimal cyclic representations of $\text{U}_q(\wh{\mfk{sl}}(2,\C)$ and their intertwiner computes the R-matrix.}
    \label{fig:reprentation of R-matrix of CPM as a vertex model}
\end{figure} 

By a procedure due to Bazhanov and Stroganov, it turns out that, at least for odd $N$, one only needs to compute the intertwiner of a minimal cyclic representation and a highest-weight spin-$1/2$ representation \cite{BazhanovStroganov199005}. One starts from the R-matrix of the six-vertex model at $q^N=1$ for odd $N$. The corresponding YB equation is satisfied by intertwiners of spin-$1/2$ representations, denoted as $\mcal{R}_{\mcal{V}_{\frac{1}{2}}\mcal{V}'_{\frac{1}{2}}}$. Let us assume that one can construct the intertwiner $\mcal{R}_{\mcal{V}_{\text{C}},\mcal{V}'_{\frac{1}{2}}}$ of one minimal cyclic representation and one spin-$1/2$ representation. Then the YB equation associated with the triple tensor product $\mcal{V}_{\text{C}}\otimes\mcal{V}'_{\text{C}}\otimes\mcal{V}_{\frac{1}{2}}$ would depend quadratically on $\mcal{R}_{\mcal{V}_{\text{C}},\mcal{V}'_{\frac{1}{2}}}$ and linearly on $\mcal{R}_{\mcal{V}_{\text{C}}\mcal{V}'_{\text{C}}}$, the interwiner of two minimal cyclic representations. Hence one can solve it for $\mcal{R}_{\mcal{V}_{\text{C}}\mcal{V}'_{\text{C}}}$ easily. $\mcal{R}_{\mcal{V}_{\text{C}}\mcal{V}'_{\text{C}}}$ is related to the R-matrix of CPM by some choice of parameters \cite[\S 4]{BazhanovStroganov199005}. Of course, according to our arguments in \S\ref{sec:curve of spectral parameter as a branched cover of P1}, any computation of this sort has to be performed in a nonperturbative completion of the 4d CS theory. We will not attempt it in this work.

\subsubsection{Lack of the Rapidity-Difference Property}\label{sec:lack of rapidity-difference property}

Up to this point, we have discussed the engineering of CPM within the 4d CS theory. However, to boost our proposal, we still need to argue for the main property of the R-matrix of CPM, i.e.\ the lack of rapidity-difference property. This property means that the R-matrix of an integrable spin model satisfies
    \begin{equation}
        \mcal{R}(z_1,z_2)=\mcal{R}(z_1-z_2) \;,
    \end{equation}
    i.e.\ the R-matrix, as a function of the rapidity parameter, is a {\it single-valued} function of the difference of rapidities on the curve of the spectral parameter. In the Belavin--Drinfeld classification of solutions to classical YB equation, the rapidity-different property is assumed \cite{BelavinDrinfeld198207,BelavinDrinfeld1998}. This assumption restricts the possible choice of the curve of the spectral parameter to $\C,\C^\times,$ and the elliptic curve $\mbb{E}_\tau$, where in the latter case the complex structure $\tau$ can appear as a dynamical parameter in the YB equation \cite{Felder199407,Felder199412}. This result is compatible with the formal group law. Once, the rapidity-difference property is relaxed, it would not be surprising that there could be R-matrices where the curve of the spectral parameter has a genus higher than one. This is obvious because even the meaning of $z_1-z_2$ on a higher-genus curve is unclear.

\smallskip In our discussion, the lack of this property follows from the structure of $\omega_{\Pbb^1}$, given in \eqref{eq:one-form relevant for engineering cpm on P1}, in particular having branch cuts. We consider $\cpn$ as a branched cover of $\Pbb^1$ and studied everything on $C\times\Pbb^1$ including the absence of the rapidity-difference property of the R-matrix. Although we do not set up a perturbation theory, we can still get a glimpse into the origin of this feature of the model in the formulation on $C\times\Pbb^1$. Recall that the equations of motion of the 4d CS theory on $C\times\Pbb^1$ are given by
\begin{equation}
    \omega_{\Pbb^1}\wedge F=0 \;,
\end{equation}
where $F:=\rd A+\frac{1}{2}[A,A]$ is the curvature of the connection $A$ and $\rd=\rd x\partial_x+\rd y\partial_y+\rd\bar{z}\partial_{\bar{z}}$ with $x$ and $y$ denote the local coordinate on the topological plane $C$. Therefore, the free propagator two-form $P$ of the theory satisfies the following relation
\begin{equation}
    \frac{\mfk{i}}{2\pi}\omega_{\Pbb^1}\wedge\rd P=\delta^{(4)}(x,y,z,\bar{z})\,\rd\text{Vol}_{C\times\Pbb^1} \;,
\end{equation}
where $\delta^{(4)}(x,y,z,\bar{z})$ denotes the four-dimensional delta-function distribution on $C\times\Pbb^1$ and $\rd\text{Vol}_{C\times\Pbb^1}$ is the volume form of $C\times\Pbb^1$. Using \eqref{eq:vanishing of zbar derivative of omega}, we can write this equation as
\begin{equation}
    \frac{\mfk{i}}{2\pi}\rd z\wedge\rd\left(\frac{P}{\sqrt[\leftroot{-2}\uproot{2}N]{\prod_{i=1}^{2N}(z-z_{Q_i})}}\right)=\delta^{(4)}(x,y,z,\bar{z})\,\rd\text{Vol}_{C\times\Pbb^1} \;.
\end{equation}
What appears inside the bracket has already been computed in a Lorentz-like gauge in \cite[eq.\ (4.5)]{CostelloWittenYamazaki201709}. Hence ($P=\text{Tr}_{\mfk{g}}(P^{ab}J_aJ_b)$ with $\{J_1,\ldots,J_{\dim\mfk{g}}\}$ is a basis for the gauge Lie algebra $\mfk{g}$ of the 4d CS theory)
\begin{equation}\label{eq:propagator for 4d CS in the presence of branch points}
    P^{ab}(x,y,z,\bar{z})=\left(\frac{\delta^{ab}}{2\pi}\cdot \frac{\sqrt[\leftroot{-2}\uproot{2}N]{\prod_{i=1}^{2N}(z-z_{Q_i})}}{(x^2+y^2+z\bar{z})^2}\right)\left(x\rd y\wedge\rd\bar{z}+y\rd \bar{z}\wedge\rd x+2\bar{z}\rd x\wedge\rd y\right) \;.
\end{equation}
In the cases that the R-matrix of the integrable spin model can be computed perturbatively, the classical R-matrix is essentially calculated by integrating the $\rd x\wedge \rd y$ component of $P$ along $x$ and $y$ where two Wilson lines are stretched. In those cases, where the factor $\sqrt[\leftroot{-2}\uproot{2}N]{\prod_{i=1}^{2N}(z-z_{Q_i})}$ is absent, the rapidity-difference property is preserved as the R-matrix is a {\it single-valued} function of the rapidity parameter \cite[eqs.\ (4.11)--(4.14)]{CostelloWittenYamazaki201709}. However, in the case of CPM, the presence of the factor $\sqrt[\leftroot{-2}\uproot{2}N]{\prod_{i=1}^{2N}(z-z_{Q_i})}$ in \eqref{eq:propagator for 4d CS in the presence of branch points} is a sign of breaking of the rapidity-difference property. It is lost simply because \eqref{eq:propagator for 4d CS in the presence of branch points} is not a single-valued function of $z$. In a full, nonperturbative treatment of the 4d CS theory, the propagator \eqref{eq:propagator for 4d CS in the presence of branch points} is certainly modified but it is expected that the modification would not restore the rapidity-difference property. Therefore, the lack of this crucial feature of the R-matrix of the CPM  can be seen from the formulation on $C\times\Pbb^1$, although in a non-rigorous and heuristic manner. 

\subsection{Generalization to gCPM}\label{sec:generalization to gCPM}

What we have explained so far for CPM can be readily generalized to the gCPM. We would only quote the final results, as the interested reader can easily fill in the details by an obvious generalization of the arguments for the CPM.

\smallskip The curve $\pc$ defined in \eqref{eq:ZNxn-1 quotient of the curve of spectral parameter of gCPM} can be realized as a branched cover of $\Pbb^1$.  From \eqref{eq:data of gCPM curve as a complete intersection},\footnote{Recall that $\pc$ is the free $\mbb{Z}_N^{n-1}$-quotient of $\cpc$.} the branched covering map $\pi_{N,n}:\pc\to\Pbb^1$ is given by explicitly by
(in a suitable coordinate $z$ defined in a local neighborhood)
\begin{equation}
    \pi_{N,n}\sim 
    \left\{
    \begin{aligned}
        &z \;, &\qquad  &{\text{away from the ramification points} \;,}
        \\
        & z^{N^{(n-1)}} \;, &\qquad  &\hphantom{away}{\text{near a ramification point} \;,}
    \end{aligned}
    \right.
\end{equation} 
and has degree
\begin{equation}
    \deg \pi_{N,n}=N^{n-1} \;.
\end{equation}
The degree of the ramification divisor again follows from \eqref{eq:data of gCPM curve as a complete intersection} to be
\begin{equation}
    \deg R_{\pi_{N,n}}=2N^{(n-1)}(N-1)(n-1) \;,
\end{equation}
hence there are $2N^{n-1}$ ramification points each of which has the ramification index $(N-1)(n-1)+1$. We denote these points as $\{P_1,\ldots,P_{2N^{n-1}}\}$ and the corresponding branch points as $\{Q_1,\ldots,Q_{2N^{n-1}}\}$. 

\smallskip The one-forms on $\Pbb^1$ and $\pc$ are given by $(\#,r_i,s_i)=(2,N^{n-1},N^{n-1})$  in \eqref{eq:most general form of omegaP1 without zero})
\begin{equation}\label{eq:one-form relevant for engineering gcpm on P1}
    \omega_{\Pbb^1}=\frac{1}{\sqrt[\leftroot{-2}\uproot{2}N^{n-1}]{\prod_{i=1}^{2N^{n-1}}(z-z_{Q_i})}}\rd z \;,
\end{equation}
and
\begin{equation}\label{eq:one-form on the curve of spectral parameter of gCPM}
    \omega_{\pc}\sim
    \left\{
    \begin{aligned}
        &\frac{1}{\sqrt[\leftroot{-2}\uproot{2}N^{n-1}]{\prod_{i=1}^{2N^{n-1}}(z-z_{Q_i})}}\rd z \;, &\qquad &\text{away from all $P_i$s \;,}
        \\
        &\frac{N^{n-1}z^{N^{n-1}-1}}{\sqrt[\leftroot{-2}\uproot{2}N^{n-1}]{\prod_{i=1}^{2N(n-1)}(z^{N^{n-1}}-z_{Q_i})}}\rd z \;, &\qquad &\text{near one of $P_i$s \;.}
    \end{aligned}
    \right.
\end{equation}
The following relation still holds
\begin{equation}\label{eq:vanishing of zbar derivative of omegaP1 for gCPM}
    \partial_{\bmb{z}}\omega_{\Pbb^1}=0 \;.
\end{equation}
The relevant gauge Lie algebra of the 4d CS would be $\mfk{sl}(n,\C)$. The question of boundary condition and the lack of rapidity-difference property are verbatim for the case of the CPM. 
 
\smallskip With these results at hand, we conclude our discussion of the engineering of gCPM within 4d CS theory.

\section{Origin of the Hyperbolic Monopole/gCPM Correspondence}
\label{sec:on the origin of the correspondence}

In \S\ref{sec:generalized corresponence}, we established a correspondence between the spectral data of hyperbolic $\SU{n}$-monopoles and a class of $\mathbb{Z}_N^{n-1}$ gCPMs. An immediate question is whether this is just pure coincidence or there is a deeper reason for it. In this section, we argue that this correspondence is tied to the existence of 6d and eventually 10d theories whose two different incarnations describe the two sides of the hyperbolic monopole/gCPM correspondence. 

\smallskip The six-dimensional theory turns out to be a holomorphic version of CS theory in six dimensions introduced by Witten in the context of the open-string field theory \cite{Witten199207}
\begin{equation}\label{eq:6dCS action}
    S=\frac{1}{2\pi}\bigintsss_{X}\Omega\wedge\CS(\mcal{A}) \;,
\end{equation}
with 
\begin{equation}\label{eq:CS term in the case of 6d hCS theory}
    \CS(\mcal{A}):=\text{Tr}_{\mfk{g}}\left(\mathcal{A}\wedge\bmb{\partial}\mathcal{A}+\frac{2}{3}\mathcal{A}\wedge\mathcal{A}\wedge\mathcal{A}\right) \;.
\end{equation}
Here, $X$ is a complex three-manifold, $\Omega$ is a section of the canonical bundle of $X$ of degree $(3,0)$, $\mathcal{A}$ is a connection on a principal $G$-bundle over $X$ with $\mfk{g}:=\text{Lie}(G)$, and $\bmb{\partial}$ is the $(0,1)$-part of the exterior derivative $\rd=\partial+\bmb{\partial}$. As usual, we are working with the corresponding associated bundle $\mcal{W}\to X$. Let $(z,w_1,w_2)$ be the local holomorphic coordinates on $X$, then
\begin{equation}\label{eq:expansion of A and del bar in 6d hCS in terms of local coordinate}
    \mcal{A}=\mcal{A}_{\bmb{z}}\rd\bmb{z}+\sum_{i=1}^2 \mcal{A}_{\bmb{w}_i}\rd\bmb{w}_i \;,
    \qquad 
    \bmb{\partial}=\rd\bmb{z}\partial_{\bmb{z}}+\sum_{i=1}^2\rd\bmb{w}_i\partial_{\bmb{w}_i} \;.
\end{equation}
In general, to integrate $\Omega \wedge \CS(\mcal{A})$, a global section $\Omega$ of the canonical bundle of $M$ is chosen.  To have a global section, the canonical bundle has to be trivial, which is the case of Calabi--Yau (CY) manifolds (see also \cite{BarberisDottiVerbitsky200712} and references in \cite[pg.\ 1]{Tosatti201401} for other examples). The condition for being CY is not necessary and \eqref{eq:6dCS action} can be defined on more general complex three-manifolds \cite[\S 2]{DonaldsonThomas199606}. In such situations, $\Omega$ is not a holomorphic volume form but is still a global section of the canonical bundle of $X$. 
The complex structure is irrelevant for integration, and one can always work with the underlying smooth manifold. One can define the integration on the underlying smooth manifold by choosing a partition of unity $\{\phi_\alpha\}$ supported on open sets of an open covering $\{U_\alpha\}$ of $X$. Then, the integration on $X$ can be defined using the partition of unity arguments by defining the integral only on $U_\alpha$. Therefore, working on general complex manifolds for which there is no holomorphic volume form poses no issue. 

\smallskip To relate this theory to hyperbolic monopoles on the one hand and the gCPM on the other hand, the basic idea is that both of these systems are intimately connected to four-dimensional physics: As we have seen in \S\ref{sec:hyperbolic monopoles and spectral data}, hyperbolic monopoles of arbitrary mass can be constructed from instantons on $\mcal{U}$, defined in \eqref{eq:definition of future-directed time-like vectors}, the solid cone of future-directed time-like vectors on the Minkowski space $\mbb{R}^{1,3}$. Hence, the relevant four-dimensional spacetime in this case is $\R^{1,3}$ or an open set $\mcal{U}$ thereof. On the other hand, based on our observations in \S\ref{sec:gCPM and 4d CS theory}, gCPM is related to 4d CS theory on $C\times\Pbb^1$ with a certain choice of the one-form $\omega$ given in \eqref{eq:one-form relevant for engineering cpm on P1} and \eqref{eq:one-form relevant for engineering gcpm on P1}. It is then natural to think that this theory is the dimensional reduction of a six-dimensional theory on $T^2\times C \times\Pbb^1$, where $T^2$ is a torus with a small volume and becomes a copy of $\R^2$ in the microscopic regime. Therefore, the relevant spacetime in this case is $\R^2\times C$, which can be taken to be $\R^4$ if we take $C=\R^2$. Due to the involvement of $\Pbb^1$, the relevant six-dimensional space turns out to be the (right-handed) projective spinor bundle\footnote{To be more pedantic, we should say right-handed or left-handed projective spinor bundle and write $\Pbb S_\pm(M)$ for these bundles. However, to avoid cluttering, we sometimes omit the right-/left-handed adjective. In the following, we assume this is understood.} of a four-dimensional spacetime $M$ which we denote as $\sbm$. This space is defined through an incidence relation which relates events in spacetime $M$ to points in its twistor space $\mcal{Z}_M$. For a spacetime $M$ of arbitrary signature, $\sbm\simeq M\times\Pbb^1$ as a smooth manifold, which has a natural almost-complex structure \cite[\S 9.1]{WardWells199108}. This almost-complex structure is integrable if and only if $M$ is antiself-dual \cite[Theorem 4.1]{AtiyahHitchinSinger197809}.\footnote{Under the Hodge-star operation, the Weyl tensor of a Riemannian four-manifold $M$ is decomposed as $W=W_++W_-$, where $W_\pm$ corresponds to the two eigenvalues of the star operation. $M$ is self-dual if $W_-=0$ and anti-self-dual if $W_+=0$. For more details, see \cite[\S 1, pg.\ 428]{AtiyahHitchinSinger197809}. Everything can be stated in terms of the left-handed projective spinor bundle.} In both cases of interest, $\R^4$ and $\R^{1,3}$, the spacetime is antiself-dual and indeed flat. Therefore, $\sbm$ has an integrable complex structure, hence a three-dimensional complex manifold. Since the 6d hCS theory, defined in \eqref{eq:6dCS action}, depends only on the choice of a complex structure, we can define it on $\sbm$
\begin{equation}\label{eq:6d hCS action on projective spinor bundle}
    S=\frac{1}{2\pi}\bigintsss_{\sbm}\Omega\wedge\CS(\mcal{A}) \;.
\end{equation}
The salient feature of $\Pbb S(M)$ is that it defines a double fibration
\begin{equation}
\begin{tikzcd}\label{fig:correspondence defined by projective spinor bundle}
    & \arrow[ld,"\pi_{\mcal{Z}_M}" swap] \sbm\arrow[rd,"\pi_M"] &
    \\
    \mcal{Z}_M & & M
\end{tikzcd} 
\end{equation}
where $\pi_M$ and $\pi_{\mcal{Z}_M}$ would be defined momentarily. We will not need many details of twistor geometry. In the notation of \cite[\S 1.4]{Adamo201712}, a point of the projective spinor bundle is the equivalence class of pairs $[(x^{\alpha\dot{\alpha}},\psi_\beta)]$, where $\alpha,\dot{\alpha}=1,2$ are the 2-spinor indices,
\begin{equation}\label{eq:2-spinor notation of spacetime events}
    [x^{\alpha\dot{\alpha}}]:=
    \begin{pmatrix}
        x^0+x^3 & x^1-\mfk{i}x^2
        \\
        x^1+\mfk{i}x^2 & x^0-x^3
    \end{pmatrix} \;,
\end{equation}
with $(x^0,x^1,x^2,x^3)$ is a spacetime point and $\psi_\beta$ is a projective spinor defined up to a multiplication by $\lambda\in\C^{\times}$. Hence $([x^{\alpha\dot{\alpha}}],\psi_\beta)\sim ([x^{\alpha\dot{\alpha}}],\lambda\psi_\beta)$. $\pi_{M}(([x^{\alpha\dot{\alpha}}],\psi_\beta))=[x^{\alpha\dot{\alpha}}]$ is a the projection to the spacetime points while $\pi_{\mcal{Z}_M}(([x^{\alpha\dot{\alpha}}],\psi_\beta))=(x^{\beta\dot{\alpha}}\psi_\beta,\psi_\alpha)$ imposes the incidence relation.

\smallskip To realize the two sides of the correspondence, we proceed as follows:

\begin{enumerate}
    \item [--] {\small\bf The Hyperbolic Monopole Side.} We take $M=\R^{1,3}$ or the open subset $\mcal{U}\subset\R^{1,3}$,\footnote{The condition of being an (anti)self-dual manifold descends from $\R^{1,3}$ to its open sets including $\mcal{U}$. It is an open subset of $\R^{1,3}$ endowed with the usual Euclidean topology. This can be seen as follows: The cone of time-like vectors with respect to an event $x\in\R^{1,3}$, which we denote as $\text{Con}^{\text{T}}(x)$, is defined by $\left\{y\in\R^{1,3}\,|\,-(x_0-y_0)^2+\sum_{i=1}^3(x_i-y_i)^2 <0 \right\}$. We then define $f:\R^{1,3}\to\R$ by $f(u):=-(x_0-u_0)^2+\sum_{i=1}^3(x_i-u_i)^2$, which is obviously a continuous function. We have $f^{-1}(-\infty,0)=\text{Con}^{\text{T}}(x)-\{x\}$. As $(-\infty,0)$ is open in $\R$ and $f$ is continuous, we conclude that $\text{Con}^{\text{T}}(x)-\{x\}$ is open in Euclidean topology. On the other hand, $\text{Con}^{\text{T}}(x)-\{x\}=\mcal{U}_+\cup\,\mcal{U}_-$, where $\mcal{U}_+=\mcal{U}$ and $\mcal{U}_-$ denote the cone of future-directed and past-directed time-like vectors, respectively. Since $\mcal{U}_+$ and $\mcal{U}_-$ are disjoint, and their union is open, both of them have to be open in Euclidean topology. If either or both of $\mcal{U}_\pm$ is closed, then their union would be closed. We emphasize that a general embedded submanifold of an (anti)self-dual manifold is not necessarily (anti)self-dual.} defined in \eqref{eq:definition of future-directed time-like vectors}. Then, the spaces appearing in \eqref{fig:correspondence defined by projective spinor bundle} become
    \begin{equation}\label{eq:choice of spaces in double fibration for hyperbolic monopole side}
        M=\mcal{U} \;,
        \qquad \mcal{Z}_M=\tsu \;,
        \qquad \sbm=\mscr{C}_{\mcal{U}} \;,
    \end{equation}
    as in \eqref{fig:correspondence space for u}.
    A further reduction of this double fibration by the dilatation (i.e.\ $\R_+$-action) will give the double fibration \eqref{fig:correspondence space for twistor space} and the construction explained in \S\ref{sec:hyperbolic monopoles and spectral data}. We shall explain how the 6d hCS theory describes antiself-dual instantons on $\mcal{U}$, and hence hyperbolic monopoles.

    \item[--] {\small\bf The gCPM Side.} We take $M=\R^4$, and the space appearing in \eqref{fig:correspondence defined by projective spinor bundle} become
    \begin{equation}\label{eq:choice of spaces in double fibration for gCPM side}
        M=\R^4 \;, 
        \qquad \mcal{Z}_M=\tsrf=\R^4\times\Pbb^1\to\Pbb^1 \;,
        \qquad \sbm=\tsrf \;,
    \end{equation}
    where $\tsrf$ denotes the twistor space of $\R^4$ and the fibration is holomorphic. We then argue how the dimensional reduction of 6d hCS theory to $\R^2\times\Pbb^1$ gives the 4d CS theory that describes the gCPM. 
\end{enumerate}

\smallskip Let us now explain how it works in detail.

\subsection{Hyperbolic Monopoles from Six Dimensions}\label{sec:hyperbolic monopoles from six dimensions}

We first discuss how to get hyperbolic monopoles from six dimensions. 

\smallskip Using \eqref{eq:choice of spaces in double fibration for hyperbolic monopole side}, the double fibration \eqref{fig:correspondence defined by projective spinor bundle} becomes
\begin{equation}\label{fig:double fibration in the case of instantons on U}
\begin{tikzcd}
    & \arrow[ld,"\pi_{\tsu}" swap] \Pbb S(\mcal{U})\arrow[rd,"\pi_\mcal{U}"] &
    \\
    \tsu & & \mcal{U}
\end{tikzcd}
\end{equation}
where $\Pbb S(\mcal{U})$ is an $\R_+$-fibration over $\mcal{Z}_\mcal{U}$ \cite{Porter198305}. We first explain how to realize self-dual connections on $\mcal{U}$ in terms of the 6d hCS theory. Upon $\R_+$-action, these connections will give solutions of the Bogomolny equation in the hyperbolic space. Finally, we elaborate on how to realize the characteristic features of hyperbolic monopoles, i.e.\ their masses and charges, and their spectral data in this picture.  

\subsubsection{Instantons from 6d Chern--Simons Theory}\label{sec:instantons from 6d CS theory}

Let us explain how instantons on $\mcal{U}$ descend from six dimensions. The basic idea is to use the results of Murray and Singer in \cite[Propositions 3.2 and 3.3]{MurraySinger199607}, as we briefly recalled in \hyperlink{correspondence between bundles over the double fibration}{Correspondence Between Bundles over the Double Fibration \eqref{fig:correspondence space for u}} in \S\ref{sec:the twistor correspondence}. According to \cite[Proposition 3.2]{MurraySinger199607}, there is a one-to-one correspondence between CR-bundles over $\tsu$ and integrable bundles over $\mscr{C}_\mcal{U}$, and according to \cite[Proposition 3.3]{MurraySinger199607}, there is a one-to-one correspondence between the latter and bundles with self-dual connections on a vector bundle over $\mcal{U}$. These results together with \eqref{fig:correspondence space for u} and \eqref{fig:double fibration in the case of instantons on U} give the following diagram
\begin{equation}\label{fig:square of correspondences for instantons on U}
\begin{tikzcd}
    & \arrow[ld,"\mu" swap] \mscr{C}_\mcal{U}\arrow[rd,"\nu"] &
    \\
    \mcal{Z}_{\mcal{U}} & & \mcal{U}
    \\
    & \arrow[ul,"\pi_{\tsu}"] \Pbb S(\mcal{U})\arrow[ur,"\pi_\mcal{U}" swap] &
\end{tikzcd}.
\end{equation}
The missing link is the following
\begin{prop}\normalfont \label{prop:one-to-one correspondence between CR bundles on ZU and holomorphic vector bundles on PS+U}
    Let $\mcal{U}$ be the open cone of future-directed timelike vectors in Minkowski space, $\tsu$ be its five-dimensional twistor space, and $\Pbb S(\mcal{U})$ be its projective spinor bundle. Then, there is a one-to-one correspondence between CR-bundles over $\tsu$ and holomorphic vector bundles over $\Pbb S(\mcal{U})$ equipped with a Hermitian structure. 
\end{prop}

\begin{proof}
    Let $(W,\bmb{\partial}_W)$ be a CR bundle over $\tsu$. Then, we define a vector bundle over $\Pbb S(\mcal{U})$ and a $\bmb{\partial}$ operator on it by 
    \begin{equation}
        \mcal{W}:=\pi_{\tsu}^*W \;,
        \qquad  \bmb{\partial}_{\mcal{W}}:=\pi_{\tsu}^*\bmb{\partial}_W \;.
    \end{equation}
    As a CR vector bundle is always a complex vector bundle equipped with a CR structure, its pullback is also a complex vector bundle. It is easy to see that $\bmb{\partial}_{\mcal{W}}$ induces a holomorphic structure on the complex vector bundle $\mcal{W}$
    \begin{eqaligned}
        \bmb{\partial}_{\mcal{W}}^2&=\bmb{\partial}_{\mcal{W}}\wedge \bmb{\partial}_{\mcal{W}}
        \\
        &=\pi_{\tsu}^*\bmb{\partial}_W\wedge \pi_{\tsu}^*\bmb{\partial}_W
        \\
        &=\pi_{\tsu}^*(\bmb{\partial}_W\wedge \bmb{\partial}_W)
        \\
        &=0 \;,
    \end{eqaligned}
    where in the last line, we have used the fact that $\bmb{\partial}_W$ defines a CR structure on $\tsu$ (i.e.\ $\bmb{\partial}_W\wedge \bmb{\partial}_W=0$). Hence, $(\mcal{W},\bmb{\partial}_\mcal{W})$ is a holomorphic vector bundle over $\Pbb S(\mcal{U})$.

    \smallskip Conversely, assume $(\mcal{W},\bmb{\partial}_\mcal{W})$ is a holomorphic vector bundle on $\Pbb S(\mcal{U})$. Then, we define a bundle over $\tsu$ by
    \begin{equation}
        W:=s^*_{\tsu}\mcal{W} \;, \qquad \bmb{\partial}_W:=s^*_{\tsu}\mcal{W} \;,
    \end{equation}
    where $s_{\tsu}$ is a section of the $\R_+$-projection $\Pbb S(\mcal{U})\to\tsu$. A similar argument shows that $(W,\bmb{\partial}_W)$ has a CR structure. If $s'_{\tsu}$ is another section of $\Pbb S(\mcal{U})\to\tsu$, with the corresponding CR vector bundle $(W',\bmb{\partial}_{W'})$ given by $W'={s'}_{\tsu}^*\mcal{W}$ and $\bmb{\partial}_{W'}={s'}_{\tsu}^*\bmb{\partial}_{\mcal{W}}$, we need to show that $(W,\bmb{\partial}_W)$ and $(W',\bmb{\partial}_{W'})$ are isomorphic CR vector bundles. Let $\gamma\in\tsu$ and define a curve $c:[0,1]\to\Pbb S(\mcal{U})$ with $c(0)=s_{\tsu}(\gamma)$ and $c(1)=s'_{\tsu}(\gamma)$. As $\mcal{W}$ is equipped with a Hermitian structure, $\bmb{\partial}_{\mcal{W}}$ defines the $(0,1)$-part of a unique connection $\nabla$ on $\mcal{W}$ whose curvature is of type $(1,1)$.\footnote{For a holomorphic vector bundle $\mcal{W}\to X$ with a Hermitian structure $H$ on a complex manifold $X$, there exists a unique compatible $H$-connection $\mcal{A}$ whose $(0,1)$-part coincides with the holomorphic structure on $\mcal{W}$ \cite[Theorem 10.3]{Moroianu200703}. It is easy to see that the curvature of this connection is of type $(1,1)$.} We define a parallel section $\psi\in\Gamma(\mcal{W},\Pbb S(\mcal{U}))$, which satisfies $\nabla_{\dot{c}}\psi=0$ and $\psi(c(0))=w$ with $w\in\rho^{-1}(c(0))$ is fixed, $\rho:\mcal{W}\to\Pbb S(\mcal{U})$ is the vector bundle projection, and $\dot{c}(t)$ is the tangent to the curve $c(t)$. Such a section is unique, as it satisfies a first-order differential equation and a fixed boundary condition. Then, by construction, $\psi(c(1))=w'\in\rho^{-1}(c(1))$, which provides a linear isomorphism between the fibers $\mcal{W}_{c(0)}$ and $\mcal{W}_{c(1)}$ over $c(0)$ and $c(1)$. As $W_\gamma=s^*_{\tsu}\mcal{W}_{c(0)}$ and $W'_\gamma={s'}^*_{\tsu}\mcal{W}_{c(1)}$, we end-up with  a unique bundle map $\Xi_{\psi}:W\to W'$ which intertwines the holomorphic structures, i.e.\ $\bmb{\partial}_{W'}\circ\Xi_\psi=\Xi_\psi\circ\bmb{\partial}_W$. Therefore, $(W,\bmb{\partial}_W)$ and $(W',\bmb{\partial}_{W'})$ are isomorphic CR vector bundles. 

    \smallskip We can thus conclude the proposition.  
\end{proof}

This result has an immediate consequence

\begin{cor}\normalfont\label{cor:one-to-one correspondence between self-dual connection on U and holomorphic vector bundle on PS+(U)}
    There is a one-to-one correspondence between bundles with self-dual connections on $\mcal{U}$ and holomorphic vector bundle on $\Pbb S(\mcal{U})$ equipped with a Hermitian structure.
\end{cor}

\begin{proof}
    It immediately follows from \cite[Propositions 3.2 and 3.3]{MurraySinger199607} and Proposition \ref{prop:one-to-one correspondence between CR bundles on ZU and holomorphic vector bundles on PS+U}.
\end{proof}

\smallskip Recall that for any connection $\mcal{A}$ on a complex vector bundle $\mcal{W}\to X$ over a complex manifold, we have 
\begin{equation}
    \mcal{F}^{(0,2)}=(\bmb{\partial}_{\mcal{A}})^2 \;,
\end{equation}
with $\mcal{F}=\rd\mcal{A}+\mcal{A}\wedge\mcal{A}$ is the curvature of $\mcal{A}=\mcal{A}^{(1,0)}+\mcal{A}^{(0,1)}$, and $\bmb{\partial}_{\mcal{A}}:=\bmb{\partial}+\mcal{A}^{(0,1)}$.  Therefore, if $\mcal{F}^{(0,2)}=0$, the $(0,1)$-part of the connection defines a holomorphic structure on $\mcal{W}$ given by $\bmb{\partial}_{\mcal{A}}^{(0,1)}$, which then would be an integrable operator on $\mcal{W}$.\footnote{The existence of a holomorphic structure on a complex vector bundle over a complex manifold is non-trivial. For a result in this direction see \cite{BanicaPutinar198706}.} Putting these results together, there is a clear way to obtain self-dual connections on $\mcal{U}$ from the 6d hCS theory: it amounts to showing that the 6d hCS theory, as formulated in \eqref{eq:6d hCS action on projective spinor bundle}, is the field theory description of connections on a vector bundle $\mcal{W}\to\Pbb S(\mcal{U})$ whose curvature is of type $(1,1)$. 

\smallskip It is easy to show that the 6d hCS theory is such a theory. The equations of motion of \eqref{eq:6d hCS action on projective spinor bundle} are
\begin{equation}\label{eq:equation of motion of 6d hCS theory}
    \Omega\wedge\mcal{F}=0 \;.
\end{equation}
$\mcal{F}$ has the decomposition $\mcal{F}=\mcal{F}^{(2,0)}+\mcal{F}^{(1,1)}+\mcal{F}^{(0,2)}$ and  \eqref{eq:equation of motion of 6d hCS theory} implies $\mcal{F}^{(0,2)}=0$. There are no other constraints on the components of $\mcal{F}$. On the other hand, recall that $\mcal{A}$ in \eqref{eq:expansion of A and del bar in 6d hCS in terms of local coordinate} has only anti-holomorphic components, hence $\mcal{F}^{(2,0)}=0$ automatically. Therefore, $\mcal{F}=\mcal{F}^{(1,1)}$ and is of type $(1,1)$. As such, the 6d hCS theory serves the purpose of describing connections on a holomorphic bundle over $\Pbb S(\mcal{U})$ whose curvature is of type $(1,1)$. By Corollary \ref{cor:one-to-one correspondence between self-dual connection on U and holomorphic vector bundle on PS+(U)}, it then follows that any connection of this sort would give rise to a self-dual connection $\wt{A}$ on $\mcal{U}$. To construct this self-dual connection explicitly, we choose a section $s_{\mcal{U}}:\mcal{U}\to\Pbb S(\mcal{U})$ and define 
\begin{equation}\label{eq:self-dual connection on U in terms of connection on PS+(U)}
    \wt{A}:=s_\mcal{U}^*\mcal{A} \;.
\end{equation}
The commutativity of the diagram \eqref{fig:square of correspondences for instantons on U} ensures the existence and uniqueness of $\wt{A}$.

 \smallskip As a final comment, we would like to highlight \cite{Popov199803,Popov199806} where the following was realized:
 Let $X$ be a three-dimensional complex manifold realized as a $\Pbb^1$-fibration over a self-dual Riemannian spacetime $M$.
 Then the condition of holomorphicity of a bundle $\mcal{W}\to X$ is described by the 6d hCS theory on $X$ and leads to self-duality equations on $M$. In particular, $X$ could be the twistor space of $\R^4$. However, the method used in loc.\ cit. is different than the ones we used here. This relation was later proposed more concretely in \cite{Costello202004} based on which the realization of self-duality in the form of the Yang equation \cite{Yang197706} was worked out explicitly in \cite{BittlestonSkinner202011}.

\subsubsection{Hyperbolic Monopoles as Dilatation-Invariant Instantons}\label{sec:hyperbolic monopoles as dilatation-invariant instantons}

As we explained in \S\ref{sec:hyperbolic monopoles and spectral data}, to construct hyperbolic $\SU{n}$-monopoles with arbitrary masses, Murray and Singer realize hyperbolic $\SU{n}$-monopoles as dilatation-invariant instantons on the open cone $\mcal{U}$ of future-pointing time-like vectors on the Minkowski spacetime \cite{MurraySinger199607}. Let us first establish this fact in an elementary fashion.

\smallskip Recall that $\mcal{U}$, defined in \eqref{eq:definition of future-directed time-like vectors}, is an open set of the Minkowski space and itself is a semi-Riemannian manifold. From the description of $\Hbb^3$ given in \eqref{eq:hyperboloid model of H3}, we can find an orthogonal decomposition of the tangent bundle of $\mcal{U}$ using the usual Minkowski inner product. Consider the embedding $\iota:\Hbb^3\to\mcal{U}$. Then,
\begin{equation}
    \iota^*T\mcal{U}\simeq T\Hbb^3\oplus N_{\Hbb^3/\mcal{U}} \;,
\end{equation}
where $N_{\Hbb^3/\mcal{U}}$ is the normal bundle of $\Hbb^3$ in $\mcal{U}$. Dually, we have
\begin{equation}\label{eq:decomposition of cotangent bundle}
    \iota^*T^*\mcal{U}\simeq T^*\Hbb^3\oplus N^\vee_{\Hbb^3/\mcal{U}} \;,
\end{equation}
where $N^\vee_{\Hbb^3/\mcal{U}}$ is the conormal bundle dual to $N_{\Hbb^3/\mcal{U}}$. Then, the metric can be decomposed as
\begin{equation}
    g_{\mcal{U}}=g_{\Hbb^3}\oplus g_{N_{\Hbb^3/\mcal{U}}} \;,
\end{equation}
where $g_{\Hbb^3}$ and $g_{N_{\Hbb^3/\mcal{U}}}$ are metrics on  $T\Hbb^3$ and $N_{\Hbb^3/\mcal{U}}$, respectively. Since $g_{N_{\Hbb^3/\mcal{U}}}\in \Gamma(\text{Sym}^2(N^\vee_{\Hbb^3/\mcal{U}}))$, a section of the second symmetric power of the conormal bundle $N^\vee_{\Hbb^3/\mcal{U}}$, we can always choose a local coordinate such that its only component is $-1$, as it is time-like, and hence $\det g_{\mcal{U}}=-\det g_{\Hbb^3}$. From \eqref{eq:decomposition of cotangent bundle}, we have a decomposition of a connection $\wt{A}$ on a vector bundle on $\mcal{U}$ as 
\begin{equation}\label{eq:decomposition of A on U in terms of A on H3 and Higgs field}
    s_\pi^*\wt{A}=A_{\Hbb^3}+\phi\,\rd x^\perp \;,
\end{equation}
with $s_\pi$ is a section of $\mcal{U}\to\Hbb^3$, and $A_{\Hbb^3}$ and $\phi$ denote the components of $\wt{A}$ along $\Hbb^3$ and $x^\perp$, a local coordinate normal to $\Hbb^3$.  Since we take $\mcal{U}$ as an $\R_+$-fibration over $\Hbb^3$, the normal direction is $\R_+$ whose local coordinate is $x^\perp$. On the other hand, the self-duality equation is\footnote{To avoid cluttering, we avoid using $s_\pi^*\wt{A}$. Everything involving gauge potential $s_\pi^*\wt{A}$ should be understood as $\wt{A}\circ s_\pi$.}
\begin{eqaligned}
    \wt{F}&=\wt{F}_{\rho'\sigma'}\rd x^{\rho'}\wedge\rd x^{\sigma'}=(\star_{\mcal{U}}\wt{F})=\wt{F}_{\mu\nu}\star_{\mcal{U}}\rd x^\mu\wedge\rd x^\nu
    \\
    &=\frac{1}{2}\sqrt{|\det g_\mcal{U}|} g_\mcal{U}^{\mu\mu'}g_\mcal{U}^{\nu\nu'}\varepsilon_{\mu'\nu'\rho'\sigma'}\wt{F}_{\mu\nu}\rd x^{\rho'}\wedge\rd x^{\sigma'} \;.
\end{eqaligned}
Assuming that $\wt{A}$ is $x^\perp$-independent (i.e.\ dilatation-invariant) $\wt{F}_{ix^\perp}=D_i\phi$, and we get
\begin{eqaligned}
    {F_{\Hbb}}_{ij}&=\frac{1}{2}\varepsilon_{ijkx^\perp}\sqrt{\det g_{\Hbb^3}}g_{\mcal{U}}^{kl}g_{\mcal{U}}^{x^\perp x^\perp}\wt{F}_{lx^\perp}+\frac{1}{2}\varepsilon_{ijx^\perp k}\sqrt{\det g_{\Hbb^3}}g_{\mcal{U}}^{x^\perp x^\perp}g_{\mcal{U}}^{lk} \wt{F}_{x^\perp l}
    \\
    &=-\frac{1}{2}\varepsilon_{x^\perp ijk}\sqrt{\det g_{\Hbb^3}}g_{\mcal{U}}^{kl}g_{\mcal{U}}^{x^\perp x^\perp}\wt{F}_{lx^\perp}-\frac{1}{2}\varepsilon_{x^\perp ij k}\sqrt{\det g_{\Hbb^3}}g_{\mcal{U}}^{x^\perp x^\perp}g_{\mcal{U}}^{lk} \wt{F}_{l x^\perp}
    \\
    &=\sqrt{\det g_{\Hbb^3}}\varepsilon_{ij}{}{}^k D_k\phi=(\star_{\Hbb^3}D\phi)_{ij} \;,
\end{eqaligned}
where $F_{\Hbb^3}:=\rd A_{\Hbb^3}+A_{\Hbb^3}\wedge A_{\Hbb^3}$ and in the last line, we have used $g_{\mcal{U}}^{x^\perp x^\perp}=-1$. Comparing to \eqref{eq:Bogomolny equation}, we see that the dilatation-invariant self-dual instantons on $\mcal{U}$ give the Bogomolny equation on $\Hbb^3$.  

\smallskip Under the $\R_+$-quotient, we arrive at the following double fibration

\begin{equation}\label{fig:double fibration in the case of monopoles on H3}
\begin{tikzcd}
    & \arrow[ld,"\mu" swap] \tsu \arrow[rd,"\nu"] &
    \\
    \ts & & \Hbb^3
\end{tikzcd}.
\end{equation}
with $\tsu\simeq \Pbb S(\mcal{U})/\R_+$, and by abuse of notation, we have used the same notation for the projection maps as is used for the case of instantons in diagram \eqref{fig:double fibration in the case of instantons on U}.

\begin{rmk}[Bundle of Self-Dual Two Forms on the Minkowski Space]\label{rmk:bundle of self-dual two-forms on the Minkowski space}\normalfont
    We emphasize that the bundle of self-dual two-forms on the Minkowski space is not real and hence there could not be any real self-dual two-form on the Minkowski space. In particular no real solution to the self-duality equation on the Minkowski space exists. To get such real solutions, one has to impose certain reality conditions on complex solutions, as was implemented by Murray and Singer \cite[\S 3.3]{MurraySinger199607} and we briefly reminded in \hyperlink{reality condition}{The Reality Condition} in \S\ref{sec:the twistor correspondence}. \qed 
\end{rmk}

With the relation \eqref{eq:self-dual connection on U in terms of connection on PS+(U)}, the decomposition \eqref{eq:decomposition of A on U in terms of A on H3 and Higgs field}, and the double fibration \eqref{fig:double fibration in the case of monopoles on H3} at hand, we can easily deduce the properties of hyperbolic monopoles. These are done in \S\ref{sec:masses, charges, and spectral data from six dimensions}.

\subsubsection{Masses, Charges, and Spectral Data from Six Dimensions}\label{sec:masses, charges, and spectral data from six dimensions}

Let us explain how to recover the masses, charges, and spectral data of hyperbolic monopoles from six dimensions. 

\paragraph{Masses and Charges from Six Dimensions.} Having the six-dimensional gauge field $\mcal{A}$, we can define a self-dual gauge field on $\mcal{U}$ using \eqref{eq:self-dual connection on U in terms of connection on PS+(U)}, and then decompose it using \eqref{eq:decomposition of A on U in terms of A on H3 and Higgs field}, where both $A_{\Hbb^3}$ and $\phi$ are $x^\perp$-independent. Hence, we can identify 
\begin{equation}\label{eq:connection and Higgs field on H3 in terms of 6d connection}
    A_{\Hbb^3}=((s_\pi\circ s_{\mcal{U}})^*\mcal{A})_{\Hbb^3} \;, 
    \qquad \phi=((s_\pi\circ s_{\mcal{U}})^*\mcal{A})_{x^\perp} \;,
\end{equation}
where $s_\pi:\Hbb^3\to\mcal{U}$ is a section of the $\R_+$ projection, and the subscripts denote the components of $s_{\mcal{U}}^*\mcal{A}$ along those directions. 

\smallskip The definition of charges of a hyperbolic $G$-monopole, where $\mfk{g}=\text{Lie}(G)$ in \eqref{eq:6d hCS action on projective spinor bundle}, with the maximal symmetry-breaking pattern is given in \eqref{eq:definition of magnetic charges of a monopole}. From \eqref{eq:connection and Higgs field on H3 in terms of 6d connection}, we have
\begin{eqaligned}
    m_i=\omega_i(2\star_{S^2_\infty}{F_{\Hbb^3}}\big|_{S^2_\infty}) \;,
    \qquad i=1,\ldots,\rnk{g} \;,
\end{eqaligned}
and $F_{\Hbb^3}$ is given explicitly in terms of $\mcal{A}$ as
\begin{eqaligned}
    \qquad F_{\Hbb^3}:&=\rd A_{\Hbb^3}+A_{\Hbb^3}\wedge A_{\Hbb^3}
    \\
    &=\rd ((s_\pi\circ s_{\mcal{U}})^*\mcal{A})_{\Hbb^3}+((s_\pi\circ s_{\mcal{U}})^*\mcal{A})_{\Hbb^3}\wedge ((s_\pi\circ s_{\mcal{U}})^*\mcal{A})_{\Hbb^3} \;.
\end{eqaligned}
Recall that $\omega_i$s are the set of fundamental weights of $\mfk{g}$.

\smallskip On the other hand, the masses of a hyperbolic $G$-monopole are defined in \eqref{eq:definition of masses of hyperbolic monopoles}. We can thus identify the masses of the hyperbolic monopoles as
\begin{equation}
    p_i=\omega_i\left((s_\pi\circ s_{\mcal{U}})^*\mcal{A})_{x^\perp}\big|_{S^2_\infty}\right) \;,
    \qquad i=1,\ldots,\rnk{\mfk{g}} \;.
\end{equation}

We thus successfully expressed the masses and charges of hyperbolic $G$-monopoles very explicitly from six dimensions. 

\begin{rmk}[The Choice of Gauge Lie Algebra for 6d hCS Theory]\label{rmk:choice of gauge lie algebra for 6d hCS theory, the hyperbolic monopole side}\normalfont
    We still need to specify the gauge Lie algebra of 6d theory in \eqref{eq:6d hCS action on projective spinor bundle}, which through the above procedure gives the gauge group $\SU{n}$ for the hyperbolic monopole. From \eqref{eq:connection and Higgs field on H3 in terms of 6d connection}, we see that the natural choice is $\mfk{g}=\mfk{su}(n)$. However, the 6d hCS is naturally formulated using a complex Lie algebra. Hence, $\mfk{g}=\mfk{sl}(n,\C)$ is the natural candidate. We have explained in Remark \ref{rmk:bundle of self-dual two-forms on the Minkowski space} that some reality conditions must be imposed to make real solutions possible. Imposing the reality conditions leaves us with the real forms of $\mfk{sl}(n,\C)$, which are either the compact real form $\mfk{su}(n)$, or the split (i.e.\ noncompact) real form $\mfk{sl}(n,\R)$. The compactness of the gauge group follows from the charge quantization \cite{Yang197004}. The latter, through the Dirac quantization condition, is connected to the existence of monopoles \cite{GoddardNuytsOlive1977,GoddardOlive197809}. Hence, there are no magnetic monopole solutions for non-compact gauge groups. We are thus left with the only possibility for the gauge group of the hyperbolic monopole which is $\SU{n}$ whose gauge Lie algebra $\mfk{su}(n)$ is the compact real form of the gauge Lie algebra $\mfk{sl}(n,\C)$ of the 6d hCS theory. See also Remark \ref{rmk:choice of gauge lie algebra for 6d hCS theory, the gCPM side}   \qed 
\end{rmk}

The next task is to work out the spectral data, which we do next.

\paragraph{Spectral Data from Six Dimensions.}

Let us finally explain how to recover the spectral data of a hyperbolic $\SU{n}$-monopole from six dimensions. Let $\mcal{W}\to\Pbb S(\mcal{U})$ be a holomorphic vector bundle equipped with a holomorphic structure $\bmb{\partial}_W=\bmb{\partial}+\mcal{A}$, where $\mcal{A}$ is a solution to the equations of motion \eqref{eq:equation of motion of 6d hCS theory} of the 6d hCS theory. By Corollary \ref{cor:one-to-one correspondence between self-dual connection on U and holomorphic vector bundle on PS+(U)}, we can choose a section $s_{\mcal{U}}:\mcal{U}\to\Pbb S(\mcal{U})$ and define the corresponding bundle on $\mcal{U}$ by
\begin{equation}\label{eq:bundle on U from bundle W on PS+(U)}
    \wt{V}:=s^*_\mcal{U}\mcal{W} \;,
\end{equation}
with a self-dual connection as \eqref{eq:self-dual connection on U in terms of connection on PS+(U)}. Furthermore, from \eqref{eq:bundles on H3 from bundles on U}, we can define the corresponding bundles $\wt{E}\to\Hbb^3$ as
\begin{equation}\label{eq:bundle on H3 from W bundle on PS+(U)}
    \wt{E}:=s_{\pi}^*\wt{V}=(s_\pi\circ s_\mcal{U})^*\mcal{W} \;.
\end{equation}
It is equipped with a connection and Higgs field that by \eqref{eq:decomposition of A on U in terms of A on H3 and Higgs field} are given in \eqref{eq:connection and Higgs field on H3 in terms of 6d connection}. This is essentially what one needs to construct the spectral data of hyperbolic $\SU{n}$-monopole. Let us briefly summarize the construction of the spectral data
\begin{enumerate}
    \item [(1)] We first impose boundary conditions, as explained in \hyperlink{reality condition}{Boundary Conditions} in \S\ref{sec:the twistor correspondence}.

    \item [(2)] By the breaking of gauge invariance near $S^2_\infty$ as $\SU{n}\to\U{1}^{n-1}$, from \eqref{eq:vector bundle EU as a direct sum of EUi}, the bundle \eqref{eq:bundle on H3 from W bundle on PS+(U)} would have a direct-sum decomposition
    \begin{equation}
        \wt{E}=\bigoplus_{i=1}^n \wt{E}_i \;,
    \end{equation}
    where each one-dimensional bundle $E_i$ is given by \eqref{eq:definition of 1d bundles tildeE on U} and as such is determined by the Higgs field \eqref{eq:connection and Higgs field on H3 in terms of 6d connection}.

    \item [(3)] From these bundles, one can define bundles
    \begin{equation}
        E^\pm=\pi^*_\pm\wt{E} \;,
    \end{equation}
    with $\pi_\pm$ are given in \eqref{eq:sections pipm of the correspondence}, which in turn, by the virtue of \eqref{eq:definition of bundles Epmi over ZH3}, will have the decompositions
    \begin{equation}
        E^+_i=\pi^*_+\left(\bigoplus_{j=1}^{i}\wt{E}_{n-j+1}\right) \;, \qquad 
        E^-_i=\pi^*_-\left(\bigoplus_{j=1}^{i}\wt{E}_{j}\right) \;.
    \end{equation}
    By Theorem \ref{thr:twistor correspondence theorem}, these subbundles sit in filtrations
    \begin{equation}
        0\simeq E^\pm_0\subset E^\pm_1\subset \ldots\subset E^\pm_{n-1}\subset E^\pm_n\simeq E^\pm \;.
    \end{equation}
    Note that we have to impose boundary conditions to be able to use Theorem \ref{thr:twistor correspondence theorem}.
    
    \item[(4)] One can then define the $i$\textsuperscript{th} spectral curve $S_i$, using \eqref{eq:the map Gamma-}, as the divisor of the map
    \begin{equation}
        \det E_i^-\to \det E^+/E^+_{n-1} \;, \qquad i=1,\ldots,n-1 \;,
    \end{equation}
    which defines the spectral data of a generic monopole by the definition given in \hyperlink{spectral data}{Spectral Data} in \S\ref{sec:spectral data and recovering monopole solutions}. 

    \item[(5)] The spectral curves, by construction, would recover the monopole solution by what we have explained in \S\ref{sec:spectral data and recovering monopole solutions}.
\end{enumerate}

One can thus reconstruct the spectral data of hyperbolic $\SU{n}$-monopoles from six dimensions. With this result, we conclude our construction of hyperbolic $\SU{n}$-monopole from the 6d hCS theory. 

\subsection{Generalized CPM from Six Dimensions}\label{sec:gCPM from six dimensions}

In this section, we discuss the other side of the hyperbolic monopole/gCPM correspondence. Based on discussions in \S\ref{sec:gCPM and 4d CS theory}, we explain how gCPM can be realized from the viewpoint of 6d hCS theory.

\smallskip From \eqref{fig:correspondence defined by projective spinor bundle} and \eqref{eq:choice of spaces in double fibration for gCPM side}, we consider the following double fibration
\begin{equation}\label{fig:double fibration in the case of gCPM}
\begin{tikzcd}
    & \arrow[ld,"\pi_{\tsrf}" swap] \sbe\arrow[rd,"\pi_{\R^4}"] &
    \\
    \tsrf & & \R^4
\end{tikzcd},
\end{equation}
where $\pi_{\tsrf}$ is trivial in this case since $\sbe\simeq \tsrf$ with $\tsrf$ being the twistor space of $\R^4$. From this, we would perform the following steps: We consider the 6d hCS theory on the $\sbe\simeq \tsrf$. These spaces have many equivalent descriptions. We take the description  $\R^4\times\Pbb^1\to\Pbb^1$, where the projection is holomorphic. We explain how the integration over an $\R^2$ plane (or rather its compactification to a 2-torus) inside $\R^4\times\Pbb^1$ will result in the 4d holomorphic BF (hBF) theory on $\R^2\times\Pbb^1$. The latter theory would reduce to the 4d CS on $\R^2\times\Pbb^1$ in the limit of the small volume of the torus. We use a specific holomorphic volume form $\Omega$ on $\R^4\times\Pbb^1$, which will coincide with the one-form $\omega_{\Pbb^1}$ given in \eqref{eq:one-form relevant for engineering gcpm on P1}. Based on our discussion in \S\ref{sec:gCPM and 4d CS theory}, we recover the gCPM from this description of the 4d CS theory, hence from six dimensions. 

\paragraph{Twistor Space of $\R^4$.} We have $\sbe=\R^4\times\Pbb^1$ as a smooth manifold \cite[\S 4, pg.\ 438, Example 1]{AtiyahHitchinSinger197809}. Topologically, $\tsrf$ can be considered as $\Pbb^3$ from which a line is removed,\footnote{Recall from \S\ref{sec:more details on twistor space of hyperbolic space} that the twistor space of $S^4$ is $\Pbb^3$. To get to $\R^4$, one needs to remove the point at infinity from $S^4$, which corresponds to removing the corresponding line $\Pbb^1$ from $\Pbb^3$.} hence $\tsrf\simeq \Pbb^3\backslash\Pbb^1$. $\tsrf$ can also be thought of as the total space of the fibration $\tsrf\to\Pbb^1$ where the map is given by the projection onto the second factor. Since $\mbb{R}^4$ is a hyper-K\"ahler manifold, it admits a $\Pbb^1$-worth of complex structures, which is the $\Pbb^1$ factor in $\tsrf$, i.e.\ a complex structure on $\mbb{R}^4$ is given by a point in $\Pbb^1$. An almost complex structure on $\tsrf$ is given by a pair $(\mcal{J}_z,\mcal{J})$, where $\mcal{J}_z$ is the particular complex structure associated to a point $z\in\Pbb^1$ and $\mcal{J}$ is the standard complex structure on $\Pbb^1$ in which $z$ is a holomorphic coordinate. One can then use the Newlander--Nirenberg theorem to show that this is integrable and indeed defines a complex structure on $\mbb{R}^4\times\Pbb^1$ with respect to which $\R^4\times\Pbb^1\to\Pbb^1$ is holomorphic. Another description of $\tsrf$ is as the total space of $\mcal{O}_{\Pbb^1}(1)\oplus\mcal{O}_{\Pbb^1}(1)\to\Pbb^1$. 

\paragraph{The Choice of Sections of Canonical Bundle.} We will see that the section $\Omega$ of canonical bundle in \eqref{eq:6d hCS action on projective spinor bundle} will take the following form
\begin{equation}\label{eq:choice of section of canonical bundle of twistor space}
    \Omega=\Omega(z)\,\rd w_1\wedge\rd w_2\wedge \rd z \;,
\end{equation}
with $\Omega(z)$ a function of only $z$, which we will identify its form. To connect to the 4d CS theory, certain assumptions are needed and then this form will be dictated and there is not much freedom in choosing it. To have a well-defined action \eqref{eq:6d hCS action on projective spinor bundle}, the integrated should be invariant under the projective transformations of either number of holomorphic coordinates, a condition that is automatically satisfied since $\Omega\wedge\CS(\mcal{A})$ is a differential form. We only demand it to be locally nonvanishing (i.e.\ on an open set), not having a pole, and also closed. As we will see below, $\Omega(z)$ would be identified with $\omega_{\Pbb^1}$, given in \eqref{eq:one-form relevant for engineering cpm on P1} or \eqref{eq:one-form relevant for engineering gcpm on P1}, and as such is locally nonvanishing and closed top-degree holomorphic form and a well-defined section of the canonical bundle of $\sbm$. 

\subsubsection{4d Holomorphic BF Theory from 6d Chern--Simons Theory}\label{sec:4d hBF theory from 6d hCS theory}

\smallskip As one of the main ingredients of our arguments is the dimensional reduction of the 6d hCS theory along a plane, let us briefly explain how the procedure will always give rise to the 4d holomorphic BF (hBF) theory in the orthogonal directions. The hBF theories have been introduced in \cite{Popov199909} and have been further studied in \cite{IvanovaPopov200002,BaulieuTanzini200412}.

\smallskip Consider the 6d hCS theory on $\sbe=\R^2_{w_1}\times\R^2_{w_2}\times\Pbb^1_z$, where the subscripts show the holomorphic coordinates along the corresponding directions, the integrand of \eqref{eq:6d hCS action on projective spinor bundle} can be written as
\begin{equation}\label{eq:rewriting the lagrangian of 6d hCS}
    \Omega\wedge\CS(\mcal{A})=\left(\Omega\wedge\rd\bmb{w}_1\wedge\rd\bmb{w}_2\wedge\rd\bmb{z}\right)\,\text{Tr}_{\mfk{g}}(\mcal{A}_{\bmb{w}_1}\partial_{\bmb{w}_2}\mcal{A}_{\bmb{z}}-\mcal{A}_{\bmb{z}}\partial_{\bmb{w}_2}\mcal{A}_{\bmb{w}_1}+2\mcal{A}_{\bmb{w}_2}\mcal{F}_{\bmb{z}\bmb{w}_1}) \;,
\end{equation}
where $\mcal{F}_{\bmb{z}\bmb{w}_1}=\partial_{\bmb{z}}\mcal{A}_{\bmb{w}_1}-\partial_{\bmb{w}_1}\mcal{A}_{\bmb{z}}+[\mcal{A}_{\bmb{z}},\mcal{A}_{\bmb{w}_1}]$, and we have performed some integration by parts, and have used (see Remark \ref{rmk:importance of vanishing of dbar derivative of omegaP1} for more details)
\begin{equation}\label{eq:vanishing of anti-holomorphic derivative of Omega}
    \partial_{\bmb{z}}\Omega=0=\partial_{\bmb{w}_1}\!\Omega \;,
\end{equation}
which follows from \eqref{eq:choice of section of canonical bundle of twistor space} and our assumption that $\Omega(z)$ does not have a pole, and some properties of $\text{Tr}_{\mfk{g}}$. We would like to integrate over the $\R^2_{w_2}$ part of the fiber of the fibration $\R^4\times\Pbb^1\to\Pbb^1$. First notice that the 6d theory does not depend on the choice of a metric and only a complex structure. To integrate over the fiber, we first compactify $\R^2_{w_2}$ to a torus $T^2$ and choose a metric of the form
\begin{equation}
    g_{\sbe}=g_{T^2}\oplus g_{\R^2_{{w}_1}}\oplus g_{\Pbb^1_{z}} \;.
\end{equation}
The metric must be compatible with the chosen complex structure $\mcal{J}_z$ associated with $z\in\Pbb^1$. It is always possible to choose such a metric on $\R^4\times\Pbb^1$. For a complex structure
\begin{equation}\label{eq:choice of complex structure}
    \mathcal{J}_z=\mfk{i}\left(\partial_z\otimes\rd z-\partial_{\bmb{z}}\otimes\rd\bbar{z}\right)\oplus\mfk{i}\sum_{i=1}^2\left(\partial_{w_i}\otimes\rd w_i-\partial_{\bmb{w}_i}\otimes\rd \bbar{w}_i\right) \;,
\end{equation}
it is easy to see that the metric
\begin{equation}
    \rd s^2_{\sbe}=(g_{z\bmb{z}}\rd z\otimes\rd \bbar{z}) \oplus\sum_{i=1}^2 g_{w_i\bmb{w}_i}\rd w_i\otimes\rd\bbar{w}_i \;,
\end{equation}
for some functions $g_{z\bmb{z}}$ and $g_{w_i\bmb{w}_i}$, which in principle could depend on all of the coordinates $(z,\bmb{z},w_i,\bmb{w}_i)$, is clearly compatible with the complex structure \eqref{eq:choice of complex structure}.

\begin{rmk}\normalfont 
    For the purpose of integration, it does not matter whether we think of $\sbe$ as a complex manifold or a smooth real manifold. We only need to keep the complex structure of $\Pbb^1\subset\sbe$. We think of it as a smooth real manifold isomorphic to $\R^4\times\Pbb^1$ where $\R^4$ is equipped with a specific complex structure explained in Remark \ref{rmk:complex structure on PS+R4}.
\end{rmk}

\begin{rmk}[The Choice of Complex Structure on $\sbe$]\normalfont
    In the above construction, we have used a specific complex structure on $\sbe$ in which $w_1$ and $w_2$ are holomorphic. As we will explain in the proof of Proposition \ref{prop:PS+(M) is holomorphically embedded in PS+C13}, the complex structure needs to be a very specific one, the one defined in \eqref{eq:complex structure on R4}. More specifically, the coordinates $(w_1,w_2)$ must coincide with coordinates $(x^+_E,x^-_E)$ defined in \eqref{eq:complex structure on R4}. This is mandatory to be able to use the result discussed in \S\ref{sec:6d hCS on PS+(M) from 10d hCS on PS+C13}. As it will be clear, the construction does not have much freedom and indeed is very rigid. \label{rmk:complex structure on PS+R4}
\end{rmk}

\smallskip As we would like to perform a dimensional reduction, we will assume that all fields and also $\Omega=\Omega(z,w_1,w_2)\rd w_1\wedge\rd w_2 \wedge \rd z$ are $(w_2,\bmb{w}_2)$-independent. We choose the metric $g_{T^2}$ judiciously such that (1) it is only a function of $(w_2,\bmb{w}_2)$ and (2) $\text{Vol}(T^2)\to 0$ in the IR limit. Then, the first term inside the trace in \eqref{eq:rewriting the lagrangian of 6d hCS} vanishes and by integrating over $(w_2,\bmb{w}_2)$, we arrive at the action of the hBF theory on $\R^2_{w_1}\times\Pbb^1_z$
\begin{equation}\label{eq:action for hBF theory in 4d CS side}
    S=\frac{1}{2\pi}\bigintsss_{\R^2\times\Pbb^1}\text{Tr}_{\mfk{g}}(B^{(2,0)}\wedge F^{(0,2)}) \;,
\end{equation}
with the following field redefinition ($w_1\to w$) 
\begin{eqaligned}
    \label{eq:fields of the 4d hBF theory in terms of fields of 6d hCS theory}
    B^{(2,0)}&:=\text{Vol}(T^2)\Omega(z,w)\rd z\wedge \mcal{A}_{\bmb{w}_2}\rd w \;,
    \\
    F^{(0,2)}&:=\mcal{F}_{\bmb{w}\bar{z}}\rd\bmb{w}\wedge\rd\bmb{z} \;.
\end{eqaligned}
We now explain how this action can be recovered from the action of the 4d CS theory.

\begin{rmk}\normalfont 
    In principle, we could have chosen a more general form of $\Omega$ which depends on $w_2$. This would not affect the analysis drastically. The only difference would have been that integration of $\R^2_{w_2}$ would produce a prefactor $c\times\text{Vol}(T^2)$, for some finite constant $c$, instead of $\text{Vol}(T^2)$. This constant can always be reabsorbed in the field redefinition of four-dimensional fields. However, one can argue that $\Omega$ should not depend on $w_2$. By \eqref{eq:vanishing of anti-holomorphic derivative of Omega} and Remark \ref{rmk:importance of vanishing of dbar derivative of omegaP1} below, ${\Omega}$ should not have a pole, and furthermore it should not have a zero, as $\Omega$ should be nonzero on an open set. This will reduce the possibilities to two cases: (1) Either one needs to impose appropriate boundary conditions at the location of poles of $\Omega$; or (2) Assume independence from $w_1,w_2,$ and $z$. As we will see, it is enough for us to take $\Omega$ to be independent of $w_1$ and $w_2$. We will realize in \S\ref{sec:emergence of 4d CS theory} that connection to the 4d CS theory implies that $\Omega$ cannot be independent of $z$ but \eqref{eq:vanishing of zbar derivative of omegaP1 for gCPM} comes to the rescue. See in particular Remark \ref{rmk:importance of vanishing of dbar derivative of omegaP1}.
    
\end{rmk}

\subsubsection{Emergence of 4d Chern--Simons Theory} \label{sec:emergence of 4d CS theory}

The theory \eqref{eq:action for hBF theory in 4d CS side} is not immediately connected to the 4d CS theory. However, the latter reduces to \eqref{eq:action for hBF theory in 4d CS side} in the appropriate limit, as we will now explain. 

\smallskip Consider the action \eqref{eq:4dCS action} and the connection $A=A_{w}\rd w+A_{\bmb{w}}\rd\bmb{w}+A_{\bmb{z}}\rd\bmb{z}$. The integrated multiplied with $\text{Vol}(T^2)$, after some integration by parts, can be written as 
\begin{eqaligned}\label{eq:action of 4dCS multiplied with volume of torus}
    \omega\wedge\CS(A)&=2\text{Vol}(T^2)\text{Tr}_{\mfk{g}}\left(\omega\wedge (A_w\rd w)\wedge (F_{\bmb{w}\bmb{z}}\rd\bmb{w}\rd\bmb{z})\right)
    \\
    &+2\text{Vol}(T^2)\text{Tr}_{\mfk{g}}\left(\omega\wedge(A_{\bmb{z}}\partial_{w}A_{\bmb{w}})\rd\bmb{z}\wedge\rd w\wedge\rd\bmb{w}\right)
    \\
    &=\text{Tr}_{\mfk{g}}(B\wedge F)+2\text{Vol}(T^2)\text{Tr}_{\mfk{g}}\left(\omega\wedge(A_{\bmb{z}}\partial_{w}A_{\bmb{w}})\rd\bmb{z}\wedge\rd w\wedge\rd\bmb{w}\right) \;,
\end{eqaligned}
where we have introduced the following field redefinitions
\begin{equation}\label{eq:fields of the 4d hBF in terms of fields of 4d CS theory}
    B:=2\text{Vol}(T^2)\,\omega\wedge A_w\rd w \;,
    \qquad F:=F_{\bmb{w}\bmb{z}}\rd\bmb{w}\rd\bmb{z} \;.
\end{equation}
If we now send $\text{Vol}(T^2)\to 0$, we see that the second term of \eqref{eq:action of 4dCS multiplied with volume of torus} drops out and the first term can be identified with \eqref{eq:action for hBF theory in 4d CS side} provided that we assume $\Omega(z,w)=\Omega(z)$ and
\begin{equation}\label{eq:identification of holomorphic volume form in 6d and one-form of hCS theory}
    \Omega(z)\rd z=\omega(z)\rd z \;,
\end{equation}
which in turn, using \eqref{eq:fields of the 4d hBF theory in terms of fields of 6d hCS theory}, will give the following identifications
\begin{equation}
    B= B^{(0,2)} \;, \qquad  F= F^{(0,2)} \;.
\end{equation}
By comparing \eqref{eq:fields of the 4d hBF theory in terms of fields of 6d hCS theory} and \eqref{eq:fields of the 4d hBF in terms of fields of 4d CS theory}, the fields of the 4d and 6d hCS theories can be identified as
\begin{equation}
    2A_w=\mcal{A}_{\bmb{w}_2} \;, \qquad F_{\bmb{w}\bmb{z}}=\mcal{F}_{\bmb{w}\bmb{z}} \;.
\end{equation}
Note that $(\text{Vol}(T^2))^{1/2}$ behaves like the scale of the effective field theory. The above computation shows that in the deep IR when $\text{Vol}(T^2)\to 0$, the effective theory is the 4d hBF theory while if we keep $\text{Vol}(T^2)$ small but finite, the effective theory is the 4d CS theory. 

\smallskip We still need to make contact with the gCPM. From \eqref{eq:identification of holomorphic volume form in 6d and one-form of hCS theory}, this is achieved by setting
\begin{equation}\label{eq:identification of holomorphic volume form in 6d and one-form describing gCPM}
    \Omega(z)\rd z=\omega(z)\rd z=\omega_{\Pbb^1} \;,
\end{equation}
where the latter is given in \eqref{eq:one-form relevant for engineering gcpm on P1}. Recall from \eqref{eq:vanishing of zbar derivative of omegaP1 for gCPM} that $\partial_{\bmb{z}}\omega_{\Pbb^1}=0$, and it is perfectly fine to identify it with $\Omega(z)$ due to \eqref{eq:vanishing of anti-holomorphic derivative of Omega}. We thus see that, at the finite scale of effective field theory set by the volume of extra dimensions of $\sbe$, the action of the 6d hCS theory on $\sbe$ reduces to the following action of 4d CS theory on $\R^2\times \Pbb^1$ 
\begin{equation}
    S=\frac{1}{2\pi}\bigintsss_{\R^2\times\Pbb^1}\omega_{\Pbb^1}\wedge\CS(A) \;,
\end{equation}
with the identification \eqref{eq:identification of holomorphic volume form in 6d and one-form describing gCPM} of 4d and 6d fields. Hence, by what we have explained in \S\ref{sec:gCPM and 4d CS theory}, this reduction realizes the discrete integrable-model side of the hyperbolic monopole/gCPM correspondence. 

\begin{rmk}[The Importance of $\partial_{\bmb{z}}\omega_{\Pbb^1}=0$]\label{rmk:importance of vanishing of dbar derivative of omegaP1}\normalfont
    The identification \eqref{eq:identification of holomorphic volume form in 6d and one-form describing gCPM} shows the importance of \eqref{eq:vanishing of zbar derivative of omegaP1 for gCPM}, which is two-fold:

    \begin{enumerate}
        \item [--] {\small\bf Potential Breaking of the Gauge Invariance.} One of the main concerns of \cite{Costello202004,BittlestonSkinner202011} was the danger of potential breaking of the gauge invariance at the location of poles of $\Omega$. This is due to the fact that by performing the gauge transformation $\mcal{A}+\bmb{\partial}\alpha$, for some $\mfk{g}$-valued function $\alpha$, the first term of \eqref{eq:CS term in the case of 6d hCS theory} transforms as
        \begin{equation}
        \text{Tr}_{\mfk{g}}(\mcal{A}\bmb{\partial}\mcal{A})\quad\mapsto\quad \text{Tr}_{\mfk{g}}(\mcal{A}\bmb{\partial}\mcal{A})+\text{Tr}_{\mfk{g}}(\bmb{\partial}\alpha\wedge\bmb{\partial}\mcal{A}) \;.
        \end{equation}
        By performing an integration by parts, the second term leads to $\bmb{\partial}\Omega\wedge\alpha\wedge\bmb{\partial}\mcal{A}$. Arguments similar to what we have presented in \S\ref{sec:gauge Lie algebra, boundary conditions, R-matrix} show that if $\Omega$ has a pole, then $\bmb{\partial}\Omega$ is proportional to (possibly a derivative of) delta functions. This can potentially lead to the breaking of gauge invariance. In our case, the identification \eqref{eq:identification of holomorphic volume form in 6d and one-form describing gCPM} and \eqref{eq:vanishing of zbar derivative of omegaP1 for gCPM} guarantee that there are no contributions of this sort and the gauge invariance is preserved.

        \item [--] {\small\bf Peculiarity of gCPM.} One may naively think that the reductions from the 6d hCS to the 4d CS can be done such that the resulting theory describes other integrable models studied in \cite{CostelloWittenYamazaki201709} and possibly relate it to an integrable field theory. However, this is not possible precisely because of \eqref{eq:vanishing of zbar derivative of omegaP1 for gCPM}. In the case of other integrable spin models, $\partial_{\bmb{z}}\omega_{\Pbb^1}$ gives a (possibly derivative of) delta function $\delta^{(2)}(z,\bmb{z})$ \cite[\S 9.1]{CostelloWittenYamazaki201709}, and as a result, we cannot simply identify it with $\Omega$, and in turn, we cannot apply the argument of reducing the 6d hCS theory to the 4d CS theory. 
        This provides another explanation of the peculiarity of the gCPM. See also Remark \ref{rmk:why gCPM is peculiar?}.
    \end{enumerate}
    
    These arguments demonstrate the essentiality of the relation \eqref{eq:vanishing of zbar derivative of omegaP1 for gCPM}.\qed
\end{rmk}

\begin{rmk}[The Choice of Gauge Lie Algebra for 6d hCS Theory]\label{rmk:choice of gauge lie algebra for 6d hCS theory, the gCPM side}\normalfont
    We still need to specify the gauge Lie algebra of the 6d theory in \eqref{eq:6d hCS action on projective spinor bundle}. Based on our discussion in \S\ref{sec:gauge Lie algebra, boundary conditions, R-matrix}, the natural choice of the gauge Lie algebra for the 4d CS theory is $\mfk{sl}(n,\C)$. This naturally descends from the 6d hCS theory with the gauge algebra $\mfk{sl}(n,\C)$. Recall from Remark \ref{rmk:choice of gauge lie algebra for 6d hCS theory, the hyperbolic monopole side} that we had to impose reality condition to get instanton solutions on the Minkowski space. Here, we do not need to impose such a reality condition and hence the gauge algebra of the 4d CS theory will be $\mfk{sl}(n,\C)$. \qed 
\end{rmk}

\begin{rmk}[No Parameters Corresponding to Masses of Higgs Fields in gCPM]\normalfont
    In the gCPM side of the correspondence, there are no parameters corresponding to the boundary values $\{p_1,\ldots,p_n\}$ of the Higgs field. Recall that at the boundary sphere $S^2_\infty$, the gauge invariance breaks as $\SU{n}\to\U{1}^{n-1}$, due to the VEV of the Higgs field, which is an adjoint-valued scalar field emerges by reduction from 4d to 3d. These parameters deform the theory. However, there is no scalar field in the 6d and 4d CS theories for which we can turn on background VEV. It is not possible to give a background VEV to a gauge field as it breaks gauge invariance, which we do not want to do. This explains why there are no parameters corresponding to the masses of the Higgs fields in the gCPM.\label{rmk:no parameters corresponding to masses in the gCPM side}   \qed
\end{rmk}

\begin{rmk}[No Reality Condition on the Data]\normalfont\label{rmk:no reality condition on the data}
    We have seen two types of reality conditions one needs to impose on the data of hyperbolic monopoles (1) the reality condition to have real solutions of the Bogomolny equation (see \hyperlink{reality condition}{The Reality Condition} in \S\ref{sec:the twistor correspondence}), and (2) And the reality of spectral curves in $\ts$, in the sense of \eqref{eq:reality of spectral curves of hyperbolic monopoles}. Both of these reality conditions are absent in the gCPM side of the correspondence since we start from the 6d hCS theory on a different space $\sbe$ and compare it to the 6d hCS theory on $\sbrm$ that gives rise to the hyperbolic monopole side of the correspondence. In particular, the curve of the spectral parameter of the gCPM is a complex and not real curve.
\end{rmk}

\subsection{The Correspondence from Ten Dimensions}

Up to now, we have established that the two sides of the hyperbolic monopole/gCPM correspondence are described in terms of the 6d hCS theory on $\sbm$, where $M$ is the Minkowski space $\R^{1,3}$ (or the open subset $\mcal{U}\subset\R^{1,3}$) in the monopole side and the Euclidean space $\R^4$ in the gCPM side. In this section, we explain how these two sides can be unified in ten-dimensional physics.

\smallskip The basic idea is clear: $\R^{1,3}$ and $\R^4$ can both be realized as real slices of the complexified Minkowski spacetime, which we denote as $\C^{1,3}$. This space is simply $\C^4$ equipped with the Minkowski quadratic pairing $\langle\cdot,\cdot\rangle:\C^{1,3}\times\C^{1,3}\to\C$
\begin{equation}
    \langle x,y\rangle:=-x_0y_0+\sum_{i=1}^3x_iy_i \;, \qquad \forall x,y\in\C^{1,3} \;.
\end{equation}
It has a projective spinor bundle $\sbcm\simeq \C^{1,3}\times\Pbb^1$ of dimension $\dim_{\C}(\sbcm)=5$ and sits in a double fibration
\begin{equation}\label{fig:double fibration in the case of complexified Minkowski space}
\begin{tikzcd}
    & \arrow[ld,"\pi_{\mcal{Z}_{\C^{1,3}}}" swap] \sbcm\arrow[rd,"\pi_{\C^{1,3}}"] &
    \\
    \mcal{Z}_{\C^{1,3}} & & \C^{1,3}
\end{tikzcd},
\end{equation}
where $\mcal{Z}_{\C^{1,3}}\simeq\Pbb^3\backslash\Pbb^1$ is the twistor space of $\C^{1,3}$. $\pi_{\C^{1,3}}$ is simply projection to the spacetime points while $\pi_{\mcal{Z}_{\C^{1,3}}}$ imposes the incidence relation (see below \eqref{eq:2-spinor notation of spacetime events}).

\paragraph{Inclusion of $\sbm$ in $\sbcm$.} The most obvious real slice is that of $\R^{1,3}$, which can be included in $\C^{1,3}$ as
\begin{equation}\label{eq:embedding of R13 in C13}
    \iota_{\R^{1,3}}(x_0,x_1,x_2,x_3)=(x_0,x_1,x_2,x_3) \;,
\end{equation}
while $\R^4$ can be included as
\begin{equation}\label{eq:embedding of R4 in C13}
    \iota_{\R^4}(x_0,x_1,x_2,x_3)=(\mfk{i}x_0,x_1,x_2,x_3) \;.
\end{equation}

\smallskip Consider the following diagram with $M=\R^{1,3}$ or $\R^4$
\begin{equation}
\begin{tikzcd}[column sep=1.5cm, row sep=1cm]
    i_M^*\sbcm \arrow[d, "\pi_M" swap] \arrow[r] &\sbcm\arrow[d,"\pi_{\C^{1,3}}"] 
    \\
    M \arrow[r, "i_M" swap] & \C^{1,3}
\end{tikzcd},
\end{equation}
The fiber of $i_M^*\sbcm$ at a point $x\in M$ is simply the fiber over its image in $\C^{1,3}$. From what we explained below \eqref{eq:2-spinor notation of spacetime events}, we see that this fiber is just a projective spinor. Hence, $i_M^*\sbcm$ can be identified with $\sbm$. We can thus include $\sbm\hookrightarrow\sbcm$.  

\subsubsection{6d hCS Theory from 10d hCS Theory}\label{sec:6d hCS on PS+(M) from 10d hCS on PS+C13}

We would like to identify a 10d theory on $\sbcm$ which reduces to the 6d hCS theory on $\sbm$ upon dimensional reduction. The basic idea is again clear from \cite[pg.\ 34, Lemma 7.1.1]{CostelloLi201606} where the following is proven: The dimensional reduction of the hCS on $\C^5$ to any $\C^k\subset \C^5$ gives the hCS theory on $\C^k$, where the starting theory on $\C^5$ had been worked out by Baulieu \cite{Baulieu201009}. The BV action of the hCS on $\C^k$ is given by
\begin{equation}\label{eq:action of 10 hCS theory}
    S=\frac{1}{2\pi}\bigintsss_{\C^{k|5-k}}\left(\prod_{i=1}^{5-k}{\rd\epsilon_i}\wedge\prod_{j=1}^k\rd w_j\right)\wedge\CS(\mcal{A}) \;,
\end{equation}
with $\CS(\mcal{A})$ is the same as \eqref{eq:CS term in the case of 6d hCS theory}. Here, ${w_1,\ldots,w_k}$ are coordinates along $\C^k$, $\epsilon_i$s are odd variables keeping track of the cohomological degree of the fields in the BV quantization, and $\mcal{A}\in\Omega^{0,*}(\C^{k|5-k})\otimes\mfk{gl}(n,\C)[1]$, where $[1]$ means the shift in the cohomological degrees. One substitute $\mcal{A}$, expanded in polynomials of $\epsilon_i$, in \eqref{eq:action of 10 hCS theory} and picks up only the term in $\epsilon_1\ldots\epsilon_{5-k}\Omega^{0,k}(\C^k)$. 

\smallskip However, we cannot directly use this result since $\sbcm$ is not, strictly speaking, a copy of $\C^5$. Instead, we prove the following 

\begin{lem}\normalfont\label{lem:6d hCS from 10d hCS}
    Consider the hCS theory with gauge Lie algebra $\mfk{g}$ on $X$ with $\dim_{\C}X=5$, and let $Y$ be a holomorphically-embedded submanifold of $X$ of $\dim_{\C}Y=k$. Then, the dimensional reduction of the hCS theory from $X$ to $Y$ gives the hCS theory on $Y$. 
\end{lem}

\begin{proof}
    The basic idea is the same as loc.\ cit. We consider the dimensional reduction of the field content of the theory. The BV field content of hCS on $X$ is given by
    \begin{equation}
        \mcal{A}^{\text{BV}}_X\in \Omega^{(0,*)}(X)\otimes\mfk{g}[1] \;.
    \end{equation}
    Consider the holomorphic embedding $\iota:Y\to X$. If we equip $X$ with a metric $g_X$ compatible with its complex structure $\mcal{J}_X$, we would have
    \begin{equation}\label{eq:pullback of anti-holomorphic cotangent bundle of ambient space}
        \iota^* T^{*(0,1)}X=T^{*(0,1)}Y\oplus {N^{\vee (0,1)}_{Y/X}} \;,
    \end{equation}
    where ${N^{\vee(0,1)}_{Y/X}}$ is the anti-holomorphic conormal bundle of $Y$ in $X$. This means $\bigwedge\nolimits^*(\iota^* T^{*(0,1)}X)$ is the bundle
    \begin{equation}\label{eq:bundle on the ambient space in which fields take value}
        \left(\bigoplus_{n=0}^5\bigoplus_{p=0}^k\bigoplus_{n-p=0}^{5-k}\bigwedge\nolimits^p T^{*(0,1)}Y\bigwedge\nolimits^{n-p} N^\vee_{Y/X}\right)\otimes\mfk{g}[1] \;.
    \end{equation}
    Therefore, by identifying the anti-holomorphic conormal directions as $\rd\bmb{v}_i\sim\epsilon_i,\,i=1,\ldots,5-k$, where $\bmb{v}_i$s are those of $\bmb{w}_i$s that are normal to $Y$. Then, \eqref{eq:bundle on the ambient space in which fields take value} gives a general field on $Y$ as a section of the bundle
    \begin{equation}
        \left(\bigoplus_{n=0}^5\bigoplus_{p=0}^k\bigoplus_{n-p=0}^{5-k}\bigwedge\nolimits^p T^{*(0,1)}Y\bigwedge\nolimits^{n-p} \C^{0|5-k}\right)\otimes\mfk{g}[1] \;,
    \end{equation}
    where $\C^{0|5-k}$ is a trivial vector bundle over $X$ with purely $(5-k)$-dimensional odd fibers, generated by odd variables $\epsilon_1,\ldots,\epsilon_{5-k}$. The space of sections of this bundle can thus be identified with $\Omega^{0,*}(Y^{k|5-k})\otimes\mfk{g}[1]$, where $Y^{k|5-k}$ is the split supermanifold $Y\times \C^{0|5-k}$. The latter is the field content of hCS theory on $Y$ with the action
    \begin{equation}\label{eq:action of hCS on an embedded submanfold of the ambient space}
        S_Y=\frac{1}{2\pi}\bigintsss_{Y^{k|5-k}}\left(\prod_{i=1}^{5-k}\rd\epsilon_i\wedge\Omega_Y\right)\wedge\CS(\mcal{A}^{\text{BV}}_Y) \;,
    \end{equation}
    where $\Omega_Y$ is a section of the canonical bundle of $Y$, $\CS(\mcal{A}^{\text{BV}}_Y)$ is given by \eqref{eq:CS term in the case of 6d hCS theory}, and $\mcal{A}^{\text{BV}}_Y$ can be expressed in terms of $\mcal{A}^{\text{BV}}_X$. Recall that $\mcal{A}^{\text{BV}}_X$ is the following formal sum
    \begin{equation}\label{eq:field of 10d hCS on X as a formal sum of various fields}
        \mcal{A}^{\text{BV}}_X=\sum_{p=0}^5\sum_{i_1,\ldots,i_p=1}^5\frac{1}{p!}\mcal{A}^{\text{BV,(p)}}_{X,\bmb{w}_{i_1}\ldots\bmb{w}_{i_p}}\rd\bmb{w}_{i_1}\ldots\rd\bmb{w}_{i_p} \;,
    \end{equation}
    then $\mcal{A}^{\text{BV}}_Y$ is determined by the identification $\rd\bmb{v}_i\simeq \epsilon_i$, where $(v_1,v_2)$ are normal directions to $Y$. For example, the one-form part of $\mcal{A}^{\text{BV}}_Y$ is given by
    \begin{equation}\label{eq:one-form part of the BV field on Y in terms of fields on X}
        \mcal{A}^{(1),\text{BV}}_Y=\mcal{A}^{(1),\text{BV}}_{Y,\bmb{w}_i}\rd\bmb{w}_i=\left(\mcal{A}^{(3),\text{BV}}_{X,\bmb{w}_i\bmb{v}_1\bmb{v}_2}\epsilon_1\epsilon_2+\sum_{j=1}^2\mcal{A}^{(2),\text{BV}}_{X,\bmb{w}_i\bmb{v}_j}\epsilon_j+\mcal{A}^{(1),\text{BV}}_{X,\bmb{w}_i}\right)\rd\bmb{w}_i \;,
    \end{equation}
    where $\mcal{A}^{(1),\text{BV}}_Y$ denotes the one-form part of the BV fields on $Y$. This completes the proof.
\end{proof}

\smallskip The case we are interested in is the following: 
\begin{equation}\label{eq:space X, Y and value k of interest}
    X=\sbcm \;, \qquad Y=\sbm \;, \qquad M=\R^{1,3},\R^4 \;, \qquad k=3 \;.
\end{equation}

Note that \eqref{eq:pullback of anti-holomorphic cotangent bundle of ambient space} is true if and only if $Y$ is a holomorphically-embedded submanifold of $X$. Hence, we first need to show that $\sbm$ can be holomorphically embedded in $\sbcm$. The following proposition is the proof of this fact

\begin{prop}\label{prop:PS+(M) is holomorphically embedded in PS+C13}\normalfont
    Let $\sbm$ and $\sbcm$ be the projective spinor bundles of $M=\R^{1,3}$ or $\R^4$, and $\C^{1,3}$, respectively. Then, $\sbm$ can be holomorphically embedded in $\sbcm$ equipped with a fixed complex structure.
\end{prop}

\begin{proof} Consider the inclusion map $\mcal{I}_M:\sbm\hookrightarrow\sbcm$. By what we have explained around \eqref{eq:2-spinor notation of spacetime events}, this is defined by
    \begin{equation}\label{eq:inclusion of PS+(M) in PS+ C13}
        \mcal{I}_M([x^{\alpha\dot{\alpha}},\lambda_\beta]):=[(\iota_M(x^{\alpha\dot{\alpha}}),\lambda_\beta] \;,
    \end{equation}
    where $\iota_M$ is given by \eqref{eq:embedding of R13 in C13} or \eqref{eq:embedding of R4 in C13}. Note that
    \begin{equation}
        \iota_M[x^{\alpha\dot{\alpha}}]=
        \begin{pmatrix}
            \delta_M x^0+x^3 & x^1-\mfk{i}x^2
            \\
            x^1+\mfk{i}x^2 & \delta_M x^0-x^3
        \end{pmatrix} \;,
    \end{equation}
    with
    \begin{equation}
        \delta_{M}=
        \left\{
        \begin{aligned}
        &1 \;, &\qquad M&=\R^{1,3} \;,
        \\
        &\mfk{i} \;, &\qquad M&=\R^{4}  \;.
        \end{aligned}
        \right.
    \end{equation}
    To prove the result, we proceed in two steps:
    \begin{enumerate}
        \item [(1)] We first show that the inclusion $\mcal{I}_M$ realizes $\sbm$ as a smooth embedded submanifold of $\sbcm$;

        \item [(2)] We then show that there is a complex structure on $\sbm$ in which the inclusion $\mcal{I}_M$ is holomorphic. 
    \end{enumerate}

\smallskip To prove (1), we need to show that $\mcal{I}_M$ is a smooth immersion as well as a topological embedding. The tangent space of $\sbm$ at the point $[(x^{\alpha\dot{\alpha}},\psi_\beta)]$ is generated by $\{\partial_{x^{\alpha\dot{\alpha}}},\partial_{\psi_\beta}\}$. Under $\mcal{I}_M$, the basis is mapped to the basis $\{\partial_{y^{\alpha\dot{\alpha}}},\partial_{\lambda_\beta}\}$ over the point $[([y^{\alpha\dot{\alpha}}],\chi_\beta)]=\mcal{I}_M[([x^{\alpha\dot{\alpha}}]),\psi_\beta)]\in\sbcm$ with
\begin{eqaligned}\label{eq:coordinates y on PS+C13 in terms of coordinates x on PS+M}
    y^{1\dot{1}}&=\delta_M x^0+x^3 \;, &\qquad y^{1\dot{2}}&=x^1-\mfk{i}x^2 \;,
    \\
    y^{2\dot{1}}&=x^1+\mfk{i}x^2 \;, &\qquad  y^{2\dot{2}}&=\delta_M x^0-x^3 \;,
\end{eqaligned}
and $\chi_\beta=\psi_\beta$. Then, the differential $\rd\mcal{I}_M$ maps $\partial_{y^{\alpha\dot{\alpha}}}=\delta_M^{-1}\partial_{x^{\alpha\dot{\alpha}}}$ for $(\alpha,\dot{\alpha})=(1,\dot{1}),(2,\dot{2})$ while $\partial_{y^{\alpha\dot{\alpha}}}=\partial_{x^{\alpha\dot{\alpha}}}$ for $(\alpha,\dot{\alpha})=(1,\dot{2}),(2,\dot{1})$. This is clearly an injective map. As we move the point $[([x^{\alpha\dot{\alpha}}],\psi_\beta)]$ around, the map changes smoothly. Hence, $\mcal{I}_M$ is indeed a smooth immersion. We next need to show that it is a topological embedding, i.e.\ it is a homeomorphism into its image. \eqref{eq:inclusion of PS+(M) in PS+ C13} shows that $\mcal{I}_M$ is smooth, injective, and surjective into its image with an obvious continuous inverse. Hence, $\mcal{I}$ is a topological embedding. We thus conclude that $\sbm$ is a smooth embedded submanifold of $\sbcm$.

\smallskip To prove (2), we would be constructive and explicitly construct complex structures on $\sbm$ and $\sbcm$. Note from  \eqref{eq:coordinates y on PS+C13 in terms of coordinates x on PS+M} that
\begin{eqaligned}
    y^{1\dot{1}}&=\frac{1}{2}\left(1+\delta_M\right)x^{1\dot{1}}-\frac{1}{2}\left(1-\delta_M\right)x^{2\dot{2}} \;,
    &\qquad y^{1\dot{2}}=x^{1\dot{2}} \;,
    \\
    y^{2\dot{2}}&=\frac{1}{2}\left(-1+\delta_M\right)x^{1\dot{1}}+\frac{1}{2}\left(1+\delta_M\right)x^{2\dot{2}} \;,
    &\qquad y^{2\dot{1}}=x^{2\dot{1}} \;.
\end{eqaligned}

For $M=\R^{1,3}$, we define a complex structure $\mcal{J}_{\C^{1,3}}$ on $\sbcm$ in which the following coordinates are holomorphic:
\begin{equation}\label{eq:master complex structure on PS+C13}
   \mcal{J}_{\C^{1,3}}: \qquad \left\{y^\pm_1:=y^{1\dot{1}}\pm\mfk{i}y^{1\dot{2}},y^\pm_2:=y^{2\dot{2}}\pm\mfk{i}y^{2\dot{1}};\chi_\beta\right\} \;.
\end{equation}
Note that $y^{\alpha\dot{\alpha}}$s are all complex coordinates and hence $y^+_i$ and $y^-_i$ are not complex conjugates of each other. We also define a complex structure $\mcal{J}_{\R^{1,3}}$ on $\sbrm$ in which the following coordinates are holomorphic
\begin{equation}\label{eq:complex structure on R13}
    \mcal{J}_{\R^{1,3}}: \qquad 
    \left\{x^+_L:=x^{1\dot{1}}+\mfk{i}x^{1\dot{2}},x^-_L:=x^{2\dot{2}}-\mfk{i}x^{2\dot{1}},\psi_\beta\right\} \;.
\end{equation}
Then 
\begin{equation}
    (y^+_1,y^-_1,y^+_2,y^-_2;\chi_\beta)=\mcal{I}_{\R^{1,3}}(x^+_L,x^-_L,\psi_\beta)=(x^+_L,0,0,x^-_L;\psi_\beta) \;.
\end{equation}
It is evident that the inclusion $\mcal{I}_M$ is holomorphic in these complex structures. 

\smallskip On the other hand for $M=\R^4$, it follows from \eqref{eq:coordinates y on PS+C13 in terms of coordinates x on PS+M} that
\begin{equation}
    y^{1\dot{1}}=\frac{1+\mfk{i}}{2}\left(x^{1\dot{1}}+\mfk{i}x^{2\dot{2}}\right) \;, \qquad y^{2\dot{2}}=\frac{-1+\mfk{i}}{2}\left(x^{1\dot{1}}-\mfk{i}x^{2\dot{2}}\right) \;.
\end{equation}
We define a complex structure $\mcal{J}_{\R^4}$ on $\sbe$ in which the following coordintes are holomorphic
\begin{equation}\label{eq:complex structure on R4}
    \mcal{J}_{\R^4}: \qquad 
    \left\{x^+_E:=x^{1\dot{1}}+\mfk{i}x^{2\dot{2}}+\frac{2\mfk{i}}{1+\mfk{i}}x^{1\dot{2}},x^-_E:=-x^{1\dot{1}}+\mfk{i}x^{2\dot{2}}+\frac{2\mfk{i}}{1-\mfk{i}}x^{2\dot{1}},\psi_\beta\right\} \;.
\end{equation}
Then in the complex structure on $\sbcm$ in which the holomorphic coordinates are \eqref{eq:master complex structure on PS+C13}, we have
\begin{equation}
    (y^+_1,y^-_1,y^+_2,y^-_2,\chi_\beta)=\mcal{I}_{\R^4}(x^+_E,x^-_E,\chi_\beta)=\left(\frac{1+\mfk{i}}{2}x^+_E,0,0,\frac{1-\mfk{i}}{2}x^-_E,\psi_\beta\right) \;,
\end{equation}
It is again evident that $\mcal{I}_{\R^4}$ is holomorphic in these complex structures. 

\smallskip We have thus shown that $\sbm$ with the complex structure $\mcal{J}_M$ given in \eqref{eq:complex structure on R13} or \eqref{eq:complex structure on R4} can be holomorphically embedded in $\sbcm$ with the complex structure $\mcal{J}_{\C^{1,3}}$ given in \eqref{eq:master complex structure on PS+C13}. Hence $\mcal{I}_M$ is a holomorphic embedding for both $M=\R^{1,3}$ and $M=\R^4$. This concludes the proof.

\end{proof}

\begin{rmk}[The Relevance of Open-String Field Theory]\normalfont \label{rmk:relevance of open-string field theory}
    A B-brane in the B-model topological string theory is supported on a holomorphic submanifold $Y$ of the target $X$, which is a CY manifold and is represented by the complex of coherent sheaves on $Y$ modulo quasi-isomorphisms. The space of open-string states stretched between branes represented by the complexes $\mcal{B}^\bullet_1$ and $\mcal{B}_2^{\bullet}$, which determines the field content of the theory living on $Y$, is given by $\text{Hom}(\mcal{B}^\bullet_1,\mcal{B}_2^{\bullet})$. For the simplest case of $n$ copies of the structure sheaf of $\C^k,k=1,\ldots,5$, it is shown in \cite[\S 7.1]{CostelloLi201606} that the field content of the theory matches with the field content of the hCS on $\C^k$. This is the basic idea of the relation between the holomorphic twist of the theory living on the brane and the hCS theory. In our case, $X=\sbcm=\C^{1,3}\times\Pbb^1$ and $Y=\sbm$ with $M=\R^{1,3}$ or $\R^4$. As such, $\sbcm$ is not a CY manifold, and hence there cannot be a topological B-model onto it. However, the hCS theory on a complex five-dimensional manifold such as $\sbcm$ still makes sense, at least classically, which is the point we are using here. Whether the hCS theory on $\sbcm$ makes sense quantum-mechanically is not what we would consider in this work. Therefore, the relation of our construction to open-string field theory is not immediately clear. However, recall that the condition of being CY for the target comes from the requirement of conformal invariance of the sigma models into these targets. In the full formulation of string theory, it is expected that the theory can be defined on general backgrounds, including nonconformal ones \cite{Zwiebach199606}. In that context, conformal backgrounds are only classical solutions to the string-field-theory equations of motion, and it might be possible to define a topological version of the theory on non-CY targets.  \qed 
\end{rmk}

With these results at hand, we are now ready to track the origin of the correspondence in ten dimensions. 

\subsubsection{Chasing the Correspondence to Ten Dimensions}\label{sec:chasing the correspondence to ten dimensions}

We first would like to see explicitly that the dimensional reduction of the 10d hCS theory on $\sbcm$ to $\sbm$ gives the 6d hCS theory on $\sbm$. This is easy to establish. We only need to substitute \eqref{eq:one-form part of the BV field on Y in terms of fields on X} into \eqref{eq:action of hCS on an embedded submanfold of the ambient space}, by which we can explicitly see that, after integration over $\C^{0|2}$ in the split supermanifold $Y^{3|2}=Y\times\C^{0|2}$, the BV action on $Y$ is given by 
\begin{equation}\label{eq:action of 6d CS on a holomorphically-embedded submanifold Y}
    S_Y=\frac{1}{2\pi}\bigintsss_{Y}\Omega_Y\wedge\CS(\mcal{A}_Y^{(1)})+\ldots \;,
\end{equation}
where from \eqref{eq:one-form part of the BV field on Y in terms of fields on X}, we have the ordinary gauge field $\mcal{A}_Y^{(1)}$ expressed in terms of 10d fields as
\begin{equation}\label{eq:ordinary gauge field on Y in terms of fields on X}
    \mcal{A}_Y^{(1)}=\sum_{i=1}^3\mcal{A}_{Y,\bmb{w}_i}\rd\bmb{w}_i=\sum_{i=1}^3\left(\mcal{A}_{X,\bmb{w}_i}^{(1),\text{BV}}+\sum_{j=1}^2\mcal{A}_{X,\bmb{w}_i\bmb{v}_j}^{(2),\text{BV}}+\mcal{A}_{X,\bmb{w}_i\bmb{v}_1\bmb{v}_2}^{(3),\text{BV}}\right)\rd\bmb{w}_i \;.
\end{equation}
Note that here $\bmb{v}_i$ are not differential-form indices as those directions lie normal to $Y$. Furthermore, $\Omega_Y$ can be expressed in terms of $\Omega_X$ by first picking a metric $g_X$ compatible with the complex structure, using which we would have the decomposition \eqref{eq:pullback of anti-holomorphic cotangent bundle of ambient space}. We then choose two orthonormal sections $\mbs{n}_1$ and $\mbs{n}_2$ of $N^{(1,0)}_{Y/X}$, the holomorphic normal bundle of $Y$ in $X$, and define
\begin{equation}\label{eq:section of canonical bundle of Y in terms of that of X}
    \Omega_Y:=\iota_{\mbs{n}_1}\iota_{\mbs{n}_2}\Omega_X \;,
\end{equation}
where $\iota_{\mbs{n}_i}$ denotes the contraction with $\mbs{n}_i$. The construction does not depend on the metric $g_X$ and any metric compatible with the complex structure $\mcal{J}_X$ would do the job.  The first term of \eqref{eq:action of 6d CS on a holomorphically-embedded submanifold Y} coincides with the action \eqref{eq:6d hCS action on projective spinor bundle} of 6d hCS theory on the projective spinor bundle of $M$ provided $Y=\sbm$, and $\ldots$ denotes the rest of the BV action. 

\smallskip We can thus realize the origin of the correspondence in ten dimensions as follows: One starts from the following BV action of 10 hCS theory
    \begin{equation}\label{eq:BV action of 10d hCS theory}
        S_{X}=\frac{1}{2\pi}\bigintsss_{X}\Omega_{X}\wedge\CS(\mcal{A}_{X}) \;,
        \qquad X=\sbcm \;,
    \end{equation}
    where $\Omega_X$ is a section of the canonical bundle\footnote{Recall that holomorphic volume forms exist only for CY manifolds.} of $\sbcm$ of degree $(5,0)$, and the field $\mcal{A}_X\in\Omega^{(0,*)}(\sbcm)\otimes\mfk{sl}(n,\C)[1]$, hence it is given by the formal sum \eqref{eq:field of 10d hCS on X as a formal sum of various fields}. We equip $\sbcm$ with the complex structure $\mcal{J}_{\C^{1,3}}$ given in \eqref{eq:master complex structure on PS+C13}. One can then get to the two sides of the correspondence as follows:

\paragraph{The Hyperbolic $\SU{n}$-Monopole Side.} We equip $\sbrm$ with the complex structure $\mcal{J}_{\R^{1,3}}$ defined in \eqref{eq:complex structure on R13}, and holomorphcally-embed it inside $\sbcm$ by the results of Proposition \ref{prop:PS+(M) is holomorphically embedded in PS+C13}. We then dimensionally-reduce \eqref{eq:BV action of 10d hCS theory} to $\sbrm$, which, by what we have explained around \eqref{eq:action of 6d CS on a holomorphically-embedded submanifold Y} gives the following BV action of the 6d hCS theory
\begin{equation}\label{eq:BV action on PS+R13}
    S_{Y}=\frac{1}{2\pi}\bigintsss_{Y}\Omega_{Y}\wedge\CS(\mcal{A}^{(1)}_{Y})+\ldots \;,
    \qquad Y=\sbrm \;,
\end{equation}
where $\Omega_{Y}$ is a (not necessarily global) section of the canonical bundle of $\sbrm$, that can be constructed as in \eqref{eq:section of canonical bundle of Y in terms of that of X}, $\mcal{A}^{(1)}_{Y}$ is the gauge field of the 6d hCS theory and can be written explicitly in terms of 10d fields as \eqref{eq:ordinary gauge field on Y in terms of fields on X}, and $\ldots$ denotes the rest of the BV action. Once we are in six dimensions, one can recover the data of hyperbolic $\SU{n}$-monopoles, including its masses, charges, and spectral data, by what we have explained in \S\ref{sec:masses, charges, and spectral data from six dimensions}. In particular, we have to impose a reality condition, which reduces the gauge Lie algebra from $\mfk{sl}(n,\C)$ to $\mfk{su}(n)$, as we explained in Remark \ref{rmk:bundle of self-dual two-forms on the Minkowski space}. 

\paragraph{The gCPM Side.} We proceed similarly. We equip $\sbe$ with the complex structure $\mcal{J}_{\R^4}$ defined in \eqref{eq:complex structure on R4}, and embed it holomorphically inside $\sbcm$ by the results of Proposition \ref{prop:PS+(M) is holomorphically embedded in PS+C13}. We then dimensionally-reduce \eqref{eq:BV action of 10d hCS theory} to $\sbe$, which, by what we have explained around \eqref{eq:action of 6d CS on a holomorphically-embedded submanifold Y} gives the following BV action of 6d hCS theory
\begin{equation}\label{eq:BV action on PS+R4}
    S_{Y}=\frac{1}{2\pi}\bigintsss_{Y}\Omega_{Y}\wedge\CS(\mcal{A}^{(1)}_{Y})+\ldots \;,
    \qquad Y=\sbe \;,
\end{equation}
where $\Omega_{Y}$ is a section of the canonical bundle of $\sbe$, that again can be constructed as in \eqref{eq:section of canonical bundle of Y in terms of that of X}, $\mcal{A}^{(1)}_{Y}$ is the gauge field of 6d hCS theory and can be written explicitly in terms of 10d fields as \eqref{eq:ordinary gauge field on Y in terms of fields on X}, and $\ldots$ denotes the rest of the BV action. Once we are in six dimensions, the gCPM can be recovered from what we have explained in \S\ref{sec:gCPM from six dimensions}. 

\smallskip As $\sbcm$ is equipped with a fixed complex structure  \eqref{eq:master complex structure on PS+C13}, and the physics of the 10d hCS theory only depends on the choice of complex structure, it follows that we can successfully recover the two sides of the hyperbolic monopoles/gCPM correspondence as two different manifestations of the same ten-dimensional physics of hCS theory on $\sbcm$. 

\smallskip In a naive sense, the correspondence can be summarized as performing a Wick rotation and imposing or lifting reality conditions on the data. In more detail, we start in 10d and dimensionally reduce to the 6d hCS on $\sbrm$, which realizes the hyperbolic monopole side after imposing reality conditions. We now Wick-rotate $\R^{1,3}\to\R^4$, relax the reality condition on the spectral data, and realize the gCPM side of the correspondence. Conversely, we can start in 10d, and dimensionally reduce the theory to the 6d hCS on $\sbe$, which then realizes the gCPM side of the correspondence. We then perform an inverse Wick rotation $\R^4\to\R^{1,3}$, impose reality condition on the data, and recover the hyperbolic monopole side of the correspondence. With this heuristic comment, we finalize our discussion.

\section{Discussion and Future Directions}
\label{sec:discussion and future direction}

In this work, we have explored the relationship between the spectral data of magnetic monopoles in hyperbolic space and the curve of the spectral parameter of the gCPM. We generalized the observation of Atiyah and Murray \cite{Atiyah199106,AtiyahMurray1995} for the group $\SU{2}$ to the group $\SU{n}$ in \S\ref{sec:generalized corresponence}. We then proposed a realization of the gCPM inside the 4d CS theory in \S\ref{sec:gCPM and 4d CS theory}, which explains its various features. Finally, we explored the origin of the correspondence in \S\ref{sec:on the origin of the correspondence}. This involved three steps: (1) we first showed in \S\ref{sec:hyperbolic monopoles from six dimensions} and \S\ref{sec:gCPM from six dimensions} that the two sides of the correspondence can be realized using the 6d hCS theory formulated on projective spinor bundle of the Minkowski space $\R^{1,3}$, in the case of hyperbolic monopoles, and the projective spinor bundle of $\R^4$, in the case of the gCPM; (2) We then explained that these projective spinor bundles can be holomorphically-embedded in the projective spinor bundle of the complexified Minkowski space $\C^{1,3}$, which is a complex five-dimensional manifold whose complex structure we fix. (3) We then explained that the 6d hCS on projective spinor bundles of $\R^{1,3}$ and $\R^4$ can be realized as the dimensional reduction of the 10d hCS formulated on the projective spinor bundle of $\C^{1,3}$. These are the content of \S\ref{sec:6d hCS on PS+(M) from 10d hCS on PS+C13}. Putting these three points together, we concluded in \S\ref{sec:chasing the correspondence to ten dimensions} that the hyperbolic monopole/gCPM correspondence is the incarnation of the single ten-dimensional physics. 

\smallskip Finally, we collect some of the most interesting puzzles for future investigation.

\paragraph{Explicit Computation of the R-Matrix of the gCPM.} Our discussion of the relationship between the gCPM and the 4d CS theory lacks the explicit computation of the R-matrix of the gCPM. Our discussion in \S\ref{sec:CPM from 4d CS theory} indicates that this would not be a perturbative computation. This observation is consistent with the fact that gCPM does not have a classical R-matrix, and hence the full quantum R-matrix cannot be realized as a formal expansion whose first non-trivial term is the classical R-matrix. Therefore, the explicit computation of R-matrix is an outstanding challenge on multiple levels:

\begin{itemize}
  \item[(1)]{\small\bf Non-Perturbative Definition of 4d CS Theory.} The first issue is to formulate CS theory non-perturbatively. This is more than just writing the path integral formally as 
  \begin{equation}\label{eq:formal definition of path integral of 4d CS theory}
      \langle\mbs{L}_h\mbs{L}_v\rangle\sim\bigintsss\mcal{D}A\,\,\mbs{L}_h\mbs{L}_v\exp\left(-\frac{1}{2\pi\hbar}\bigintsss_{C\times\Pbb^1}\omega_{\Pbb^1}\wedge\CS(A)\right) \;,
  \end{equation}
  for some horizontal and vertical line defects $\mbs{L}_h$ and $\mbs{L}_v$, respectively, and possibly sum over the saddle points, with $\omega_{\Pbb^1}$ as in \eqref{eq:one-form relevant for engineering gcpm on P1}. One possible route to such a non-perturbative formulation is the realization of the theory within string theory. The latter is well-known \cite{CostelloYagi201810,AshwinkumarTanZhao201806,IshtiaqueMoosavianRaghavendranYagi202110}. Following \cite{Witten201101}, one of the most serious attempts to provide a nonperturbative formulation of the 4d CS theory is \cite{AshwinkumarTanZhao201806}. However, the recipe given in this work \cite[eq.\ (3.14)]{AshwinkumarTanZhao201806} is formal and probably not useful for practical computations. One thus needs to provide a more practical nonperturbative recipe for the R-matrix computation. As is usual with any field theory, a proper nonperturbative formulation is a challenge; that 4d CS theory is a semi-holomorphic field theory might simplify the procedure but does not trivialize it.

  \smallskip Another route could be based on what we have explained in \S\ref{sec:gCPM from six dimensions}, where the 4d CS theory emerges as the dimensional reduction of the 6d hCS theory. Therefore, it is conceivable that the non-perturbative definition of the 4d CS theory would involve the 6d hCS theory, and eventually, based on what we have explained \S\ref{sec:6d hCS on PS+(M) from 10d hCS on PS+C13}, the 10d hCS theory.   

  \item[(2)] {\small\bf 4d CS Theory at the Roots of Unity.} It is known that the gCPMs are related to certain irreducible representations of the quantum group $\text{U}_q(\widehat{\mfk{sl}}(n,\mbb{C}))$ (or more precisely a trivial extension thereof)  at a root of unity, the so-called minimal cyclic representations \cite{DateJimboKeiMiwa199008,DateJimboKeiMiwa199103,DateJimboMikiMiwa199104}. These are families of $N^{n-1}$-dimensional representations parameterized by $z\in(\mbb{C}^\times)^{3n-1}$, where  $N\ge 3$ is an odd number determining an $N$\textsuperscript{th} root of unity $q$. This means that a more elaborate incorporation of the gCPM in the 4d CS theory would involve the formulation of the latter at the roots of unity. Since this is intimately related to the nonperturbative formulation of the theory, it is expected that the loop-counting parameter $\hbar$ in \eqref{eq:formal definition of path integral of 4d CS theory} will be related to parameters $n$ and $N$ through some relation. There are three different limits that might be helpful in figuring out the precise relationship between $(n,N)$ and $\hbar$: (1) $N\to\infty$ while keeping $n$ finite; (2) $N,n\to\infty$ while keeping $n/N$ finite; (3) $N,n\to\infty$ while $n/N\to 0$ \cite{AuYangPerk199305,AuYangPerk199906}. It was also realized in \cite{Tarasov199204,Tarasov199211} that the intertwiners of some other irreducible cyclic representations are related to the gCPM.  

  \item[(3)]  {\small\bf Line Defects Carrying Minimal Cyclic Representations.} The R-matrix $\mcal{R}(z,z')$ is the intertwiner of certain tensor product $\pi_{z,z'}$ of minimal cyclic representations \cite{DateJimboKeiMiwa199008,DateJimboKeiMiwa199103,DateJimboMikiMiwa199104}. The existence of an R-matrix forces the parameters $(z,z')$ to lie on a certain algebraic variety, i.e. the curve of the spectral parameter of the model. Therefore, based on what we explained in \S\ref{sec:correspondence for classical and exceptional groups}, one needs to consider two line defects carrying minimal cyclic representations of $\text{U}_q(\wh{\mfk{sl}}(n,\C))$ and perform the computation of the R-matrix. Hence, the construction of such line defects is a significant and essential step. It is expected that the representation theory of quantum groups $\text{U}_q(\widehat{\mfk{so}}(2n+1,\mbb{C}))$, $\text{U}_q(\widehat{\mfk{sp}}(2n,\mbb{C}))$, and $\text{U}_q(\widehat{\mfk{so}}(2n,\mbb{C}))$ and more generally $\text{U}_q(\widehat{\mfk{g}})$ should provide insight into the structure of the gCPM for these groups and may even lead to the realization of new integrable models with higher-genus curve of the spectral parameter. Certain irreducible representations of non-affine quantum groups have been considered in the literature \cite{DateJimboMikiMiwa199104,DeConciniKac1990,ChariPressley1991,Schnizer199201,Schnizer199305}. See also the discussion in \S\ref{sec:gauge Lie algebra, boundary conditions, R-matrix} for Bazhanov--Stroganov procedure, a generalization of which for gCPM should be straightforward.
\end{itemize}

It is clear that what we have explored in this work is just the tip of a huge iceberg. 

\paragraph{Integrable Spin Models with Higher-Genus Curve of Spectral Parameter.} From the basic philosophy of the 4d CS theory, the existence of an integrable lattice model is tied to the preservation of topological symmetry along the topological plane. Based on this idea, we have argued in \S\ref{sec:generalization to gCPM} that together with the requirement of $\mbb{Z}_N^{n-1}$-invariance, this would almost uniquely fix the form of the one-form of the 4d CS theory. Let us recall the most general form of one-form, other than the rational one, which does not have a zero
\begin{equation}
    \omega_{\Pbb^1}=\frac{1}{\prod_{i=1}^{\#}(z^{r_i}-z_i)^{1/s_i}}\rd z \;,
\end{equation}
with both $r_i$ and $s_i$ are positive integers. If we perform $z\to 1/z$ transformation, the requirement of not having a zero gives
\begin{equation}\label{eq:requirement of not having a zero for omegaP1}
    \sum_{i=1}^{\#}\frac{r_i}{s_i}-2=0 \;.
\end{equation}
This is a very stringent constraint. However, we cannot proceed further without further information. For example, one can ask about the existence of a model in which spins take value in $\mbb{Z}_{N_1}\times\ldots\times\mbb{Z}_{N_n}$ with $N_i\ne N_j$ for $i\ne j$. This would set $r_i=N_i,\,i=1,\ldots,n$, and \eqref{eq:requirement of not having a zero for omegaP1} would reduce to
\begin{equation}
    \sum_{i=1}^{n}\frac{N_{i}}{s_i}-2=0 \;.
\end{equation}
One of the minimal solutions is
\begin{equation}
    \#=2 \;,
    \qquad (r_1,r_2)=(N_1,N_2) \;,
    \qquad (s_1,s_2)=(N_1,N_2) \;.
\end{equation}
The gCPM is a particular instance of this solution for $r_i=s_i=N^{n-1},\,i=1,2$. This simple analysis suggests that there may be an integrable model in which the spins take value in $\mbb{Z}_{N_1}\times\mbb{Z}_{N_2}$ and the one-form $\omega_{\Pbb^1}$ would take the form
\begin{equation}
    \omega_{\Pbb^1}=\frac{1}{\sqrt[\leftroot{-2}\uproot{2} N_1]{\prod_{i=1}^{N_1}(z-z_i)}\sqrt[\leftroot{-2}\uproot{2} N_2]{\prod_{i=1}^{N_2}(z-z'_i)}}\rd z \;, \qquad i,j=1,2,\,i\ne j \;,
\end{equation}
for some $\{z_1,\ldots,z_{N_1}\}$ and $\{z'_1,\ldots,z'_{N_2}\}$. This still satisfies $\partial_{\bmb{z}}\omega_{\Pbb^1}=0$. It would be very interesting to explore this direction for two specific questions: (1) Would the requirement of not having a zero in the topological plane be enough to conclude that there is an integrable model whose associated one-form is determined by that requirement? (2) Would it be possible for $\omega_{\Pbb^1}$ to have zeroes with the possibility of breaking topological invariance but then it is restored at the quantum level by some mechanism?

\paragraph{Treating Hyperbolic Monopoles as Circle-Invariant Instantons.} In this work, we have used the work of Murray and Singer to work out the hyperbolic monopole/gCPM correspondence \cite{MurraySinger199607}. Using this work, it is at least meaningful to take the limit of vanishing boundary values of the Higgs field and define the spectral data. However, there are some partial results on the construction of spectral data of hyperbolic $\SU{n}$-monopole for integral boundary values of the Higgs field in the spirit of the original work of Atiyah \cite[\S 14]{Chan2017}. This work uses a different set of boundary conditions than Murray and Singer. It is unclear how one deals with vanishing boundary values of the Higgs field since there is no limit of this sort. It might be possible to set the boundary values to zero without taking a limit. It would be interesting to (1) complete the work of loc.\ cit.\ and provide a proper generalization of the work of Atiyah to the group $\SU{n}$; in particular, it would be interesting to show that the spectral data uniquely recovers the hyperbolic monopole solution, and (2) work out the correspondence with gCPM and compare the results with what has been explored in \S\ref{sec:generalized corresponence}. It might also be useful to provide a formal formulation of spectral data in the language of algebraic geometry along the lines of \cite{HurtubiseMurray1990}.

\paragraph{Possible Role of K\"ahler CS Theory.} We have seen in \S\ref{sec:instantons from 6d CS theory} that instantons can be realized from the 6d CS theory (see also \cite{Popov199803,Popov199806,Costello202004,BittlestonSkinner202011}). Although the case we discussed was related to instantons on the Minkowski space, it is desirable to have an action whose equations of motion describe (anti-)instantons on $\R^4$, as hyperbolic monopoles with integer boundary values of Higgs field have been constructed as circle-invariant instantons on $\R^4$ \cite{Atiyah1984,Chakrabarti198601,Nash198608}. Such an action does actually exists and is dubbed K\"ahler CS theory \cite{NairSchiff199008,NairSchiff199203}. Consider a five-dimensional product manifold of the form $M\times \R$, where $(M,\omega)$ is a K\"ahler manifold equipped with the K\"ahler structure $\omega$. Then, the action of K\"ahler CS theory is given by
\begin{equation}
    S=\frac{k}{2\pi}\bigintsss_{M\times \R}\omega\wedge\CS(\mcal{A})+\bigintsss_{M\times \R}\text{Tr}_{\mfk{g}}\left((\Phi+\overline{\Phi})\mcal{F}\right) \;,
\end{equation}
where $k$ denotes the coupling of the theory, $\mcal{A}$ is the gauge field on $M\times \R$, and $\Phi$ is a scalar field. If we denote by $F$ the field strength on $M$, the relevant equations of motion of the theory are
\begin{equation}
    F^{(2,0)}=0=F^{(0,2)} \;,
    \qquad \omega\wedge F=0 \;.
\end{equation}
It turns out that $\Phi$ does not affect the analysis, and the space of solutions of these equations up to gauge transformations turns out to be the moduli space of anti-instantons on $M$. Therefore, one can get hyperbolic monopoles with integer boundary values of the Higgs field by circle-reduction of these equations as in \cite{Atiyah1984}. However, it is not quite clear how to connect this theory to the gCPM. One possibility is that one can connect this theory to the 6d hCS theory on $\sbe$ and then to the gCPM following what has been explained in \S\ref{sec:gCPM from six dimensions}. It would be interesting to explore this direction further. 

\paragraph{Possible Roles of Supersymmetric Gauge Theories, and String Field Theory.}  It is known that supersymmetric gauge theories can be formulated in twistor spaces \cite{BoelsMasonSkinner200604}. On the one hand, the holomorphic twist of the ten-dimensional maximally supersymmetric gauge theory on $\C^5$ is the 10d hCS theory \cite{Baulieu201009}. To connect these ideas to this work, it would be interesting to see whether holomorphic twists of supersymmetric gauge theories on $\sbm$ for $M=\R^{1,3},\R^4$ and $\sbcm$ are related to the hCS theories, in particular in six and ten dimensions. This may be a route to connect our results to string theory.

\smallskip  On the other hand, the action of the 6d CS theory describes the open sector of string field theory (i.e.\ spacetime) realization of the B-model topological string theory \cite{Witten199207}. As we already mentioned in Remark \ref{rmk:relevance of open-string field theory}, $\sbcm$, $\sbrm$, and $\sbe$ are not CY manifolds: As smooth Riemannian manifolds, all of these are of the form $\R^{\#}\times S^2$ for some $\#$. The Ricci curvature $\text{R}$ of such manifolds is $\text{R}(\R^{\#}\times  S^2)=\text{R}(\R^{\#})+\text{R}(S^2)={1}/{2}$. Therefore, it is not immediately clear how the construction of this work is related to string field theory. As we previously mentioned in Remark \ref{rmk:relevance of open-string field theory}, such manifolds are expected to appear as backgrounds in the full formulation of the theory \cite{Zwiebach199606}. It would be highly desirable but extremely challenging to understand the set of possible backgrounds of string field theory. This work provides further evidence of the importance of this mostly-neglected question at the foundation of string theory. We leave this highly non-trivial follow-up of this work to ambitious readers. Another possible alternative route to string theory is detailed below.

\paragraph{Holomorphic Chern--Simons Theories in Complex $(2k+1)$-Dimensions.} We have discussed in \S\ref{sec:on the origin of the correspondence} how the action of the hCS theory on $\sbcm$ reduces to the action of the 6d CS theory on $\sbm$ with $M=\R^{1,3}$ or $\R^4$. Instead of the action of the BV action of the hCS theory, it is tempting to propose the following action in complex dimension $2k+1$
\begin{equation}\label{eq:action of hCS theory in 2k+1 complex dimension}
    S_{2k+1}=\frac{1}{2\pi}\bigintsss_{X}\Omega_{2k+1}\wedge\text{CS}_{2k+1}(\mcal{A}) \;,
    \qquad k=1,2,\ldots \;,
\end{equation}
where $X$ is a $(2k+1)$-dimensional complex manifold with the holomorphic volume form $\Omega_{2k+1}$ and $\text{CS}_{2k+1}(\mcal{A})$ is the Chern--Simons $(2k+1)$-form depending on a connection $\mcal{A}=\sum_{i=1}^{2k+1}\mcal{A}_{\bmb{z}_i}\rd\bmb{z}_i$. In the simplest situation of $k=1$, this is just the action of the 6d CS theory which we have visited many times in this work. The next simplest case with $k=2$ gives a 10d theory. We can perform a dimensional reduction to six dimensions following the procedure explained in \S\ref{sec:4d hBF theory from 6d hCS theory}: One needs to choose a metric compatible with the complex structure of $\sbcm$, choose a suitable metric along the two complex directions, say $v_{i},\,i=1,2$, along which once to reduce the theory, assume the fields and $\Omega_{2k+1}$ are independent of $v_1$ and $v_2$, and finally integrate over $v_1$ and $v_2$. Taking $X=Y\times T^4$, with $Y$ being complex three-dimensional and $T^4$ denotes a four-torus, the resulting action is 
\begin{eqaligned}
    \wt{S}_3=\frac{1}{2\pi}\cdot\text{Vol}(T^4)\bigintsss_Y\Omega_3\wedge \text{Tr}_{\mfk{g}}\left([\mcal{A}_{\bmb{v}_1},\mcal{A}_{\bmb{v}_2}]\text{CS}_3(\mcal{A}_3)\right)+\ldots \;,
\end{eqaligned}
where $\mcal{A}_3=\sum_{i=1}^3\mcal{A}_{\bmb{w}_i}\rd\bmb{w}_i$ and $\ldots$ denotes other terms coming from the dimensional reduction. If we take $\mcal{A}_{\bmb{v}_i}$ to be constants, then the first term almost resembles the action of the 6d hCS theory. To the best of our knowledge, the possible role of the hCS theories described by the actions \eqref{eq:action of hCS theory in 2k+1 complex dimension} in integrability, gauge theory, string theory (for the particular case of $k=2$) have not been explored in the literature. It would be interesting to understand more about these theories.

\paragraph{Non-Maximal Symmetry Breaking.} In this work, we have considered the maximal symmetry-breaking pattern of hyperbolic monopole solutions. To the best of our knowledge, the explicit proof that hyperbolic monopoles with arbitrary symmetry-breaking patterns are integrable and can be determined by their spectral data has not appeared in the literature. However, hyperbolic $\SU{n}$-monopoles are expected to be integrable for any values and numbers of monopole charges $(m_1,\ldots,m_k)\in\mbb{Z}_{\ge}^{k}$\footnote{$\mbb{Z}_{\ge}$ denotes the set of non-negative integers.} with $k\le n-1$. This is associated with the symmetry-breaking pattern $\SU{n}\to H\times\U{1}^{k}$ with $H\subset \SU{n}$ has rank $r-k$. Hence, the construction discussed in \S\ref{sec:gCPM from hyperbolic SU(n) monopoles} suggests that there may be a corresponding lattice integrable model. It is plausible that such symmetry-breaking patterns are associated with gCPMs with spin variables taking values in $\mbb{Z}_N^{k-2}$. It would be interesting to investigate whether one can construct new integrable models for such symmetry-breaking patterns of hyperbolic $\SU{n}$ monopoles.

\paragraph{Possible Generalization of the gCPM.} The hyperbolic monopole/gCPM correspondence established in \S\ref{sec:generalized corresponence} may pave the way for the construction of generalizations of the gCPM. As hyperbolic $\SU{n}$-monopoles are expected to be integrable for any values of the monopole charge $(N_1, \ldots, N_{n-1})\in\mbb{Z}^{n-1}$, the construction discussed in \S\ref{sec:gCPM from hyperbolic SU(n) monopoles} suggests that the corresponding lattice integrable system, if exists, is also integrable. The curve $\Sigma$ of the spectral parameter would be the curve defined by the following set of equations
    \begin{equation}
        \begin{pmatrix}
            (z_i^+)^{N_i}
            \\
            (z_i^-)^{N_i}
        \end{pmatrix}
        =K_{ij}
        \begin{pmatrix}
            (z_j^+)^{N_j}
            \\
            (z_j^-)^{N_j}
        \end{pmatrix} \;,
        \qquad i,j=1,\ldots,n \;,
    \end{equation}
    where the matrices $K_{ij}$ still satisfy the identity and cocycle conditions \eqref{eq:relations satisfied by matrices Kij of gCPM}, and would have genus (see \eqref{eq:genus of Zm1x...Zmn-1 cover of the curve in product of n copies of P1 with of arbitrary degree} or \eqref{eq:genus of a general complete intersection})
    \begin{equation}
        g_{\Sigma}=1+N^2_1\ldots N^2_{n-1}\left(\sum_{i=1}^{n-1}N_i-n\right) \;.
    \end{equation}
    This supposed-to-be-integrable lattice model would be a further generalization of the generalized chiral Potts model where the $n-1$ spin variables are taking values in $\mbb{Z}_{N_1}\times\cdots\times\mbb{Z}_{N_{n-1}}$. The existence of such an integrable model would be highly non-trivial. This is partially due to comments around \eqref{eq:requirement of not having a zero for omegaP1}.

\paragraph{Discrete Nahm Equations and String Theory.}  It is well-known \cite{Donaldson198409,HurtubiseMurray1989} that monopoles on $\R^3$ are in one-to-one correspondence with the solutions of Nahm's equations \cite{Nahm198109,Nahm198211}. On the other hand, Nahm's equations have a natural realization within string theory \cite{Diaconescu199608}. 

\smallskip There is an analog to Nahm's equations and the ADHM data for hyperbolic monopoles. For the gauge group $\SU{2}$, this was proposed by Braam and Austin \cite{BraamAustin199008} and their integrability is shown in \cite{MurraySinger199903}. Their generalization to the group $\SU{n}$ was discussed in \cite{Chan201506,Chan2017}. Therefore, there are two natural questions (1) can hyperbolic monopoles be constructed in string theory? (2) have discrete Nahm's equations a realization in string theory? The answers to these questions pave the way for a possible understanding of the hyperbolic monopoles/gCPM correspondence as a (sequence) of string dualities, along the same lines as the Bethe/Gauge correspondence \cite{CostelloYagi201810,IshtiaqueMoosavianRaghavendranYagi202110}. 

\paragraph{Hyperbolic Monopoles and Rational Maps into Flags.} It is well-known that the moduli space of monopoles on $\R^3$ can be realized as the space of rational maps into flag manifolds. This conjecture was first put forward by Atiyah \cite{Atiyah198412} and proved by Donaldson for the gauge group $\SU{2}$ \cite{Donaldson198409}. For the case of maximal symmetry-breaking with classical gauge groups, it is proven by Hurtubise and Murray \cite{Hurtubise198506,Hurtubise198912,HurtubiseMurray1989}. A further generalization for arbitrary compact semisimple gauge group and arbitrary symmetry-breaking pattern has been proved by Jarvis \cite{Jarvis199807,Jarvis200001}.

\smallskip There is a similar story for hyperbolic monopoles, conjectured by Atiyah \cite{Atiyah1984}. The case with $\SU{2}$ gauge group was established in \cite{JarvisNorbury199709}, while the proof is still lacking for other gauge groups. This would provide a connection between the space of such maps and the moduli space of hyperbolic monopoles, studied in \cite{Atiyah1984,Braam198902,Nash2006,Nash200810,Hitchin2008,BielawskiSchwachhofer201104,BielawskiSchwachhofer201201,FigueroaOFarrillGharamti201311,Gharamti2015,FranchettiRoss202302} (see also \cite{Sutcliffe202112} for the connection to the hyperbolic analog of Atiyah--Hitchin manifolds), along the lines of \cite{MannMilgram199307}. Furthermore, a natural question related to this work is which points/loci of the moduli space of hyperbolic monopoles or rational maps lead to a generalization of the CPM. 

\section*{Acknowledgments} 
We would like to thank Roland Bittleston, Kevin Costello, and Edward Witten for stimulating discussion. The research of SFM is funded by the ERC Consolidator Grant
\#864828 “Algebraic Foundations of Supersymmetric Quantum Field Theory” (SCFTAlg).
The work of MY was supported in
part by the JSPS Grant-in-Aid for Scientific Research (Grant No.\ 20H05860, 23K17689, 23K25865), and
by JST, Japan (PRESTO Grant No. JPMJPR225A, Moonshot R\& D Grant No.\ JPMJMS2061).
Kavli IPMU is supported by World Premier International Research Center Initiative (WPI), MEXT, Japan.

\appendix

\section{More Details on Twistor Space of Hyperbolic Space}
\label{sec:more details on twistor space of hyperbolic space}
For the sake of completeness, in this appendix, we explain the construction of twistor space of hyperbolic space from the $\ct$-action on the twistor space of $S^4$, i.e.\ $\Pbb^3$.

\smallskip Recall that instantons on $\R^4$ are solutions of the (anti)self-duality equation
\begin{equation}\label{eq:antiselfduality equation}
    F=\pm \star_{\R^4}F \;,
\end{equation}
on $\R^4$, where $\star_{\R^4}$ denotes the Hodge star operation on $\R^4$. For convenience, we work with the antiself-dual solutions or anti-instantons. As an equality of differential forms, this equation is invariant under any coordinate transformation of $\R^4$. In particular, it is invariant under a conformal rescaling. Denote the coordinate on $\R^4$ to be $(x_1,x_2,r,\theta)$, in which the flat metric on $\R^4$ is
\begin{equation}\label{eq:the metric on R4 in terms of metric on H3}
	ds^2_{\R^4}=r^2\left(\frac{\rd x_1^2+\rd x_2^2+\rd r^2}{r^2}+\rd \theta^2\right) \;.
\end{equation} 
The invariance of \eqref{eq:antiselfduality equation} under conformal rescaling implies that we can work with the conformally-rescaled metric on $\R^4$
\begin{equation}\label{eq:conformally-rescaled metric on R4}
    {ds'}^2_{\R^4}=\frac{\rd x_1^2+\rd x_2^2+\rd r^2}{r^2}+\rd \theta^2 \;.
\end{equation}
This metric is singular at $r=0$, which is simply a copy of $x_1x_2$-plane that we denote as $\R^2_{12}$. This metric is thus smooth only along $\R^4-\R^2$. Identifying the first term of \eqref{eq:conformally-rescaled metric on R4} as the Poincar\'e metric on $\Hbb^3$, we realize the conformal equivalence
\begin{equation}\label{eq:conformal equivalence}
	\R^4-\R^2\simeq \Hbb^3\times S^1 \;.
\end{equation}
This relation naturally leads to the study of $S^1$-invariant solutions of \eqref{eq:antiselfduality equation}, which can then be interpreted as hyperbolic monopoles on $\Hbb^3$. More precisely, this equivalence sends $(x_1,x_2,r,e^{\mfk{i}\theta})\in\Hbb^3\times S^1$ to $(x_1,x_2,re^{\mfk{i}\theta})\in \R^4-\R^2$. If we consider the conformal compactification of $\R^4-\R^2$ i.e. $S^4-S^2$, this equivalence shows that $S^4$ is an $S^1$-equivariant conformal compactification of $\Hbb^3\times S^1$. Finally, the boundary of $\Hbb^3$ (or more precisely its closure $\overline{\Hbb}^3$), located at $r=0$\footnote{This is similar to the more familiar case of the upper half-plane $\Hbb^2$ endowed with the hyperbolic metric where $\text{Im}(z)=0$ is the boundary of $\overline{\Hbb}^2$.} is denoted as $S^2_\infty:=\partial\overline{\Hbb}^3$. On the other hand, it is well-known that an anti-instanton $\SU{n}$ solution on $S^4$ corresponds to a rank-$n$ holomorphic vector bundle on $\Pbb^3$ \cite[Theorem 2.9]{Atiyah1979}. Therefore, studying hyperbolic monopoles on $\Hbb^3$ naturally leads to studying the $S^1$-action (or its complexification $\C^\times$) on $\Pbb^3$.\footnote{We denote the complex projective space simply as $\Pbb^{\#}$ while for the quaternionic projective space we use the notation $\mbb{HP}^{\#}$.} This will lead to the notion of twistor space of $\Hbb^3$, which we elaborate on in the following. 

\begin{rmk} \normalfont
    By defining hyperbolic monopoles as circle-invariant instantons, the role of the Higgs field is played by the component of the 4d gauge field along the circle, i.e. $A_\theta$. Rescaling $R\to \lambda R$, changes $R\rd\theta\to \lambda R\rd\theta$ leads to $A_\theta\to\lambda^{-1}A_\theta$. Therefore, the values of the Higgs field and especially its masses would be rescaled. On the other hand, the scalar curvature $\mbf{R}$ of $\Hbb^3$ will be rescaled as $\mbf{R}\to\lambda^{-1}\mbf{R}$ too. This means that $p_i/\mbf{R}$ is an invariant parameter of the configuration. If we send $p\to 0$, then $\mbf{R}\to\infty$, and vice versa. Fixing $\mbf{R}=-1$ means the $i$\sth component of the Higgs field is $p_i$, and unlike the case of Euclidean monopoles cannot be gauged away.  Due to the lack of an absolute scale in the Euclidean case, there is no such rescaling and we can always take the $i$\sth component of the Higgs field to be $1$.  \qed  
\end{rmk}

\smallskip $\mbb{R}^4$ can be identified as the set of quaternions and its compactification $S^4$ can then be identified with the quaternionic projective line. The set of quaternions is defined as
\begin{equation}
	\Hbb:=\{a_0+a_1\mfk{i}+a_2\mfk{j}+a_3\mfk{k}|\,a_i\in\R\} \;,
\end{equation}
with $\mfk{i}^2=\mfk{j}^2=\mfk{k}^2=-1$ and $\mfk{i}\mfk{j}=\mfk{k}$ and its cyclic permutations. We have
\begin{eqaligned}
    a_0+a_1\mfk{i}+a_2\mfk{j}+a_3\mfk{k}&=a_0+a_1\mfk{i}+a_2\mfk{j}+a_3\mfk{i}\mfk{j}
    \\
    &=a_0+a_1\mfk{i}+(a_2+a_3\mfk{i})\mfk{j} \;.
\end{eqaligned}
Hence, $\mbb{H}$ can be identified with $\C^2$ as $\Hbb\simeq\{z_1+z_2\mfk{j}\,|\,z_1,z_2\in\C\}=\C^2$. We define the scalar multiplication from the left. The quaternionic projective line $\mbb{HP}^1$ is the quotient of $\Hbb^2$, the two-dimensional quaternionic space, by the right multiplication: we identify $(h_1,h_2)\sim  \lambda(h_1,h_2)$ for $\lambda\in\Hbb^\times$ and $(h_1,h_2) \in \Hbb^2\simeq \C^4$. One then has the projection $\pi:\Pbb^3\to \mbb{HP}^1$ is given by
\begin{equation}\label{eq:the twistor fibration}
	\pi([z_1:z_2:z_3:z_4]_{\C})=[z_1+z_3\mfk{j}:z_2+z_4\mfk{j}]_{\Hbb} \;,
\end{equation}
where the notation $[z_1:z_2:z_3:z_4]_{\C}$ and $[z_1+z_3\mfk{j}:z_2+z_4\mfk{j}]_{\Hbb}$ denote the homogeneous coordinates on $\Pbb^3$ and $\mbb{HP}^1$, respectively. The fibration \eqref{eq:the twistor fibration} defines the twistor fibration. 

\smallskip The fiber $[z_1:z_2:z_3:z_4]_{\C}$ over a generic point $[w_1+w_2\mfk{j}:1]_{\Hbb}$ can be easily determined
\begin{equation}\label{eq:fiber over a generic point of twistor fibration}
    \begin{pmatrix}
    z_1
    \\
    z_3
    \end{pmatrix}
    =
    \begin{pmatrix}
        w_1 & -\bmb{w}_2
        \\
        w_2 & \hphantom{-}\bmb{w}_1
    \end{pmatrix}
    \begin{pmatrix}
    z_2
    \\
    z_4
    \end{pmatrix} \;.
\end{equation}
The only parameters are $z_2$ and $z_4$ coordinates up to a rescaling by an element of $\ct$, and as such, the fiber over a generic point is a copy of $\Pbb^1$. 
Furthermore, the fiber over the point $[w_1+w_2\mfk{j}:0]_{\Hbb}\sim[1:0]_{\mbb{H}}$ is just $[1:0:0:0]_{\C}$. 

\smallskip There is an antiholomorphic involution $\wt{\sigma}$\footnote{It is an involution since $\wt{\sigma}^2=1$.} on $\Pbb^3$
\begin{equation}\label{eq:real structure on P3}
	\wt{\sigma}([z_1:z_2:z_3:z_4]_{\C})=[\bmb{z}_1:\bmb{z}_2:\bmb{z}_3:\bmb{z}_4]_{\C} \;.
\end{equation}
There is no fixed point for the action of $\wt{\sigma}$ but there are fixed lines which are called {\it real lines}. Therefore, $\wt{\sigma}$ defines a real structure on $\Pbb^3$.

\smallskip We can now consider the $S^1$-action on the twistor fibration. Define the circle action on $S^4\simeq\mbb{HP}^1$ by\footnote{It is more convenient to define the $S^1$ action from the left.}
\begin{equation}
	[h_1:h_2]_{\Hbb}\mapsto [h_1e^{+\mfk{i}\frac{\theta}{2}}:h_2e^{+\mfk{i}\frac{\theta}{2}}]_{\Hbb} \;, \qquad h_1,h_2\in\Hbb \;.
\end{equation}
Since $\R^4\subset S^4$ due to $h\mapsto[h:1]_{\Hbb}$ (recall that $\Hbb\simeq\C^2\simeq\R^4$), the $S^1$-action on $\R^4$ is
\begin{equation}
	[h:1]_{\Hbb}\mapsto [e^{-\mfk{i}\frac{\theta}{2}}h e^{+\mfk{i}\frac{\theta}{2}}:1]_{\Hbb} \;.
\end{equation}
If we take $h=z_1+z_3\mfk{j}$, then $z_1+z_3\mfk{j} \mapsto z_1+e^{-\mfk{i}\theta}z_3\mfk{j}$. The fixed loci of this action is a copy of $\R^2$ ($z_3=0$) or its conformal compactification $S^2$, as before, which is the same as $S^2_\infty$. On the other hand, if we identify $\R^4$ as $[1:h]$ with $h=z_2+z_4\mfk{j}$, then the fixed loci is given by the plane $z_4=0$. 

\smallskip The lift and complexification of this action to $\Pbb^3$ is a $\C^{\times}$-action given by
\begin{equation}\label{eq:the C* action on CP3}
    \lambda\cdot [z_1:z_2:z_3:z_4]_{\C}=[ z_1:  z_2:  \lambda z_3:\lambda z_4]_{\C} \;,
    \qquad \lambda\in\C^{\times} \;.
\end{equation}
This action has two fixed lines
\begin{equation}\label{eq:the fixed lines of the C* action on CP3}
	\Pbb^1_+:= \{z_1=z_2=0\} \;,
    \qquad \Pbb^1_-:= \{z_3=z_4=0\} \;.
\end{equation}
Therefore, the space on which $\C^\times$-action is non-trivial is $\Pbb^3-(\Pbb^1_+\cup\Pbb^1_-)$, and are clearly exchanged by \eqref{eq:real structure on P3}. Under the twistor map \eqref{eq:the twistor fibration}, these lines are sent to $S^2_\infty$ 
\begin{eqaligned}\label{eq:projection of fixed-lines of C*-action on CP3}
     \pi([z_1:z_2:0:0]_\C)&=[z_1:z_2]_\Hbb\in S^2_\infty \;,
     \\
      \pi([0:0:z_3:z_4]_\C)&=[z_3\mfk{j},z_4\mfk{j}]_\Hbb=[\bmb{z}_3:\bmb{z}_4]_{\Hbb}\in S^2_\infty \;,
\end{eqaligned}
where in the second line and the second equality, we multiplied by $\mfk{j}$ from the right. Therefore, $\Pbb^1_+\to S^2_\infty$ is orientation-preserving while $\Pbb^1_-\to S^2$ is orientation-reversing.  

\smallskip Finally, it follows from \eqref{eq:the C* action on CP3} that
\begin{eqaligned}\label{eq:small and large limits of lambda}
\lim_{\lambda\to 0}\lambda\cdot [z_1:z_2:z_3:z_4]_{\C}&=[z_1:z_2:0:0]_{\C}\in\Pbb^1_- \;, %
	\\
\lim_{\lambda\to \infty}\lambda\cdot[z_1:z_2:z_3:z_4]_{\C} &=[0:0:z_3:z_4]_{\C}\in\Pbb^1_+ \;. %
\end{eqaligned}
Hence, taking the quotient of $\C^{\times}$-action \eqref{eq:the C* action on CP3}, we find that
\begin{equation}\label{eq:minitwistor space of hyperbolic space}
	\frac{\Pbb^3-(\Pbb^1_+\cup\Pbb^1_-)}{\C^{\times}}\simeq \Pbb_+^1\times\Pbb^1_-  \;.
\end{equation} 
This is the minitwistor space of $\overline{\Hbb}^3$ which parameterizes geodesics on $\Hbb^3$ \cite{Atiyah1984}.

\section{Some Basic Facts about Complete Intersections}
\label{sec:basic facts about complete intersections}
In this appendix, we explain some basic facts about curves on projective space that are complete intersections. The motivation is that the curve of the spectral parameter of the gCPM, defined in \eqref{eq:curve of spectral parameter of gCPM}, is a special case of complete intersections. 

\paragraph{Smooth Complete Intersections.} Consider $n-1$ polynomials $\{F_1,\ldots,F_{n-1}\}$ of degree $\{d_1,\ldots,d_{n-1}\}$. Consider the zero-locus
\begin{equation}
    C:=\left\{[x_0:\ldots,x_n]\in\Pbb^n\,\big|\,F_1(x_0,\ldots,x_n)=\ldots=F_{n-1}(x_0,\ldots,x_n)=0\right\} \;,
\end{equation}
which is a curve in $\Pbb^n$. This curve is called a smooth complete intersection if the $(n-1)\times (n+1)$ matrix $[\partial F_i/\partial x_\alpha]$ has the maximum rank $n-1$ at every point of $C$.

\smallskip We would like to express $C$ as a branched cover of $\Pbb^1$. By the Riemann's Existence Theorem, 
every Riemann surface carries a nonconstant meromorphic function, and as a result can be realized as a branched cover of $\Pbb^1$. The basic idea is that a meromorphic function on $C$ can be written locally as $P(z)/Q(z)$ for some holomorphic functions $P(z)$ and $Q(z)$ and a local coordinate $z$ on $C$. Then, the branched cover map $\pi:C\to\Pbb^1$ is given by $\pi(z)=[P(z):Q(z)]$. 

\smallskip For a complete intersection $C$, consider the meromorphic function $x_n/x_0$, which defines the branched covering map $\pi([x_0:\ldots:x_n])=[x_0:x_n]$. Using this map, we can compute the genus of $C$. Recall the Riemann--Hurwitz formula \cite{Hurwitz189103} 
\begin{equation}\label{eq:Riemann-Hurwitz applied to complete intersections}
    2g_C-2=\deg\pi(2g_{\Pbb^1}-2)+\deg R_\pi \;,
\end{equation}
where $g_C$ is the genus of $C$, $\deg\pi$ is the degree of $\pi$, and $\deg R_\pi$ is the degree of the ramification divisor $R_\pi$ of $\pi$. Hence, we need to determine $\deg\pi$ and $\deg R_\pi$ first. Consider a generic point $q:=[x_0:x_n]=[1:\lambda]\in\Pbb^1$. A fiber over $q$ is given by points of the form $[1:x_1 :x_2:\ldots:\lambda]$ which satisfy $F_i(1,x_1,\ldots,x_{n-1},\lambda)=0,\,i=1,\ldots,n-1$. Using say $F_{n-1}$, we can express $x_{n-1}$ as a function of $\lambda,x_1,\ldots,x_{n-2}$ as $\prod_{i=1}^{d_{n-1}}(x_{n-1}-\zeta^{(n-1)}_i(\lambda,x_1,\ldots,x_{n-2}))$ for some functions $\zeta^{(n-1)}_i$. Plugging one of the roots into $F_{n-2}$, we can similarly express $x_{n-2}$ as a function of $\lambda,x_1,\ldots,x_{n-3}$ with $d_{n-2}$ roots $\{\zeta^{(n-2)}_1(\lambda,x_1\ldots,x_{n-3}),\ldots,\zeta^{(n-2)}_{d_{n-2}}(\lambda,x_1\ldots,x_{n-3})\}$. Continuing in this fashion, we end up with $x_1$ in terms of $\lambda$, with $d_1$ roots $\zeta^{(1)}_1(\lambda),\ldots\zeta^{(2)}_{d_1}(\lambda)$. Therefore, there are $d_1\ldots d_{n-1}$ points in the fiber over $q$, each of multiplicity $+1$. We thus conclude
\begin{equation}\label{eq:degree of branched cover map of complete intersections}
    \deg\pi=d_1\ldots d_{n-1} \;.
\end{equation}

Next, we determine $\deg R_\pi$. The map $\pi$ is ramified at a point $p\in C$ if and only if
\begin{equation}\label{eq:definition of J(p)}
    J(p):=\det\left(\frac{\partial (F_1,\ldots,F_{n-1})}{\partial (x_1,\ldots,x_{n-1})}\right)(p)=0 \;,
\end{equation}
where
\begin{equation}
    \frac{\partial (F_1,\ldots,F_{n-1})}{\partial (x_1,\ldots,x_{n-1})}:=
    \begin{pmatrix}
        \partial F_1/\partial x_1 & \ldots & \partial F_1/\partial x_{n-1}
        \\
        \vdots & \ddots & \vdots 
        \\
        \partial F_{n-1}/\partial x_1 & \ldots & \partial F_{n-1}/\partial x_{n-1}
    \end{pmatrix} \;.
\end{equation}
We can see this as follows. Let $J(p)\ne 0$ and define the map $F:\C^{n+1}\to\C^{n-1}$ given by $F(x_0,\ldots,x_n):=(F_1(x_0,\ldots,x_n),\ldots,F_{n-1}(x_0,\ldots,x_n))$. Then, by Implicit Function Theorem, there is a unique function $g: U\to \C^{n-1}$, for some open set $U\subset\C^2$ such $F=F(x_0,g_1(x_0,x_n),\ldots,g_{n-1}(x_0,x_n),x_n)$ with $g=(g_1,\ldots,g_{n-1})$. This would imply that, near $p$, $x_{n0}:=x_n/x_0$ is a local coordinate on the curve $C$, which can be locally written as $[1:g_1(x_{n0}):\ldots:g_{n-1}(x_{n0}):x_{n0}]$. Therefore, in terms of local coordinates, the map $\pi$ is given by $\pi(x_{n0})=x_{n0}$, which is clearly unramified. Hence $\pi$ is unramified at points $p\in C$ that $J(p)\ne 0$. Conversely, assume that $J(p)=0$. Then, $x_{n0}$ would not be a good local coordinate, and instead $x_{\alpha\beta}:=x_\alpha/x_\beta$ with $1\le\alpha\ne\beta\le n-1$ would serve as a local description of the curve. $C$ can thus be locally written as $[h_0(x_{\alpha\beta}):\ldots:h_n(x_{\alpha\beta})]$ with $h_\alpha(x_{\alpha\beta})=1$ and $h_\beta(x_{\alpha\beta})=x_{\alpha\beta}$. 
Therefore, the map $\pi$ can be described locally as $\pi(x_{\alpha\beta})=[h_0(x_{\alpha\beta}):h_n(x_{\alpha\beta})]$. To show that there is a ramification, we notice $\partial_{x_{\alpha\beta}}F_i=\partial_{x_{\alpha\beta}}h_\gamma\partial_{x_\gamma}F_i=0$, especially at $p$. This would imply
\begin{equation}
    \begin{pmatrix}
        \partial F_1/\partial x_0 & \ldots & \partial F_1/\partial x_n
        \\
        \vdots & \ddots & \vdots
        \\
        \partial F_{n-1}/\partial x_0 & \ldots & \partial F_{n-1}/\partial x_n
    \end{pmatrix}
    \begin{pmatrix}
        \partial_{x_{\alpha\beta}}h_0(x_{\alpha\beta})
        \\
        \vdots
        \\
        \partial_{x_{\alpha\beta}}h_n(x_{\alpha\beta})
    \end{pmatrix}
    =0 \;.
\end{equation}
As $[\partial F_i/\partial x_\alpha]$ has rank $n-1$, we conclude that $\partial_{x_{\alpha\beta}}h_\gamma(p)=0$ with $\gamma\ne \alpha,\beta$. Therefore, we get $\partial_{x_{\alpha\beta}}\pi(x_{\alpha\beta})=h_0(x_{\alpha\beta})^{-2}[\partial_{x_{\alpha\beta}}h_n(x_{\alpha\beta})h_0(x_{\alpha\beta})-h_n(x_{\alpha\beta})\partial_{x_{\alpha\beta}}h_0(x_{\alpha\beta})](p)=0$. Recall that the degree of ramification of $\pi$ at a point $p$ is given by $1+\text{ord}_p \pi'$, where $\text{ord}_p\pi'$ is the order of vanishing of $\pi'$ at $p$. As we saw $\partial_{x_{\alpha\beta}}\pi(x_{\alpha\beta})(p)=0$ from which we see that the order of ramification of $\pi$ at $p$ is bigger than one, and hence $\pi$ is ramified. Putting together, we conclude that $\pi$ is ramified at $p\in C$ if and only if $J(p)=0$, where $J(p)$ is defined in \eqref{eq:definition of J(p)}. 

\smallskip To determine $\deg R_\pi$, we use Bezout's Theorem according to which the degree of an intersection divisor, say $\text{div}\,G$ for a polynomial $G$, of a curve $C$ is given generically by $\deg G\times\deg C$. In the case of a complete intersection, the relevant intersection divisor is given by the polynomial $\det[\partial (F_1,\ldots,F_{n-1}/\partial (x_1,\ldots,x_{n-1})]$, and we have $R_\pi=\text{div}(\det[\partial F_i/\partial x_\alpha])$. The degree of this polynomial is $(d_1+\ldots+d_{n-1}-n+1)$.\footnote{Since $F_i$ is homogeneous of degree $d_i$, $\partial_{x_\alpha} F_i$ is also homogeneous of degree $d_i-1$.} On the other hand, recall that the degree of a curve $C$ is given by the degree of any hyperplane divisor on $C$. We consider the polynomial $G(x_0,\ldots,x_n)=x_n$, and assume $[0:\ldots:0:1:0:\ldots:0]$ with $1$ is put in the $i$\textsuperscript{th} slots with $1\le i\le n-1$. Hence, we can consider the function $x_n/x_0$ to compute the hyperplane divisor $\text{div}\,G$, which is the same as the divisors of zeroes of $x_n/x_0$. The latter is the same as the inverse image divisor $\pi^*(0)$, with $\pi$ defined above. We only need the degree of $\pi^*(0)$, which is the same as the degree of $\pi$\footnote{Recall that for a finite morphism $f: X\to Y$ between Riemann surfaces $X$ and $Y$ and a divisor $D$ on $Y$, we have $\deg(f^*(D))=\deg f\times\deg D$.} given in \eqref{eq:degree of branched cover map of complete intersections}. We thus conclude that 
\begin{equation}\label{eq:degree of the ramification divisor of the map pi of a complete intersection}
    \deg R_\pi=d_1\ldots d_{n-1}\left(d_1+\ldots+d_{n-1}-n+1\right) \;.
\end{equation}
Putting this and \eqref{eq:degree of branched cover map of complete intersections} into  \eqref{eq:Riemann-Hurwitz applied to complete intersections}, we arrive at the genus of $C$
\begin{equation}\label{eq:genus of a general complete intersection}
    g_C=1+\frac{1}{2}d_1\ldots d_{n-1}\left(d_1+\ldots+d_{n-1}-n-1\right) \;.
\end{equation}

\paragraph{Curve of the Spectral Parameter of gCPM as a Complete Intersection.} The curve $\cpc$ of the spectral parameter of gCPM is a special case of complete intersections. This can be seen by first noticing that for $[\partial F_i/\partial x_\alpha]$ has rank $2N$, which follows from \eqref{eq:curve of spectral parameter of gCPM}. We can thus obtain the information about the curve from what we have explained for a complete intersection by setting $n\to 2n-1$ and $d_1=\ldots=d_{2n-2}=N$ in \eqref{eq:degree of branched cover map of complete intersections}, \eqref{eq:degree of the ramification divisor of the map pi of a complete intersection}, and \eqref{eq:genus of a general complete intersection}
\begin{eqgathered}\label{eq:data of gCPM curve as a complete intersection}
    \deg{\wt\pi}_{N,n}=N^{2n-2} \;, \qquad \deg R_{{\wt\pi}_{N,n}}=2N^{2n-2}(N-1)(n-1) \;,
    \\
    g_{\cpc}=N^{2n-2}(N(n-1)-n)+1 \;,
\end{eqgathered}
where we have denoted the branched covering map as ${\wt\pi}_{N,n}:\cpc\to\Pbb^1$. The second line of \eqref{eq:data of gCPM curve as a complete intersection} is the genus of the curve of the spectral parameter of gCPM given in \eqref{eq:genus of the curve of spectral parameter of gCPM}.

\section{Some Basic Lie Algebra Facts and Manipulations}
\label{sec:some basic facts and Lie algebra manipulations}
In this appendix, we collect some basic facts about classical and exceptional Lie algebras that would be useful to construct spectral curves of hyperbolic monopoles with gauge groups different than $\SU{n}$. Based on the discussion in \S\ref{sec:correspondence for classical and exceptional groups}, we spell out the details of the relation between the fundamental weights of a classical or exceptional Lie algebra $\mfk{g}$  and $\mfk{su}(n)$, which leads to the relation between spectral curves and charges associated with the two Lie algebras. In the following, we use \cite[Appendix C]{Knapp1996}, \cite[\S 2.4.3]{GoodmanNolan2009}, and \cite[Chapter X, \S 3.3]{Helgason2001}, to which we refer for further details.

\smallskip We start with classical Lie algebras $A_n, B_n, C_n,$ and $D_n$, and then spell out some details of exceptional Lie algebras. The compact real forms of the simple Lie groups over $\C$ are described by Cartan \cite{Cartan191403,Cartan192703}. In particular, it is known that \cite{Cartan1894}
\begin{equation}\label{eq:the embedding of exceptional Lie algebras in classical Lie algebras}
	\begin{gathered}
		G_2\subset B_3=\mfk{so}(7) \;,
        \quad F_4\subset D_{13}=\mfk{so}(26) \;,
        \quad E_6\subset A_{26}=\mfk{sl}(27) \;,
		\\
		\quad E_7\subset C_{56}=\mfk{sp}({56}) \;,
        \quad E_8\subset D_{124}=\mfk{so}(248) \;.
	\end{gathered}
\end{equation}
In constructing the corresponding monopole, we embed these algebras in $\mfk{su}(n)$ for some $n$ in each case. 

\paragraph{Lie Algebra $A_{n-1}:=\mfk{sl}(n)$.} Let us start with (the compact real form) of $A_{n-1}=\mfk{sl}(n)$ ($\mfk{su}(n)$). The simple roots are given by
\begin{equation}
	\alpha^{A_n}_i=e_{i}-e_{i+1} \;, \qquad i=1,\ldots,n-1 \;,
\end{equation}
where $\{e_1,\cdots,e_n\}$ is an orthonormal basis ($\langle e_i,e_j\rangle=\delta_{ij}$) of $\mbb{R}^{n}$. Therefore, the corresponding fundamental weights are given by (see \eqref{eq:defining equation of fundamental weights})
\begin{equation}\label{eq:fundamental weights of su(n)}
	\omega_i^{A_n}=\sum_{j=1}^i e_j \;,
    \qquad i=1,\cdots,n-1 \;.
\end{equation}
The highest weight representations corresponding to fundamental weights are 
\begin{equation}
    \omega^{A_n}_i \qquad\Longleftrightarrow \qquad  \bigwedge\nolimits^{i}\mbb{C}^n \;,
    \qquad i=1,\ldots,n-1 \;,
\end{equation}
while $\bigwedge\nolimits^n\mbb{C}^n$ is one-dimensional and corresponds to the trivial representation of $\mfk{su}(n)$. 

\paragraph{Lie Algebra $B_n:=\mfk{so}(2n+1)$.} The algebra $B_{n}=\mfk{so}(2n+1)$ has rank $n$. A set of simple roots are
\begin{eqgathered}
    \alpha^{B_n}_i=e_i-e_{i+1} \;, \qquad i=1,\ldots,n-1 \;, \qquad 
    \\
    \alpha^{B_n}_{n}=e_n \;,
\end{eqgathered}
    and the corresponding fundamental weights are given by
\begin{eqgathered}\label{eq:fundamental weights of so(2n+1)}
    \omega^{B_n}_i=\sum_{j=1}^i e_j \;, \qquad i=1,\cdots,n-1 \;,
		\\
		\omega^{B_n}_n=\frac{1}{2}(e_1+\cdots+e_n) \;.
\end{eqgathered}
Hence, fundamental highest-weight representations are 
\begin{eqgathered}
    \omega^{B_n}_i\qquad\Longleftrightarrow\qquad \bigwedge\nolimits^i\mbb{C}^{2n+1} \;,
    \qquad i=1,\ldots,n \;.
\end{eqgathered}
$\mfk{so}(2n+1)$ can be embedded in $\mfk{su}(2n+1)$. Under the embedding $\mfk{so}(2n+1)\hookrightarrow\mfk{su}(2n+1)$ and by comparing \eqref{eq:fundamental weights of su(n)} and \eqref{eq:fundamental weights of so(2n+1)}, we see the relation between fundamental weights\begin{eqgathered}\label{eq:relation between fundamental weights of Bn and A2n}
    \omega_i^{B_n}=\omega^{A_{2n}}_i \;, \qquad i=1,\ldots,n-1 \;,
    \\
    \omega_n^{B_n}=\frac{1}{2}\omega^{A_{2n}}_n \;.
\end{eqgathered}
We can order the fundamental weights as
\begin{equation}\label{eq:fundamental weights of so(2n+1) in terms of su(n+1)}
    \left(\omega_1^{A_2n},\ldots,\omega^{A_{2n}}_{2n}\right)=\left(\omega^{B_n}_1,\ldots,\omega_{n-1}^{B_n},2\omega_n^{B_n},2\omega_n^{B_n},\omega^{B_n}_{n-1},\ldots,\omega^{B_n}_1\right) \;,
\end{equation}
hence the corresponding fundamental highest-weight representations of $\mfk{so}(2n+1)$ are identified with the following fundamental representations of $\mfk{su}(2n+1)$
\begin{equation}
    \left(\mbb{C}^{2n+1},\ldots,\bigwedge\nolimits^{n-1}\mbb{C}^{2n+1},\bigwedge\nolimits^{n}\mbb{C}^{2n+1},\bigwedge\nolimits^{n}\mbb{C}^{2n+1},\bigwedge\nolimits^{n-1}\mbb{C}^{2n+1},\ldots,\mbb{C}^{2n+1}\right) \;.
\end{equation}
Since spectral curves are labeled by fundamental weights, \eqref{eq:relation of spectral curves: so(2n+1) and su(2n+1)} is evident. 

\smallskip Furthermore, using \eqref{eq:definition of magnetic charges of a monopole} and \eqref{eq:fundamental weights of so(2n+1) in terms of su(n+1)}, we can find the relation between charges of an $\mfk{su}(2n+1)$ monopole and the embedded $\mfk{so}(2n+1)$ monopole as
\begin{equation}\label{eq:relation between Bn and A2n charges I}
    m_i^{B_n}=\omega^{B_n}_i(2\star_{S^2_\infty}F)=\omega^{A_{2n}}_i(2\star_{S^2_\infty}F)=m_i^{A_{2n}} \;,
    \qquad i=1,\ldots,n-1 \;,
\end{equation}
while
\begin{equation}\label{eq:relation between Bn and A2n charges II}
    m_n^{B_n}=\omega^{B_n}_n(2\star_{S^2_\infty}F)=\frac{1}{2}\omega^{A_{2n}}_n(2\star_{S^2_\infty}F)=\frac{1}{2}m_n^{A_{2n}} \;.
\end{equation}
As magnetic charges should be integers, this means that the $n$\textsuperscript{th} monopole charge of the corresponding $\mfk{su}(2n+1)$-monopole should be even. As we would like to get the spectral data of a gCPM for which all charges of monopole associated with the special unitary group should be the same, we set $m_i^{A_{2n}}=2N,\,i=1,\ldots,n$, which then gives \eqref{eq:charges of so(2n+1) monopoles associated with a gCPM}.

\paragraph{Lie Algebra $C_n:=\mfk{sp}(2n)$.} The rank-$n$ algebra $C_n=\mfk{sp}(2n)$ has a set of simple roots given by 
\begin{eqgathered}
    \alpha_1^{C_n}=e_i-e_{i+1} \;, \qquad i=1,\ldots,n-1 \;,
    \\
    \alpha_n^{C_n}=2e_n \;,
\end{eqgathered}
and fundamental weights are
\begin{equation}\label{eq:fundamental weights of sp(2n)}
	\begin{gathered}
		\omega_i^{C_n}=\sum_{j=1}^i e_j \;,
        \qquad i=1,\cdots,n \;.
	\end{gathered}
\end{equation}
The fundamental weights and corresponding highest-weights are\footnote{Note that these representations are constructed with respect to a nondegenerate antisymmetric form $\mbs{\Omega}$ on $\mbb{C}^{2n}$. The fundamental modules are in fact the space of $\mbs{\Omega}$-harmonic elements. Consider the space
	\begin{equation*}
		H(\bigwedge\mbb{C}^{2n},\mbs{\Omega})=\bigoplus_{p=0}^nH(\bigwedge\nolimits^p\mbb{C}^{2n},\mbs{\Omega}) \;,
	\end{equation*}
	where $H(\bigwedge\nolimits^p\mbb{C}^{2n},\mbs{\Omega})$ is the space of $\mbs{\Omega}$-harmonic $p$-vectors. Then, the fundamental representations are 
	\begin{equation*}
		\left(H(\mbb{C}^{2n},\mbs{\Omega},\cdots,H(\bigwedge\nolimits^n\mbb{C}^{2n},\mbs{\Omega}))\right) \;.
	\end{equation*}
	Also $H(\bigwedge\nolimits^p\mbb{C}^{2n},\mbs{\Omega})=0$ for $p>n$. In the above construction, we assume that this subtlety is understood. For more details see \cite[\S 5.5.2]{GoodmanNolan2009}.}
\begin{equation}
    \omega^{C_n}_i\qquad\Longleftrightarrow\qquad \bigwedge\nolimits^{i}\C^{2n} \;,
    \qquad i=1,\ldots,n \;.
\end{equation}
This Lie algebra can be embedded in $\mfk{su}(2n)$. Under the embedding $\mfk{sp}(2n)\hookrightarrow\mfk{su}(2n)$ and by comparing \eqref{eq:fundamental weights of su(n)} and \eqref{eq:fundamental weights of sp(2n)}, the relation between fundamental weights is
\begin{equation}\label{eq:relation between fundamental weights of Cn and A2n-1}
    \omega^{C_n}_i=\omega^{A_{2n-1}}_i \;, \qquad i=1,\ldots,n \;.
\end{equation}
The fundamental weights can be ordered as
\begin{equation}
    \left(\omega_1^{A_{2n-1}},\ldots,\omega^{A_{2n-1}}_{2n-1}\right)=\left(\omega^{C_n}_1,\ldots,\omega_{n-1}^{C_n},\omega_n^{C_n},\omega^{C_n}_{n-1},\ldots,\omega^{C_n}_1\right) \;,
\end{equation}
from which it follows that the relation between fundamental highest-weight representations is
\begin{equation}
    \left(\mbb{C}^{2n}, \bigwedge\mbb{C}^{2n},\ldots,\bigwedge\nolimits^{n-1}\mbb{C}^{2n},\bigwedge\nolimits^{n}\mbb{C}^{2n},\bigwedge\nolimits^{n-1}\mbb{C}^{2n},\ldots,\mbb{C}^{2n}\right) \;,
\end{equation}
from which \eqref{eq:relation between spectral curves of hyperbolic sp(n) and su(2n) monopoles} is deduced. 

\smallskip Furthermore, from \eqref{eq:definition of magnetic charges of a monopole} and \eqref{eq:relation between fundamental weights of Cn and A2n-1}, we get the $n$ charges of $\mfk{sp}(2n)$-monopoles in terms of charges of the corresponding $\mfk{su}(2n)$-monopoles
\begin{equation}\label{eq:relation between Cn and A2n-1 charges}
    m_i^{C_n}=\omega_i^{C_n}(2\star_{S^2_\infty}F)=\omega_i^{A_{2n-1}}(2\star_{S^2_\infty}F)=m_i^{A_{2n-1}} \;,
    \qquad i=1,\ldots,n \;.
\end{equation}
Setting $m_i^{A_{2n-1}}=N$, we get \eqref{eq:charges of Cn monopoles}. 

\paragraph{Lie Algebra $D_n:=\mfk{so}(2n)$.} This algebras has rank $n$. Simple roots are  
\begin{eqgathered}
    \alpha^{D_n}_i=e_i-e_{i+1} \;, \qquad i=1,\ldots,n-1 \;,
    \\
    \alpha^{D_n}_n=e_{n-1}+e_n \;,
\end{eqgathered}
and the corresponding fundamental weights are
\begin{eqgathered}\label{eq:fundamental weights of Dn}
    \omega^{D_n}_i=\sum_{j=1}^i e_j \;,\qquad i=1,\cdots,n-2 \;,
		\\
		\omega^{D_n}_{n-1}=\frac{1}{2}(e_1+\cdots+e_{n-1}-e_n) \;, \qquad \omega^{D_n}_{n}=\frac{1}{2}(e_1+\cdots+e_n) \;.
\end{eqgathered}
The highest-weight representations corresponding to these fundamental weights are given by
\begin{eqaligned}
    \omega^{D_n}_i\hphantom{\omega_i}&\qquad &\Longleftrightarrow\qquad &\bigwedge\nolimits^{i}\mbb{C}^{2n} \;,
    \\
    \omega^{D_n}_{n-1}+\omega^{D_n}_{n} &\qquad &\Longleftrightarrow\qquad &\bigwedge\nolimits^{n-1}\mbb{C}^{2n} \;,
\end{eqaligned}
and $\bigwedge\nolimits^{n}\mbb{C}^{2n}$ is the direct sum of two irreducible $\mfk{so}(2n)$-module with highest weights $2\omega^{D_n}_{n-1}$ and $2\omega^{D_n}_{n}$. 

\smallskip $D_n$ can be embedded in $\mfk{su}(2n)$. Under the embedding $\mfk{so}(2n)\hookrightarrow\mfk{su}(2n)$ and by comparing \eqref{eq:fundamental weights of su(n)} and \eqref{eq:fundamental weights of Dn}, the fundamental weights are related as
\begin{eqgathered}\label{eq:relation between fundamental weights of Dn and A2n-1}
    \omega^{D_n}_i=\omega^{A_{2n-1}}_i \;, \qquad i=1,\ldots,n-2 \;,
    \\
    \omega^{D_n}_{n-1}+\omega^{D_n}_n=\omega^{A_{2n-1}}_{n-1} \;,\qquad
    2\omega^{D_n}_{n}=\omega^{A_{2n-1}}_n \;.
\end{eqgathered}
If we order the fundamental weights as  
\begin{equation}
    \left(\omega^{A_{2n-1}}_1,\ldots,\omega_n^{A_{2n-1}},\ldots,\omega^{A_{2n-1}}_{2n-1}\right)=\left(\omega^{D_n}_1,\ldots,\omega_{n-2}^{D_n},\omega_{n-1}^{D_n}+\omega^{D_n}_{n},2\omega^{D_n}_n,\omega_{n-1}^{D_n}+\omega^{D_n}_{n}, \omega^{D_n}_{n-2},\ldots,\omega^{D_n}_1\right) \;,
\end{equation}
the corresponding highest-weight representations are
\begin{equation}
    \left(\mbb{C}^{2n}, \bigwedge\mbb{C}^{2n},\ldots,\bigwedge\nolimits^{n-1}\mbb{C}^{2n},\bigwedge\nolimits^{n}\mbb{C}^{2n},\bigwedge\nolimits^{n-1}\mbb{C}^{2n},\ldots,\mbb{C}^{2n}\right) \;.
\end{equation}

\smallskip Again, from \eqref{eq:definition of magnetic charges of a monopole} and \eqref{eq:relation between fundamental weights of Dn and A2n-1}, the relation of magnetic charges is
\begin{equation}\label{eq:relation between Dn and A2n-1 charges I}
    m^{D_n}_i=\omega_i^{D_n}(2\star_{S^2_\infty} F_\infty)=\omega_i^{A_{2n-1}}(2\star_{S^2_\infty} F_\infty)=m_i^{A_{2n-1}} \;,
    \quad  i=1,\cdots,n-2 \;,
\end{equation}
while
\begin{equation}\label{eq:relation between Dn and A2n-1 charges II}
    m^{D_n}_{-}+m^{D_n}_{+}=(\omega_{n-1}^{D_n}+\omega_n^{D_n})(2\star_{S^2_\infty} F_\infty)=\omega_{n-1}^{A_{2n-1}}(2\star_{S^2_\infty} F_\infty)=m_{n-1}^{A_{2n-1}} \;,
\end{equation}
and
\begin{equation}\label{eq:relation between Dn and A2n-1 charges III}
    m^{D_n}_{n}=\omega^{D_n}_n(2\star_{S^2_\infty} F_\infty)=\frac{1}{2}\omega^{A_{2n-1}}_n(2\star_{S^2_\infty} F_\infty)=\frac{1}{2}m^{A_{2n-1}}_n \;.
\end{equation}
We obtain \eqref{eq:charges of so(2n) monopoles} by setting $m^{A_{2n-1}}_i=2N,\,i=1,\ldots,n$, as all charges have to be integer to get a meaningful generalization of the CPM. 

\paragraph{Exceptional Lie Algebra $\mfk{g}_2$.} The Lie algebra $\mfk{g}_2$ has rank $2$, dimension $14$. There are two fundamental representations of dimensions $7$ and $14$. The dimension of adjoint representation is $14$. The simple roots are given by
\begin{equation}\label{eq:simple roots of g2}
    \{\alpha^{\mfk{g}_2}_1,\alpha^{\mfk{g}_2}_2\}=\left\{e_1-e_2,-2e_1+e_2+e_3\right\} \;,
\end{equation}
and the corresponding fundamental weights are
\begin{equation}\label{eq:fundamental weights of g2}
    \omega^{\mfk{g}_2}_1=-e_2+e_3 \;,
    \qquad \omega^{\mfk{g}_2}_2=-e_1-e_2+2e_3 \;.
\end{equation}
As we noticed in \eqref{eq:the embedding of exceptional Lie algebras in classical Lie algebras}, the Lie algebra $\mfk{g}_2$ can be embedded into $\mfk{so}(7)$. This can be seen by a technique called root folding (for more details, see \cite[Chapter 30]{Bump2013}). The Dynkin diagram of $\mfk{so}(7)$ can be folded to the Dynkin diagram of $\mfk{g}_2$. This has been illustrated in Fig. \ref{fig:the folding of so(7) Dynkin diagram to G2 Dynkin diagram}.
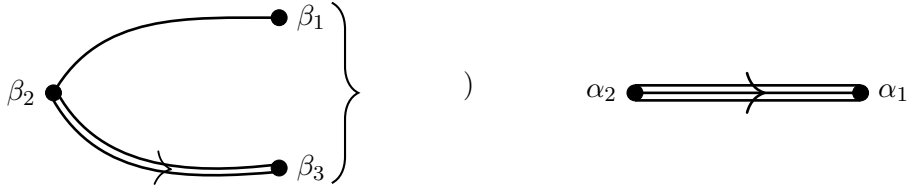
\begin{figure}[H]\centering 
	\begin{minipage}{.35\textwidth}\centering
		\begin{tikzpicture}
			\draw[fill] (0,0)node[left, xshift=-.1cm]{$\beta_2$} circle (3pt);
			\draw[fill] (3,1)node[right, xshift=.1cm]{$\beta_1$} circle (3pt);
			\draw[fill] (3,-1)node[right, xshift=.1cm]{$\beta_3$} circle (3pt);
			\draw [decorate,decoration={brace,amplitude=10pt},line width=.8pt]
			(3.7,1.2) -- (3.7,-1.2);
			Node at (3.7,0) [right,xshift=.4cm]{\text{sqeeze}};
			\draw[line width=1pt] (0,0) to [out=60, in=180] (3,1);
			\draw[line width=1pt] (0,-.1) to [out=-60, in=185] (3,-1.05);
			\draw[line width=1pt] (0,.1) to [out=-64, in=187] (3,-.95);
			\path[->-,line width=.7pt] (0,0) to [out=-65, in=185] (3,-1);
		\end{tikzpicture}
	\end{minipage}
	~
	\begin{minipage}{.1\textwidth}\centering 
		\begin{tikzpicture}
			\begin{scope}[] 
				Node at (2,0){};
				Node at (3.2,0){$\xRightarrow{\makebox[1cm]{}}$};
			\end{scope}
		\end{tikzpicture}
	\end{minipage}
	~
	\begin{minipage}{.35\textwidth}\centering 
		\begin{tikzpicture}
			\draw[fill] (0,0)node[left, xshift=-.1cm]{$\alpha_2$} circle (3pt);
			\draw[fill] (3,0)node[right, xshift=.1cm]{$\alpha_1$} circle (3pt);
			\draw[line width=1pt] (0,.1) -- (3,.1);
			\draw[line width=1pt,->-] (0,0) -- (3,0);
			\draw[line width=1pt] (0,-.1) -- (3,-.1);
		\end{tikzpicture}
	\end{minipage}
	\caption{The folding of $\mfk{so}(7)$ Dynkin diagram (left) onto $\mfk{g}_2$ Dynkin diagram (right).}\label{fig:the folding of so(7) Dynkin diagram to G2 Dynkin diagram}
\end{figure}
The fundamental weights of $\mfk{g}_2$ can be written in terms of fundamental weights of $\mfk{so}(7)$, which in turn can be written in terms of fundamental weights of $\mfk{su}(7)$. The overall result is 
\begin{eqaligned}\label{eq:fundamental weights of g2 in terms of those of so(7) and su(7)}
    \omega^{\mfk{g}_2}_1&=-\omega^{B_3}_1-\omega^{B_3}_2+\omega^{B_3}_3=-\omega^{A_{6}}_1-\omega^{A_6}_2+\frac{1}{2}\omega^{A_6}_3 \;,
    \\
    \omega^{\mfk{g}_2}_2&=-3\omega^{B_3}_2+\omega^{B_3}_3=-3\omega^{A_6}_2+\frac{1}{2}\omega^{A_6}_3 \;,
\end{eqaligned}
through which we can relate the magnetic charges of the corresponding monopoles analogously
\begin{eqaligned}\label{eq:magnetic charges of g2 monopoles in terms of those of so(7) and su(7) monopoles}
    m^{\mfk{g}_2}_1&=-m^{B_3}_1-m^{B_3}_2+m^{B_3}_3=-m^{A_{6}}_1-m^{A_6}_2+
    \frac{1}{2}m^{A_6}_3 \;,
    \\
    m^{\mfk{g}_2}_2&=-3m^{B_3}_2+m^{B_3}_3=-3m^{A_6}_2+\frac{1}{2}m^{A_6}_3 \;.
\end{eqaligned}
Recall that to connect to gCPM, all charges associated with $\mfk{su}(7)$ monopoles should be the same. However, it is easy to see that it is impossible to have a $\mfk{g}_2$-monopole whose corresponding $\mfk{su}(7)$-monopole has equal charges: we simply set $m^{\mfk{su}(7)}_i=2N,\,i=1,2,3$ for some $N\in\mbb{Z}_+$. Then,
\begin{equation}
    m^{\mfk{g}_2}_1=-3N \;,
    \qquad m^{\mfk{g}_2}_2=-5N \;,
\end{equation}
which are both negative.

\paragraph{Exceptional Lie Algebra $\mfk{f}_4$} The Lie algebra $\mfk{f}_4$ has rank $4$, dimension $52$. The dimension of fundamental representation is $26$ and the dimension of adjoint representation is $52$. The simple roots are given by
\begin{eqaligned}\label{eq:simple roots of f4}
    \alpha_1^{\mfk{f}_4}&=\frac{1}{2}(e_1-e_2-e_3-e_4) \;,
    &\qquad \alpha_2^{\mfk{f}_4}&=e_4 \;,
    \\
    \alpha_3^{\mfk{f}_4}&=e_3-e_4 \;, &\qquad 
    \alpha_4^{\mfk{f}_4}&=e_2-e_3 \;,
\end{eqaligned}
and the corresponding fundamental weights are
\begin{eqaligned}\label{eq:fundamental weights of f4}
    \omega_1^{\mfk{f}_4}&=e_1 \;,
    &\qquad \omega_2^{\mfk{f}_4}&=\frac{1}{2}(3e_1+e_2+e_3+e_4) \;,
    \\
    \omega_3^{\mfk{f}_4}&=2e_1+e_2+e_3 \;, &\qquad 
    \omega_4^{\mfk{f}_4}&=e_1+e_2 \;.
\end{eqaligned}
From \eqref{eq:fundamental weights of su(n)} and \eqref{eq:fundamental weights of Dn}, these can be written as fundamental weights of $\mfk{so}(26)$ and in turn those of $\mfk{su}(26)$ as follows
    
\begin{eqaligned}
    \omega_1^{\mfk{f}_4}&=\omega_1^{D_{13}}=\omega_1^{A_{25}} \;, &\qquad 
    \omega_2^{\mfk{f}_4}&=\omega_1^{D_{13}}+\frac{1}{2}\omega_4^{D_{13}}=\omega_1^{A_{25}}+\frac{1}{2}\omega_4^{A_{25}} \;,
    \\
    \omega_4^{\mfk{f}_4}&=\omega_2^{D_{13}}=\omega_2^{A_{25}} \;, &\qquad 
    \omega_3^{\mfk{f}_4}&=\omega_1^{D_{13}}+\frac{1}{2}\omega_3^{D_{13}}=\omega_1^{A_{25}}+\frac{1}{2}\omega_3^{A_{25}} \;.
\end{eqaligned}

Using this relation and \eqref{eq:definition of magnetic charges of a monopole}, we have
\begin{eqaligned}
    m_1^{\mfk{f}_4}&=m_1^{D_{13}}=m_1^{A_{25}} \;, &\qquad 
    m_2^{\mfk{f}_4}&=m_1^{D_{13}}+\frac{1}{2}m_4^{D_{13}}=m_1^{A_{25}}+\frac{1}{2}m_4^{D_{25}} \;,
    \\
    m_4^{\mfk{f}_4}&=m_2^{D_{13}}=m_2^{A_{25}} \;, &\qquad 
    m_3^{\mfk{f}_4}&=m_1^{D_{13}}+\frac{1}{2}m_3^{D_{13}}=m_1^{A_{25}}+\frac{1}{2}m_3^{D_{25}} \;.
\end{eqaligned}
Therefore, if we take $m_i^{A_{25}}=2N,\,i=1,\ldots 25$, we get the charges of an $\mfk{f}_4$-monopole associated with a gCPM
\begin{eqaligned}
    m_1^{\mfk{f}_4}&=2N, &\qquad m_2^{\mfk{f}_4}&=3N \;,
    \\
    m_4^{\mfk{f}_4}&=2N, &\qquad m_3^{\mfk{f}_4}&=3N \;.
\end{eqaligned}

\paragraph{Exceptional Lie Algebra $\mfk{e}_6, \mfk{e}_7,$ and $\mfk{e}_8$.}  The Lie algebra $\mfk{e}_6$ has rank $6$, dimension $78$. The dimension of fundamental representation is $27$ and the dimension of adjoint representation is $78$. The Lie algebra $\mfk{e}_7$ has rank $7$, dimension $133$. The dimension of fundamental representation is $56$ and the dimension of adjoint representation is $133$. The Lie algebra $\mfk{e}_8$ has rank $8$, dimension $248$. The dimension of the adjoint representation is $248$. The rest of the computations are similar to what has been explained previously and can be figured out using \cite[Appendix C]{Knapp1996}. We leave the details to the interested reader.

\bibliography{References}

\providecommand{\href}[2]{#2}\begingroup\raggedright\begin{thebibliography}{100}

\bibitem{Donaldson198302}
S.~K. Donaldson, \emph{{An Application of Gauge Theory to Four-Dimensional Topology}}, \href{https://doi.org/10.4310/jdg/1214437665}{\emph{J. Diff. Geom.} {\bfseries 18} (Feb, 1983) 279--315}.

\bibitem{Donaldson198501}
S.~K. Donaldson, \emph{{Anti Self-Dual Yang--Mills Connections Over Complex Algebraic Surfaces and Stable Vector Bundles}}, \href{https://doi.org/10.1112/plms/s3-50.1.1}{\emph{Proc. Lond. Math. Soc.} {\bfseries 50} (Jan, 1985) 1--26}.

\bibitem{Ward198508}
R.~S. Ward, \emph{{Integrable and Solvable Systems, and Relations Among Them}}, \href{https://doi.org/10.1098/rsta.1985.0051}{\emph{Phil. Trans. Roy. Soc. Lond. A} {\bfseries 315} (Aug, 1985) 451--457}.

\bibitem{Ward1977}
R.~S. Ward, \emph{\href{https://ora.ox.ac.uk/objects/uuid:8ea3ffe7-c739-4e83-8437-7a136747b267/files/m0561cfb146657983394658a0814000b1}{Curved Twistor Spaces}}, Ph.D. thesis, University of Oxford, 1977.

\bibitem{Ward197704}
R.~S. Ward, \emph{{On Selfdual Gauge Fields}}, \href{https://doi.org/10.1016/0375-9601(77)90842-8}{\emph{Phys. Lett. A} {\bfseries 61} (Apr, 1977) 81--82}.

\bibitem{AtiyahHitchinSinger197707}
M.~F. Atiyah, N.~J. Hitchin and I.~M. Singer, \emph{{Deformations of Instantons}}, \href{https://doi.org/10.1073/pnas.74.7.2662}{\emph{Proc. Nat. Acad. Sci.} {\bfseries 74} (Jul, 1977) 2662--2663}.

\bibitem{AtiyahHitchinSinger197809}
M.~F. Atiyah, N.~J. Hitchin and I.~M. Singer, \emph{{Selfduality in Four-Dimensional Riemannian Geometry}}, \href{https://doi.org/10.1098/rspa.1978.0143}{\emph{Proc. Roy. Soc. Lond. A} {\bfseries 362} (Sep, 1978) 425--461}.

\bibitem{Bogomolny197604}
E.~B. Bogomolny, \emph{\href{https://www.osti.gov/biblio/7309001}{The Stability of Classical Solutions}}, {\emph{Sov. J. Nucl. Phys.} {\bfseries 24} (1976) 449}.

\bibitem{RebbiSoliani198412}
E.~B. Bogomolny, \emph{{The Stability of Classical Solutions}},  in \emph{\href{https://www.worldscientific.com/worldscibooks/10.1142/0046}{Solitons and Particles}} (C.~Rebbi and G.~Soliani, eds.).
\newblock World Scientific, Dec, 1984.

\bibitem{Braam1987}
P.~J. Braam, \emph{\href{https://ora.ox.ac.uk/objects/uuid:daa73d43-6d58-404c-9926-ebf23f59cfc6/files/mac33fa6e784449c73b2e83d7a358df07}{Magnetic Monopoles and Hyperbolic Three-Manifolds}}, Ph.D. thesis, University of Oxford, 1987.

\bibitem{Braam198902}
P.~J. Braam, \emph{{Magnetic Monopoles on Three-Manifolds}}, \href{https://doi.org/10.4310/jdg/1214443597}{\emph{J. Differ. Geom} {\bfseries 30} (Feb, 1989) 425--464}.

\bibitem{AtiyahHitchin198812}
M.~F. Atiyah and N.~Hitchin, \emph{\href{https://press.princeton.edu/books/hardcover/9780691633312/the-geometry-and-dynamics-of-magnetic-monopoles}{The Geometry and Dynamics of Magnetic Monopoles}}.
\newblock Princeton University Press, Dec, 1988.

\bibitem{Atiyah1984}
M.~Atiyah, \emph{{Magnetic Monopoles in Hyperbolic Spaces}},  in \emph{\href{http://www.math.tifr.res.in/~publ/studies/Vector-Bundles-On-Algebraic-Varieties.pdf}{Proc. of Bombay Colloquium 1984 on Vector Bundles on Algebraic Varieties}}, pp.~1--34, Oxford University Press, 1987.

\bibitem{JarvisNorbury199711}
S.~Jarvis and P.~Norbury, \emph{{Zero and Infinite Curvature Limits of Hyperbolic Monopoles}}, \href{https://doi.org/10.1112/s0024609397003500}{\emph{Bull. Lond. Math. Soc.} {\bfseries 29} (Nov, 1997) 737–744}.

\bibitem{Nash200810}
O.~Nash, \emph{{Singular Hyperbolic Monopoles}}, \href{https://doi.org/10.1007/s00220-007-0368-2}{\emph{Commun. Math. Phys.} {\bfseries 277} (Oct, 2008) 161--187}.

\bibitem{BraamAustin199008}
P.~J. {Braam} and D.~M. {Austin}, \emph{{Boundary Values of Hyperbolic Monopoles}}, \href{https://doi.org/10.1088/0951-7715/3/3/012}{\emph{Nonlinearity} {\bfseries 3} (Aug, 1990) 809--823}.

\bibitem{MurraySinger199607}
M.~{Murray} and M.~{Singer}, \emph{{Spectral Curves of Non-Integral Hyperbolic Monopoles}}, \href{https://doi.org/10.1088/0951-7715/9/4/009}{\emph{Nonlinearity} {\bfseries 9} (Jul, 1996) 973--997}.

\bibitem{Norbury199911}
P.~Norbury, \emph{{Asymptotic Values of Hyperbolic Monopoles}}, \href{https://doi.org/10.1017/s0024610701002186}{\emph{J. London Math. Soc.} {\bfseries 64} (Aug, 2001) 245--256}, [\href{https://arxiv.org/abs/arXiv:math/9911146}{{\ttfamily arXiv:math/9911146}}].

\bibitem{Chakrabarti198601}
A.~Chakrabarti, \emph{{Construction of Hyperbolic Monopoles}}, \href{https://doi.org/10.1063/1.527338}{\emph{J. Math. Phys.} {\bfseries 27} (Jan, 1986) 340}.

\bibitem{Nash198608}
C.~Nash, \emph{{Geometry of Hyperbolic Monopoles}}, \href{https://doi.org/10.1063/1.526985}{\emph{J. Math. Phys.} {\bfseries 27} (1986) 2160}.

\bibitem{Chan201506}
J.~Y. Chan, \emph{{Discrete Nahm Equations for {\normalfont $\SU{n}$} Hyperbolic Monopoles}}, \href{https://doi.org/10.1016/j.geomphys.2018.01.012}{\emph{J. Geom. Phys.} {\bfseries 132} (Oct, 2018) 239--256}, [\href{https://arxiv.org/abs/arXiv:1506.08736}{{\ttfamily arXiv:1506.08736}}].

\bibitem{Chan2017}
J.~Y.~C. Chan, \emph{\href{https://minerva-access.unimelb.edu.au/bitstream/handle/11343/129707/Thesis.pdf}{On Hyperbolic Monopoles}}, Ph.D. thesis, University of Melbourne, 2017.

\bibitem{ForgacsHorvathPalla198102}
P.~Forgacs, Z.~Horvath and L.~Palla, \emph{{Exact, Fractionally Charged Self-Dual Solution}}, \href{https://doi.org/10.1103/PhysRevLett.46.392}{\emph{Phys. Rev. Lett.} {\bfseries 46} (Feb, 1981) 392}.

\bibitem{ComtetForgacsHorvathy198407}
A.~Comtet, P.~Forgacs and P.~A. Horvathy, \emph{{Bogomolny-Type Equations in Curved Space-Time}}, \href{https://doi.org/10.1103/PhysRevD.30.468}{\emph{Phys. Rev. D} {\bfseries 30} (Jul, 1984) 468}.

\bibitem{MantonSutcliffe201207}
N.~S. Manton and P.~M. Sutcliffe, \emph{{Platonic Hyperbolic Monopoles}}, \href{https://doi.org/10.1007/s00220-013-1864-1}{\emph{Commun. Math. Phys.} {\bfseries 325} (Jan, 2014) 821--845}, [\href{https://arxiv.org/abs/1207.2636}{{\ttfamily 1207.2636}}].

\bibitem{Ward199811}
R.~S. Ward, \emph{{Two Integrable Systems Related to Hyperbolic Monopoles}}, \href{https://doi.org/10.4310/ajm.1999.v3.n1.a12}{\emph{Asian J. Math.} {\bfseries 3} (Jan, 1999) 325–332}, [\href{https://arxiv.org/abs/arXiv:solv-int/9811012}{{\ttfamily arXiv:solv-int/9811012}}].

\bibitem{Hitchin198212}
N.~J. {Hitchin}, \emph{{Monopoles and Geodesics}}, \href{https://doi.org/10.1007/BF01208717}{\emph{Comm. Math. Phys.} {\bfseries 83} (Dec, 1982) 579--602}.

\bibitem{Hitchin198302}
N.~J. Hitchin, \emph{{On the Construction of Monopoles}}, \href{https://doi.org/10.1007/bf01211826}{\emph{Comm. Math. Phys.} {\bfseries 89} (Jun, 1983) 145--190}.

\bibitem{Ward198109}
R.~S. Ward, \emph{{A Yang--Mills--Higgs Monopole of Charge 2}}, \href{https://doi.org/10.1007/BF01208497}{\emph{Commun. Math. Phys.} {\bfseries 79} (Sep, 1981) 317--325}.

\bibitem{Prasad198003}
M.~K. Prasad, \emph{{Exact Yang--Mills--Higgs Monopole Solutions of Arbitrary Topological Charge}}, \href{https://doi.org/10.1007/BF01213599}{\emph{Commun. Math. Phys.} {\bfseries 80} (Mar, 1981) 137}.

\bibitem{PrasadRossi198103}
M.~K. Prasad and P.~Rossi, \emph{{Construction of Exact Yang-Mills-Higgs Multimonopoles of Arbitrary Charge}}, \href{https://doi.org/10.1103/PhysRevLett.46.806}{\emph{Phys. Rev. Lett.} {\bfseries 46} (Mar, 1981) 806--809}.

\bibitem{Murray198306}
M.~K. Murray, \emph{{Monopoles and Spectral Curves for Arbitrary Lie Groups}}, \href{https://doi.org/10.1007/BF01205507}{\emph{Commun. Math. Phys.} {\bfseries 90} (Jun, 1983) 263--271}.

\bibitem{Murray198412}
M.~K. {Murray}, \emph{{Non-Abelian Magnetic Monopoles}}, \href{https://doi.org/10.1007/BF01212534}{\emph{Comm. Math. Phys.} {\bfseries 96} (Dec, 1984) 539--565}.

\bibitem{Murray1983}
M.~K. Murray, \emph{\href{http://www.maths.adelaide.edu.au/michael.murray/thesis.pdf}{Non-Abelian Magnetic Monopoles}}, Ph.D. thesis, Oxford University, 1983.

\bibitem{HurtubiseMurray1989}
J.~Hurtubise and M.~K. Murray, \emph{{On the Construction of Monopoles for the Classical Groups}}, \href{https://doi.org/10.1007/BF01221407}{\emph{Commun. Math. Phys.} {\bfseries 122} (Mar, 1989) 35--89}.

\bibitem{HurtubiseMurray1990}
J.~{Hurtubise} and M.~K. {Murray}, \emph{{Monopoles and their Spectral Data}}, \href{https://doi.org/10.1007/BF02097006}{\emph{Comm. Math. Phys.} {\bfseries 133} (Nov., 1990) 487--508}.

\bibitem{Penrose196702}
R.~Penrose, \emph{{Twistor Algebra}}, \href{https://doi.org/10.1063/1.1705200}{\emph{J. Math. Phys.} {\bfseries 8} (Feb, 1967) 345}.

\bibitem{Penrose196805}
R.~Penrose, \emph{{Twistor Quantization and Curved Space-Time}}, \href{https://doi.org/10.1007/BF00668831}{\emph{Int. J. Theor. Phys.} {\bfseries 1} (May, 1968) 61--99}.

\bibitem{Penrose196901}
R.~Penrose, \emph{{Solutions of the Zero-Rest-Mass Equations}}, \href{https://doi.org/10.1063/1.1664756}{\emph{J. Math. Phys.} {\bfseries 10} (Jan, 1969) 38--39}.

\bibitem{Eastwood198501}
M.~G. Eastwood, \emph{{The Generalized Penrose--Ward Transform}}, \href{https://doi.org/10.1017/s030500410006271x}{\emph{Math. Proc. Camb. Philos. Soc.} {\bfseries 97} (Jan, 1985) 165–187}.

\bibitem{Sutcliffe202008}
P.~Sutcliffe, \emph{{Spectral Curves of Hyperbolic Monopoles from ADHM}}, \href{https://doi.org/10.1088/1751-8121/abe5cc}{\emph{J. Phys. A} {\bfseries 54} (Mar, 2021) 165401}, [\href{https://arxiv.org/abs/2008.00435}{{\ttfamily 2008.00435}}].

\bibitem{Bethe193103}
H.~{Bethe}, \emph{{Zur Theorie der Metalle - I. Eigenwerte und Eigenfunktionen der Linearen Atomkette}}, \href{https://doi.org/10.1007/BF01341708}{\emph{Z. Phys.} {\bfseries 71} (Mar, 1931) 205--226}.

\bibitem{Onsager194310}
L.~Onsager, \emph{{Crystal Statistics. 1. A Two-Dimensional Model with an Order Disorder Transition}}, \href{https://doi.org/10.1103/PhysRev.65.117}{\emph{Phys. Rev.} {\bfseries 65} (Feb, 1944) 117--149}.

\bibitem{FaddeevSklyaninTakhtajan1980}
E.~K. Sklyanin, L.~A. Takhtajan and L.~D. Faddeev, \emph{{The Quantum Inverse Problem Method. I}}, \href{https://doi.org/10.1007/BF01018718}{\emph{Theor. Math. Phys.} {\bfseries 40} (Aug, 1979) 688--706}.

\bibitem{FaddeevTakhtajan1979}
L.~Faddeev and L.~Takhtajan, \emph{{The Quantum Method of the Inverse Problem and the Heisenberg XYZ Model}}, \href{https://doi.org/10.1070/rm1979v034n05abeh003909}{\emph{Russ. Math. Surveys} {\bfseries 34} (Oct, 1979) 11--68}.

\bibitem{Sklyanin1982}
E.~Sklyanin, \emph{{Quantum Version of the Method of Inverse Scattering Problem}}, \href{https://doi.org/10.1007/BF01091462}{\emph{J. Sov. Math.} {\bfseries 19} (Jul, 1982) 1546--1596}.

\bibitem{Lax196809}
P.~D. Lax, \emph{{Integrals of Nonlinear Equations of Evolution and Solitary Waves}}, \href{https://doi.org/10.1002/cpa.3160210503}{\emph{Commun. Pure Appl. Math.} {\bfseries 21} (Sep, 1968) 467--490}.

\bibitem{McGuire196405}
J.~B. {McGuire}, \emph{{Study of Exactly Soluble One-Dimensional $N$-Body Problems}}, \href{https://doi.org/10.1063/1.1704156}{\emph{J. Math. Phys.} {\bfseries 5} (May, 1964) 622--636}.

\bibitem{Yang196711}
C.-N. Yang, \emph{{Some Exact Results for the Many Body Problems in One Dimension with Repulsive Delta Function Interaction}}, \href{https://doi.org/10.1103/PhysRevLett.19.1312}{\emph{Phys. Rev. Lett.} {\bfseries 19} (Dec, 1967) 1312--1314}.

\bibitem{Yang196712}
C.-N. Yang, \emph{{S-Matrix for the One-Dimensional $N$-Body Problem with Repulsive or Attractive Delta-Function Interaction}}, \href{https://doi.org/10.1103/PhysRev.168.1920}{\emph{Phys. Rev.} {\bfseries 168} (Apr, 1968) 1920--1923}.

\bibitem{Baxter197203}
R.~J. Baxter, \emph{{Partition Function of the Eight-Vertex Lattice Model}}, \href{https://doi.org/10.1016/0003-4916(72)90335-1}{\emph{Annals Phys.} {\bfseries 70} (Mar, 1972) 193--228}.

\bibitem{Baxter197805}
R.~J. Baxter, \emph{{Solvable Eight-Vertex Model on an Arbitrary Planar Lattice}}, \href{https://doi.org/10.1098/rsta.1978.0062}{\emph{Phil. Trans. Roy. Soc. Lond. A} {\bfseries 289} (May, 1978) 315--346}.

\bibitem{Kundu199612}
A.~{Kundu}, \emph{{Quantum Integrable Systems: Construction, Solution, Algebraic Aspect}}, {\emph{arXiv e-prints} (Dec., 1996) }, [\href{https://arxiv.org/abs/hep-th/9612046}{{\ttfamily hep-th/9612046}}].

\bibitem{KulishReshetikhinSklyanin198109}
P.~P. Kulish, N.~Y. Reshetikhin and E.~K. Sklyanin, \emph{{Yang--Baxter Equation and Representation Theory: I}}, \href{https://doi.org/10.1007/bf02285311}{\emph{Lett. Math. Phys.} {\bfseries 5} (Sep, 1981) 393--403}.

\bibitem{BelavinDrinfeld198207}
A.~A. Belavin and V.~G. Drinfeld, \emph{{Solutions of the Classical Yang--Baxter Equation for Simple Lie Algebras}}, \href{https://doi.org/10.1007/BF01081585}{\emph{Funct. Anal. Its. Appl.} {\bfseries 16} (Jul, 1982) 1--29}.

\bibitem{BelavinDrinfeld198307}
A.~A. Belavin and V.~G. Drinfeld, \emph{{Classical Young--Baxter Equation for Simple Lie Algebras}}, \href{https://doi.org/10.1007/bf01078107}{\emph{Funct. Anal. Its. Appl.} {\bfseries 17} (Jul, 1984) 220--221}.

\bibitem{BelavinDrinfeld1998}
A.~A. Belavin and V.~G. Drinfeld, \emph{\href{https://www.amazon.com/Triangle-Equations-Algebras-Mathematics-Mathematical/dp/9057022699}{Triangle Equations and Simple Lie Algebras}}.
\newblock Harwood Academic, 1998.

\bibitem{KulishSklyanin198207}
P.~P. Kulish and E.~K. Sklyanin, \emph{{Solutions of the Yang--Baxter Equation}}, \href{https://doi.org/10.1007/bf01091463}{\emph{J. Sov. Math.} {\bfseries 19} (Jul, 1982) 1596--1620}.

\bibitem{BazhanovStroganov198207}
V.~V. Bazhanov and Y.~G. Stroganov, \emph{{Trigonometric and $S_n$-Symmetric Solutions of Triangle Equations with Variables on the Faces}}, \href{https://doi.org/10.1016/0550-3213(82)90075-x}{\emph{Nucl. Phys. B.} {\bfseries 205} (Jul, 1982) 505--526}.

\bibitem{Bazhanov198709}
V.~V. Bazhanov, \emph{{Integrable Quantum Systems and Classical Lie Algebras}}, \href{https://doi.org/10.1007/bf01221256}{\emph{Commun. Math. Phys.} {\bfseries 113} (Sep, 1987) 471--503}.

\bibitem{Bazhanov198710}
V.~V. Bazhanov, \emph{{Quantum R-Matrices and Matrix Generalizations of Trigonometric Functions}}, \href{https://doi.org/10.1007/bf01022960}{\emph{Theor. Math. Phys.} {\bfseries 73} (Oct, 1987) 1035--1039}.

\bibitem{Reshetikhin198710}
N.~Y. Reshetikhin, \emph{{The Spectrum of the Transfer Matrices Connected with Kac-Moody Algebras}}, \href{https://doi.org/10.1007/bf00416853}{\emph{Lett. Math. Phys.} {\bfseries 14} (Oct, 1987) 235--246}.

\bibitem{Stolin199112}
A.~Stolin, \emph{{On Rational Solutions of Yang-Baxter Equation for $\mfk{sl}(n)$.}}, \href{https://doi.org/10.7146/math.scand.a-12369}{\emph{Math. Scand.} {\bfseries 69} (Dec, 1991) 57}.

\bibitem{Kulish198611}
P.~P. Kulish, \emph{{Integrable Graded Magnets}}, \href{https://doi.org/10.1007/bf01083770}{\emph{J. Sov. Math."} {\bfseries 35} (Nov, 1986) 2648--2662}.

\bibitem{LeitesSerganova198401}
D.~A. Leites and V.~V. Serganova, \emph{{Solutions of the Classical Yang--Baxter Equation for Simple Superalgebras}}, \href{https://doi.org/10.1007/bf01031030}{\emph{Theor. Math. Phys.} {\bfseries 58} (Jan, 1984) 16--24}.

\bibitem{BazhanovShadrikov198712}
V.~V. Bazhanov and A.~G. Shadrikov, \emph{{Trigonometric Solutions of Triangle Equations. Simple Lie Superalgebras}}, \href{https://doi.org/10.1007/bf01041913}{\emph{Theor. Math. Phys.} {\bfseries 73} (Dec, 1987) 1302--1312}.

\bibitem{IshtiaqueMoosavianZhou202308}
N.~Ishtiaque, S.~F. Moosavian and Y.~Zhou, \emph{{Elliptic Stable Envelopes for Certain Non-Symplectic Varieties and Dynamical R-Matrices for Superspin Chains from the Bethe/Gauge Correspondence}}, \href{https://doi.org/10.3842/SIGMA.2024.099}{\emph{SIGMA} {\bfseries 20} (Oct, 2024) 099}, [\href{https://arxiv.org/abs/2308.12333}{{\ttfamily 2308.12333}}].

\bibitem{EtingofKazhdan199506}
P.~Etingof and D.~Kazhdan, \emph{{Quantization of Lie Bialgebras, I}}, \href{https://doi.org/10.1007/bf01587938}{\emph{Sel. Math.} {\bfseries 2} (Sep, 1996) 1–41}, [\href{https://arxiv.org/abs/arXiv:q-alg/9506005}{{\ttfamily arXiv:q-alg/9506005}}].

\bibitem{Martins201410}
M.~J. Martins, \emph{{An Integrable Nineteen Vertex Model Lying on a Hypersurface}}, \href{https://doi.org/10.1016/j.nuclphysb.2015.01.018}{\emph{Nucl. Phys. B} {\bfseries 892} (Jan, 2015) 306--336}, [\href{https://arxiv.org/abs/1410.6749}{{\ttfamily 1410.6749}}].

\bibitem{Shastry198801}
B.~{Sriram Shastry}, \emph{{Decorated Star-Triangle Relations and Exact Integrability of the One-Dimensional Hubbard Model}}, \href{https://doi.org/10.1007/BF01022987}{\emph{J. Stat. Phys} {\bfseries 50} (Jan, 1988) 57--79}.

\bibitem{Au-YangMcCoyBarryPerkTangYan1987}
H.~Au-Yang, B.~M. McCoy, J.~H.~H. Perk, S.~Tang and M.-L. Yan, \emph{{Commuting Transfer Matrices in the Chiral Potts Models: Solutions of Star Triangle Equations with Genus $>1$}}, \href{https://doi.org/10.1016/0375-9601(87)90065-X}{\emph{Phys. Lett. A} {\bfseries 123} (Aug, 1987) 219--223}.

\bibitem{BaxterPerkAu-Yang1988}
R.~J. Baxter, J.~H.~H. Perk and H.~Au-Yang, \emph{{New Solutions of the Star Triangle Relations for the Chiral Potts Model}}, \href{https://doi.org/10.1016/0375-9601(88)90896-1}{\emph{Phys. Lett. A} {\bfseries 128} (Mar, 1988) 138--142}.

\bibitem{BazhanovKashaevMangazeev1990}
V.~V. Bazhanov, R.~Kashaev and V.~Mangazeev, ``{\it\href{https://inis.iaea.org/collection/NCLCollectionStore/_Public/22/061/22061391.pdf}{$\mbb{Z}_N\times\mbb{Z}_N$ Generalization of Chiral Potts Model}}.'' Institute for High Energy Physics, Protvino, Russia, 1990.

\bibitem{BazhanovKashaevMangazeevStroganov199105}
V.~V. Bazhanov, R.~M. Kashaev, V.~V. Mangazeev and Y.~G. Stroganov, \emph{{${\mbb{Z}_{N}}^{\times {n-1}}$ Generalization of the Chiral Potts Model}}, \href{https://doi.org/10.1007/bf02099497}{\emph{Comm. Math. Phys.} {\bfseries 138} (May, 1991) 393--408}.

\bibitem{HowesKadanoffDenNijs198301}
S.~Howes, L.~P. Kadanoff and M.~Den~Nijs, \emph{{Quantum Model for Commensurate-Incommensurate Transitions}}, \href{https://doi.org/10.1016/0550-3213(83)90212-2}{\emph{Nucl. Phys. B} {\bfseries 215} (Jan, 1983) 169--208}.

\bibitem{vonGehlenRittenberg1985}
G.~von Gehlen and V.~Rittenberg, \emph{{$\mbb{Z}_N$-Symmetric Quantum Chains with an Infinite Set of Conserved Charges and $\mbb{Z}_N$ Zero Modes}}, \href{https://doi.org/10.1016/0550-3213(85)90350-5}{\emph{Nucl. Phys. B} {\bfseries 257} (1985) 351--370}.

\bibitem{ScottPearce198902}
H.~Scott and P.~Pearce, \emph{{Calculation of Intermolecular Interaction Strengths in the $P_{\beta'}$ Phase in Lipid Bilayers. Implications for Theoretical Models}}, \href{https://doi.org/10.1016/s0006-3495(89)82810-3}{\emph{Biophys. J.} {\bfseries 55} (Feb, 1989) 339--345}.

\bibitem{Baxter198808}
{Baxter, R.~J.}, \emph{{Free Energy of the Solvable Chiral Potts Model}}, \href{https://doi.org/10.1007/BF01019722}{\emph{J. Statist. Phys.} {\bfseries 52} (Aug, 1988) 639--667}.

\bibitem{AlbertiniMcCoyPerkTang198903}
G.~Albertini, B.~M. McCoy, J.~H.~H. Perk and S.~Tang, \emph{{Excitation Spectrum and Order Parameter for the Integrable $N$-State Chiral Potts Model}}, \href{https://doi.org/10.1016/0550-3213(89)90415-X}{\emph{Nucl. Phys. B} {\bfseries 314} (Mar, 1989) 741--763}.

\bibitem{BaxterBazhanovPerk199004}
R.~J. Baxter, V.~V. Bazhanov and J.~H.~H. Perk, \emph{{Functional Relations for Transfer Matrices of the Chiral Potts Model}}, \href{https://doi.org/10.1142/S0217979290000395}{\emph{Int. J. Mod. Phys. B} {\bfseries 4} (Apr, 1990) 803--870}.

\bibitem{BazhanovStroganov199005}
V.~V. Bazhanov and Y.~G. Stroganov, \emph{{Chiral Potts Model as a Descendant of the Six Vertex Model}}, \href{https://doi.org/10.1007/BF01025851}{\emph{J. Statist. Phys.} {\bfseries 59} (May, 1990) 799--817}.

\bibitem{Baxter199005}
R.~J. Baxter, \emph{{Chiral Potts Model: Eigenvalues of the Transfer Matrix}}, \href{https://doi.org/10.1016/0375-9601(90)90646-6}{\emph{Phys. Lett. A} {\bfseries 146} (May, 1990) 110--114}.

\bibitem{Baxter199302}
R.~J. Baxter, \emph{{Corner Transfer Matrices of the Chiral Potts Model. 2. The Triangular Lattice}}, \href{https://doi.org/10.1007/BF01053584}{\emph{J. Statist. Phys.} {\bfseries 70} (Feb, 1993) 535--582}.

\bibitem{Cardy199210}
J.~L. Cardy, \emph{{Critical Exponents of the Chiral Potts Model from Conformal Field Theory}}, \href{https://doi.org/10.1016/0550-3213(93)90353-Q}{\emph{Nucl. Phys. B} {\bfseries 389} (Jan, 1993) 577--586}, [\href{https://arxiv.org/abs/hep-th/9210002}{{\ttfamily hep-th/9210002}}].

\bibitem{Baxter200501}
R.~J. Baxter, \emph{{The Order Parameter of the Chiral Potts Model}}, \href{https://doi.org/10.1007/s10955-005-5534-3}{\emph{J. Statist. Phys.} {\bfseries 120} (Aug, 2005) 1--36}, [\href{https://arxiv.org/abs/cond-mat/0501226}{{\ttfamily cond-mat/0501226}}].

\bibitem{Baxter200504}
R.~J. Baxter, \emph{{Derivation of the Order Parameter of the Chiral Potts Model}}, \href{https://doi.org/10.1103/PhysRevLett.94.130602}{\emph{Phys. Rev. Lett.} {\bfseries 94} (Apr, 2005) 130602}, [\href{https://arxiv.org/abs/cond-mat/0501227}{{\ttfamily cond-mat/0501227}}].

\bibitem{Au-YangPerk1989}
H.~Au-Yang and J.~H. Perk, \emph{\href{https://projecteuclid.org/ebooks/advanced-studies-in-pure-mathematics/Integrable-Systems-in-Quantum-Field-Theory-and-Statistical-Mechanics/chapter/Onsagers-Star-Triangle-Equation-Master-Key-to-Integrability/10.2969/aspm/01910057?tab=ChapterArticleLink}{Onsager's Star-Triangle Equation: Master Key to Integrability}},  in \emph{\href{https://projecteuclid.org/proceedings/advanced-studies-in-pure-mathematics/integrable-systems-in-quantum-field-theory-and-statistical-mechanics/toc/10.2969/aspm/01910000}{Integrable Systems in Quantum Field Theory and Statistical Mechanics}} (M.~Jimbo, T.~Miwa and A.~Tsuchiya, eds.), pp.~57--94.
\newblock Elsevier, 1989.

\bibitem{AuYangPerk199609}
H.~{Au-Yang} and J.~H.~H. {Perk}, \emph{{The Many Faces of the Chiral Potts Model}}, \href{https://doi.org/10.1142/S0217979297000046}{\emph{Int. J. Mod. Phys. B} {\bfseries 11} (Jan, 1997) 11--26}, [\href{https://arxiv.org/abs/q-alg/9609003}{{\ttfamily q-alg/9609003}}].

\bibitem{Perk201511}
J.~H. Perk, \emph{{The Early History of the Integrable Chiral Potts Model and the Odd-Even Problem}}, \href{https://doi.org/10.1088/1751-8113/49/15/153001}{\emph{J. Phys. A} {\bfseries 49} (Mar, 2016) 153001}, [\href{https://arxiv.org/abs/arXiv:1511.08526}{{\ttfamily arXiv:1511.08526}}].

\bibitem{Au-YangPerk201601}
H.~Au-Yang and J.~H. Perk, \emph{{\href{https://projecteuclid.org/proceedings/proceedings-of-the-centre-for-mathematics-and-its-applications/Proceedings-of-the-2014-Maui-and-2015-Qinhuangdao-Conferences-in/Chapter/About-30-years-of-integrable-Chiral-Potts-model-quantum-groups/pcma/1487646022}{About 30 Years of Integrable Chiral Potts Model, Quantum Groups at Roots of Unity and Cyclic Hypergeometric Functions}}},  in \emph{\href{https://projecteuclid.org/proceedings/proceedings-of-the-centre-for-mathematics-and-its-applications/proceedings-of-the-2014-maui-and-2015-qinhuangdao-conferences-in-honour-of-vaughan-f-r-jones-60th-birthday/toc/pcma/1487646020}{Proceedings of the 2014 Maui and 2015 Qinhuangdao Conferences in Honour of Vaughan FR Jones’ 60th Birthday}} (S.~Morrison and D.~Penneys, eds.), pp.~1--14, Australian National University, Mathematical Sciences Institute, 2017, \href{https://arxiv.org/abs/arXiv:1601.01014}{{\ttfamily arXiv:1601.01014}}.

\bibitem{DateJimboKeiMiwa199008}
E.~Date, M.~Jimbo, K.~Miki and T.~Miwa, \emph{{R-Matrix for Cyclic Representations of $U_q(\widehat{\mfk{sl}}(3,\mbb{C}))$ at $q^3=1$}}, \href{https://doi.org/10.1016/0375-9601(90)90573-7}{\emph{Phys. Lett. A.} {\bfseries 148} (Aug, 1990) 45–49}.

\bibitem{DateJimboKeiMiwa199103}
E.~Date, M.~Jimbo, K.~Miki and T.~Miwa, \emph{{Generalized Chiral Potts Models and Minimal Cyclic Representations of $U_q(\widehat{\mfk{gl}}(n,\mbb{C}))$}}, \href{https://doi.org/10.1007/BF02099119}{\emph{Commun. Math. Phys.} {\bfseries 137} (Mar, 1991) 133--148}.

\bibitem{DateJimboMikiMiwa199104}
E.~Date, M.~Jimbo, K.~Miki and T.~Miwa, \emph{{Cyclic Representations of $U_q(\mfk{sl}(n+1,\C))$ at $q^N=1$}}, \href{https://doi.org/10.2977/prims/1195169842}{\emph{Publ. Res. Inst. Math. Sci.} {\bfseries 27} (Apr, 1991) 347–366}.

\bibitem{KashaevMangazeevNakanishi199109}
R.~Kashaev, V.~Mangazeev and T.~Nakanishi, \emph{{Yang--Baxter Equation for the $\mfk{sl}(n)$ Chiral Potts Model}}, \href{https://doi.org/10.1016/0550-3213(91)90542-6}{\emph{Nucl. Phys. B.} {\bfseries 362} (Sep, 1991) 563–582}.

\bibitem{BazhanovSergeev201006}
V.~V. Bazhanov and S.~M. Sergeev, \emph{{A Master Solution of the Quantum Yang--Baxter Equation and Classical Discrete Integrable Equations}}, \href{https://doi.org/10.4310/ATMP.2012.v16.n1.a3}{\emph{Adv. Theor. Math. Phys.} {\bfseries 16} (Jan, 2012) 65--95}, [\href{https://arxiv.org/abs/1006.0651}{{\ttfamily 1006.0651}}].

\bibitem{BazhanovSergeev201106}
V.~V. Bazhanov and S.~M. Sergeev, \emph{{Elliptic Gamma-Function and Multi-Spin Solutions of the Yang--Baxter Equation}}, \href{https://doi.org/10.1016/j.nuclphysb.2011.10.032}{\emph{Nucl. Phys. B} {\bfseries 856} (Mar, 2012) 475--496}, [\href{https://arxiv.org/abs/1106.5874}{{\ttfamily 1106.5874}}].

\bibitem{KelsYamazaki201709}
A.~P. Kels and M.~Yamazaki, \emph{{Lens Elliptic Gamma Function Solution of the Yang--Baxter Equation at Roots of Unity}}, \href{https://doi.org/10.1088/1742-5468/aaa8fd}{\emph{J. Stat. Mech.} {\bfseries 1802} (Feb, 2018) 023108}, [\href{https://arxiv.org/abs/1709.07148}{{\ttfamily 1709.07148}}].

\bibitem{Yamazaki201203}
M.~Yamazaki, \emph{{Quivers, YBE and 3-Manifolds}}, \href{https://doi.org/10.1007/JHEP05(2012)147}{\emph{J. High Energy Phys.} {\bfseries 05} (May, 2012) 147}, [\href{https://arxiv.org/abs/1203.5784}{{\ttfamily 1203.5784}}].

\bibitem{Terashima:2012cx}
Y.~Terashima and M.~Yamazaki, \emph{{Emergent 3-Manifolds from 4d Superconformal Indices}}, \href{https://doi.org/10.1103/PhysRevLett.109.091602}{\emph{Phys. Rev. Lett.} {\bfseries 109} (2012) 091602}, [\href{https://arxiv.org/abs/1203.5792}{{\ttfamily 1203.5792}}].

\bibitem{Yamazaki201307}
M.~Yamazaki, \emph{{New Integrable Models from the Gauge/YBE Correspondence}}, \href{https://doi.org/10.1007/s10955-013-0884-8}{\emph{J. Statist. Phys.} {\bfseries 154} (Nov, 2014) 895}, [\href{https://arxiv.org/abs/1307.1128}{{\ttfamily 1307.1128}}].

\bibitem{Yamazaki201808}
M.~Yamazaki, \emph{{Integrability As Duality: The Gauge/YBE Correspondence}}, \href{https://doi.org/10.1016/j.physrep.2020.01.006}{\emph{Phys. Rept.} {\bfseries 859} (May, 2020) 1--20}, [\href{https://arxiv.org/abs/1808.04374}{{\ttfamily 1808.04374}}].

\bibitem{DateJimboMikiMiwa199201}
E.~Date, M.~Jimbo, K.~Miki and T.~Miwa, \emph{{Braid Group Representations Arising from the Generalized Chiral Potts Models}}, \href{https://doi.org/10.2140/pjm.1992.154.37}{\emph{Pac. J. Math.} {\bfseries 154} (May, 1992) 37–66}.

\bibitem{Atiyah199106}
M.~Atiyah, \emph{{Magnetic Monopoles and the Yang--Baxter Equations}}, \href{https://doi.org/10.1142/s0217751x91001349}{\emph{Int. J. Mod. Phys.} {\bfseries A06} (Jul, 1991) 2761--2774}.

\bibitem{AtiyahMurray1995}
M.~F. Atiyah and M.~K. Murray, \emph{{Monopoles and Yang--Baxter Equations}},  in \emph{\href{https://www.routledge.com/Further-Advances-in-Twistor-Theory-Volume-II-Integrable-Systems-Conformal/Mason-Hughston/p/book/9780582004658}{Further Advances in Twistor Theory: Integrable Systems, Conformal Geometry and Gravitation}} (L.~J. Mason, L.~P. Hughston and P.~Z. Kobak, eds.), vol.~II of \emph{Monographs and Surveys in Pure and Applied Mathematics}, pp.~13--14.
\newblock CRC Press, 1st~ed., 1995.

\bibitem{Costello201303}
K.~Costello, \emph{{Supersymmetric Gauge Theory and the Yangian}},  \href{https://arxiv.org/abs/1303.2632}{{\ttfamily 1303.2632}}.

\bibitem{Witten201611}
E.~Witten, \emph{{Integrable Lattice Models from Gauge Theory}}, \href{https://doi.org/10.4310/ATMP.2017.v21.n7.a10}{\emph{Adv. Theor. Math. Phys.} {\bfseries 21} (Nov, 2017) 1819--1843}, [\href{https://arxiv.org/abs/1611.00592}{{\ttfamily 1611.00592}}].

\bibitem{CostelloWittenYamazaki201709}
K.~Costello, E.~Witten and M.~Yamazaki, \emph{{Gauge Theory and Integrability, I}}, \href{https://doi.org/10.4310/ICCM.2018.v6.n1.a6}{\emph{ICCM Not.} {\bfseries 06} (Jul, 2018) 46--119}, [\href{https://arxiv.org/abs/1709.09993}{{\ttfamily 1709.09993}}].

\bibitem{CostelloYamazakiWitten201802}
K.~Costello, E.~Witten and M.~Yamazaki, \emph{{Gauge Theory and Integrability, II}}, \href{https://doi.org/10.4310/ICCM.2018.v6.n1.a7}{\emph{ICCM Not.} {\bfseries 06} (Jul, 2018) 120--146}, [\href{https://arxiv.org/abs/1802.01579}{{\ttfamily 1802.01579}}].

\bibitem{CostelloYamazaki201908}
K.~Costello and M.~Yamazaki, \emph{{Gauge Theory And Integrability, III}},  \href{https://arxiv.org/abs/1908.02289}{{\ttfamily 1908.02289}}.

\bibitem{Vicedo201908}
B.~Vicedo, \emph{{4D Chern--Simons Theory and Affine Gaudin Models}}, \href{https://doi.org/10.1007/s11005-021-01354-9}{\emph{Lett. Math. Phys.} {\bfseries 111} (Feb, 2021) 24}, [\href{https://arxiv.org/abs/1908.07511}{{\ttfamily 1908.07511}}].

\bibitem{DelducLacroixMagroVicedo201909}
F.~Delduc, S.~Lacroix, M.~Magro and B.~Vicedo, \emph{{A Unifying 2D Action for Integrable $\sigma$-Models from 4D Chern--Simons Theory}}, \href{https://doi.org/10.1007/s11005-020-01268-y}{\emph{Lett. Math. Phys.} {\bfseries 110} (Feb, 2020) 1645--1687}, [\href{https://arxiv.org/abs/1909.13824}{{\ttfamily 1909.13824}}].

\bibitem{BeniniSchenkelVicedo202008}
M.~Benini, A.~Schenkel and B.~Vicedo, \emph{{Homotopical Analysis of 4d Chern--Simons Theory and Integrable Field Theories}}, \href{https://doi.org/10.1007/s00220-021-04304-7}{\emph{Commun. Math. Phys.} {\bfseries 389} (Jan, 2022) 1417--1443}, [\href{https://arxiv.org/abs/2008.01829}{{\ttfamily 2008.01829}}].

\bibitem{LiniadoVicedo202301}
J.~Liniado and B.~Vicedo, \emph{{Integrable Degenerate $\mscr{E}$-Models from 4d Chern--Simons Theory}}, \href{https://doi.org/10.1007/s00023-023-01317-x}{\emph{Ann. Henri Poincar\'e} {\bfseries 24} (Apr, 2023) 3421--3459}, [\href{https://arxiv.org/abs/2301.09583}{{\ttfamily 2301.09583}}].

\bibitem{FukushimaSakamotoYoshida202003}
O.~Fukushima, J.-i. Sakamoto and K.~Yoshida, \emph{{Comments on $\eta$-Deformed Principal Chiral Model from 4D Chern--Simons Theory}}, \href{https://doi.org/10.1016/j.nuclphysb.2020.115080}{\emph{Nucl. Phys. B} {\bfseries 957} (Aug, 2020) 115080}, [\href{https://arxiv.org/abs/2003.07309}{{\ttfamily 2003.07309}}].

\bibitem{CostelloStefanski202005}
K.~Costello and B.~Stefa\'nski, \emph{{Chern-Simons Origin of Superstring Integrability}}, \href{https://doi.org/10.1103/PhysRevLett.125.121602}{\emph{Phys. Rev. Lett.} {\bfseries 125} (Sep, 2020) 121602}, [\href{https://arxiv.org/abs/2005.03064}{{\ttfamily 2005.03064}}].

\bibitem{Tian202005}
J.~Tian, \emph{{Comments on $\lambda$--Deformed Models from 4D Chern--Simons Theory}},  \href{https://arxiv.org/abs/2005.14554}{{\ttfamily 2005.14554}}.

\bibitem{TianHeChen202007}
J.~Tian, Y.-J. He and B.~Chen, \emph{{$\lambda$-Deformed AdS$_5\times S^5$ Superstring from 4D Chern--Simons Theory}}, \href{https://doi.org/10.1016/j.nuclphysb.2021.115545}{\emph{Nucl. Phys. B} {\bfseries 972} (Nov, 2021) 115545}, [\href{https://arxiv.org/abs/2007.00422}{{\ttfamily 2007.00422}}].

\bibitem{Stedman202009}
J.~Stedman, \emph{{Four-Dimensional Chern--Simons and Gauged Sigma Models}},  \href{https://arxiv.org/abs/2109.08101}{{\ttfamily 2109.08101}}.

\bibitem{FukushimaSakamotoYoshida202105}
O.~Fukushima, J.-i. Sakamoto and K.~Yoshida, \emph{{Integrable Deformed T$^{1,1}$ Sigma Models from 4D Chern--Simons Theory}}, \href{https://doi.org/10.1007/JHEP09(2021)037}{\emph{J. High Energy Phys.} {\bfseries 09} (Sep, 2021) 037}, [\href{https://arxiv.org/abs/2105.14920}{{\ttfamily 2105.14920}}].

\bibitem{FukushimaSakamotoYoshida202112}
O.~Fukushima, J.-i. Sakamoto and K.~Yoshida, \emph{{Non-Abelian Toda Field Theories from a 4D Chern--Simons Theory}}, \href{https://doi.org/10.1007/JHEP03(2022)158}{\emph{J. High Energy Phys.} {\bfseries 03} (Mar, 2022) 158}, [\href{https://arxiv.org/abs/2112.11276}{{\ttfamily 2112.11276}}].

\bibitem{HeTianChen202105}
Y.-J. He, J.~Tian and B.~Chen, \emph{{Deformed Integrable Models from Holomorphic Chern--Simons Theory}}, \href{https://doi.org/10.1007/s11433-022-1931-x}{\emph{Sci. China Phys. Mech. Astron.} {\bfseries 65} (Sep, 2022) 100413}, [\href{https://arxiv.org/abs/2105.06826}{{\ttfamily 2105.06826}}].

\bibitem{Levine202309}
N.~Levine, \emph{{Equivalence of 1-Loop RG Flows in 4d Chern--Simons and Integrable 2d Sigma-Models}},  \href{https://arxiv.org/abs/2309.16753}{{\ttfamily 2309.16753}}.

\bibitem{AshwinkumarSakamotoYamazaki202309}
M.~Ashwinkumar, J.-i. Sakamoto and M.~Yamazaki, \emph{{Dualities and Discretizations of Integrable Quantum Field Theories from 4d Chern--Simons Theory}},  \href{https://arxiv.org/abs/2309.14412}{{\ttfamily 2309.14412}}.

\bibitem{ColeWeck202407}
L.~T. Cole and P.~Weck, \emph{{Integrability in Gravity from Chern--Simons Theory}}, \href{https://doi.org/10.1007/JHEP10(2024)080}{\emph{J. High Energy Phys.} {\bfseries 10} (Oct, 2024) 080}, [\href{https://arxiv.org/abs/2407.08782}{{\ttfamily 2407.08782}}].

\bibitem{NorburyRomao200512}
P.~Norbury and N.~M. Romao, \emph{{Spectral Curves and the Mass of Hyperbolic Monopoles}}, \href{https://doi.org/10.1007/s00220-006-0148-4}{\emph{Commun. Math. Phys.} {\bfseries 270} (Dec, 2007) 295--333}, [\href{https://arxiv.org/abs/math-ph/0512083}{{\ttfamily math-ph/0512083}}].

\bibitem{Weinberg198005}
E.~J. {Weinberg}, \emph{{Fundamental Monopoles and Multimonopole Solutions for Arbitrary Simple Gauge Groups}}, \href{https://doi.org/10.1016/0550-3213(80)90245-X}{\emph{Nucl. Phys.} {\bfseries B167} (May, 1980) 500--524}.

\bibitem{Weinberg198208}
E.~J. {Weinberg}, \emph{{Fundamental Monopoles in Theories with Arbitrary Symmetry Breaking}}, \href{https://doi.org/10.1016/0550-3213(82)90324-8}{\emph{Nucl. Phys.} {\bfseries B203} (Aug., 1982) 445--471}.

\bibitem{GoddardNuytsOlive1977}
P.~Goddard, J.~Nuyts and D.~I. Olive, \emph{{Gauge Theories and Magnetic Charge}}, \href{https://doi.org/10.1016/0550-3213(77)90221-8}{\emph{Nucl. Phys.} {\bfseries B125} (Jul, 1977) 1--28}.

\bibitem{GoddardOlive197809}
P.~{Goddard} and D.~I. {Olive}, \emph{{Magnetic Monopoles in Gauge Field Theories}}, \href{https://doi.org/10.1088/0034-4885/41/9/001}{\emph{Rep. Prog. Phys.} {\bfseries 41} (Sep, 1978) 1357--1437}.

\bibitem{CorriganGoddard198108}
E.~Corrigan and P.~Goddard, \emph{{An $n$-Monopole Solution with $4n-1$ Degrees of Freedom}}, \href{https://doi.org/10.1007/bf01941665}{\emph{Commun. Math. Phys} {\bfseries 80} (Aug, 1981) 575--587}.

\bibitem{WeinbergYi200609}
E.~J. Weinberg and P.~Yi, \emph{{Magnetic Monopole Dynamics, Supersymmetry, and Duality}}, \href{https://doi.org/10.1016/j.physrep.2006.11.002}{\emph{Phys. Rept.} {\bfseries 438} (Jan, 2007) 65--236}, [\href{https://arxiv.org/abs/hep-th/0609055}{{\ttfamily hep-th/0609055}}].

\bibitem{BernsteinGelfandGelfand197306}
I.~N. Bernstein, I.~M. Gel’fand and S.~I. Gel’fand, \emph{{Schubert Cells and Cohomology of the Spaces $G/P$}}, \href{https://doi.org/10.1070/rm1973v028n03abeh001557}{\emph{Russ. Math. Surv.} {\bfseries 28} (Jun, 1973) 1--26}.

\bibitem{Georgiou2009}
N.~Georgiou, \emph{\href{https://sword.cit.ie/cgi/viewcontent.cgi?article=1348&context=allthe}{The Geometry Of the Space Of Oriented Geodesics Of Hyperbolic 3-Space}}, Ph.D. thesis, Munster Technological University, 2009.

\bibitem{GeorgiouGuilfoyle200702}
N.~Georgiou and B.~Guilfoyle, \emph{{On the Space of Oriented Geodesics of Hyperbolic 3-Space}}, \href{https://doi.org/10.1216/rmj-2010-40-4-1183}{\emph{Rocky Mt. J. Math.} {\bfseries 40} (Aug, 2010) }, [\href{https://arxiv.org/abs/arXiv:math/0702276}{{\ttfamily arXiv:math/0702276}}].

\bibitem{Wells197903}
R.~O. Wells, \emph{{Complex Manifolds and Mathematical Physics}}, \href{https://doi.org/10.1090/s0273-0979-1979-14596-8}{\emph{Bull. Am. Math. Soc.} {\bfseries 1} (Mar, 1979) 296–336}.

\bibitem{EastwoodPenroseWells198101}
M.~G. Eastwood, R.~Penrose and R.~O. Wells, \emph{{Cohomology and Massless Fields}}, \href{https://doi.org/10.1007/BF01942327}{\emph{Commun. Math. Phys.} {\bfseries 78} (Jan, 1981) 305--351}.

\bibitem{Grothendieck195701}
A.~Grothendieck, \emph{{Sur La Classification des Fibres Holomorphes Sur La Sphere de Riemann}}, \href{https://doi.org/10.2307/2372388}{\emph{Am. J. Math.} {\bfseries 79} (Jan, 1957) 121}.

\bibitem{SibnerSibner199202}
L.~M. Sibner and R.~J. Sibner, \emph{{Classification of Singular Sobolev Connections by their Holonomy}}, \href{https://doi.org/10.1007/bf02101096}{\emph{Comm. Math. Phys.} {\bfseries 144} (Feb, 1992) 337--350}.

\bibitem{SibnerSibner1995}
L.~M. Sibner and R.~J. Sibner, ``{\it\href{https://archive.mpim-bonn.mpg.de/id/eprint/668/1/preprint_1995_44.pdf}{Existence of Hyperbolic Monopoles with Arbitrary Mass at Infinity}}.'' Online, 1995.

\bibitem{SibnerSibner201210}
L.~M. Sibner and R.~J. Sibner, \emph{{Hyperbolic Multi-Monopoles with Arbitrary Mass}}, \href{https://doi.org/10.1007/s00220-012-1562-4}{\emph{Commun. Math. Phys.} {\bfseries 315} (2012) 383--399}, [\href{https://arxiv.org/abs/arXiv:1210.0856}{{\ttfamily arXiv:1210.0856}}].

\bibitem{JaffeTaubes1980}
A.~M. Jaffe and C.~H. Taubes, \emph{\href{https://onlinelibrary.wiley.com/doi/10.1002/zamm.19820620624}{Vortices and Monopoles: Structure of Static Gauge Theories}}.
\newblock Progress in Physics 2. Birkhäuser Verlag, 1980.

\bibitem{Mendizabal202306}
J.~Mendizabal, \emph{{A Hyper-K\"ahler Metric on the Moduli Spaces of Monopoles with Arbitrary Symmetry Breaking}}, \href{https://doi.org/10.1007/s10455-024-09954-z}{\emph{Ann. Glob. Anal. Geom.} {\bfseries 66} (Jul, 2024) }, [\href{https://arxiv.org/abs/arXiv:2306.03069}{{\ttfamily arXiv:2306.03069}}].

\bibitem{Mendizabal2024}
J.~Mendizabal, \emph{\href{https://discovery.ucl.ac.uk/id/eprint/10196914/}{On Monopoles with Arbitrary Symmetry Breaking and their Moduli Spaces}}, Ph.D. thesis, University Coll. London, 2024.

\bibitem{ShnirZhilin201508}
Y.~Shnir and G.~Zhilin, \emph{{$G_2$ Monopoles}}, \href{https://doi.org/10.1103/PhysRevD.92.045025}{\emph{Phys. Rev. D} {\bfseries 92} (Aug, 2015) 045025}, [\href{https://arxiv.org/abs/1508.01871}{{\ttfamily 1508.01871}}].

\bibitem{Shnir2005}
Y.~M. Shnir, \emph{\href{https://link.springer.com/book/10.1007/3-540-29082-6}{Magnetic Monopoles}}.
\newblock Theoretical and Mathematical Physics. Springer, Berlin/Heidelberg, 2005.

\bibitem{Costello201308}
K.~Costello, \emph{\href{http://www.ams.org/books/pspum/088/}{Integrable Lattice Models from Four-Dimensional Field Theories}},  in \emph{Proceedings of String-Math 2013} (R.~Donagi, M.~R. Douglas, L.~Kamenova and M.~Rocek, eds.), vol.~88, pp.~3--24, 2014, \href{https://arxiv.org/abs/1308.0370}{{\ttfamily 1308.0370}}.

\bibitem{Witten201101}
E.~Witten, \emph{{Fivebranes and Knots}}, \href{https://doi.org/10.4171/qt/26}{\emph{Quantum Topol.} {\bfseries 3} (Nov, 2011) 1–137}, [\href{https://arxiv.org/abs/arXiv:1101.3216}{{\ttfamily arXiv:1101.3216}}].

\bibitem{CostelloYagi201810}
K.~Costello and J.~Yagi, \emph{{Unification of Integrability in Supersymmetric Gauge Theories}}, \href{https://doi.org/10.4310/ATMP.2020.v24.n8.a1}{\emph{Adv. Theor. Math. Phys.} {\bfseries 24} (Dec, 2020) 1931--2041}, [\href{https://arxiv.org/abs/1810.01970}{{\ttfamily 1810.01970}}].

\bibitem{AshwinkumarTanZhao201806}
M.~Ashwinkumar, M.-C. Tan and Q.~Zhao, \emph{{Branes and Categorifying Integrable Lattice Models}}, \href{https://doi.org/10.4310/ATMP.2020.v24.n1.a1}{\emph{Adv. Theor. Math. Phys.} {\bfseries 24} (Jan, 2020) 1--24}, [\href{https://arxiv.org/abs/1806.02821}{{\ttfamily 1806.02821}}].

\bibitem{IshtiaqueMoosavianRaghavendranYagi202110}
N.~Ishtiaque, S.~F. Moosavian, S.~Raghavendran and J.~Yagi, \emph{{Superspin Chains from Superstring Theory}}, \href{https://doi.org/10.21468/SciPostPhys.13.4.083}{\emph{SciPost Phys.} {\bfseries 13} (Oct, 2022) 083}, [\href{https://arxiv.org/abs/2110.15112}{{\ttfamily 2110.15112}}].

\bibitem{Wu197110}
F.~W. Wu, \emph{{Ising Model with Four-Spin Interactions}}, \href{https://doi.org/10.1103/PhysRevB.4.2312}{\emph{Phys. Rev. B} {\bfseries 4} (Oct, 1971) 2312--2314}.

\bibitem{KadanoffWegner197112}
L.~P. Kadanoff and F.~J. Wegner, \emph{{Some Critical Properties of the Eight-Vertex Model}}, \href{https://doi.org/10.1103/PhysRevB.4.3989}{\emph{Phys. Rev. B} {\bfseries 4} (Dec, 1971) 3989--3993}.

\bibitem{Felder199407}
G.~Felder, \emph{{Conformal Field Theory and Integrable Systems Associated to Elliptic Curves}},  in \emph{\href{https://link.springer.com/chapter/10.1007/978-3-0348-9078-6_119}{Proceedings of the International Congress of Mathematicians}}, pp.~1247--1255.
\newblock Birkhäuser Basel, 1995.
\newblock \href{https://arxiv.org/abs/arXiv:hep-th/9407154}{{\ttfamily arXiv:hep-th/9407154}}.

\bibitem{Felder199412}
G.~Felder, \emph{{Elliptic Quantum Groups}},  in \emph{{11th International Conference on Mathematical Physics (ICMP-11) (Satellite colloquia: New Problems in the General Theory of Fields and Particles, Paris, France, 25-28 Jul 1994)}} (D.~Iagolnitzer, ed.), pp.~211--218, 7, 1994, \href{https://arxiv.org/abs/hep-th/9412207}{{\ttfamily hep-th/9412207}}.

\bibitem{Witten199207}
E.~Witten, \emph{\href{https://link.springer.com/chapter/10.1007/978-3-0348-9217-9_28}{Chern--Simons Gauge Theory as a String Theory}},  in \emph{{The Floer Memorial Volume}} (H.~Hofer, C.~H. Taubes, A.~Weinstein and E.~Zehnder, eds.), vol.~133 of \emph{Progress in Mathemaics}, pp.~637--678.
\newblock Birkhäuser, 1995.
\newblock \href{https://arxiv.org/abs/hep-th/9207094}{{\ttfamily hep-th/9207094}}.

\bibitem{BarberisDottiVerbitsky200712}
M.~L. Barberis, I.~G. Dotti and M.~Verbitsky, \emph{{Canonical Bundles of Complex Nilmanifolds, with Applications to Hypercomplex Geometry}}, \href{https://doi.org/10.4310/mrl.2009.v16.n2.a10}{\emph{Math. Res. Lett.} {\bfseries 16} (Mar, 2009) 331--347}, [\href{https://arxiv.org/abs/arXiv:0712.3863}{{\ttfamily arXiv:0712.3863}}].

\bibitem{Tosatti201401}
V.~Tosatti, \emph{{Non-K\"ahler Calabi--Yau Manifolds}},  in \emph{\href{http://www.ams.org/books/conm/644/}{Analysis, Complex Geometry, and Mathematical Physics: In Honor of Duong H. Phong}} (P.~M.~N. Feehan, J.~Song, B.~Weinkove and R.~A. Wentworth, eds.), pp.~261--277.
\newblock American Mathematical Society, 2015.
\newblock \href{https://arxiv.org/abs/arXiv:1401.4797}{{\ttfamily arXiv:1401.4797}}.

\bibitem{DonaldsonThomas199606}
S.~K. Donaldson and R.~P. Thomas, \emph{\href{https://www.ma.imperial.ac.uk/~rpwt/skd.pdf}{Gauge Theory in Higher Dimensions}},  in \emph{{Conference on Geometric Issues in Foundations of Science in honor of Sir Roger Penrose's 65th Birthday}}, pp.~31--47, 6, 1996.

\bibitem{WardWells199108}
R.~S. Ward and R.~O. Wells, \emph{\href{https://www.cambridge.org/core/books/twistor-geometry-and-field-theory/B1E6211CC3D935029ABD1D30B68B9360}{Twistor Geometry and Field Theory}}.
\newblock Cambridge Monographs on Mathematical Physics. Cambridge University Press, 8, 1991.

\bibitem{Adamo201712}
T.~Adamo, \emph{{Lectures on Twistor Theory}},  in \emph{\href{https://pos.sissa.it/323/}{Proceedings of XIII Modave Summer School in Mathematical Physics {\textemdash} PoS(Modave2017)}}, vol.~323, p.~003, PoS Proc. Sci., Mar, 2018, \href{https://arxiv.org/abs/1712.02196}{{\ttfamily 1712.02196}}.

\bibitem{Porter198305}
J.~R. Porter, \emph{{Self-Dual Yang--Mills Fields on Minkowski Space-Time}}, \href{https://doi.org/10.1063/1.525802}{\emph{J. Math. Phys.} {\bfseries 24} (May, 1983) 1233–1239}.

\bibitem{Moroianu200703}
A.~Moroianu, \emph{\href{https://www.cambridge.org/core/books/lectures-on-kahler-geometry/EB2FCE9AF46904A6651920E161EB339F}{Lectures on K\"ahler Geometry}}.
\newblock Cambridge University Press, Mar, 2007.

\bibitem{BanicaPutinar198706}
C.~B\^anic\^a and M.~Putinar, \emph{{On Complex Vector Bundles on Projective Threefolds}}, \href{https://doi.org/10.1007/bf01388917}{\emph{Invent. Math.} {\bfseries 88} (Jun, 1987) 427–438}.

\bibitem{Popov199803}
A.~D. Popov, \emph{{Selfdual Yang--Mills: Symmetries and Moduli Space}}, \href{https://doi.org/10.1142/S0129055X99000350}{\emph{Rev. Math. Phys.} {\bfseries 11} (Oct, 1999) 1091--1149}, [\href{https://arxiv.org/abs/hep-th/9803183}{{\ttfamily hep-th/9803183}}].

\bibitem{Popov199806}
A.~D. Popov, \emph{{Holomorphic Chern--Simons--Witten Theory: From 2D to 4D Conformal Field Theories}}, \href{https://doi.org/10.1016/S0550-3213(99)00227-8}{\emph{Nucl. Phys. B} {\bfseries 550} (Jun, 1999) 585--621}, [\href{https://arxiv.org/abs/hep-th/9806239}{{\ttfamily hep-th/9806239}}].

\bibitem{Costello202004}
K.~Costello, ``{\it\href{https://youtu.be/ZlDNpPHvA8A?t=103}{Topological Strings, Twistors, and Skyrmions}}.'' Talk at the Western Hemisphere Colloquium on Geometry and Physics, Apr, 2020.

\bibitem{Yang197706}
C.~N. Yang, \emph{{Condition of Selfduality for SU(2) Gauge Fields on Euclidean Four-Dimensional Space}}, \href{https://doi.org/10.1103/PhysRevLett.38.1377}{\emph{Phys. Rev. Lett.} {\bfseries 38} (Jun, 1977) 1377}.

\bibitem{BittlestonSkinner202011}
R.~Bittleston and D.~Skinner, \emph{{Twistors, the ASD Yang--Mills Equations and 4d Chern--Simons Theory}}, \href{https://doi.org/10.1007/JHEP02(2023)227}{\emph{J. High Energy Phys.} {\bfseries 02} (2023) 227}, [\href{https://arxiv.org/abs/2011.04638}{{\ttfamily 2011.04638}}].

\bibitem{Yang197004}
C.-N. Yang, \emph{{Charge Quantization, Compactness of the Gauge Group, and Flux Quantization}}, \href{https://doi.org/10.1103/PhysRevD.1.2360}{\emph{Phys. Rev. D} {\bfseries 1} (Apr, 1970) 2360}.

\bibitem{Popov199909}
A.~D. Popov, \emph{{Holomorphic Analogs of Topological Gauge Theories}}, \href{https://doi.org/10.1016/S0370-2693(99)01464-1}{\emph{Phys. Lett. B} {\bfseries 473} (Jan, 2000) 65--72}, [\href{https://arxiv.org/abs/hep-th/9909135}{{\ttfamily hep-th/9909135}}].

\bibitem{IvanovaPopov200002}
T.~A. Ivanova and A.~D. Popov, \emph{{Dressing Symmetries of Holomorphic BF Theories}}, \href{https://doi.org/10.1063/1.533261}{\emph{J. Math. Phys.} {\bfseries 41} (May, 2000) 2604--2615}, [\href{https://arxiv.org/abs/hep-th/0002120}{{\ttfamily hep-th/0002120}}].

\bibitem{BaulieuTanzini200412}
L.~Baulieu and A.~Tanzini, \emph{{Topological Symmetry of Forms, $N=1$ Supersymmetry and S-Duality on Special Manifolds}}, \href{https://doi.org/10.1016/j.geomphys.2005.12.006}{\emph{J. Geom. Phys.} {\bfseries 56} (Nov, 2006) 2379--2401}, [\href{https://arxiv.org/abs/hep-th/0412014}{{\ttfamily hep-th/0412014}}].

\bibitem{CostelloLi201606}
K.~Costello and S.~Li, \emph{{Twisted Supergravity and its Quantization}},  \href{https://arxiv.org/abs/1606.00365}{{\ttfamily 1606.00365}}.

\bibitem{Baulieu201009}
L.~Baulieu, \emph{{{\normalfont $\SU{5}$}-Invariant Decomposition of Ten-Dimensional Yang--Mills Supersymmetry}}, \href{https://doi.org/10.1016/j.physletb.2010.12.044}{\emph{Phys. Lett. B} {\bfseries 698} (Mar, 2011) 63--67}, [\href{https://arxiv.org/abs/1009.3893}{{\ttfamily 1009.3893}}].

\bibitem{Zwiebach199606}
B.~Zwiebach, \emph{{Building String Field Theory Around Nonconformal Backgrounds}}, \href{https://doi.org/10.1016/S0550-3213(96)00502-0}{\emph{Nucl. Phys. B} {\bfseries 480} (Dec, 1996) 541--572}, [\href{https://arxiv.org/abs/hep-th/9606153}{{\ttfamily hep-th/9606153}}].

\bibitem{AuYangPerk199305}
H.~Au-Yang and J.~H.~H. Perk, \emph{{The $N\to\infty$ Limit of the Chiral Potts Model}},  in \emph{\href{https://www.worldscientific.com/worldscibooks/10.1142/1863}{Differential Geometric Methods in Theoretical Physics - Proceedings of the XXI International Conference}} (C.~N. Yang, M.~L. Ge and X.~W. Zhou, eds.), p.~624pp, Sep, 1993, \href{https://arxiv.org/abs/hep-th/9305171}{{\ttfamily hep-th/9305171}}.

\bibitem{AuYangPerk199906}
H.~Au-Yang and J.~H.~H. Perk, \emph{{The Large-$N$ Limits of the Chiral Potts Model}}, \href{https://doi.org/10.1016/S0378-4371(98)00386-0}{\emph{Physica A} {\bfseries 268} (Jun, 1999) 175--206}, [\href{https://arxiv.org/abs/math/9906029}{{\ttfamily math/9906029}}].

\bibitem{Tarasov199204}
V.~O. Tarasov, \emph{{Cyclic Monodromy Matrices for the R-Matrix of the Six Vertex Model and the Chiral Potts Model with Fixed Spin Boundary Conditions}}, \href{https://doi.org/10.1142/S0217751X92004129}{\emph{Int. J. Mod. Phys. A} {\bfseries 7S1B} (Apr, 1992) 963--975}.

\bibitem{Tarasov199211}
V.~Tarasov, \emph{{Cyclic Monodromy Matrices for $\mfk{sl}(n)$ Trigonometric R-Matrices}}, \href{https://doi.org/10.1007/BF02096799}{\emph{Commun. Math. Phys.} {\bfseries 158} (Dec, 1993) 459--484}, [\href{https://arxiv.org/abs/hep-th/9211105}{{\ttfamily hep-th/9211105}}].

\bibitem{DeConciniKac1990}
C.~De~Concini and V.~G. Kac, \emph{\href{https://www1.mat.uniroma1.it/people/deconcini/DeConcini-Kac.pdf}{Representations of Quantum Groups at Roots of $1$}},  in \emph{\href{https://worldscientific.com/worldscibooks/10.1142/1085}{Modern Quantum Field Theory}{ - Proceedings of the International Colloquium on Modern Quantum Field Theory, TIFR, Bombay, India, 8 – 14 January 1990}}, pp.~333--335, World Scientific Publishing River Edge, NJ, USA, 1991.

\bibitem{ChariPressley1991}
V.~Chari and A.~Pressley, \emph{{Minimal Cyclic Representations of Quantum Groups at Roots of Unity}}, {\emph{\href{https://gallica.bnf.fr/ark:/12148/bpt6k57325582/f433}{C. R. Acad. Sci. Paris Sr. I Math.}} {\bfseries \href{https://gallica.bnf.fr/ark:/12148/bpt6k57325582/f433}{313}} (\href{https://gallica.bnf.fr/ark:/12148/bpt6k57325582/f433}{{1991}}) \href{https://gallica.bnf.fr/ark:/12148/bpt6k57325582/f433}{429}}.

\bibitem{Schnizer199201}
A.~W. Schnizer, \emph{{Roots of Unity: Representations for Symplectic and Orthogonal Quantum Groups}}, \href{https://doi.org/10.1063/1.530004}{\emph{J. Math. Phys.} {\bfseries 34} (Jul, 1993) 4340--4363}.

\bibitem{Schnizer199305}
W.~A. Schnizer, \emph{{Roots of Unity: Representations of Quantum Groups}}, \href{https://doi.org/10.1007/BF02102010}{\emph{Commun. Math. Phys.} {\bfseries 163} (1994) 293--306}, [\href{https://arxiv.org/abs/hep-th/9305180}{{\ttfamily hep-th/9305180}}].

\bibitem{NairSchiff199008}
V.~P. Nair and J.~Schiff, \emph{{A K\"ahler--Chern--Simons Theory and Quantization of Instanton Moduli Spaces}}, \href{https://doi.org/10.1016/0370-2693(90)90624-F}{\emph{Phys. Lett. B} {\bfseries 246} (Aug, 1990) 423--429}.

\bibitem{NairSchiff199203}
V.~P. Nair and J.~Schiff, \emph{{K\"ahler Chern--Simons Theory and Symmetries of Antiselfdual Gauge Fields}}, \href{https://doi.org/10.1016/0550-3213(92)90239-8}{\emph{Nucl. Phys. B} {\bfseries 371} (Mar, 1992) 329--352}.

\bibitem{BoelsMasonSkinner200604}
R.~Boels, L.~J. Mason and D.~Skinner, \emph{{Supersymmetric Gauge Theories in Twistor Space}}, \href{https://doi.org/10.1088/1126-6708/2007/02/014}{\emph{J. High Energy Phys.} {\bfseries 02} (Feb, 2007) 014}, [\href{https://arxiv.org/abs/hep-th/0604040}{{\ttfamily hep-th/0604040}}].

\bibitem{Donaldson198409}
S.~K. Donaldson, \emph{{Nahm's Equations and the Classification of Monopoles}}, \href{https://doi.org/10.1007/BF01214583}{\emph{Commun. Math. Phys.} {\bfseries 96} (Sep, 1984) 387--407}.

\bibitem{Nahm198109}
W.~Nahm, \emph{\href{https://lib-extopc.kek.jp/preprints/PDF/1981/8111/8111193.pdf}{All Self-Dual Multimonopoles for Arbitrary Gauge Group}},  in \emph{\href{https://link.springer.com/book/10.1007/978-1-4613-3509-2}{12th NATO Advanced Summer Institute on Theoretical Physics: Structural Elements in Particle Physics and Statistical Mechanics}} (J.~Honerkamp, K.~Pohlmeyer and H.~H.~Römer, eds.), p.~378 pp, 9, 1981.

\bibitem{Nahm198211}
W.~Nahm, \emph{\href{https://lib-extopc.kek.jp/preprints/PDF/1982/8212/8212065.pdf}{The Algebraic Geometry of Multimonopoles}},  in \emph{\href{https://link.springer.com/book/10.1007/3-540-12291-5}{11th International Colloquium on Group Theoretical Methods in Physics}} (M.~Serdaroglu and E.~Inoenue, eds.), p.~569 pp, 11, 1982.

\bibitem{Diaconescu199608}
D.-E. Diaconescu, \emph{{D-Branes, Monopoles and Nahm Equations}}, \href{https://doi.org/10.1016/S0550-3213(97)00438-0}{\emph{Nucl. Phys. B} {\bfseries 503} (Oct, 1997) 220--238}, [\href{https://arxiv.org/abs/hep-th/9608163}{{\ttfamily hep-th/9608163}}].

\bibitem{MurraySinger199903}
M.~K. Murray and M.~A. Singer, \emph{{On the Complete Integrability of the Discrete Nahm Equations}}, \href{https://doi.org/10.1007/s002200050789}{\emph{Commun. Math. Phys.} {\bfseries 210} (Mar, 2000) 497--519}, [\href{https://arxiv.org/abs/math-ph/9903017}{{\ttfamily math-ph/9903017}}].

\bibitem{Atiyah198412}
M.~F. Atiyah, \emph{{Instantons in Two and Four Dimensions}}, \href{https://doi.org/10.1007/BF01212288}{\emph{Commun. Math. Phys.} {\bfseries 93} (Dec, 1984) 437--451}.

\bibitem{Hurtubise198506}
J.~Hurtubise, \emph{{Monopoles and Rational Maps: A Note on a Theorem of Donaldson}}, \href{https://doi.org/10.1007/BF01212447}{\emph{Commun. Math. Phys.} {\bfseries 100} (Jun, 1985) 191--196}.

\bibitem{Hurtubise198912}
J.~Hurtubise, \emph{{The Classification of Monopoles for the Classical Groups}}, \href{https://doi.org/10.1007/bf01260389}{\emph{Comm. Math. Phys.} {\bfseries 120} (Dec, 1989) 613--641}.

\bibitem{Jarvis199807}
S.~Jarvis, \emph{{Euclidian Monopoles and Rational Maps}}, \href{https://doi.org/10.1112/s0024611598000434}{\emph{Proc. Lond. Math. Soc.} {\bfseries 77} (Jul, 1998) 170–192}.

\bibitem{Jarvis200001}
S.~Jarvis, \emph{{A Rational Map of Euclidean Monopoles via Radial Scattering}}, \href{https://doi.org/10.1515/crll.2000.055}{\emph{J. Reine Angew. Math.} {\bfseries 2000} (Jan, 2000) 17--41}.

\bibitem{JarvisNorbury199709}
S.~Jarvis and P.~Norbury, \emph{{Compactification of Hyperbolic Monopoles}}, \href{https://doi.org/10.1088/0951-7715/10/5/005}{\emph{Nonlinearity} {\bfseries 10} (Sep, 1997) 1073–1092}.

\bibitem{Nash2006}
O.~Nash, \emph{{Differential Geometry of Monopole Moduli Spaces}}, Ph.D. thesis, University of Oxford, 2006.
\newblock \href{https://arxiv.org/abs/math/0610295}{{\ttfamily math/0610295}}.

\bibitem{Hitchin2008}
N.~Hitchin, \emph{\href{http://www.numdam.org/item/AST_2008__321__5_0/}{Einstein Metrics and Magnetic Monopoles}},  in \emph{G\'eom\'etrie Diff\'erentielle, Physique Math\'ematique, Math\'ematiques et soci\'et\'e (I) : Volume en l'honneur de Jean Pierre Bourguignon} (H.~Oussama, ed.), no.~321 in Ast\'erisque, pp.~5--29.
\newblock Soci\'et\'e math\'ematique de France, 2008.

\bibitem{BielawskiSchwachhofer201104}
R.~Bielawski and L.~Schwachhofer, \emph{{Pluricomplex Geometry and Hyperbolic Monopoles}}, \href{https://doi.org/10.1007/s00220-013-1761-7}{\emph{Commun. Math. Phys.} {\bfseries 323} (Jul, 2013) 1--34}, [\href{https://arxiv.org/abs/1104.2270}{{\ttfamily 1104.2270}}].

\bibitem{BielawskiSchwachhofer201201}
R.~Bielawski and L.~Schwachhofer, \emph{{Hypercomplex Limits of Pluricomplex Structures and the Euclidean Limit of Hyperbolic Monopoles}},  \href{https://arxiv.org/abs/1201.0781}{{\ttfamily 1201.0781}}.

\bibitem{FigueroaOFarrillGharamti201311}
J.~Figueroa-O'Farrill and M.~Gharamti, \emph{{Supersymmetry of Hyperbolic Monopoles}}, \href{https://doi.org/10.1007/JHEP04(2014)074}{\emph{J. High Energy Phys.} {\bfseries 04} (2014) 074}, [\href{https://arxiv.org/abs/1311.3588}{{\ttfamily 1311.3588}}].

\bibitem{Gharamti2015}
M.~Gharamti, \emph{\href{https://era.ed.ac.uk/handle/1842/10479}{Supersymmetry and Geometry of Hyperbolic Monopoles}}, Ph.D. thesis, Edinburgh U., Edinburgh U., Sch. Math., Jul, 2015.

\bibitem{FranchettiRoss202302}
G.~Franchetti and C.~Ross, \emph{{The Asymptotic Structure of the Centred Hyperbolic 2-Monopole Moduli Space}}, \href{https://doi.org/10.3842/SIGMA.2023.043}{\emph{SIGMA} {\bfseries 19} (Jul, 2023) 043}, [\href{https://arxiv.org/abs/2302.13792}{{\ttfamily 2302.13792}}].

\bibitem{Sutcliffe202112}
P.~Sutcliffe, \emph{{A Hyperbolic Analogue of the Atiyah-Hitchin Manifold}}, \href{https://doi.org/10.1007/JHEP01(2022)090}{\emph{J. High Energy Phys.} {\bfseries 01} (2022) 090}, [\href{https://arxiv.org/abs/2112.02949}{{\ttfamily 2112.02949}}].

\bibitem{MannMilgram199307}
B.~M. Mann and R.~J. Milgram, \emph{{On the Moduli Space of $\SU{n}$ Monopoles and Holomorphic Maps to Flag Manifolds}}, \href{https://doi.org/10.4310/jdg/1214454095}{\emph{J. Diff. Geom.} {\bfseries 38} (Jul, 1993) 39--103}.

\bibitem{Atiyah1979}
M.~F. Atiyah, \emph{\href{https://books.google.ca/books/about/Geometry_of_Yang_Mills_Fields.html?id=VcjvAAAAMAAJ&redir_esc=y}{Geometry of Yang-Mills Fields}}.
\newblock Edizioni della Normale, 1979.

\bibitem{Hurwitz189103}
A.~Hurwitz, \emph{{Ueber Riemann’sche Fl\"achen mit Gegebenen Verzweigungspunkten}}, \href{https://doi.org/10.1007/bf01199469}{\emph{Math. Ann.} {\bfseries 39} (Mar, 1891) 1–60}.

\bibitem{Knapp1996}
A.~W. Knapp, \emph{\href{https://link.springer.com/book/10.1007/978-1-4757-2453-0}{Lie Groups Beyond an Introduction}}, vol.~140 of \emph{Progress in Mathematics}.
\newblock Birkh\"auser Boston, 1996.

\bibitem{GoodmanNolan2009}
R.~Goodman and N.~R. Wallach, \emph{\href{https://link.springer.com/book/10.1007/978-0-387-79852-3}{Symmetry, Representations, and Invariants}}.
\newblock Springer New York, 2009.

\bibitem{Helgason2001}
S.~Helgason, \emph{\href{https://bookstore.ams.org/gsm-34}{Differential Geometry, Lie Groups, and Symmetric Spaces}}.
\newblock Graduate Studies in Mathematics. American Mathematical Society, Jun, 2001.

\bibitem{Cartan191403}
E.~Cartan, \emph{{Les Groupes R\'eels Simples, Finis et Continus}}, \href{https://doi.org/10.24033/asens.676}{\emph{Ann. Sci. \'Ec. Norm. Sup\'er.} {\bfseries 31} (1914) 263–355}.

\bibitem{Cartan192703}
E.~Cartan, \emph{{Sur Certaines Formes Riemanniennes Remarquables des G\'eom\'etries \'a Groupe Fondamental Simple}}, \href{https://doi.org/10.24033/asens.781}{\emph{Ann. Sci. \'Ec. Norm. Sup\'er.} {\bfseries 44} (1927) 345–467}.

\bibitem{Cartan1894}
E.~Cartan, \emph{\href{https://archive.org/details/surlastructured00bourgoog/page/n12/mode/2up}{Sur la Structure des Groupes de Transformations Finis et Continus}}, Ph.D. thesis, Sorbonne Universit\'e, 1894.

\bibitem{Bump2013}
D.~Bump, \emph{\href{https://www.springer.com/gp/book/9781475740943}{Lie Groups}}.
\newblock Springer New York, 2013.

\end{thebibliography}\endgroup
\bibliographystyle{JHEP}

\end{document}